\documentclass[12pt]{article}
\usepackage[utf8]{inputenc}
\usepackage{xcolor}
\usepackage[normalem]{ulem} 
\usepackage{colortbl}
\usepackage{booktabs}
\usepackage{multirow}
\usepackage{threeparttable}
\usepackage{enumitem}
\usepackage{bm}
\usepackage{dsfont}
\usepackage{amsmath}
\usepackage{amsthm}
\usepackage{amssymb}
\usepackage[letterpaper,margin=1in]{geometry}
\usepackage{setspace}
\usepackage[authoryear]{natbib}
\usepackage[bookmarks=true,bookmarksnumbered=false,bookmarksopen=false,
 breaklinks=true,pdfborder={0 0 1},backref=false,colorlinks=true]
 {hyperref}
\hypersetup{
 citecolor=blue,linkcolor=blue,urlcolor=blue}

\makeatletter

\numberwithin{equation}{section}

\usepackage{graphicx}
\usepackage{subcaption}
\usepackage{placeins}
\usepackage{algorithm}
\usepackage{algpseudocode}
 
\usepackage[toc,page]{appendix}
\usepackage{booktabs}
\usepackage[dvipsnames]{xcolor}
\usepackage{bbm}
\usepackage{mathrsfs}
\usepackage{tikz}
\usetikzlibrary{calc}

\tikzset{%
  nodeBase/.style={circle, minimum size=0.92cm, text width=0.92cm, align=center,
                   inner sep=0pt, font=\scriptsize},%
  inTBorder/.style={draw=black, dashed, line width=0.8pt},%
  notInTBorder/.style={draw=black, line width=0.8pt},%
  sPivotColor/.style={draw=red,  fill=red!15,  text=red},%
  sPeriphColor/.style={draw=blue, fill=blue!15, text=blue},%
  otherColor/.style={draw=black, fill=white, text=black},%
  sPivot_inT/.style={nodeBase, inTBorder, sPivotColor},%
  sPeriph_inT/.style={nodeBase, inTBorder, sPeriphColor},%
  sPivot_notInT/.style={nodeBase, notInTBorder, sPivotColor},%
  sPeriph_notInT/.style={nodeBase, notInTBorder, sPeriphColor},%
  other_inT/.style={nodeBase, inTBorder, otherColor},%
  other_notInT/.style={nodeBase, notInTBorder, otherColor},%
  eShared/.style={draw=black, line width=1.3pt},%
  eSonly/.style={draw=blue,  line width=0.6pt},%
  eTonly/.style={draw=red, dashed, line width=0.6pt}%
}%

\usepackage{amsthm}
\usepackage{array}
\usepackage{tabularx}
\usepackage{booktabs}

\newcolumntype{L}[1]{>{\raggedright\arraybackslash}p{#1}}
\newcounter{ovcase}
\newcounter{ovsubcase}[ovcase]

\makeatother

\theoremstyle{plain}
\newtheorem{assumption}{\protect\assumptionname}[section]
\newtheorem{prop}{\protect\propositionname}[section]
\newtheorem{lem}{\protect\lemmaname}[section]
\theoremstyle{remark}

\theoremstyle{plain}
\newtheorem{thm}{\protect\theoremname}[section]
\newtheorem{cor}{\protect\corollaryname}[section]
\providecommand{\assumptionname}{Assumption}
\providecommand{\corollaryname}{Corollary}
\providecommand{\lemmaname}{Lemma}
\providecommand{\propositionname}{Proposition}
\providecommand{\remarkname}{Remark}
\providecommand{\theoremname}{Theorem}

\begin{document}
\hypersetup{linkcolor=black}
\renewcommand{\thefootnote}{\fnsymbol{footnote}}
\begin{center}
{\Large \bfseries Moment Restrictions for Dyadic Network Formation Models with Nontransferable Utility\par}
\vspace{1.2em}
{\large Hanping Chen\textsuperscript{a}\footnote{Email: \href{mailto:hanpingchen@link.cuhk.edu.cn}{\texttt{hanpingchen@link.cuhk.edu.cn}}.}\quad Zeqi Wu\textsuperscript{b}\footnote{Email: \href{mailto:wuzeqi@ruc.edu.cn}{\texttt{wuzeqi@ruc.edu.cn}}.}\par}
\vspace{0.6em}
{\small \textsuperscript{a}\,School of Management and Economics, The Chinese University of Hong Kong, Shenzhen\par}
{\small \textsuperscript{b}\,Institute of Statistics and Big Data, Renmin University of China\par}
\end{center}
\renewcommand{\thefootnote}{\arabic{footnote}}
\setcounter{footnote}{0}
\hypersetup{linkcolor=blue}
\vspace{1em}
\global\long\def\arraystretch{1.4}%
\global\long\def\Pr{\mathbb{P}}%
\global\long\def\1{\mathds{1}}%
\global\long\def\E{\mathbb{E}}%
\global\long\def\sgn{\mathrm{sgn}}%
\global\long\def\bs#1{\boldsymbol{#1}}%
\global\long\def\bf#1{\mathbf{#1}}%
\global\long\def\Cov{\mathrm{Cov}}%
\global\long\def\Var{\mathrm{Var}}%
\global\long\def\tp{\xrightarrow{P}}%
\global\long\def\td{\xrightarrow{D}}%
\global\long\def\piv{\mathrm{piv}}%
\global\long\def\per{\mathrm{per}}%
\global\long\def\supp{\mathrm{supp}}%
\providecommand{\Pbb}{\mathbb{P}}%
\providecommand{\norm}[1]{\left\lVert #1 \right\rVert}%
\providecommand{\abs}[1]{\left\lvert #1 \right\rvert}%
\providecommand{\op}{\mathrm{o}_{p}}%
\providecommand{\plim}{\operatorname*{plim}}%
\providecommand{\iid}{\stackrel{\mathrm{i.i.d.}}{\sim}}%
\providecommand{\La}{\Lambda}%
\providecommand{\argmin}{\operatorname*{argmin}}%
\providecommand{\as}{\text{a.s.}}%
\providecommand{\bbeta}{\widehat{\bs{\beta}}}%
\providecommand{\bbetaone}{\widehat{\bs{\beta}}_{1}}%
\providecommand{\bOmegaUniv}{\widehat{\bf\Omega}_{N}^{\mathrm{univ}}}%
\providecommand{\tr}{\operatorname{tr}}%

\begin{abstract}
This paper investigates the construction of moment restrictions in dyadic
network formation models with unobserved individual heterogeneity under nontransferable utility.
Using observed links and covariates from five-node pentads, we construct moment restrictions
that do not depend on individual fixed effects.
For a broad class of covariate specifications, the construction is minimal
in the sense that it uses the least possible number of nodes and dyads.
Based on these moment restrictions, we propose  the pentad-GMM estimator.
We establish asymptotic normality of the pentad-GMM estimator in dense, sparse,
and ultra-sparse network regimes, with regime-specific convergence rates and asymptotic variances. These results provide a basis for inference across all
three regimes.
To make our method computationally efficient, we develop an algorithm that reduces the
computational cost of the estimator from na\"ive \(O(N^5)\) to \(O(N^3)\).
We apply the proposed method to an academic-discussion network, and find academic
homophily and a positive association between  link formation and a potential partner's openness
to different perspectives.
\end{abstract}



\section{Introduction}

Networks are pervasive in both economic and social contexts. The relationships that
individuals form with each other reflect their characteristics and preferences. Lenders, for
example, decide whom to lend to based on borrowers' credit and repayment histories,
while researchers and firms seek collaborators with complementary expertise
and resources. Understanding how these characteristics and preferences shape
link formation is central to explaining social and economic
interactions.

An important distinction in network formation modeling is whether utility
is transferable between the two agents.
Under transferable utility (TU), utilities or payoffs can
be redistributed through compensating mechanisms, and a link forms when
the joint surplus of the two individuals is nonnegative.
The transferable-utility assumption, however, is restrictive and may not
hold in many applications. In friendship networks, for instance, a link requires
mutual recognition by both individuals, but there is typically no
transfer that can compensate one party for
entering a relationship they would otherwise decline.
This setting is more naturally described by
non-transferable-utility (NTU) models. Under NTU, each individual makes a
unilateral decision about whether to form the relationship, and a link
is established if and only if both agents agree
(\citealp{myerson1991game}, \citealp{jackson1996strategic}).

These linking decisions depend not only on observed characteristics
but also on unobserved traits such as sociability, productivity, and risk tolerance.
A highly extroverted individual, for example,
forms relationships across a broad range of observed types, a pattern that
could otherwise be attributed to weak homophily \citep{graham2017econometric}.
Moreover, ignoring such heterogeneity when it is correlated with observed
characteristics can bias estimates of homophily and other covariate effects.

One way to account for this heterogeneity is to estimate the
parameter of interest and the individual fixed effects by joint maximum likelihood.
Treating individual heterogeneity as fixed effects to be estimated, however, 
leads to the well-known incidental-parameter problem of \citet{neyman1948consistent},
which can induce asymptotic bias in the estimate of the parameter of interest.
Several methods have been proposed to correct incidental-parameter bias in
nonlinear panel models with individual fixed effects
(e.g., \citealp{hahn2004jackknife}; \citealp{dhaene2015split};
\citealp{fernandez2016individual}), and related work extends these methods to
network models with individual fixed effects
(e.g., \citealp{hughes2022estimating,li2024estimation}). 
However, a limitation of these methods is their reliance on
compactness restrictions on the individual fixed effects.
With bounded fixed effects, such restrictions typically exclude 
sparse-network sequences in which link
probabilities vanish as the network grows.
Yet sparsity is a common feature of empirical networks
\citep{graham2024sparse}.
An alternative way is to eliminate and difference out the individual fixed effects.

For nonlinear panel models, approaches to eliminating fixed effects
include methods based on sufficient statistics, functional differencing,
and maximum score restrictions
(e.g., \citealp{chamberlain1980analysis}; \citealp{manski1987semiparametric};
\citealp{bonhomme2012functional}).
Analogous approaches are available for TU network formation
(e.g., \citealp{graham2017econometric}; \citealp{toth2017semiparametric};
\citealp{bonhomme2023functional}).
The joint-surplus structure in TU models preserves additive separability in
the fixed effects, making these differencing methods comparatively
straightforward to apply.
However, these arguments do not carry over directly to NTU models.
The bilateral structure of NTU models poses a different challenge.
Under NTU, researchers usually observe only the realized undirected link,
not the two underlying unilateral decisions. A realized link reveals mutual
consent, whereas an absence of a link does not reveal whether one or both individuals
refuse to link. Moreover, the two individual fixed effects enter the observed link probability
through separate asymmetric latent utilities. As emphasized by
\citet{gao2023logical} and \citet{li2024estimation}, this asymmetry, combined
with partial observability, breaks the additive separability exploited by
standard sufficient-statistic and differencing arguments.

To address these challenges, we propose a pentad differencing
strategy for logistic NTU network formation models with individual
unobserved heterogeneity. The construction is analogous in spirit to
functional differencing in nonlinear panel models
(\citealp{bonhomme2012functional,honore2024moment,bonhomme2025feedback}).
The main idea is to fix two pivotal nodes and form three triangles sharing
the pivotal dyad, using three peripheral nodes and seven dyads in total.
Eliminating each peripheral node's fixed effect gives a homogeneous linear equation in two pivotal nodes'
fixed effects and a constant term. Together, stacking the
three equations implies that a \(3\times3\) coefficient matrix is singular, and thus has zero
determinant. Using this zero-determinant restriction, we construct a moment
function involving only observed links and covariates that has zero conditional
expectation at \(\bs\beta_0\) and forms the basis for GMM estimation.
With the individual fixed effects eliminated, the method does not
require compactness restrictions on their support and accommodates sparse
networks in which link probabilities vanish as the network size grows.
We also show that, for a broad class of covariate specifications, five
nodes and seven dyads are necessary for constructing a nontrivial fixed-effect-free
moment restriction. Among constructions using five nodes and seven dyads,
ours is unique up to relabeling, and the resulting moment restrictions are unique up to rescaling.

We then establish asymptotic normality of the resulting pentad-GMM estimator across
dense, sparse, and ultra-sparse networks.
Conditional on the observed characteristics and unobserved individual
heterogeneity, we apply a Hoeffding decomposition over dyads
(\citealp{hoeffding1948class}) and show that the leading projection varies
with network sparsity.
These three sparsity regimes are characterized by how the average
link probability \(\rho_N\) scales with network size \(N\).
In dense and mildly sparse networks, with
\(N\rho_N\to\infty\), the one-dyad projection dominates. In the
sparse regime, with
\(\rho_N\asymp N^{-1}\), several graph projections contribute
at the same order. In the ultra-sparse regime, with
\(N\rho_N\to0\) and \(N^5\rho_N^4\to\infty\),
the leading variance contribution comes from projections associated
with some five-node, four-edge trees.
The change in the leading projection contrasts with the tetrad-logit estimator of
\citet{graham2017econometric}, whose asymptotic representation is
driven by dyad-level projections under both the dense and sparse
networks considered there.

To accommodate the different leading projections across sparsity regimes, we
propose a tractable covariance estimator. The estimator is
consistent under the appropriate normalization in all three regimes and
requires no prior knowledge of \(\rho_N\) or the sparsity regime.
For implementation, we exploit the determinant structure to reduce the
computational cost of the estimator from na\"ive  \(O(N^5)\)
to \(O(N^3)\) operations.
The empirical application demonstrates the proposed method using data from an
academic-discussion network among Master of Social Work (MSW) students.
The asymmetric NTU estimates reveal academic homophily and a tendency to
form discussion links within the same cohort. They also suggest that students
who are more open to different perspectives are more attractive discussion
partners.

\textit{Related Literature}. This paper contributes to the literature
on econometric models of dyadic link formation in a single large network with
unobserved heterogeneity. A major
strand of this literature studies dyadic link formation models in which
unobserved heterogeneity enters through individual-specific fixed effects.
An incomplete list includes 
\citet{chatterjee2011random}, \citet{yan2013central},
\citet{graham2017econometric},
\citet{charbonneau2017multiple}, \citet{toth2017semiparametric},
\citet{jochmans2018semiparametric}, \citet{dzemski2019empirical},
\citet{gao2020nonparametric}, \citet{gao2023logical},
\citet{hughes2022estimating}, \citet{candelaria2024semiparametric} and
\citet{zeleneev2020identification}.
For comprehensive reviews, see
\citet{de2020econometric} and \citet{graham2020network}.
Much of this work focuses on transferable-utility (TU) network formation
models with individual fixed effects. Under logistic errors, fixed
effects can be eliminated by conditioning on sufficient statistics such
as the degree sequence (e.g., \citealp{graham2017econometric};
\citealp{charbonneau2017multiple}). Semiparametric approaches also
eliminate fixed effects without imposing a parametric distribution on the
idiosyncratic shocks. For example, \citet{toth2017semiparametric} uses rank-based arguments,
while \citet{candelaria2024semiparametric} combines a special-regressor
transformation with arithmetic differencing.
Related work also studies fixed-effect estimation and bias correction in
nonlinear dyadic models, including the jackknife approach developed by
\citet{hughes2022estimating}. More recently, \citet{yan2026penalized} develop a penalized-likelihood
approach that allows the fixed effects to diverge at a logarithmic rate and
accommodates sparse degree sequences. Their primary specification observes
both directed links within each dyad.
Our paper differs from this work in focusing on undirected NTU link formation
with individual fixed effects,
giving rise to partial observability.

The papers closest to our setting are \citet{gao2023logical} and \citet{li2024estimation}.
\citet{gao2023logical} study undirected network formation under
nontransferable utility (NTU) and propose \textit{logical differencing},
which uses logical contraposition to eliminate fixed effects without imposing
parametric assumptions on idiosyncratic shocks (see also
\citealp{gao2020robust}). Their analysis focuses on identification and
consistent estimation rather than large-network inference. Furthermore, their analysis is based 
on dense-network asymptotics, with conditional linking probabilities 
remaining nonvanishing as the network grows.
\citet{li2024estimation} develop
estimation and inference methods for NTU models with individual fixed effects
by combining a joint method-of-moments initial estimator, a Le Cam one-step
refinement, and a split-network jackknife with bagging. Their procedure
corrects incidental-parameter bias and delivers asymptotic normality for the
homophily estimator. Nevertheless, they impose compactness restrictions on the fixed
effects, which excludes sparse networks. We instead construct moment conditions that 
eliminate the fixed effects and develop asymptotic normality for both dense and 
sparse networks without imposing compactness on the fixed effects.

Our work also relates to the literature on sparse network asymptotics
and inference under dyadic dependence. A key insight from this literature
is that the leading variance component can change with network sparsity.
\citet{graham2024sparse} shows that, in logistic regression with dyadic
dependence, variance components that are negligible in dense networks can
become first order in sparse networks. Relatedly,
\citet{chandrasekhar2025network} study subgraph generated models and
develop asymptotic theory for estimators based on higher-order network
structures. This paper shares the emphasis on variance decompositions
generated by overlapping network configurations and extends it to 
regime-adaptive inference for NTU link formation models with fixed effects.

The paper also builds on the panel data literature on eliminating fixed
effects in nonlinear models. In static and dynamic logit panels, fixed
effects can be eliminated by conditioning on sufficient statistics or
state-switching events, as in \citet{chamberlain1980analysis} and
\citet{honore2000panel}. Functional differencing provides a more general
route by constructing moment restrictions that are free of fixed effects
(\citealp{bonhomme2012functional}). Recent contributions extend and
systematize this logic in binary choice and dynamic panel settings,
including \citet{honore2019panel}, \citet{honore2021identification},
\citet{kitazawa2022transformations},
\citet{davezies2023fixed}, \citet{dano2023transition},
\citet{bonhomme2023functional}, \citet{honore2024moment},
\citet{bonhomme2025feedback}, and \citet{dobronyi2021identification}. This paper
follows the same broad principle by constructing moment restrictions that
do not contain the fixed effects.
While our approach builds on the idea of functional differencing, we develop
an explicit and minimal construction for NTU networks, which is, to the best of our knowledge, new to the literature.

Another line of the network formation literature studies 
strategic link interdependence, 
where the payoff from a link may depend on other links in the network. 
This literature often models the observed network as an equilibrium object using pairwise stability or related network-game concepts, 
with pairwise stability tracing back to \citet{jackson1996strategic} 
(see, e.g., \citealp{miyauchi2016structural}, \citealp{mele2017structural}, 
\citealp{de2018identifying}, \citealp{sheng2020structural}, 
and \citealp{menzel2024strategic}). 
Recent work also incorporates unobserved heterogeneity into
strategic network formation models; for example,
\citet{gao2026tractable} study strategic link interdependence 
with individual fixed effects and develop tractable 
subnetwork-based identifying restrictions under TU models. 
The present paper abstracts from strategic link interdependence
and focuses instead on the fixed-effect problem in NTU link formation.

The remainder of the paper is organized as follows.
Sections~\ref{sec:model} and~\ref{sec:identification} present the NTU model,
the pentad moment function, and the GMM estimator. Section~\ref{sec:asymptotics}
develops the regime-specific asymptotic theory, while
Sections~\ref{sec:implementation} and~\ref{sec:empirical} present the
implementation details and empirical application. Section~\ref{sec:conclusion}
concludes. Appendix~\ref{app:graph_notation} introduces graph notation and terminology.
Appendices~\ref{sec:feasible_covariance} and~\ref{sec:simulation}
develop a tractable covariance estimator for inference and report the simulations, respectively.
The remaining appendices contain all proofs.
\section{Network Formation Model}

\label{sec:model} In this paper, the researcher
observes $(\bf L,\bf X)$, where $\bf L$ is the $N\times N$
adjacency matrix with binary $(i,j)$th entry $L_{ij}\in\{0,1\}$
indicating the presence of a link between individuals $i$ and $j$, and
$\bf X:=(\bs X_{1},\ldots,\bs X_{N})^\top$ is the $N\times d$
matrix of observed characteristics, with $\bs X_{i}\in\mathbb{R}^{d}$
for all $i\in[N]:=\{1,\ldots,N\}$. We focus on an undirected
non-transferable utility (NTU) network, so that $L_{ij}=L_{ji}$,
formed among $N$ individuals (nodes) in a single large network. Self-links 
are ruled out by convention (i.e., $L_{ii}=0$). The parameter of interest,
$\bs{\beta}_{0}\in\mathbb{R}^{d}$, captures how observed dyadic covariates
affect each individual's unilateral linking decision.
Let $\bs\Gamma:=(\Gamma_{1},\ldots,\Gamma_{N})^{\top}\in \mathbb{R}^N$
denote the vector of individual-specific unobserved effects, which capture
heterogeneity in individuals' latent propensities to form links. Let the binary random variable $G_{ij}\in\{0,1\}$
indicate whether individual $i$ proposes to form a link with
individual $j$. Individual \(i\)'s unilateral linking decision is specified as:
\[
G_{ij}=\1\left\{ \Gamma_{i}+\bs X_{ij}^\top\bs{\beta}_{0}-\epsilon_{ij}\geq 0\right\} ,\quad1\leq i\neq j\leq N.
\]
In this specification, individual $i$'s latent utility depends on individual-specific
unobserved heterogeneity $\Gamma_{i}$ and observed dyad-level
covariates $\bs X_{ij}=w(\bs X_i,\bs X_j)$, where
$w:\mathbb R^d\times\mathbb R^d\to\mathbb R^d$. The map $w$ need not be
symmetric, so in general $\bs X_{ij}\neq\bs X_{ji}$. These
dyadic covariates capture homophily or complementarity between
individuals $i$ and $j$ and may enter their respective latent utilities
asymmetrically, even though the realized link is undirected.
In addition to the observed and unobserved
covariates, individual $i$'s utility also depends on the idiosyncratic
shock $\epsilon_{ij}$. Under the NTU model, link formation 
requires the consent of both individuals:
\[
L_{ij}=\underbrace{G_{ij}}_{\text{\ensuremath{i} proposes to \ensuremath{j}}}\cdot\underbrace{G_{ji}}_{\text{\ensuremath{j} proposes to \ensuremath{i}}},\quad1\leq i\neq j\leq N.
\]
That is, a symmetric link $L_{ij}=L_{ji}$ is formed between $i$
and $j$ if and only if both individuals propose to link
and each obtains nonnegative utility from the connection. Therefore,
the link indicator $L_{ij}$ is the product of two binary indicators.
Each indicator equals one when that individual’s latent utility is nonnegative:
\begin{equation}
L_{ij}=\1\left\{ \Gamma_{i}+\bs X_{ij}^\top\bs{\beta}_{0}-\epsilon_{ij}\geq0\right\} 
\cdot\1\left\{ \Gamma_{j}+\bs X_{ji}^\top\bs{\beta}_{0}-\epsilon_{ji}\geq0\right\} .\label{mod:mod2}
\end{equation}

The mutual-consent structure in \eqref{mod:mod2} captures the key economic
distinction between NTU and TU. It parallels the bilateral-consent requirement
in the canonical pairwise-stability framework, under which both agents must
favor a new link (\citealp{jackson1996strategic}). This distinction matters when
consent cannot be replaced by compensation. Friendships, peer
ties, and research collaborations often lack a price or complete contract that
can freely redistribute the gains from the relationship, so a large benefit to
one party does not offset a negative payoff to the other. By contrast, models
with transfers allow players to bargain over side payments
(\citealp{bloch2007formation}), and econometric TU models reduce link formation
to the sign of joint surplus (e.g., \citealp{graham2017econometric}). The NTU
formulation therefore better reflects relationships in which consent is
individual and compensation is limited.

We complete the model with the following baseline assumption.
\begin{assumption}[Random Sampling and Conditional Independence]\label{ass:dyad_ind}
  For each \(N\), the node characteristics $\{(\bs X_i,\Gamma_i)\}_{i=1}^N$
  are i.i.d. across \(i\), and their common distribution may depend on \(N\).
  The conditional likelihood of the network $\bf L$, given
  $\bf X$ and $\bs\Gamma$, is
  \[
  \Pr\left(\bf L=\bf l\mid\bf X,\bs\Gamma\right)=\prod_{i<j}\Pr\left(L_{ij}=l_{ij}\mid\bs X_{i},\bs X_{j},\Gamma_{i},\Gamma_{j}\right),
  \]
  with 
  \[
  \begin{aligned}
  \Pr\left(L_{ij}=l_{ij}\mid\bf X,\bs\Gamma\right)
  &=
  \left[
  F\left(\Gamma_{i}+\bs X_{ij}^\top\bs{\beta}_{0}\right)
  F\left(\Gamma_{j}+\bs X_{ji}^\top\bs{\beta}_{0}\right)
  \right]^{l_{ij}}\\
  &\quad\times
  \left[
  1-F\left(\Gamma_{i}+\bs X_{ij}^\top\bs{\beta}_{0}\right)
  F\left(\Gamma_{j}+\bs X_{ji}^\top\bs{\beta}_{0}\right)
  \right]^{1-l_{ij}}
  \end{aligned}
  \]
  for all $i<j$, where $F(\cdot)$ is the standard logistic cumulative distribution function (cdf).
  \end{assumption}
  Assumption~\ref{ass:dyad_ind} imposes the random sampling and conditional
  dyadic independence structure of the model. First, within each \(N\), the
  node characteristics \((\bs X_i,\Gamma_i)\) are sampled i.i.d. across
  individuals. This condition applies to the joint distribution of
  \((\bs X_i,\Gamma_i)\) and does not require \(\bs X_i\) and \(\Gamma_i\)
  to be independent. Second,
  conditional on the full collection of node characteristics, unordered dyads are
  independent. For each \(i<j\), the conditional probability of a realized
  link has the mutual-consent product-logit form
  \[
  P_{ij}(\bs\beta_0)
  := \Pr(L_{ij}=1\mid \bf X, \bs\Gamma)=
  \Pr(L_{ij}=1\mid \bs X_i,\bs X_j,\Gamma_i,\Gamma_j)
  =
  F(\Gamma_i+\bs X_{ij}^{\top}\bs\beta_0)
  F(\Gamma_j+\bs X_{ji}^{\top}\bs\beta_0).
  \]
The product-logit form is consistent with conditionally independent standard
  logistic shocks \(\epsilon_{ij}\) and \(\epsilon_{ji}\) for the two directed
  decisions within each dyad, while the likelihood factorization separately
  imposes conditional independence across unordered dyads.
  This sampling and conditional link-independence structure parallels
  Assumptions 1 and 3 in  \citet{graham2017econometric} and
  \citet{li2024estimation}.

\section{Moment Restrictions and Estimation}

\label{sec:identification}

\subsection{Differencing via Pentads}

The mutual-consent rule in \eqref{mod:mod2} captures the bilateral nature of
link formation across a broad range of networks, as discussed in
Section~\ref{sec:model}. However, this structure makes the individual fixed
effects $\bs\Gamma$ difficult to difference out. As implied by the model in \eqref{mod:mod2},
both agents' decisions are revealed only when $L_{ij}=1$, in which case
$\left(G_{ij},G_{ji}\right)=(1,1)$. An absent link ($L_{ij}=0$), however,
can arise from three possible outcomes:
$\left(G_{ij},G_{ji}\right)\in\{(0,0),(1,0),(0,1)\}$.
The researcher therefore cannot tell whether $i$, $j$, or both declined to
form the link. If the directed proposal decisions $G_{ij}$ were observed, the
individual fixed effects could be eliminated using arguments analogous to
those used in TU models
(\citealp{graham2017econometric}).
Moreover, under Assumption~\ref{ass:dyad_ind}, the fixed effects
\(\Gamma_i\) and \(\Gamma_j\) do not enter the conditional link probability
\(P_{ij}(\bs\beta_0)\) through a single additive index
\citep{gao2023logical,li2024estimation}.
This further complicates the direct application of arithmetic-differencing arguments
developed for TU models (e.g., \citealp{toth2017semiparametric};
\citealp{candelaria2024semiparametric}).

We therefore develop a new differencing argument for the NTU model. Under
Assumption~\ref{ass:dyad_ind}, we construct a nontrivial moment function\footnote{A
moment function \(\psi(\bf L,\bf X;\bs{\beta})\) is trivial if
\(\psi(\cdot,\bf X;\bs\beta_0)\equiv0\) almost surely.}
$\psi(\bf L,\bf X;\bs{\beta})$ of the observed links and covariates that
satisfies
\begin{equation}
\E\big[\psi(\bf L,\bf X;\bs{\beta}_{0})\,\big|\,\bf X,\bs\Gamma\big]=0.
\label{eq:cond_moment_A-1}
\end{equation}
By the law of iterated expectations, \eqref{eq:cond_moment_A-1} implies
\(\E[\psi(\bf L,\bf X;\bs{\beta}_{0})\mid\bf X]=0\).
The resulting restriction, which conditions only on $\bf X$, forms the basis
for GMM estimation without requiring an estimator of $\bs\Gamma$. The
conditional moment in \eqref{eq:cond_moment_A-1} is in the spirit of
functional differencing (\citealp{bonhomme2012functional}), but the
construction developed here for the NTU model is new. It uses a five-node,
seven-dyad configuration with two pivot nodes, \(i,j\), and three peripheral
nodes, \(k,l,m\), as illustrated in
Figure~\ref{fig:five-node-triangles}.

\begin{figure}[htbp!]
\centering
\begin{tikzpicture}[scale=0.4]
        \def\R{2.5}

        \node[circle, draw=blue, fill=blue!15, text=blue, minimum size=0.8cm] (l) at (90:\R) {$l$};
        \node[circle, draw=blue, fill=blue!15, text=blue, minimum size=0.8cm] (m) at (162:\R) {$m$};
        \node[circle, draw=red,  fill=red!15,  text=red,  minimum size=0.8cm] (i) at (234:\R) {$i$};
        \node[circle, draw=red,  fill=red!15,  text=red,  minimum size=0.8cm] (j) at (306:\R) {$j$};
        \node[circle, draw=blue, fill=blue!15, text=blue, minimum size=0.8cm] (k) at (18:\R) {$k$};

        \draw (i) -- (k);
        \draw (i) -- (l);
        \draw (i) -- (m);
        \draw (i) -- (j);
        \draw (j) -- (k);
        \draw (j) -- (l);
        \draw (j) -- (m);
\end{tikzpicture}
\caption{Five-node pivot-peripheral configuration.}
\label{fig:five-node-triangles}
\end{figure}
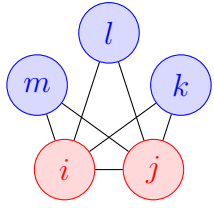

We first derive a linear relation from one pivot--peripheral triangle and
then combine the three relations generated by \(k,l,m\). Throughout this
construction, probabilities are conditional on \((\bf X,\bs\Gamma)\) and
evaluated at \(\bs{\beta}_{0}\). Write \(\alpha_u:=e^{-\Gamma_u}\) for a
node \(u\), and write
\(W_{uv}(\bs{\beta}_{0}):=\exp(-\bs X_{uv}^\top\bs{\beta}_{0})\). For any dyad
\((u,v)\) among these five nodes, write
\(P_{uv}:=P_{uv}(\bs{\beta}_{0})\) for the conditional link probability
defined in Section~\ref{sec:model}.
Under Assumption~\ref{ass:dyad_ind}, the product-logit probability satisfies
\[
\begin{aligned}
P_{uv}
&=\frac{1}
{(1+W_{uv}(\bs{\beta}_{0})\alpha_u)
(1+W_{vu}(\bs{\beta}_{0})\alpha_v)},\\
\frac{1-P_{uv}}{P_{uv}} &=W_{uv}(\bs{\beta}_{0})W_{vu}(\bs{\beta}_{0})\alpha_u\alpha_v
+W_{uv}(\bs{\beta}_{0})\alpha_u
+W_{vu}(\bs{\beta}_{0})\alpha_v.
\end{aligned}
\]
For the triangle \(\{i,j,k\}\), multiplying \((1-P_{uv})/P_{uv}\) by
\(P_{uv}\) for \((u,v)\in\{(i,k),(j,k),(i,j)\}\) gives
\begin{equation}
\label{eq:triangle-p-identities}
\begin{aligned}
1-P_{ik}
&=P_{ik}W_{ik}(\bs{\beta}_{0})W_{ki}(\bs{\beta}_{0})\alpha_i\alpha_k
+P_{ik}W_{ik}(\bs{\beta}_{0})\alpha_i
+P_{ik}W_{ki}(\bs{\beta}_{0})\alpha_k,\\
1-P_{jk}
&=P_{jk}W_{jk}(\bs{\beta}_{0})W_{kj}(\bs{\beta}_{0})\alpha_j\alpha_k
+P_{jk}W_{jk}(\bs{\beta}_{0})\alpha_j
+P_{jk}W_{kj}(\bs{\beta}_{0})\alpha_k,\\
1-P_{ij}
&=P_{ij}W_{ij}(\bs{\beta}_{0})W_{ji}(\bs{\beta}_{0})\alpha_i\alpha_j
+P_{ij}W_{ij}(\bs{\beta}_{0})\alpha_i
+P_{ij}W_{ji}(\bs{\beta}_{0})\alpha_j.
\end{aligned}
\end{equation}
Solving the identities for dyads \((i,k)\) and \((j,k)\) in
\eqref{eq:triangle-p-identities} for the peripheral fixed effect
\(\alpha_k\) and equating the two expressions eliminates \(\alpha_k\).
Substituting the identity for dyad \((i,j)\) then eliminates the product
\(\alpha_i\alpha_j\). The resulting equation is linear in
\(P_{ij}\alpha_i\), \(P_{ij}\alpha_j\), and a constant term.
The following lemma records the resulting relation for any
pivot--peripheral triangle.
\begin{lem}
\label{lem:lem1}
Suppose Assumption~\ref{ass:dyad_ind} holds. For any distinct nodes
\(i,j,r\), there exist coefficients \(b_{ijr}(\bs{\beta}_{0})\),
\(c_{ijr}(\bs{\beta}_{0})\), and \(d_{ijr}(\bs{\beta}_{0})\), given
explicitly in \eqref{eq:triangle-coefficients-app}, such that
\(
b_{ijr}(\bs{\beta}_{0})P_{ij}\alpha_i
+c_{ijr}(\bs{\beta}_{0})P_{ij}\alpha_j
+d_{ijr}(\bs{\beta}_{0})=0.
\)
\end{lem}

Applying Lemma~\ref{lem:lem1} to \(r=k,l,m\) gives three linear
relations in \(P_{ij}\alpha_i\), \(P_{ij}\alpha_j\), and a constant term.
Collecting the three relations gives
\[
\mathbf M_{ijklm}(\bs{\beta}_{0})
\begin{bmatrix}
P_{ij}\alpha_i\\
P_{ij}\alpha_j\\
1
\end{bmatrix}
=0,
\qquad
\mathbf M_{ijklm}(\bs{\beta}_{0}):=
\begin{bmatrix}
b_{ijk}(\bs{\beta}_{0}) & c_{ijk}(\bs{\beta}_{0}) & d_{ijk}(\bs{\beta}_{0})\\
b_{ijl}(\bs{\beta}_{0}) & c_{ijl}(\bs{\beta}_{0}) & d_{ijl}(\bs{\beta}_{0})\\
b_{ijm}(\bs{\beta}_{0}) & c_{ijm}(\bs{\beta}_{0}) & d_{ijm}(\bs{\beta}_{0})
\end{bmatrix}.
\]
Since \([P_{ij}\alpha_i,P_{ij}\alpha_j,1]^\top\) is nonzero,
\(\mathbf M_{ijklm}(\bs{\beta}_{0})\) is singular, so
\(\det(\mathbf M_{ijklm}(\bs{\beta}_{0}))=0\).

The equality \(\det(\mathbf M_{ijklm}(\bs{\beta}_{0}))=0\) is not yet a
feasible moment condition, since the entries of
\(\mathbf M_{ijklm}(\bs{\beta}_{0})\) still involve the unobserved
conditional link probabilities \(P_{uv}\).
For a generic parameter value \(\bs\beta\), define the feasible matrix
\(\widetilde{\mathbf M}_{ijklm}(\bs\beta)\) by the explicit entries in
\eqref{eq:feasible-coefficients-app} and \eqref{eq:feasible-matrix-app}.
At \(\bs{\beta}_{0}\), this construction replaces \(P_{uv}\) and
\(1-P_{uv}\) in \(\mathbf M_{ijklm}(\bs{\beta}_{0})\) with
\(L_{uv}\) and \(1-L_{uv}\), respectively.
To verify that this substitution preserves the conditional mean of the
determinant, Lemma~\ref{lem:det-row-supports} in
Appendix~\ref{app:detconstruction} shows that each dyad enters every
monomial in the determinant expansion at most once, so the expansion is
multilinear in the dyad probabilities.
Assumption~\ref{ass:dyad_ind} gives
\(\E[L_{uv}\mid\bf X,\bs\Gamma]=P_{uv}\) and, for distinct dyads \((u,v)\)
and \((a,b)\),
\(\E[L_{uv}L_{ab}\mid\bf X,\bs\Gamma]=P_{uv}P_{ab}\). More generally, the
same factorization holds for finite products over distinct dyads; the
corresponding identities for \(1-L_{uv}\) follow by linearity.
Consequently,
\[
\E\!\left[\det(\widetilde{\mathbf M}_{ijklm}(\bs{\beta}_{0}))
\mid \bf X,\bs\Gamma\right]
=\det(\mathbf M_{ijklm}(\bs{\beta}_{0}))
=0.
\]
The following theorem summarizes the resulting fixed-effect-free moment
restriction.
\begin{thm}[Fixed-Effect-Free Moment Restriction]\label{thm:moment_restriction}
For distinct nodes \(i,j,k,l,m\), define the pentad moment function 
\(\psi_{ijklm}(\bs\beta):=\det(\widetilde{\mathbf M}_{ijklm}(\bs\beta))\).
Under Assumption~\ref{ass:dyad_ind}, the following conditional moment
restrictions hold at the true value \(\bs{\beta}_0\):
\begin{equation}
\label{eq:moment_cond-free}
\begin{aligned}
\E\Bigl[\psi_{ijklm}(\bs{\beta}_{0})
\,\big|\,\bf X,\bs\Gamma\Bigr]&=0,\quad 
\E\Bigl[\psi_{ijklm}(\bs{\beta}_{0})
\,\big|\,\bf X\Bigr]&=0.
\end{aligned}
\end{equation}
\end{thm}

Theorem~\ref{thm:moment_restriction} gives an exact fixed-effect-free moment
restriction for every five-node configuration. The conditional restriction
in \eqref{eq:moment_cond-free} accommodates unrestricted support for
\(\bs\Gamma\) and arbitrary dependence between \(\bf X\) and
\(\bs\Gamma\). The pentad construction thus provides a new algebraic method
for obtaining fixed-effect-free moment restrictions in NTU network formation
models.

The key idea in our construction is to difference out fixed effects by
exploiting the singularity of a linear-system coefficient matrix. This idea
extends beyond the NTU model. As a simple illustration, consider the TU logit
model in \citet{graham2017econometric}, written in our notation as
\begin{equation}
\label{eq:tu-model}
L_{uv}
=\1\left\{\Gamma_u+\Gamma_v+\bs X_{uv}^{\top}\bs{\beta}_{0}
-\epsilon_{uv}\geq0\right\},
\qquad 1\leq u<v\leq N.
\end{equation}
Here \(\bs X_{uv}=\bs X_{vu}\) is a symmetric vector of observed dyadic
covariates and \(\epsilon_{uv}=\epsilon_{vu}\). Conditional on
\((\bf X,\bs\Gamma)\), the shocks \(\{\epsilon_{uv}:u<v\}\) are independent
and have the standard logistic cdf \(F\). The conditional link probability is
therefore
\(
P_{uv}:=\Pr(L_{uv}=1\mid\bf X,\bs\Gamma)
=F\!\left(\Gamma_u+\Gamma_v
+\bs X_{uv}^{\top}\bs{\beta}_{0}\right)
\).
Set \(\alpha_u:=e^{-\Gamma_u}\) and
\(W_{uv}(\bs{\beta}_{0}):=\exp(-\bs X_{uv}^\top\bs{\beta}_{0})\). Then
\(
(1-P_{uv})/P_{uv}
=W_{uv}(\bs{\beta}_{0})\alpha_u\alpha_v.
\)
Fixing the pivotal pair \((i,j)\), each peripheral node eliminates its own
fixed effect and yields a homogeneous linear equation in
\((\alpha_i,\alpha_j)^\top\). For example, the dyads \((i,k)\) and \((j,k)\) imply
\(\alpha_k=(1-P_{ik})/[P_{ik}W_{ik}(\bs{\beta}_{0})\alpha_i]\) and
\(\alpha_k=(1-P_{jk})/[P_{jk}W_{jk}(\bs{\beta}_{0})\alpha_j]\), respectively;
equating these expressions eliminates \(\alpha_k\) and produces the first row
of the system below. The two peripheral nodes \(k\) and \(l\) therefore give
\begin{equation}
\label{eq:tu-homogeneous-system}
\underbrace{\begin{bmatrix}
P_{ik}(1-P_{jk})W_{ik}(\bs{\beta}_{0}) &
-P_{jk}(1-P_{ik})W_{jk}(\bs{\beta}_{0})\\
P_{il}(1-P_{jl})W_{il}(\bs{\beta}_{0}) &
-P_{jl}(1-P_{il})W_{jl}(\bs{\beta}_{0})
\end{bmatrix}}_{\mathbf M_{ijkl}}
\underbrace{\begin{bmatrix}
\alpha_i\\[4pt]
\alpha_j
\end{bmatrix}}_{\bs\alpha}
=
\begin{bmatrix}
0\\[4pt]
0
\end{bmatrix}.
\end{equation}
Since \(\bs\alpha\) is nonzero, \(\mathbf M_{ijkl}\) is singular, implying
\[
P_{ik}P_{jl}(1-P_{jk})(1-P_{il})
W_{ik}(\bs{\beta}_{0})W_{jl}(\bs{\beta}_{0})
-P_{jk}P_{il}(1-P_{ik})(1-P_{jl})
W_{jk}(\bs{\beta}_{0})W_{il}(\bs{\beta}_{0})=0.
\]
Multiplying \(\det(\mathbf M_{ijkl})=0\) by
\(\exp\{(\bs X_{ik}+\bs X_{jk})^\top\bs{\beta}_{0}\}\) and using the
conditional independence of the four dyads gives
\(
\E[\psi_{ijkl}(\bf L,\bf X;\bs{\beta}_{0})
\mid\bf X,\bs\Gamma]=0,
\)
where 
\[
\begin{aligned}
\psi_{ijkl}(\bf L,\bf X;\bs\beta)
&=\1\{L_{ik}=1,L_{il}=0,L_{jk}=0,L_{jl}=1\}
\exp\big((\bs X_{jk}-\bs X_{jl})^\top\bs\beta\big)\\
&\quad-\1\{L_{ik}=0,L_{il}=1,L_{jk}=1,L_{jl}=0\}
\exp\big((\bs X_{ik}-\bs X_{il})^\top\bs\beta\big).
\end{aligned}
\]
This construction uses the same four-node tetrad as
\citet{graham2017econometric}, but it produces a different moment function.

The TU tetrad also clarifies why our NTU construction requires one more node.
Both constructions fix the pivotal pair \((i,j)\), with each peripheral
node supplying one equation after its fixed effect is eliminated.
Under TU, these equations are homogeneous in
\((\alpha_i,\alpha_j)^\top\), so two peripheral nodes give the
\(2\times2\) system and tetrad restriction in
\eqref{eq:tu-homogeneous-system}. Under NTU, eliminating the peripheral
fixed effect and then using the pivotal-dyad identity in
\eqref{eq:triangle-p-identities} to remove \(\alpha_i\alpha_j\) leaves
an equation in \(P_{ij}\alpha_i\) and \(P_{ij}\alpha_j\) with a constant
term; see Lemma~\ref{lem:lem1}. Two such equations generally determine
the two nuisance terms. Substituting these values into the equation from
a third peripheral node yields a restriction that no longer contains
either nuisance term.
Thus, augmenting the two nuisance terms with the constant coordinate \(1\)
converts the three affine equations into a \(3\times3\) homogeneous system
in \((P_{ij}\alpha_i,P_{ij}\alpha_j,1)^\top\), giving the
five-node determinant restriction in
Theorem~\ref{thm:moment_restriction}.

We next show that our pentad construction is \emph{minimal} in the sense that it uses
the minimum possible numbers of dyads and nodes for a nontrivial
fixed-effect-free moment and, among five-node, seven-dyad constructions, is
unique up to relabeling and scale. To state this result, we first introduce
notation for moments based on an arbitrary finite graph.
For a finite simple graph \(G=(V,E)\) with node set \(V\subseteq[N]\) and
dyad set \(E\), write
\(\mathbf L_G:=(L_{uv})_{(u,v)\in E}\) and
\(\mathbf X_G:=(\bs X_u)_{u\in V}\). For each such graph \(G\), consider a real-valued
measurable function
\(\psi_G(\mathbf L_G,\mathbf X_G;\bs\beta)\) such that
\begin{equation}
\E\!\left[\psi_G(\mathbf L_G,\mathbf X_G;\bs\beta_0)
\mid\bf X,\bs\Gamma\right]=0\quad\text{almost surely.}
\label{eq:graph-local-moment}
\end{equation}
In this notation, the pentad moment in
Theorem~\ref{thm:moment_restriction} corresponds to
\(V=\{i,j,k,l,m\}\),
\(E=\{(i,j),(i,k),(j,k),(i,l),(j,l),(i,m),(j,m)\}\).

\begin{thm}[Minimality of the Pentad Moment]
\label{thm:pentad_minimality}
Recall that \(\bs X_{ij}=w(\bs X_i,\bs X_j)\), and suppose that
\(w(\bs x,\bs y)
=\bigl(w_1(x_1,y_1),\allowbreak\ldots,\allowbreak
w_d(x_d,y_d)\bigr)^\top\).
Each \(w_r:\mathbb R^2\to\mathbb R\) satisfies one of the following:
\textnormal{(i)} \(w_r\) is real analytic on
\(\mathbb R^2\), with \(\frac{\partial w_r}{\partial y_r}(x_r^\circ,x_r^\circ)\neq0\) for some
\(x_r^\circ\in\mathbb R\); or \textnormal{(ii)} there exist
\(\varepsilon_r>0\), \(\kappa_r\in\{1,2\}\), and real-analytic functions
\(f_{r,+},f_{r,-}:(-\varepsilon_r,\infty)\to\mathbb R\) such that
\(w_r(s,t)=f_{r,+}(t-s)\) for \(s\leq t\),
\(w_r(s,t)=f_{r,-}(s-t)\) for \(s>t\),
\(f_{r,+}(0)=f_{r,-}(0)\),
\(f_{r,+}^{(\kappa_r)}(0)f_{r,-}^{(\kappa_r)}(0)\neq0\),
and, if \(\kappa_r=2\), \(f'_{r,+}(0)=f'_{r,-}(0)=0\).
For any \(\bs\beta_0\in\mathbb R^d\setminus\{\mathbf 0\}\) and any
data-generating process satisfying Assumption~\ref{ass:dyad_ind} with an
absolutely continuous distribution of \((\bs X_i,\Gamma_i)\), the
following statements hold for every finite simple graph \(G=(V,E)\) and every measurable \(\psi_G\) satisfying
\eqref{eq:graph-local-moment}:
\begin{enumerate}[label=\textnormal{(\alph*)},leftmargin=2em]
\item If \(\lvert E\rvert<7\) or \(\lvert V\rvert<5\), then
\(\psi_G\) is trivial, i.e., \(\psi_G(\cdot,\mathbf X_G;\bs\beta_0)\equiv0\) almost surely.
\item Suppose \(\lvert V\rvert=5\) and \(\lvert E\rvert=7\). Then
\(\psi_G\) is nontrivial only if there exist distinct
\(i,j,k,l,m\in V\) such that
\(E=\{(i,j),(i,k),(j,k),\allowbreak(i,l),(j,l),\allowbreak
(i,m),(j,m)\}\). For each such ordered tuple \((i,j,k,l,m)\),
\(\psi_G(\mathbf L_G,\mathbf X_G;\bs\beta_0)
\allowbreak=C(\mathbf X_G)\psi_{ijklm}(\bs\beta_0)\) almost surely for
some measurable function \(C\), where
\(\psi_{ijklm}(\bs\beta_0)\) denotes the pentad moment function defined in
Theorem~\ref{thm:moment_restriction}; moreover,
\(\psi_{ijklm}(\bs\beta_0)\) is nontrivial.
\end{enumerate}
\end{thm}

Theorem~\ref{thm:pentad_minimality} establishes that our pentad construction
is minimal in both dyads and nodes for a broad class of covariate
specifications. The restrictions on \(w\) are convenient sufficient
conditions for minimality, but are not necessary. Their role is to exclude
special algebraic degeneracies among the dyadic covariates \(\bs X_{ij}\).
We also exclude the degenerate case \(\bs\beta_0=\mathbf0\), in which all
covariate indices \(\bs X_{ij}^{\top}\bs\beta_0\) vanish.
The conditions on \(w\) cover symmetric maps such as
\(w_r(x_r,y_r)=\lvert x_r-y_r\rvert\) and
\(w_r(x_r,y_r)=\lvert x_r-y_r\rvert^2\), as well as asymmetric maps such as
\(w_r(x_r,y_r)=y_r\) and
\(w_r(x_r,y_r)=2\lvert x_r-y_r\rvert\1\{x_r>y_r\}
\allowbreak+\lvert x_r-y_r\rvert\1\{x_r\leq y_r\}\).
The minimality conclusion may also hold for other choices of \(w\),
but the proof of Theorem~\ref{thm:pentad_minimality} needs to be checked
for each such choice.

\subsection{Pentad-GMM Estimator}\label{sec:pentad_gmm}

We estimate \(\bs{\beta}_{0}\) by GMM using the fixed-effect-free pentad
moment restrictions in \eqref{eq:moment_cond-free}. The sample
moments used below keep the five roles \((i,j,k,l,m)\) ordered. Define
\(\mathcal S_N:=\{(i,j,k,l,m)\in[N]^5:i,j,k,l,m\text{ are all distinct}\}\);
then \(\abs{\mathcal S_N}=N(N-1)(N-2)(N-3)(N-4)\).
For \(s=(i,j,k,l,m)\in\mathcal S_N\), let
\(\psi_s(\bs\beta):=\psi_{ijklm}(\bs\beta)\). To obtain enough moment
conditions for a multidimensional parameter, let $q\ge d$ be fixed and let
\(\bs Z_s:=z(\bf X_s)\in\mathbb{R}^{q}\) be an  instrument
vector constructed from the observed
covariates in configuration $s=(i,j,k,l,m)$, where $\bf X_s$ collects
the node covariates in the order \((i,j,k,l,m)\).
The instrument may depend on the ordered roles, including the
order of the peripheral labels \(k,l,m\). Define the pentad moment function
\(\bs\Psi_s(\bs{\beta}):= \bs Z_s\psi_s(\bs{\beta})\) and the
population moment \(\bs h(\bs{\beta}):=\E[\bs\Psi_s(\bs{\beta})]\). The
fixed-effect-free conditional moment restriction \eqref{eq:moment_cond-free}
and the law of iterated expectation imply
\(
\bs h(\bs{\beta}_{0})
=\E\!\left[\bs Z_s\,\psi_s(\bs{\beta}_{0})\right]
=0.
\)
The sample moment vector is then given by
\begin{equation}
\bs h_{N}(\bs{\beta})\;:=\;\frac{1}{\abs{\mathcal S_N}}\sum_{s\in\mathcal S_N}\bs\Psi_s(\bs{\beta})\;\in\;\mathbb{R}^{q}.\label{eq:sample_moment-multi}
\end{equation}
The pentad-GMM estimator for \(\bs\beta_0\) is then defined as the minimizer:
\begin{equation}
\widehat{\bs{\beta}}_{\mathrm{GMM}}\in
\arg\min_{\bs{\beta}\in\mathcal B}
\bs h_N(\bs\beta)^\top\widehat{\bf W}_N\bs h_N(\bs\beta),
\label{eq:beta_gmm}
\end{equation}
where \(\widehat{\bf W}_{N}\in\mathbb R^{q\times q}\) is a feasible
symmetric positive definite weight matrix, and the parameter space
\(\mathcal B\subset\mathbb R^d\) is compact, excludes the origin
\(\mathbf 0\notin\mathcal B\), and satisfies
\(\bs{\beta}_{0}\in\operatorname{int}(\mathcal B)\),
where \(\operatorname{int}(\mathcal B)\) denotes the interior of \(\mathcal B\). We exclude the origin
because \(\psi_s(\mathbf 0)=0\) for every \(s\in\mathcal S_N\) and every
realization of the network. Hence,
\(\bs h_N(\mathbf 0)=\bs h(\mathbf 0)=0\), making \(\mathbf 0\) a degenerate
root of both the sample and population moment equations. The GMM objective in
\eqref{eq:beta_gmm} therefore always equals zero at the origin.\footnote{The
need to exclude the origin arises from the pentad-GMM criterion rather than
from the underlying network model.
When \(\bs\beta_0=\mathbf 0\), Assumption~\ref{ass:dyad_ind} implies that the
conditional link probability factors as
\(\Pr(L_{ij}=1\mid\bs X_i,\bs X_j)
=\E[F(\Gamma_i)\mid\bs X_i]\E[F(\Gamma_j)\mid\bs X_j]\).
This factorization provides a testable implication for examining the zero
coefficient vector separately. The distinction between zero and nonzero
coefficient vectors also arises in \citet{davezies2023fixed}, who formally
characterize the zero vector in binary choice panels through a different
conditional probability restriction.}

\section{Asymptotic Analysis}
\label{sec:asymptotics}

In this section, we derive asymptotic properties for the sample moment $\bs{h}_N(\bs \beta_0)$ and for the pentad-GMM estimator
\(\widehat{\bs{\beta}}_{\mathrm{GMM}}\) based on
\eqref{eq:sample_moment-multi}. We begin with basic regularity conditions for
the asymptotic analysis. The first imposes boundedness on the dyadic
design and instruments. Similar restrictions on covariates and
instruments are common in the network-formation literature;
see, for example, \citet{graham2017econometric},
\citet{gao2023logical}, and \citet{li2024estimation}.

\begin{assumption}[Bounded dyadic design and instruments]
\label{ass:bounded_design}
The parameter dimension \(d\) and instrument dimension \(q\) are fixed.
The induced dyadic covariates and instruments satisfy
\(\sup_{i\neq j}\norm{\bs X_{ij}}\le C_X\) and
\(\sup_{s\in\mathcal S_N}\norm{\bs Z_s}\le C_Z\) a.s. for
deterministic constants \(C_X,C_Z<\infty\).
\end{assumption}

Throughout this paper, let $\mathcal{E}_{N}:=\{(u,v):1\le u<v\le N\}$ denote the set of all unordered dyads in the network.
For an unordered dyad $e=(a,b)\in\mathcal{E}_{N}$, recall
that $P_{e}(\bs{\beta}_{0}):=\E[L_{e}\mid\bf X,\bs\Gamma]$
is the conditional link probability.

\begin{assumption}[Average link probability and relative expected degree]
\label{ass:sparse_envelope}
Let \(\rho_N:=\E[P_{12}(\bs\beta_0)]\). For each \(i\in[N]\), define
individual \(i\)'s relative expected degree by
\(\Lambda_{i,N}:=\rho_N^{-1}\E[P_{ij}(\bs\beta_0)\mid \bs X_i,\Gamma_i]\),
where \(j\ne i\) is arbitrary.
Assume that \(\sup_N\E[\Lambda_{i,N}^{16}]<\infty\).
\end{assumption}

Assumption~\ref{ass:sparse_envelope} defines \(\rho_N\) as the average
link probability and \(\Lambda_{i,N}\) as individual \(i\)'s relative
expected degree, namely, the expected degree conditional on
\(\bs X_i,\Gamma_i\) divided by the average expected degree: indeed,
\(\E[\Lambda_{i,N}]=1\) and
\(\E[\sum_{j\ne i}L_{ij}\mid \bs X_i,\Gamma_i]
=(N-1)\rho_N\Lambda_{i,N}\). The moment restriction
\(\sup_N\E[\Lambda_{i,N}^{16}]<\infty\) permits some individuals to have
expected degrees far above the average, provided that individuals with such
large relative expected degrees are sufficiently rare. This
restriction is the single-population counterpart of the third-moment
conditions used by \citet[Assumption~2(iii)]{graham2024sparse} to rule
out extreme skewness in the consumer and product degree distributions.
Although the moment restriction \(\sup_N\E[\Lambda_{i,N}^{16}]<\infty\)
restricts the tail behavior of the relative expected degree, it does not
restrict the asymptotic scale of the average expected degree
\((N-1)\rho_N\).

Since \(\rho_N\) is the average link probability, its
asymptotic scale determines network sparsity and the growth of the average
expected degree: a nonvanishing \(\rho_N\) yields a dense network sequence,
whereas \(\rho_N\to0\) yields a sparse sequence whose degree of sparsity
depends on how quickly \(\rho_N\) vanishes (see
\citealp{graham2020network,graham2024sparse}). Our analysis covers three sparsity
regimes: dense or mildly sparse networks with
\(N\rho_N\to\infty\); sparse networks with finite average expected degree,
\(N\rho_N\asymp1\); and ultra-sparse networks with
\(N\rho_N\to0\) and \(N^5\rho_N^4\to\infty\).

We organize the asymptotic analysis through the projection structure of
the sample moment. Conditional on \(\bf X\) and \(\bs\Gamma\), applying the
Hoeffding decomposition over dyads (obtained by Möbius inversion on the
lattice of dyad subsets; see \citealp{rota1964foundations}) to
\(\bs h_N(\bs{\beta}_{0})\) gives
\begin{equation}
\label{eq:main_projection_decomp}
\bs h_N(\bs{\beta}_{0})
=\sum_{A\subseteq\mathcal E_N:1\le\abs{A}\le7}\bs\phi_{A,N},
\quad
\bs\phi_{A,N}:=\sum_{B\subseteq A}(-1)^{\abs{A}-\abs{B}}
\E[\bs h_N(\bs{\beta}_{0})\mid \bf X,\bs\Gamma,\{L_e\}_{e\in B}].
\end{equation}
Here \(A\) is a labeled dyad set. For
\(s=(i,j,k,l,m)\in\mathcal S_N\), write
\[
E(s):=\{(i,j),(i,k),(i,l),(i,m),(j,k),(j,l),(j,m)\}
\]
for the seven dyads entering the five-node configuration. The five-node
kernel \(\bs\Psi_s(\bs{\beta}_{0})\) admits the finite polynomial representation
\begin{equation}
\label{eq:kernel_poly_rep}
\bs\Psi_s(\bs{\beta}_{0})
=\sum_{\nu=1}^{\nu_{\max}}\bs c_{s,\nu}(\bf X)m_{s,\nu}({\bf L}),
\end{equation}
Here the constant \(\nu_{\max}<\infty\) does not depend on \(N\),
\(\bs c_{s,\nu}(\bf X)\in\mathbb R^q\) is a coefficient vector, and
\(m_{s,\nu}({\bf L})=\prod_{e\in E_{s,\nu}}L_e\) is a square-free monomial
with support \(E_{s,\nu}\subseteq E(s)\); see Lemma~\ref{lem:poly}. Applying the
projection formula in
\eqref{eq:main_projection_decomp} term by term to
\eqref{eq:kernel_poly_rep} gives
\[
\bs\phi_{A,N}
=\abs{\mathcal S_N}^{-1}\bs C_{A,N}(\bf X,\bs\Gamma)\prod_{e\in A}\xi_e,
\]
where  $\xi_{e}:=L_{e}-P_{e}(\bs{\beta}_{0})$
is the centered link indicator and 
\[
\bs C_{A,N}(\bf X,\bs\Gamma)
:=
\sum_{s\in\mathcal S_N:\,A\subseteq E(s)}
\sum_{\nu:A\subseteq E_{s,\nu}}
\bs c_{s,\nu}(\bf X)
\prod_{f\in E_{s,\nu}\setminus A}P_f(\bs{\beta}_{0}).
\]
Thus the dyads in \(A\) enter the projection through the product of the
centered link indicators \(\{\xi_e:e\in A\}\). All other dyads appearing
in the same square-free monomial are replaced by their conditional link
probabilities. Since \(\{\xi_e:e\in\mathcal E_N\}\) are conditionally
independent, distinct dyad-set projections are conditionally orthogonal:
\[
\E[\bs\phi_{A,N}\bs\phi_{A',N}^{\top}\mid\bf X,\bs\Gamma]=0
\qquad\text{whenever }A\neq A'.
\]
Together with \eqref{eq:main_projection_decomp}, this orthogonality
implies that the conditional covariance, and hence the variance scale,
of \(\bs h_N(\bs{\beta}_{0})\) is obtained by adding the covariance
contributions of the dyad-set projections \(\bs\phi_{A,N}\). We therefore
work with the dyad-set projections rather than the five-node summands
directly. Since the variance order of \(\bs\phi_{A,N}\) depends on the
unlabeled graph shape represented by \(A\), rather than by the
particular node labels, we group the projections according to that
shape.

For a finite dyad set \(A\subseteq\mathcal E_N\), let
\(V(A)\) be the set of nodes that appear in the dyads in \(A\), and let
\(G_A:=(V(A),A)\) be the corresponding labeled graph. 
Let \([G_A]\) denote its isomorphism class, equivalently the unlabeled
graph shape determined by the dyads present among the nodes in \(A\).
Following the graph-isomorphism terminology in \citet{graham2020network},
each such isomorphism class can be used as an element of a graph class:
two dyad sets \(A\) and \(A'\) have the same isomorphism class if and only if
\(G_A\) and \(G_{A'}\) are isomorphic (see Appendix~\ref{app:graph_notation}).
For example, all singleton dyad sets are labeled copies of the one-dyad
class, and all two-edge paths are labeled copies of the two-star
class.

More generally, we use the term \emph{graph class} for a finite
collection of such unlabeled graph shapes. For such a class \(g\), let
\(\mathscr C_N(g):=\{A\subseteq\mathcal E_N:[G_A]\in g\}\) denote the
collection of all labeled dyad sets whose unlabeled isomorphism class is
one of the elements of \(g\). For example, the graph class
\(g_{\mathrm{3tree}}\) contains the three-edge path
(Figure~\ref{fig:projection_classes_main}(c)) and the three-star
(Figure~\ref{fig:projection_classes_main}(d)), so
\(\mathscr C_N(g_{\mathrm{3tree}})\) collects all labeled dyad sets with
either shape. For such a graph class \(g\), define the projection sum
\(
\bs\Pi_N(g):=\sum_{A\in\mathscr C_N(g)}\bs\phi_{A,N},
\)
which collects the projection terms whose dyad sets have a shape in \(g\).
The graph classes that can arise from the five-node determinant kernel
are characterized in Lemma~\ref{lem:kerneltaxonomy} and summarized in
Table~\ref{tab:regimemap}.

We now define four graph classes whose projection sums can be asymptotically
leading under the sparsity regimes considered in this paper.
The one-dyad class is
\(g_{\mathrm{dyad}}:=\{[G_{\{(1,2)\}}]\},\)
so that \(\bs\Pi_N(g_{\mathrm{dyad}})\) is the one-dyad projection sum.
The two-star class is
\(g_{\text{2-star}}:=\{[G_{\{(1,2),(1,3)\}}]\},\)
the two-edge path on three nodes. We also use two combined graph
classes. Let
\(g_{\mathrm{3tree}}
:=
\{
[G_{\{(1,2),(2,3),(3,4)\}}],
[G_{\{(1,2),(1,3),(1,4)\}}]
\big\},\)
which collects the three-edge path on four nodes and the three-star on
four nodes. Let
\(g_{\mathrm{4tree}}
:=
\big\{
[G_{\{(1,2),(2,3),(3,4),(4,5)\}}],
[G_{\{(1,2),(1,3),(1,4),(2,5)\}}]
\big\},\)
which collects the four-edge path on five nodes and the fork-shaped
four-edge tree with degree sequence \((3,2,1,1,1)\). The associated
projection sums are
\(\bs\Pi_N(g_{\mathrm{dyad}})\), \(\bs\Pi_N(g_{\text{2-star}})\),
\(\bs\Pi_N(g_{\mathrm{3tree}})\), and \(\bs\Pi_N(g_{\mathrm{4tree}})\), each
defined by summing \(\bs\phi_{A,N}\) over the labeled copies in the
corresponding class. Figure~\ref{fig:projection_classes_main} provides a
graphical illustration of these classes.

\begin{figure}[!htbp]
\centering
\captionsetup{font=small}
\captionsetup[subfigure]{font=small}
\begin{subfigure}[t]{0.30\textwidth}
\centering
\begin{tikzpicture}[scale=0.60,graphNode/.style={circle,fill=black,inner sep=1.35pt},graphEdge/.style={line width=0.6pt}]
  \node[graphNode] (d1) at (0,0) {};
  \node[graphNode] (d2) at (1.6,0) {};
  \draw[graphEdge] (d1)--(d2);
\end{tikzpicture}
\caption{\(g_{\mathrm{dyad}}\)}
\end{subfigure}
\begin{subfigure}[t]{0.30\textwidth}
\centering
\begin{tikzpicture}[scale=0.60,graphNode/.style={circle,fill=black,inner sep=1.35pt},graphEdge/.style={line width=0.6pt}]
  \node[graphNode] (w1) at (0.80,0.75) {};
  \node[graphNode] (w2) at (0,0) {};
  \node[graphNode] (w3) at (1.60,0) {};
  \draw[graphEdge] (w1)--(w2);
  \draw[graphEdge] (w1)--(w3);
\end{tikzpicture}
\caption{\(g_{\text{2-star}}\)}
\end{subfigure}
\hfill
\begin{subfigure}[t]{0.30\textwidth}
\centering
\begin{tikzpicture}[scale=0.60,graphNode/.style={circle,fill=black,inner sep=1.35pt},graphEdge/.style={line width=0.6pt}]
  \node[graphNode] (p1) at (0,0) {};
  \node[graphNode] (p2) at (0.72,0) {};
  \node[graphNode] (p3) at (1.44,0) {};
  \node[graphNode] (p4) at (2.16,0) {};
  \draw[graphEdge] (p1)--(p2)--(p3)--(p4);
\end{tikzpicture}
\caption{\(g_{\mathrm{3tree}}\): three-edge path}
\end{subfigure}

\vspace{0.9em}

\begin{subfigure}[t]{0.30\textwidth}
\centering
\begin{tikzpicture}[scale=0.60,graphNode/.style={circle,fill=black,inner sep=1.35pt},graphEdge/.style={line width=0.6pt}]
  \node[graphNode] (s1) at (0.80,0.45) {};
  \node[graphNode] (s2) at (0.80,1.10) {};
  \node[graphNode] (s3) at (0,0) {};
  \node[graphNode] (s4) at (1.60,0) {};
  \draw[graphEdge] (s1)--(s2);
  \draw[graphEdge] (s1)--(s3);
  \draw[graphEdge] (s1)--(s4);
\end{tikzpicture}
\caption{\(g_{\mathrm{3tree}}\): three-star}
\end{subfigure}
\hfill
\begin{subfigure}[t]{0.30\textwidth}
\centering
\begin{tikzpicture}[scale=0.60,graphNode/.style={circle,fill=black,inner sep=1.35pt},graphEdge/.style={line width=0.6pt}]
  \node[graphNode] (q1) at (0,0) {};
  \node[graphNode] (q2) at (0.62,0) {};
  \node[graphNode] (q3) at (1.24,0) {};
  \node[graphNode] (q4) at (1.86,0) {};
  \node[graphNode] (q5) at (2.48,0) {};
  \draw[graphEdge] (q1)--(q2)--(q3)--(q4)--(q5);
\end{tikzpicture}
\caption{\(g_{\mathrm{4tree}}\): four-edge path}
\end{subfigure}
\hfill
\begin{subfigure}[t]{0.30\textwidth}
\centering
\begin{tikzpicture}[scale=0.60,graphNode/.style={circle,fill=black,inner sep=1.35pt},graphEdge/.style={line width=0.6pt}]
  \node[graphNode] (t1) at (0.70,0.45) {};
  \node[graphNode] (t2) at (1.35,0.45) {};
  \node[graphNode] (t3) at (0.10,1.00) {};
  \node[graphNode] (t4) at (0.10,-0.10) {};
  \node[graphNode] (t5) at (1.95,0.45) {};
  \draw[graphEdge] (t1)--(t2);
  \draw[graphEdge] (t1)--(t3);
  \draw[graphEdge] (t1)--(t4);
  \draw[graphEdge] (t2)--(t5);
\end{tikzpicture}
\caption{\(g_{\mathrm{4tree}}\): fork-shaped four-edge tree}
\end{subfigure}
\caption{Graph classes whose projection sums can be asymptotically leading
under the sparsity regimes considered in this paper. Panels~(a) and~(b)
display \(g_{\mathrm{dyad}}\) and \(g_{\text{2-star}}\), respectively;
panels~(c) and~(d) display the two shapes in \(g_{\mathrm{3tree}}\); and
panels~(e) and~(f) display the two shapes in \(g_{\mathrm{4tree}}\).}
\label{fig:projection_classes_main}
\end{figure}
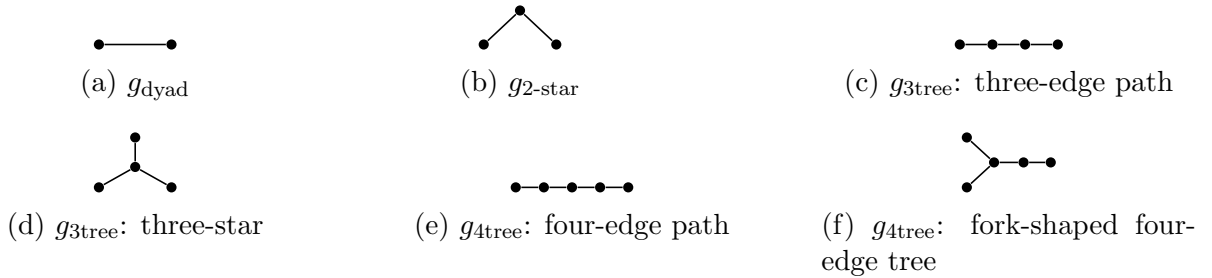

The projection construction above expresses
\(\bs h_N(\bs\beta_0)\) in terms of conditionally orthogonal projection
sums \(\bs\Pi_N(g)\). To derive the asymptotic distribution of
\(\bs h_N(\bs\beta_0)\), we first identify which of these sums form the
leading approximation by comparing their second-moment orders.

\subsection{Regime-Specific Asymptotics of the Sample Moment}

We partition the unlabeled graph shapes arising from dyad subsets
\(A\subseteq E_{s,\nu}\subseteq E(s)\) into graph classes according to
the features that determine the second-moment orders of their projection
sums. Two shapes
belong to the same graph class exactly when they
have the same number of dyads, the same number of distinct nodes, and
the same minimum of \(\abs{E_{s,\nu}}\) over the square-free monomials
\(m_{s,\nu}({\bf L})\) in \eqref{eq:kernel_poly_rep} for which the
support \(E_{s,\nu}\) contains a labeled copy of the shape. Fix one such
graph class \(g\), and let \(r(g)\) and \(v(g)\) denote, respectively,
the common number of dyads and the common number of distinct nodes among
its shapes. Define the common minimum degree of a square-free monomial
whose support contains a labeled copy of a shape in \(g\) by
\(
m(g):=\min\{\abs{E_{s,\nu}}:\exists A\subseteq E_{s,\nu}
\text{ with }[G_A]\in g\},
\)
where the minimum ranges over the square-free monomials
\(m_{s,\nu}({\bf L})\) in \eqref{eq:kernel_poly_rep}.
Proposition~\ref{prop:classbound} in
Appendix~\ref{sec:regime_complete_asymptotic} shows that
\(\E\|\bs\Pi_N(g)\|^2
=O\big(N^{-v(g)}\rho_N^{2m(g)-r(g)}\big)\).
Table~\ref{tab:regimemap} lists all graph classes $g$ arising from dyad
subsets \(A\subseteq E_{s,\nu}\) and the corresponding
values of \(r(g)\), \(v(g)\), \(m(g)\), and
\(N^{-v(g)}\rho_N^{2m(g)-r(g)}\). 

\newcommand{\taxgraph}[2][]{%
\begin{tikzpicture}[
  baseline=(current bounding box.center),
  scale=0.33,
  taxNode/.style={circle,fill=black,inner sep=0.9pt,outer sep=0pt},
  taxEdge/.style={line width=0.6pt},
  #1
]
#2
\end{tikzpicture}%
}
\newcommand{\taxsep}{\hspace{0.18em}}

\begin{table}[H]
\centering
\caption{Complete list of graph classes and variance scales induced by dyad
subsets \(A\subseteq E_{s,\nu}\) in the projection decomposition. Here
\(E_{s,\nu}\) is the dyad support of the square-free monomial
\(m_{s,\nu}({\bf L})\) in \eqref{eq:kernel_poly_rep}.}
\label{tab:regimemap}
{\scriptsize
\renewcommand{\arraystretch}{1.25}%
\renewcommand{\tabularxcolumn}[1]{m{#1}}%
\setlength{\tabcolsep}{2.2pt}
\begin{tabularx}{\textwidth}{@{}>{\raggedright\arraybackslash}m{2.05cm}*{7}{>{\centering\arraybackslash}X}@{}}
\toprule
Graph class \(g\) &
\taxgraph{\node[taxNode] (a) at (0,0) {}; \node[taxNode] (b) at (1,0) {}; \draw[taxEdge] (a)--(b);} &
\taxgraph{\node[taxNode] (a) at (0,0) {}; \node[taxNode] (b) at (.7,.65) {}; \node[taxNode] (c) at (1.4,0) {}; \draw[taxEdge] (a)--(b)--(c);} &
\taxgraph{\node[taxNode] (a) at (0,.35) {}; \node[taxNode] (b) at (.9,.35) {}; \node[taxNode] (c) at (0,-.35) {}; \node[taxNode] (d) at (.9,-.35) {}; \draw[taxEdge] (a)--(b); \draw[taxEdge] (c)--(d);} &
\taxgraph{\node[taxNode] (a) at (0,.75) {}; \node[taxNode] (b) at (-.65,0) {}; \node[taxNode] (c) at (.65,0) {}; \draw[taxEdge] (a)--(b)--(c)--(a);} &
\mbox{\taxgraph{\node[taxNode] (a) at (0,0) {}; \node[taxNode] (b) at (.7,0) {}; \node[taxNode] (c) at (1.4,0) {}; \node[taxNode] (d) at (2.1,0) {}; \draw[taxEdge] (a)--(b)--(c)--(d);}
\taxsep
\taxgraph{\node[taxNode] (a) at (0,0) {}; \node[taxNode] (b) at (0,.8) {}; \node[taxNode] (c) at (-.7,-.45) {}; \node[taxNode] (d) at (.7,-.45) {}; \draw[taxEdge] (a)--(b); \draw[taxEdge] (a)--(c); \draw[taxEdge] (a)--(d);}} &
\taxgraph{\node[taxNode] (a) at (0,.35) {}; \node[taxNode] (b) at (.55,.75) {}; \node[taxNode] (c) at (1.1,.35) {}; \node[taxNode] (d) at (.1,-.35) {}; \node[taxNode] (e) at (1,-.35) {}; \draw[taxEdge] (a)--(b)--(c); \draw[taxEdge] (d)--(e);} &
\mbox{\taxgraph{\node[taxNode] (a) at (-.55,-.45) {}; \node[taxNode] (b) at (.55,.45) {}; \node[taxNode] (c) at (-.55,.45) {}; \node[taxNode] (d) at (.55,-.45) {}; \draw[taxEdge] (a)--(b); \draw[taxEdge] (b)--(c); \draw[taxEdge] (c)--(d); \draw[taxEdge] (d)--(a);}
\taxsep
\taxgraph{\node[taxNode] (a) at (0,.75) {}; \node[taxNode] (b) at (-.65,0) {}; \node[taxNode] (c) at (.65,0) {}; \node[taxNode] (d) at (1.35,0) {}; \draw[taxEdge] (a)--(b)--(c)--(a); \draw[taxEdge] (c)--(d);}} \\
\midrule
\(r(g)\) & \(1\) & \(2\) & \(2\) & \(3\) & \(3\) & \(3\) & \(4\) \\
\(v(g)\) & \(2\) & \(3\) & \(4\) & \(3\) & \(4\) & \(5\) & \(4\) \\
\(m(g)\) & \(4\) & \(4\) & \(4\) & \(5\) & \(4\) & \(4\) & \(5\) \\
variance scale & \(N^{-2}\rho_N^7\) & \(N^{-3}\rho_N^6\) &
\(N^{-4}\rho_N^6\) & \(N^{-3}\rho_N^7\) & \(N^{-4}\rho_N^5\) &
\(N^{-5}\rho_N^5\) & \(N^{-4}\rho_N^6\) \\
\bottomrule
\end{tabularx}

\vspace{1.2ex}

\begin{tabularx}{\textwidth}{@{}>{\raggedright\arraybackslash}m{2.05cm}*{6}{>{\centering\arraybackslash}X}@{}}
\toprule
Graph class \(g\) &
\mbox{\taxgraph{\node[taxNode] (a) at (0,0) {}; \node[taxNode] (b) at (.55,0) {}; \node[taxNode] (c) at (1.1,0) {}; \node[taxNode] (d) at (1.65,0) {}; \node[taxNode] (e) at (2.2,0) {}; \draw[taxEdge] (a)--(b)--(c)--(d)--(e);}
\taxsep
\taxgraph{\node[taxNode] (a) at (0,0) {}; \node[taxNode] (b) at (.65,0) {}; \node[taxNode] (c) at (-.55,.55) {}; \node[taxNode] (d) at (-.55,-.55) {}; \node[taxNode] (e) at (1.25,0) {}; \draw[taxEdge] (a)--(b); \draw[taxEdge] (a)--(c); \draw[taxEdge] (a)--(d); \draw[taxEdge] (b)--(e);}} &
\taxgraph{\node[taxNode] (a) at (0,0) {}; \node[taxNode] (b) at (0,.85) {}; \node[taxNode] (c) at (0,-.85) {}; \node[taxNode] (d) at (-.85,0) {}; \node[taxNode] (e) at (.85,0) {}; \draw[taxEdge] (a)--(b); \draw[taxEdge] (a)--(c); \draw[taxEdge] (a)--(d); \draw[taxEdge] (a)--(e);} &
\taxgraph{\node[taxNode] (a) at (-.75,0) {}; \node[taxNode] (b) at (0,.65) {}; \node[taxNode] (c) at (.75,0) {}; \node[taxNode] (d) at (0,-.65) {}; \draw[taxEdge] (a)--(b)--(c)--(d)--(a); \draw[taxEdge] (a)--(c);} &
\mbox{\taxgraph{\node[taxNode] (a) at (-.55,-.45) {}; \node[taxNode] (b) at (.55,.45) {}; \node[taxNode] (c) at (-.55,.45) {}; \node[taxNode] (d) at (.55,-.45) {}; \node[taxNode] (e) at (1.2,-.45) {}; \draw[taxEdge] (a)--(b); \draw[taxEdge] (b)--(c); \draw[taxEdge] (c)--(d); \draw[taxEdge] (d)--(a); \draw[taxEdge] (d)--(e);}
\taxsep
\taxgraph{\node[taxNode] (a) at (0,.65) {}; \node[taxNode] (b) at (-.55,0) {}; \node[taxNode] (c) at (.55,0) {}; \node[taxNode] (d) at (-.45,1.2) {}; \node[taxNode] (e) at (.45,1.2) {}; \draw[taxEdge] (a)--(b)--(c)--(a); \draw[taxEdge] (a)--(d); \draw[taxEdge] (a)--(e);}
\taxsep
\taxgraph{\node[taxNode] (a) at (0,.65) {}; \node[taxNode] (b) at (-.55,0) {}; \node[taxNode] (c) at (.55,0) {}; \node[taxNode] (d) at (-1.1,0) {}; \node[taxNode] (e) at (1.1,0) {}; \draw[taxEdge] (a)--(b)--(c)--(a); \draw[taxEdge] (b)--(d); \draw[taxEdge] (c)--(e);}} &
\mbox{\taxgraph{\node[taxNode] (a) at (-.45,-.45) {}; \node[taxNode] (b) at (.45,-.45) {}; \node[taxNode] (c) at (-.9,.55) {}; \node[taxNode] (d) at (0,.9) {}; \node[taxNode] (e) at (.9,.55) {}; \draw[taxEdge] (a)--(c); \draw[taxEdge] (a)--(d); \draw[taxEdge] (a)--(e); \draw[taxEdge] (b)--(c); \draw[taxEdge] (b)--(d); \draw[taxEdge] (b)--(e);}
\taxsep
\taxgraph{\node[taxNode] (a) at (-.65,0) {}; \node[taxNode] (b) at (0,.6) {}; \node[taxNode] (c) at (.65,0) {}; \node[taxNode] (d) at (0,-.6) {}; \node[taxNode] (e) at (-1.15,0) {}; \draw[taxEdge] (a)--(b)--(c)--(d)--(a); \draw[taxEdge] (a)--(c); \draw[taxEdge] (a)--(e);}} &
\taxgraph{\node[taxNode] (a) at (-.45,-.45) {}; \node[taxNode] (b) at (.45,-.45) {}; \node[taxNode] (c) at (-.9,.55) {}; \node[taxNode] (d) at (0,.9) {}; \node[taxNode] (e) at (.9,.55) {}; \draw[taxEdge] (a)--(b); \draw[taxEdge] (a)--(c); \draw[taxEdge] (a)--(d); \draw[taxEdge] (a)--(e); \draw[taxEdge] (b)--(c); \draw[taxEdge] (b)--(d); \draw[taxEdge] (b)--(e);} \\
\midrule
\(r(g)\) & \(4\) & \(4\) & \(5\) & \(5\) & \(6\) & \(7\) \\
\(v(g)\) & \(5\) & \(5\) & \(4\) & \(5\) & \(5\) & \(5\) \\
\(m(g)\) & \(4\) & \(5\) & \(6\) & \(5\) & \(6\) & \(7\) \\
variance scale & \(N^{-5}\rho_N^4\) & \(N^{-5}\rho_N^6\) &
\(N^{-4}\rho_N^7\) & \(N^{-5}\rho_N^5\) & \(N^{-5}\rho_N^6\) &
\(N^{-5}\rho_N^7\) \\
\bottomrule
\end{tabularx}
}
\end{table}

The variance scales in Table~\ref{tab:regimemap} show that only the four
graph classes displayed in Figure~\ref{fig:projection_classes_main} can
enter a leading approximation in at least one of the three sparsity regimes.
Using the one-dyad scale
\(N^{-2}\rho_N^7\) for \(\bs\Pi_N(g_{\mathrm{dyad}})\) as the benchmark,
the scales for
\(\bs\Pi_N(g_{\text{2-star}})\), \(\bs\Pi_N(g_{\mathrm{3tree}})\), and
\(\bs\Pi_N(g_{\mathrm{4tree}})\) are, respectively,
\((N\rho_N)^{-1}\), \((N\rho_N)^{-2}\), and
\((N\rho_N)^{-3}\) times the benchmark. Thus, when
\(N\rho_N\to\infty\), only
the one-dyad projection is retained in the leading approximation; when
\(\rho_N\asymp N^{-1}\), all four variance scales are of order
\(N^{-9}\), so all four projections are retained; and when
\(N\rho_N\to0\), the four-tree scale dominates the other three.
Accordingly, define the leading projection sums for the three cases by
\(\bs\Pi_N^{\mathrm{D}}:=\bs\Pi_N(g_{\mathrm{dyad}})\),
\(\bs\Pi_N^{\mathrm{S}}:=\bs\Pi_N(g_{\mathrm{dyad}})
+\bs\Pi_N(g_{\text{2-star}})+\bs\Pi_N(g_{\mathrm{3tree}})
+\bs\Pi_N(g_{\mathrm{4tree}})\), and
\(\bs\Pi_N^{\mathrm{US}}:=\bs\Pi_N(g_{\mathrm{4tree}})\).
Under Assumptions~\ref{ass:dyad_ind}, \ref{ass:bounded_design}, and
\ref{ass:sparse_envelope}, the order bounds in
Table~\ref{tab:regimemap} yield the leading approximations.
If \(N\rho_N\to\infty\), then
\(N\rho_N^{-7/2}\bs h_N(\bs\beta_0)
=N\rho_N^{-7/2}\bs\Pi_N^{\mathrm{D}}+o_p(1)\).
If \(\rho_N\asymp N^{-1}\), then
\(N^{9/2}\bs h_N(\bs\beta_0)
=N^{9/2}\bs\Pi_N^{\mathrm{S}}+o_p(1)\).
If \(N\rho_N\to0\), then
\(N^{5/2}\rho_N^{-2}\bs h_N(\bs\beta_0)
=N^{5/2}\rho_N^{-2}\bs\Pi_N^{\mathrm{US}}+o_p(1)\).

For these approximations to yield nontrivial Gaussian limits, the normalized
leading projection sum in each case must have a finite, nonzero limiting
covariance. Assumption~\ref{ass:regimeSigma} requires convergence of the
normalized conditional covariance of the relevant projection sum. The
positive-trace condition rules out the zero covariance matrix as a limit, but
does not require the limiting covariance matrix to be positive definite.

\begin{assumption}
\label{ass:regimeSigma}
\textnormal{(i)} If \(N\rho_N\to\infty\), there exists a positive
semidefinite matrix \(\bf\Sigma_{\mathrm{D}}\) such that
\(\Var(N\rho_N^{-7/2}\bs\Pi_N^{\mathrm{D}}\mid\bf X,\bs\Gamma)\xrightarrow{P}\bf\Sigma_{\mathrm{D}}\) and
\(\tr(\bf\Sigma_{\mathrm{D}})>0\). \textnormal{(ii)} If
\(\rho_N\asymp N^{-1}\), there exists a positive
semidefinite matrix \(\bf\Sigma_{\mathrm{S}}\) such that
\(\Var(N^{9/2}\bs\Pi_N^{\mathrm{S}}\mid\bf X,\bs\Gamma)\xrightarrow{P}\bf\Sigma_{\mathrm{S}}\) and
\(\tr(\bf\Sigma_{\mathrm{S}})>0\). \textnormal{(iii)} If \(N\rho_N\to0\),
there exists a positive semidefinite matrix
\(\bf\Sigma_{\mathrm{US}}\) such that
\(\Var(N^{5/2}\rho_N^{-2}\bs\Pi_N^{\mathrm{US}}\mid\bf X,\bs\Gamma)\xrightarrow{P}\bf\Sigma_{\mathrm{US}}\) and
\(\tr(\bf\Sigma_{\mathrm{US}})>0\).
\end{assumption}

\begin{thm}[Asymptotic normality of the sample moment]
\label{thm:regimeCLT}
Under Assumptions~\ref{ass:dyad_ind}, \ref{ass:bounded_design},
\ref{ass:sparse_envelope}, and \ref{ass:regimeSigma}, the sample moment
\(\bs h_N(\bs{\beta}_{0})\) has the following Gaussian limits.
\textnormal{(i)} In the dense or mildly sparse regime
(\(N\rho_N\to\infty\)),
\(N\rho_N^{-7/2}\bs h_N(\bs{\beta}_{0})\xrightarrow{D} N(0,\bf\Sigma_{\mathrm{D}})\).
\textnormal{(ii)} In the sparse regime (\(\rho_N\asymp N^{-1}\)),
\(N^{9/2}\bs h_N(\bs{\beta}_{0})\xrightarrow{D} N(0,\bf\Sigma_{\mathrm{S}})\).
\textnormal{(iii)} In the ultra-sparse regime
(\(N\rho_N\to0\) and \(N^5\rho_N^4\to\infty\)),
\(N^{5/2}\rho_N^{-2}\bs h_N(\bs{\beta}_{0})
\xrightarrow{D} N(0,\bf\Sigma_{\mathrm{US}})\).
\end{thm}

The leading projection approximations and
Theorem~\ref{thm:regimeCLT} show that network sparsity determines both the
normalization of the sample moment and the graph classes entering its limiting
covariance. Table~\ref{tab:boundaryregimes} summarizes the corresponding
sparsity conditions, leading graph classes, and variance orders.

\begin{table}[!ht]
  \centering
  \caption{Leading graph classes and variance orders in the three Gaussian
  regimes. Graph labels follow Figure~\ref{fig:projection_classes_main}.}
  \label{tab:boundaryregimes}
  {\small
  \renewcommand{\arraystretch}{1.22}%
  \setlength{\extrarowheight}{1.5pt}%
  \setlength{\tabcolsep}{0.45pt}%
  \renewcommand{\tabularxcolumn}[1]{m{#1}}%
  \begin{tabularx}{\textwidth}{@{}>{\centering\arraybackslash}m{3.50cm}
  @{\hspace{0.10cm}}>{\centering\arraybackslash}m{4.15cm}
  >{\centering\arraybackslash}m{2.30cm}
  >{\centering\arraybackslash}X
  >{\centering\arraybackslash}m{1.50cm}@{}}
  \toprule
  Regime & Condition & Leading graph classes & Graph shape(s) & \shortstack{Variance\\order} \\
  \midrule
  \mbox{Dense/mildly sparse} &
  \(N\rho_N\to\infty\) &
  \(g_{\mathrm{dyad}}\) &
  \makebox[\linewidth][c]{\taxgraph[scale=0.90]{\node[taxNode] (a) at (0,0) {}; \node[taxNode] (b) at (1,0) {}; \draw[taxEdge] (a)--(b);}} &
  \(N^{-2}\rho_N^7\) \\
  \addlinespace[0.55ex]
  \parbox[c][2.15cm][c]{\linewidth}{\centering Sparse} &
  \parbox[c][2.15cm][c]{\linewidth}{\centering \(\rho_N\asymp N^{-1}\)} &
  \parbox[c][2.15cm][c]{\linewidth}{\centering
  \shortstack[c]{\(g_{\mathrm{dyad}}\); \(g_{\text{2-star}}\)\\[1.75ex]
  \(g_{\mathrm{3tree}}\); \(g_{\mathrm{4tree}}\)}} &
  \parbox[c][2.15cm][c]{\linewidth}{\centering
  \raisebox{1.50ex}{%
  \begin{tabular}[c]{@{}c@{\hspace{0.12em}}c@{\hspace{0.12em}}c@{}}
  \taxgraph[scale=0.90]{\node[taxNode] (a) at (0,0) {}; \node[taxNode] (b) at (1,0) {}; \draw[taxEdge] (a)--(b);}
  &
  \taxgraph[scale=0.90]{\node[taxNode] (a) at (0,0) {}; \node[taxNode] (b) at (.7,.65) {}; \node[taxNode] (c) at (1.4,0) {}; \draw[taxEdge] (a)--(b)--(c);}
  &
  \taxgraph[scale=0.90]{\node[taxNode] (a) at (0,0) {}; \node[taxNode] (b) at (.7,0) {}; \node[taxNode] (c) at (1.4,0) {}; \node[taxNode] (d) at (2.1,0) {}; \draw[taxEdge] (a)--(b)--(c)--(d);}
  \\[1.05ex]
  \taxgraph[scale=0.90]{\node[taxNode] (a) at (0,0) {}; \node[taxNode] (b) at (0,.8) {}; \node[taxNode] (c) at (-.7,-.45) {}; \node[taxNode] (d) at (.7,-.45) {}; \draw[taxEdge] (a)--(b); \draw[taxEdge] (a)--(c); \draw[taxEdge] (a)--(d);}
  &
  \taxgraph[scale=0.90]{\node[taxNode] (a) at (0,0) {}; \node[taxNode] (b) at (.55,0) {}; \node[taxNode] (c) at (1.1,0) {}; \node[taxNode] (d) at (1.65,0) {}; \node[taxNode] (e) at (2.2,0) {}; \draw[taxEdge] (a)--(b)--(c)--(d)--(e);}
  &
  \taxgraph[scale=0.90]{\node[taxNode] (a) at (0,0) {}; \node[taxNode] (b) at (.65,0) {}; \node[taxNode] (c) at (-.55,.55) {}; \node[taxNode] (d) at (-.55,-.55) {}; \node[taxNode] (e) at (1.25,0) {}; \draw[taxEdge] (a)--(b); \draw[taxEdge] (a)--(c); \draw[taxEdge] (a)--(d); \draw[taxEdge] (b)--(e);}
  \end{tabular}}} &
  \parbox[c][2.15cm][c]{\linewidth}{\centering \(N^{-9}\)} \\
  \addlinespace[0.55ex]
  Ultra-sparse &
  \(N\rho_N\to0,\;N^5\rho_N^4\to\infty\) &
  \(g_{\mathrm{4tree}}\) &
  \makebox[\linewidth][c]{%
  \begin{tabular}[c]{@{}c@{\hspace{0.30em}}c@{}}
  \taxgraph[scale=0.90]{\node[taxNode] (a) at (0,0) {}; \node[taxNode] (b) at (.55,0) {}; \node[taxNode] (c) at (1.1,0) {}; \node[taxNode] (d) at (1.65,0) {}; \node[taxNode] (e) at (2.2,0) {}; \draw[taxEdge] (a)--(b)--(c)--(d)--(e);}
  &
  \taxgraph[scale=0.90]{\node[taxNode] (a) at (0,0) {}; \node[taxNode] (b) at (.65,0) {}; \node[taxNode] (c) at (-.55,.55) {}; \node[taxNode] (d) at (-.55,-.55) {}; \node[taxNode] (e) at (1.25,0) {}; \draw[taxEdge] (a)--(b); \draw[taxEdge] (a)--(c); \draw[taxEdge] (a)--(d); \draw[taxEdge] (b)--(e);}
  \end{tabular}} &
  \(N^{-5}\rho_N^4\) \\
  \bottomrule
  \end{tabularx}
  }
\end{table}
\FloatBarrier

The condition \(N^5\rho_N^4\to\infty\) in
Theorem~\ref{thm:regimeCLT}(iii) has an effective sample size interpretation.
The leading variation in this regime comes from configurations of five nodes
connected by four links (see Figure~\ref{fig:projection_classes_main}(e)--(f)).
The number of possible configurations is of order
\(N^5\), and a proportion of order \(\rho_N^4\) contains the required link
pattern. Thus, \(N^5\rho_N^4\) is the effective number of informative
configurations, and \(N^5\rho_N^4\to\infty\) is the corresponding effective
sample size condition for asymptotic normality in the ultra-sparse regime.
In the sparse Erd\H{o}s--R\'enyi model, the classical central limit condition
of \citet{rucinski1988when} for the count of a fixed tree with five nodes and
four edges reduces to \(N\rho_N^{4/5}\to\infty\), which is equivalent to
\(N^5\rho_N^4\to\infty\).

\subsection{Asymptotic Properties of the Pentad-GMM Estimator}
\label{sec:gmm_asymptotics}

We now use the limit theory for the sample moment to obtain the asymptotic
distribution of the pentad-GMM estimator.  The following conditions impose a
population limit, identification, the rank of the limiting Jacobian, and
weight convergence.

\begin{assumption}
\label{ass:consistency}
(i) There exists a deterministic continuous function
\(\bs h_0:\mathcal B\to\mathbb{R}^{q}\) such that
\(\rho_N^{-4}\E[\bs h_N(\bs{\beta})]\to\bs h_0(\bs{\beta})\)
for every \(\bs{\beta}\in\mathcal B\), and
\(\bs h_0(\bs{\beta})=0\iff \bs{\beta}=\bs{\beta}_{0}\).
(ii) The matrix
\(\dot{\bs h}_0:=
\partial\bs h_0(\bs\beta_0)/\partial\bs\beta^\top
\in\mathbb{R}^{q\times d}\)
has full column rank.
(iii) The GMM weight matrix $\widehat{\bf W}_N$ used in \eqref{eq:beta_gmm} satisfies
$\widehat{\bf W}_N\xrightarrow{P} \bf W$, where $\bf W$ is deterministic and positive definite.
\end{assumption}

Assumption~\ref{ass:consistency} collects standard GMM
regularity conditions; see \citet{hansen1982large} and
\citet{newey1994large}. Assumption~\ref{ass:consistency}(i) identifies
\(\bs\beta_0\) as the unique zero of the limiting moment. For the five-node
determinant moments entering \eqref{eq:sample_moment-multi}, the smallest
nonzero link degree is four, so the population moment has the natural scale
\(\rho_N^4\), which explains the factor \(\rho_N^{-4}\) in
Assumption~\ref{ass:consistency}(i). Assumption~\ref{ass:consistency}(ii), together with
the positive definiteness in
Assumption~\ref{ass:consistency}(iii), makes
\(\dot{\bs h}_0^\top\bf W\dot{\bs h}_0\) nonsingular for the local
linearization. Assumption~\ref{ass:consistency}(iii) accommodates a data-dependent
weight matrix. In our implementation, we use the inverse empirical-Gram weight
\(\widehat{\bf W}_N=
\{\abs{\mathcal S_N}^{-1}\sum_{s\in\mathcal S_N}
\bs Z_s\bs Z_s^\top\}^{-1}\). If
\(\E[\bs Z_s\bs Z_s^\top]\) is a fixed positive definite matrix, a law of
large numbers for \(U\)-statistics
(see \citealp[Section~5.4, Theorem~A]{serfling1980approximation}) implies
that the inverse
empirical-Gram weight satisfies Assumption~\ref{ass:consistency}(iii).

\begin{thm}[Asymptotic properties of Pentad-GMM]
\label{thm:gmm}
Suppose Assumptions~\ref{ass:dyad_ind}, \ref{ass:bounded_design},
\ref{ass:sparse_envelope}, \ref{ass:regimeSigma}, and
\ref{ass:consistency} hold, and
\(N^5\rho_N^4\to\infty\). For the regime-specific normalization and
covariance, set
\[
(a_N,\bf\Sigma)
=
\begin{cases}
(N\rho_N^{-7/2},\bf\Sigma_{\mathrm{D}}),
& N\rho_N\to\infty,\\
(N^{9/2},\bf\Sigma_{\mathrm{S}}),
& \rho_N\asymp N^{-1},\\
(N^{5/2}\rho_N^{-2},\bf\Sigma_{\mathrm{US}}),
& N\rho_N\to0 \ \text{and}\ N^5\rho_N^4\to\infty,
\end{cases}
\]
where \(\bf\Sigma_{\mathrm{D}},\bf\Sigma_{\mathrm{S}}\), and
\(\bf\Sigma_{\mathrm{US}}\) are the limits in
Assumption~\ref{ass:regimeSigma}. Let
\(\widehat{\bs{\beta}}_{\mathrm{GMM}}\) denote the minimizer in
\eqref{eq:beta_gmm}. 
Define
\(
\mathbf{V}_0
:=
\big(\dot{\bs h}_0^\top\bf W\dot{\bs h}_0\big)^{-1}
\dot{\bs h}_0^\top\bf W\bf\Sigma\bf W\dot{\bs h}_0
\big(\dot{\bs h}_0^\top\bf W\dot{\bs h}_0\big)^{-1}
\). Then
\[
a_N\rho_N^4(\widehat{\bs{\beta}}_{\mathrm{GMM}}-\bs{\beta}_{0})
\xrightarrow{D}N(0,\mathbf{V}_0).
\]
\end{thm}

The normalization in Theorem~\ref{thm:gmm} implies that
\(
\widehat{\bs{\beta}}_{\mathrm{GMM}}-\bs{\beta}_{0}
=
O_p\big((N^2\rho_N)^{-1/2}\big)
\)
in the dense/mildly sparse and sparse regimes, whereas
\(
\widehat{\bs{\beta}}_{\mathrm{GMM}}-\bs{\beta}_{0}
=
O_p\big((N^5\rho_N^4)^{-1/2}\big)
\)
in the ultra-sparse regime. Since the average link probability is \(\rho_N\),
the average expected degree is \((N-1)\rho_N\), and the expected total number
of links is of order \(N^2\rho_N\). Thus, in both the dense or mildly sparse
regime and the sparse regime, the convergence rate \(N\rho_N^{1/2}\) is of the
same order as the square root of the expected number of links. It is of order
\(N\) when \(\rho_N\) is bounded away from zero and \(\sqrt{N}\) when
\(\rho_N\asymp N^{-1}\), with the intermediate rate \(N\rho_N^{1/2}\) in mildly
sparse networks. The \(N\)-rate in dense
networks and the \(\sqrt{N}\)-rate in sparse networks coincide with those
obtained for the TU tetrad-logit estimator by \citet{graham2017econometric}.
Although the convergence rates coincide, the leading variance components differ.
In the sparse regime, the tetrad-logit estimator in \cite{graham2017econometric} is driven by the
one-dyad projection, whereas our pentad-GMM estimator generally requires
higher-order projections as well.
In the ultra-sparse regime, the rate
\(N^{5/2}\rho_N^2=(N^5\rho_N^4)^{1/2}\) is the square root of the effective
sample size \(N^5\rho_N^4\), which corresponds to the number of informative
five-node, four-link configurations discussed after
Theorem~\ref{thm:regimeCLT}.

Although Theorem~\ref{thm:gmm} yields different normalizations and covariance
limits across the three regimes, the same inference procedure can be used in
all three, provided that a single covariance estimator is consistent under
every regime.

\begin{cor}[Unified inference]
\label{cor:generic_gmm_inference}
Under the conditions of Theorem~\ref{thm:gmm}, let
\(\widehat{\bf\Omega}_N\) be a covariance estimator that does
not depend on the sparsity regime and satisfies, under each of the three
regimes in Theorem~\ref{thm:gmm},
\(a_N^2\widehat{\bf\Omega}_N\xrightarrow{P}\bf\Sigma\).
Set \(\widehat{\dot{\bs h}}_N
:=\partial\bs h_N(\widehat{\bs\beta}_{\mathrm{GMM}})/
\partial\bs\beta^\top\) and
\(\widehat{\mathbf{V}}_N
:=(\widehat{\dot{\bs h}}_N^\top\widehat{\bf W}_N
\widehat{\dot{\bs h}}_N)^{-1}\allowbreak
\widehat{\dot{\bs h}}_N^\top\widehat{\bf W}_N\allowbreak
\widehat{\bf\Omega}_N\allowbreak
\widehat{\bf W}_N\widehat{\dot{\bs h}}_N\allowbreak
(\widehat{\dot{\bs h}}_N^\top\widehat{\bf W}_N
\widehat{\dot{\bs h}}_N)^{-1}\).
Then \(a_N^2\rho_N^8\widehat{\mathbf{V}}_N\xrightarrow{P}\mathbf{V}_0\). Consequently,
for every fixed \(\bs c\in\mathbb R^d\) such that
\(\bs c^\top \mathbf{V}_0\bs c>0\),
\[
\frac{\bs c^\top(\widehat{\bs\beta}_{\mathrm{GMM}}-\bs\beta_0)}
{\sqrt{\bs c^\top\widehat{\mathbf{V}}_N\bs c}}
\xrightarrow{D}N(0,1).
\]
\end{cor}
Corollary~\ref{cor:generic_gmm_inference} is stated for a generic covariance
estimator. Appendix~\ref{sec:feasible_covariance} provides a tractable
regime-adaptive covariance estimator
\(\bOmegaUniv(\widehat{\bs\beta}_{\mathrm{GMM}})\) and verifies the required
covariance consistency under all three regimes. This yields unified
inference without preliminary estimation of \(\rho_N\) or regime selection.

\section{Implementation}
\label{sec:implementation}

This section describes the practical implementation of our estimator,
including the construction of the instruments, the normalization of the GMM
criterion using pivotal-pair moments, and an algorithm that reduces the
computational cost of evaluating the pentad moments.

The conditional restriction in \eqref{eq:moment_cond-free} implies, by iterated
expectations, that any bounded instrument \(\bs Z_s=z(\bf X_s)\) is valid:
\(\E[\bs Z_s\psi_s(\bs\beta_0)]=\mathbf 0\).\footnote{The usual optimal
instrument based on the conditional expectation of the derivative of the moment function 
\citep{chamberlain1987asymptotic} requires knowledge of the conditional
distribution of the individual fixed effects given the observed covariates.
We avoid estimating this distribution. Substituting preliminary estimates of
all \(N\) individual fixed effects into the derivative would instead produce a
generated instrument and require a separate asymptotic analysis.}
We construct an instrument \(\bs Z_s\) to exploit the relabeling properties of
the pentad kernel
\(\psi_{ijklm}(\bs\beta)=\det(\widetilde{\mathbf M}_{ijklm}(\bs\beta))\).
Lemma~\ref{lem:pivsym} shows that this kernel is
invariant under interchange of the pivotal nodes \(i\) and \(j\) and
alternating under permutations of the peripheral nodes \(k,l,m\). Because
\(\bs h_N(\bs\beta)\) averages over all ordered pentads, the contributions from
odd and even peripheral permutations cancel whenever \(\bs Z_s\) is invariant
under permutations of \(k,l,m\), so that
\(\bs h_N(\bs\beta)=\mathbf 0\) for every \(\bs\beta\). Such an instrument cannot
satisfy the identification condition in Assumption~\ref{ass:consistency}(i). More
generally, only the part of an instrument that is alternating under peripheral
permutations contributes to \(\bs h_N(\bs\beta)\); all remaining parts cancel in
the ordered average. We therefore choose \(\bs Z_s\) to be alternating under
permutations of \(k,l,m\) and invariant under interchange of \(i\) and \(j\). A
determinant-based construction naturally provides both properties, as shown
below.

For the construction below, suppose \(d\geq2\), and let
\(X_{uv,1},\ldots,X_{uv,d}\) denote the coordinates of
\(\bs X_{uv}\). To include interactions between covariate coordinates, for
each pair \(1\leq a<b\leq d\), define their normalized sum and difference by
\[
X_{uv}^{ab,+}:=
\frac{X_{uv,a}+X_{uv,b}}{\sqrt 2},
\qquad
X_{uv}^{ab,-}:=
\frac{X_{uv,a}-X_{uv,b}}{\sqrt 2}.
\]
For two scalar dyadic arrays \(\mathbf{A}=(A_{uv})_{u\ne v}\) and
\(\mathbf{B}=(B_{uv})_{u\ne v}\), define the determinant difference
\[
\mathcal D_{ij;klm}(\mathbf{A},\mathbf{B})
:=
\det
\begin{pmatrix}
1 & A_{ik} & B_{jk}\\
1 & A_{il} & B_{jl}\\
1 & A_{im} & B_{jm}
\end{pmatrix}
-
\det
\begin{pmatrix}
1 & B_{ik} & A_{jk}\\
1 & B_{il} & A_{jl}\\
1 & B_{im} & A_{jm}
\end{pmatrix}.
\]
Interchanging any two of \(k,l,m\) changes the sign of both determinants and
hence of \(\mathcal D_{ij;klm}(\mathbf{A},\mathbf{B})\).
After interchanging \(i\) and \(j\), the first determinant equals the negative
of the original second determinant, while the second equals the negative of
the original first. Hence
\(\mathcal D_{ji;klm}(\mathbf{A},\mathbf{B})=\mathcal D_{ij;klm}(\mathbf{A},\mathbf{B})\).
For either sign, write \(\mathbf{X}^{ab,\pm}:=(X_{uv}^{ab,\pm})_{u\ne v}\) for the
corresponding scalar dyadic array and define its pointwise square by
\((\mathbf{X}^{ab,\pm})^2:=((X_{uv}^{ab,\pm})^2)_{u\ne v}\).
For each \(1\leq a<b\leq d\), define the three-dimensional block
\[
\bs Z_s^{ab}
:=
\Bigl(
\mathcal D_{ij;klm}\bigl(\mathbf{X}^{ab,+},(\mathbf{X}^{ab,+})^2\bigr),
\mathcal D_{ij;klm}\bigl(\mathbf{X}^{ab,-},(\mathbf{X}^{ab,-})^2\bigr),
\mathcal D_{ij;klm}\bigl(\mathbf{X}^{ab,+},\mathbf{X}^{ab,-}\bigr)
\Bigr)^\top.
\]
Stacking these blocks gives
\begin{equation}
\label{eq:covariate_only_instrument}
\bs Z_s
:=
\Bigl(
(\bs Z_s^{12})^\top,\ldots,(\bs Z_s^{1d})^\top,
(\bs Z_s^{23})^\top,\ldots,(\bs Z_s^{2d})^\top,
\ldots,(\bs Z_s^{(d-1)d})^\top
\Bigr)^\top
\in\mathbb R^{3d(d-1)/2}.
\end{equation}
Because each entry of \(\bs Z_s\) is a determinant difference of the form
\(\mathcal D_{ij;klm}(\cdot,\cdot)\), \(\bs Z_s\) changes sign whenever two of
\(k,l,m\) are interchanged and is unchanged when \(i\) and \(j\) are
interchanged, as required.
Within each block \(\bs Z_s^{ab}\), the first two entries pair the normalized
sum and difference with their respective squares, while the third forms a
cross term between the two. Thus
\eqref{eq:covariate_only_instrument} contains three instruments for each of the
\(\binom d2\) coordinate pairs, for a total of \(3d(d-1)/2\).

To improve finite-sample stability, we normalize the GMM criterion by the
average weighted squared norm of the pivotal-pair moments, with a ridge term
that keeps the denominator positive.
For \(i<j\), define the average moment for pivotal pair \(\{i,j\}\) by
\begin{equation}
\bs m_{ij,N}(\bs\beta)
:=\frac{1}{(N-2)(N-3)(N-4)}
\sum_{k,l,m\in[N]\setminus\{i,j\}:\,k,l,m\text{ distinct}}
\bs\Psi_{(i,j,k,l,m)}(\bs\beta).
\label{eq:pivotal_pair_moment}
\end{equation}
Set \(\widehat{\rho}_N:=\binom{N}{2}^{-1}\sum_{i<j}L_{ij}+N^{-2}\) and
\(\widehat c_N:=\widehat{\rho}_N^7+N^{-1}\widehat{\rho}_N^6+N^{-2}\widehat{\rho}_N^5+N^{-3}\widehat{\rho}_N^4\).
Here \(\widehat{\rho}_N\) is the empirical link density with an \(N^{-2}\) adjustment
that keeps it strictly positive even when no links are observed, and
\(\widehat c_N\) is a strictly positive data-driven ridge term constructed from
\(\widehat{\rho}_N\). The normalized GMM criterion is
\begin{equation}
Q_N^\dagger(\bs\beta)
:=
\frac{\bs h_N(\bs\beta)^\top\widehat{\bf W}_N\bs h_N(\bs\beta)}
{\binom{N}{2}^{-1}\sum_{i<j}
\bs m_{ij,N}(\bs\beta)^\top\widehat{\bf W}_N
\bs m_{ij,N}(\bs\beta)
+\widehat c_N}.
\label{eq:pair_normalized_criterion}
\end{equation}
Because \(\bs\Psi_s(\mathbf 0)=0\) for every pentad,
\(\bs h_N(\mathbf 0)=\bs m_{ij,N}(\mathbf 0)=0\). Consequently, the
unnormalized criterion in \eqref{eq:beta_gmm} may become artificially small
near the excluded origin $\bs 0$ and lead the numerical optimizer toward values close
to zero. The average weighted squared norm of the pivotal-pair moments in
the denominator of
\eqref{eq:pair_normalized_criterion} also tends to zero as
\(\bs\beta\) approaches \(\mathbf 0\) and, like the numerator, is quadratic in
averages of the same pentad moments. This normalization reduces the tendency of the criterion to fall simply
because the pentad moments shrink near the origin, thereby improving
numerical stability.
The ridge \(\widehat c_N\) adds a strictly
positive term to the denominator in finite samples. Its magnitude is calibrated
to the sparsity-dependent scale of
\(\binom{N}{2}^{-1}\sum_{i<j}
\bs m_{ij,N}(\bs\beta)^\top\widehat{\bf W}_N
\bs m_{ij,N}(\bs\beta)\), keeping the normalization well behaved across the
three sparsity regimes; see the proof of
Lemma~\ref{lem:pair_normalization} for the derivation.
Related rescalings appear
in fixed-effects binary-choice estimation. \citet{honore2024moment} normalize
each fixed-effect-free moment to bound the moment and its gradient uniformly
and improve small-sample GMM performance; the
conditional-likelihood scores in \citet{honore2000panel} have a similar form.

\begingroup
\emergencystretch=2em
The estimator used in our simulation study and empirical application is
\(\widehat{\bs\beta}^{\dagger}\in
\arg\min_{\bs\beta\in\mathcal B}Q_N^\dagger(\bs\beta)\).
Lemma~\ref{lem:pair_normalization} shows that the resulting estimator
\(\widehat{\bs\beta}^{\dagger}\) is asymptotically equivalent to the
pentad-GMM estimator \(\widehat{\bs\beta}_{\mathrm{GMM}}\) in all three regimes, i.e.,
\(a_N\rho_N^4(\widehat{\bs\beta}^{\dagger}
-\widehat{\bs\beta}_{\mathrm{GMM}})\xrightarrow{P}0\), where \(a_N\) is the regime-specific normalization factor defined in
Theorem~\ref{thm:gmm}. Consequently, under each regime, the scaled
estimators \(a_N\rho_N^4(\widehat{\bs\beta}^{\dagger}-\bs\beta_0)\) and
\(a_N\rho_N^4(\widehat{\bs\beta}_{\mathrm{GMM}}-\bs\beta_0)\) both
converge in distribution to \(N(0,\mathbf{V}_0)\).
In the simulation study and empirical application, we refer to
\(\widehat{\bs\beta}^{\dagger}\) as the pentad-GMM estimator.
\par\endgroup

Brute-force evaluation of \eqref{eq:sample_moment-multi} over all ordered
five-node configurations requires \(O(N^5)\) operations. The resulting
computational burden increases rapidly with network size \(N\). However, exploiting
the determinant structure reduces the computational complexity from
\(O(N^5)\) to \(O(N^3)\). The reduction begins by fixing an ordered pivotal
pair \((i,j)\) and ordering the remaining nodes by their labels.
Fix \(t\in\{1,\ldots,q\}\). The \(t\)-th coordinate of
\(\bs\Psi_s(\bs\beta)\) is \(Z_{s,t}\psi_s(\bs\beta)\). Both factors are
alternating under permutations of \(k,l,m\), so their product is invariant.
Hence, for any three distinct peripheral nodes, the six ways of assigning them
to the roles \((k,l,m)\) produce the same value of
\(Z_{s,t}\psi_s(\bs\beta)\). The sum over ordered peripheral triples therefore
equals six times the sum over triples with \(k<l<m\). After expanding the
determinants defining \(Z_{s,t}\) and \(\psi_s(\bs\beta)\), the contribution of
\((i,j)\) to the \(t\)th coordinate of \(\bs h_N(\bs\beta)\) is a linear
combination of a fixed number of sums of the form
\begin{equation}
\label{eq:cumulative_triple_sum}
\sum_{k<l<m}x_k y_l z_m
=
\sum_l\left(\sum_{k<l}x_k\right)y_l
\left(\sum_{m>l}z_m\right).
\end{equation}
Because \(Z_{s,t}\) does not depend on \(\bs\beta\), differentiating
\(\psi_s(\bs\beta)\) with respect to \(\bs\beta\) likewise expresses every
entry in the \(t\)th row of \(\dot{\bs h}_N(\bs\beta)\) as a linear
combination of a fixed number of sums \(\sum_{k<l<m}x_k y_l z_m\).
All prefix sums \(\sum_{k<l}x_k\) can be computed in \(O(N)\) operations, and
all suffix sums \(\sum_{m>l}z_m\) can likewise be computed in \(O(N)\)
operations; computing both collections therefore costs \(O(N)\) operations in
total. Substituting these precomputed sums into
\eqref{eq:cumulative_triple_sum} leaves a single sum over \(l\), which requires
another \(O(N)\) operations. Thus, each triple sum of the form in
\eqref{eq:cumulative_triple_sum} can be evaluated in \(O(N)\) operations.
Since \(q\) and \(d\) are fixed and each coordinate
involves only a fixed number of such terms, the contributions of a fixed
ordered pair \((i,j)\) to \(\bs h_N(\bs\beta)\) and
\(\dot{\bs h}_N(\bs\beta)\) can both be evaluated in \(O(N)\) operations.
There are \(N(N-1)=O(N^2)\) ordered pivotal pairs \((i,j)\), so evaluating both
\(\bs h_N(\bs\beta)\) and \(\dot{\bs h}_N(\bs\beta)\) requires
\(O(N^3)\) operations.
The resulting \(O(N^3)\) algorithm substantially improves the practical
feasibility of pentad-GMM for large networks.

Appendix~\ref{sec:simulation} evaluates the finite-sample behavior of the
implementation described in this section. Across dense and sparse networks
with symmetric and asymmetric dyadic covariates, root mean squared error (RMSE) declines as \(N\)
increases, while marginal coverage based on the universal covariance estimator
is broadly close to its nominal level. These findings support the practical
feasibility of the proposed estimation and inference procedure, although
sampling variability remains larger in the sparse designs.

\section{Empirical Illustration}
\label{sec:empirical}

This section applies the pentad-GMM estimator to study the
\emph{Peer Relationships among Master of Social Work Students} data
\citep{mauldin2020peer}, a longitudinal study of one entering master of social work (MSW) class.
The data combine four waves of network surveys with demographic and academic
records. We focus on the wave-4 academic-discussion network and examine how
academic similarity, shared cohort membership, and a potential partner's
universal--diverse orientation are associated with link formation after allowing
for individual-specific unobserved heterogeneity. We compare the estimates under NTU  logit models 
with those under TU logit models (\citealp{graham2017econometric}).

\subsection{Data and Variables}

The data cover MSW students at a large public university in the southern
United States. Students were surveyed at orientation in July--August 2014,
twice during their first semester in October and November--December 2014, and
at the end of the Fall 2015 semester in December 2015. At waves 2--4,
respondents used a roster of 145 students to report academic-discussion,
friendship, and professional-influence links. We use the wave-4 academic
discussion network as the outcome. Students separately nominate classmates
with whom they report studying, discussing coursework, exchanging feedback on
assignments, or asking questions about homework. We construct an
undirected measure of these bilateral interactions. For each pair \(i<j\), we
record \(L_{ij}=1\) if at least one student reports academic discussion with
the other and set \(L_{ij}=0\) otherwise. The data show only who reported
the interaction, but we do not observe who initiated or proposed it.
The final sample contains 109 students across six educational cohorts with complete
wave-4 academic-discussion information, grade point average (GPA), and scores on the
15-item short form of the Miville--Guzman Universality--Diversity Scale
(M-GUDS-S; hereafter MGUDS). These students form
5,886 unordered pairs, of which 1,246 are academic-discussion links, yielding a
link density of 0.2117.

GPA is calculated over the most recent 60 credit hours and measures prior
academic performance. MGUDS scores measure universal--diverse orientation.
The scale covers diversity of contact,
appreciation of similarities and differences, and comfort with differences.
Higher scores indicate a stronger universal--diverse orientation, reflecting
greater openness to engaging with people from different backgrounds and
perspectives \citep{fuertes2000factor}. Raw GPA and MGUDS values are measured on different
scales and are difficult to compare directly. In the subsequent analysis, we therefore use inverse-normal
rank scores \citep[see][for a similar transformation]{karadja2017richer}, which preserve the ordering of
observations within each measure and map the two measures onto a common
normal-score scale.
Table~\ref{tab:msw-estimation-summary} reports descriptive statistics for the
variables used in the analysis.

\begin{table}[H]
\centering
\caption{Sample summary statistics}
\label{tab:msw-estimation-summary}
\footnotesize
\begin{threeparttable}
\begin{tabular}{lrrrrr}
\toprule
Variable & Observations & Mean & SD & Min & Max \\
\midrule
\multicolumn{6}{l}{\textit{Panel A: Student characteristics}}\\
Raw GPA & 109 & 3.459 & 0.330 & 2.630 & 4.000 \\
Raw MGUDS score & 109 & 72.128 & 6.948 & 54.000 & 89.000 \\
GPA inverse-normal rank & 109 & -0.004 & 0.985 & -2.605 & 1.919 \\
MGUDS inverse-normal rank & 109 & 0.000 & 0.996 & -2.605 & 2.605 \\
\addlinespace
\multicolumn{6}{l}{\textit{Panel B: Network and dyadic characteristics}}\\
Academic-discussion link & 5,886 & 0.212 & 0.409 & 0.000 & 1.000 \\
GPA inverse-normal rank distance & 5,886 & 1.123 & 0.824 & 0.000 & 4.524 \\
MGUDS inverse-normal rank distance & 5,886 & 1.130 & 0.840 & 0.000 & 5.211 \\
Same cohort & 5,886 & 0.170 & 0.375 & 0.000 & 1.000 \\
\bottomrule
\end{tabular}
\par\smallskip
\begin{minipage}{\linewidth}
\footnotesize
\textit{Notes:} SD denotes standard deviation. ``Raw'' and ``inverse-normal rank'' denote values before and after the inverse-normal rank transformation, respectively.
\end{minipage}
\end{threeparttable}
\end{table}

\subsection{Results and Discussion}

We consider three specifications of the TU model in \eqref{eq:tu-model}
and the NTU model in \eqref{mod:mod2}, all of which include
GPA inverse-normal rank distance and a same-cohort indicator.
The TU specification and the symmetric
NTU specification use MGUDS inverse-normal rank distance to examine whether students with similar universal--diverse
orientations are more likely to form academic-discussion links.
We also examine whether students with higher MGUDS are more attractive
discussion partners. In the TU model, however, additive MGUDS terms
are absorbed by the individual fixed effects, so their coefficient
cannot be separately identified. This limitation motivates the
asymmetric NTU specification, in which partner MGUDS inverse-normal
rank \(\mathrm{MGUDS}_j\) enters student \(i\)'s utility.
Since \(\mathrm{MGUDS}_j\) varies across potential partners,
it cannot be absorbed by student \(i\)'s fixed effect \(\Gamma_i\).
The resulting regressors are
\begingroup
\renewcommand{\arraystretch}{1.1}
\[
\begin{aligned}
\bs X_{ij}^{\mathrm{TU}}=\bs X_{ij}^{\mathrm{NTU,S}}
&:=\begin{pmatrix}
\abs{\mathrm{MGUDS}_i-\mathrm{MGUDS}_j}\\
\abs{\mathrm{GPA}_i-\mathrm{GPA}_j}\\
\1\{C_i=C_j\}
\end{pmatrix},\  
\bs X_{ij}^{\mathrm{NTU,A}}
:=\begin{pmatrix}
\mathrm{MGUDS}_j\\
\abs{\mathrm{GPA}_i-\mathrm{GPA}_j}\\
\1\{C_i=C_j\}
\end{pmatrix},
\end{aligned}
\]
\endgroup
where \(\mathrm{GPA}_i\) and \(\mathrm{MGUDS}_i\) denote student \(i\)'s
inverse-normal rank scores, and \(C_i\) denotes the student's educational cohort.

For the TU specification, we use the tetrad-logit estimator of
\citet{graham2017econometric} and its dyad-projection covariance. We estimate
both NTU specifications using the procedures in Section~\ref{sec:implementation}.
With three regressors,
\eqref{eq:covariate_only_instrument} stacks three determinant blocks into nine
moment conditions. As the GMM objective function may have local optima, we
follow \citet{davezies2023fixed}, use 200 random starting values, and retain the
parameter vector that yields the lowest value of
\eqref{eq:pair_normalized_criterion}.

\begin{table}[!b]
\centering
\caption{TU and NTU estimates of academic-discussion link formation}
\label{tab:empirical-pentad-gmm}
\footnotesize
\renewcommand{\arraystretch}{1.1}
\begin{threeparttable}
\begin{tabular}{l@{\hspace{1.5em}}ccc}
\toprule
Variable & TU tetrad logit & \multicolumn{2}{c}{NTU pentad-GMM} \\
\cmidrule(lr){2-2}\cmidrule(l){3-4}
& (1) & (2) & (3) \\
\midrule
Partner MGUDS inverse-normal rank
    & -- & -- & \(0.352^{***}\) \\
    & & & \((0.052)\) \\
\addlinespace[0.2em]
MGUDS inverse-normal rank distance
    & \(0.036\) & \(-0.134\) & -- \\
    & \((0.076)\) & \((0.188)\) & \\
\addlinespace[0.2em]
GPA inverse-normal rank distance
    & \(-0.091\) & \(0.366\) & \(-0.339^{**}\) \\
    & \((0.078)\) & \((0.321)\) & \((0.144)\) \\
\addlinespace[0.2em]
Same cohort
    & \(3.538^{***}\) & \(3.298^{***}\) & \(2.508^{**}\) \\
    & \((0.159)\) & \((0.522)\) & \((1.134)\) \\
\bottomrule
\end{tabular}
\par\smallskip
\begin{minipage}{\linewidth}
\footnotesize
\textit{Notes:} Standard errors are in parentheses. GPA and MGUDS ranks use the
inverse-normal transformation. NTU standard errors use the
positive-semidefinite projection of the universal covariance estimator in
\eqref{eq:universalOmega}; TU standard errors use Graham's dyad-projection
covariance estimator. Significance levels are denoted by \(^{***}p<0.01\),
\(^{**}p<0.05\), and \(^{*}p<0.10\).
\end{minipage}
\end{threeparttable}
\end{table}

Columns (1) and (2) of Table~\ref{tab:empirical-pentad-gmm} compare the TU
and symmetric NTU models using the same regressors. Both show that students
in the same cohort are more likely to form academic-discussion links.
Their same-cohort coefficients are
3.538 and 3.298,
respectively, both significant at the \(1\%\) level, whereas neither GPA nor
MGUDS rank distance is statistically significant. In column (3), however,
the asymmetric NTU specification provides evidence of strong academic homophily
alongside a positive association between partner MGUDS rank and willingness
to form discussion links. These findings emerge when partner MGUDS rank
replaces MGUDS rank distance. The coefficient on partner MGUDS rank is 0.352,
significant at the \(1\%\) level. The GPA
rank-distance and same-cohort coefficients are \(-0.339\) and 2.508,
respectively, both significant at the \(5\%\) level.

Together, these findings suggest that students value both a partner's openness
to different perspectives and similarity in academic achievement. Students who are more
receptive to different perspectives may be more approachable partners for
asking questions and exchanging feedback. This interpretation accords with
evidence linking MGUDS scores to teamwork interest and aptitude among college
students \citep{kottke2011additional}.

\section{Conclusion}
\label{sec:conclusion}

This paper develops a pentad differencing strategy for logistic NTU network
formation models with individual fixed effects. Using observed links and
covariates from five-node pentads, we construct moment restrictions that do
not depend on these fixed effects. We then show that, for a broad class of
covariate specifications, this five-node, seven-dyad construction is minimal
and unique up to relabeling, and the corresponding moment restrictions are unique up to rescaling.

These moment restrictions form the basis of the pentad-GMM estimator. A
Hoeffding decomposition over dyads shows how the leading projection changes
with network sparsity, giving rise to different convergence rates and
asymptotic variances. Under the stated conditions, we establish asymptotic
normality in dense, sparse, and ultra-sparse networks. We also propose a
tractable covariance estimator that is consistent under the appropriate
normalization in all three regimes, allowing for unified inference without prior knowledge of sparsity regimes.
For implementation, we provide a computationally efficient algorithm to reduce
the computational cost of the estimator from na\"ive \(O(N^5)\) to \(O(N^3)\)
operations.

Finally, we apply our method to an academic-discussion network among Master
of Social Work students. The asymmetric NTU estimates show academic homophily and a positive association between partner openness to
different perspectives and link formation. These findings illustrate the
empirical value of our method. Future work could extend the proposed differencing strategy to network models
with strategic link interdependence and to other economic settings. 

\clearpage
\appendix
\centerline{\textbf{\LARGE{}Appendices}}

\section{Graph Notation and Terminology}
\label{app:graph_notation}

We collect the graph notation used in the appendices.
For a reference on graph terminology in network econometrics, see
\citet{graham2020network}. A finite simple
undirected graph \(G=(V,E)\) consists of a finite node set \(V\) and a
set \(E\) of unordered dyads joining distinct nodes, with no repeated
dyads. For nodes labeled in \([N]\), recall that
\(\mathcal E_N=\{(u,v):1\le u<v\le N\}\).
When a dyad is written as \((u,v)\), it denotes the ordered representative
\((\min\{u,v\},\max\{u,v\})\in\mathcal E_N\).
For \(A\subseteq E\), let \(d_A(u)\) denote the number of dyads in
\(A\) containing node \(u\). In particular, \(d_E(u)\) is the degree
of \(u\) in \(G=(V,E)\). A node of degree zero is isolated.

A subgraph of \(G=(V,E)\) has a node set contained in \(V\) and a
dyad set contained in \(E\). This containment allows additional dyads
of \(G\) among the selected nodes.
A \emph{path} consists of distinct nodes with a dyad joining each consecutive
pair. A graph is \emph{connected} if any two of its nodes can be joined by a
path. A \emph{connected component} of \(G\) is a subgraph
\(G_j=(V_j,E_j)\), with a nonempty node set \(V_j\subseteq V\) and
\(E_j=\{(u,v)\in E:u,v\in V_j\}\), satisfying two properties:
any two nodes in \(V_j\) can be joined by a path using dyads in
\(E_j\), and no dyad in \(E\) joins a node in \(V_j\) to a node in
\(V\setminus V_j\). An isolated node forms a connected component
by itself. A \emph{cycle} in \(G\) consists of dyads
\(\{(u_1,u_2),\ldots,(u_{\ell-1},u_\ell),(u_\ell,u_1)\}\subseteq E\),
where \(\ell\geq3\) and \(u_1,\ldots,u_\ell\) are distinct nodes.
A \emph{tree} is a connected graph with no cycle.
A nonempty connected graph has at least \(\lvert V\rvert-1\) dyads,
with equality if and only if it is a tree; a tree has a unique path
between any two distinct nodes.

The \emph{complement} of \(G\) has the same node set \(V\) and contains
exactly the dyads absent from \(E\).
For graphs \(G_1=(V_1,E_1)\) and \(G_2=(V_2,E_2)\), their \emph{union} has
node set \(V_1\cup V_2\) and dyad set \(E_1\cup E_2\). Their \emph{intersection}
has node set \(V_1\cap V_2\) and dyad set \(E_1\cap E_2\).
Every node belonging to both graphs is therefore retained in the
intersection. If no dyad containing that node belongs to both graphs,
the node is isolated in the intersection.

Two graphs \(G=(V,E)\) and \(G'=(V',E')\) are \emph{isomorphic},
written as \(G\cong G'\), if there exists a bijection from \(V\) to \(V'\)
such that two nodes are adjacent in \(G\) if and only if their images
are adjacent in \(G'\). A \emph{labeled copy} of a
graph is an isomorphic graph on specified node labels. A copy is
determined by its node set and dyad set and is counted once, regardless
of node ordering.
For any dyad set \(A\subseteq\mathcal E_N\), define
\(V(A):=\{i\in[N]: i\text{ is incident to some }e\in A\}\) and
\(G_A:=(V(A),A)\). Thus \(G_A\) has no isolated nodes, whereas a graph
whose node set is specified separately may include isolated nodes.
For a bijection \(\pi:V(A)\to V(B)\), write
\[
\pi(A):=\{(\min\{\pi(u),\pi(v)\},\max\{\pi(u),\pi(v)\}):(u,v)\in A\}.
\]
Then, formally,
\[
G_A\cong G_B
\quad\Longleftrightarrow\quad
\exists\text{ a bijection }\pi:V(A)\to V(B)
\text{ such that }\pi(A)=B.
\]
Let \([G_A]:=\{G_B:G_B\cong G_A\}\) denote this unlabeled isomorphism
class. Thus \([G_A]\) encodes only the shape of the dyad set \(A\),
not the numerical labels of its nodes.
For a finite graph class \(g\), meaning a finite collection of unlabeled
dyad-graph isomorphism classes, define the labeled dyad-set collection
\(\mathscr C_N(g):=\{A\subseteq\mathcal E_N:[G_A]\in g\}\).
Thus an element \(A\in\mathscr C_N(g)\) is called a labeled copy
of one of the shapes collected in \(g\).

In this paper, the term \emph{pentad} refers specifically to the graph on five distinct nodes
\(i,j,k,l,m\) with dyad set
\(\{(i,j),(i,k),\allowbreak(j,k),(i,l),(j,l),\allowbreak(i,m),(j,m)\}\), up to relabeling.
The nodes \(i,j\) are the pivot nodes, and \(k,l,m\) are the peripheral nodes.

\section{Tractable Covariance Estimator}
\label{sec:feasible_covariance}

We construct a single covariance estimator for the sample moment vector
\(\bs h_N(\bs\beta_0)\) that is consistent, under the appropriate
scaling, in each of the three regimes in
Theorems~\ref{thm:regimeCLT} and~\ref{thm:gmm}. The population target is
\(\Var\{\bs h_N(\bs\beta_0)\mid\bf X,\bs\Gamma\}\). Theorem~\ref{thm:moment_restriction}
and the \(\bf X\)-measurability of \(\bs Z_s\) give
\(\E[\bs\Psi_s(\bs\beta_0)\mid\bf X,\bs\Gamma]=\boldsymbol 0\) for every
\(s\in\mathcal S_N\), and hence
\begin{equation}
\label{eq:target_covariance}
\Var\{\bs h_N(\bs\beta_0)\mid\bf X,\bs\Gamma\}
=
\frac{1}{\abs{\mathcal S_N}^2}
\sum_{s,t\in\mathcal S_N}
\E\!\left[
\bs\Psi_s(\bs\beta_0)\bs\Psi_t(\bs\beta_0)^\top
\mid\bf X,\bs\Gamma
\right].
\end{equation}
When \(E(s)\cap E(t)=\varnothing\), conditional independence and
\(\E[\bs\Psi_s(\bs\beta_0)\mid\bf X,\bs\Gamma]=\boldsymbol 0\)
imply \(\E[\bs\Psi_s(\bs\beta_0)\bs\Psi_t(\bs\beta_0)^\top
\mid\bf X,\bs\Gamma]=\mathbf 0\). Hence, only pentad
pairs (\(s,t\)) sharing at least one dyad can contribute to the conditional covariance.

A natural starting point for estimation is to sum the moment products
\(\bs\Psi_s(\bs\beta_0)\bs\Psi_t(\bs\beta_0)^\top\) over pairs with
\(E(s)\cap E(t)\neq\varnothing\). Since
\(\abs{\mathcal S_N}\asymp N^5\), checking every ordered pentad pair
would involve \(O(N^{10})\) comparisons. To reduce the computational cost,
we first group pentads by a shared dyad. Recall that
\(\mathcal E_N=\{(u,v):1\le u<v\le N\}\) denotes the set of all dyads.
For each \(e\in\mathcal E_N\), we sum the moments of the pentads containing
\(e\) and multiply the sum by its transpose. Summing over
\(e\in\mathcal E_N\) and expanding gives
\(\sum_{e\in\mathcal E_N}
\left(\sum_{s\in\mathcal S_N:\,e\in E(s)}\bs\Psi_s(\bs\beta_0)\right)\allowbreak
\left(\sum_{t\in\mathcal S_N:\,e\in E(t)}\bs\Psi_t(\bs\beta_0)\right)^\top
\allowbreak=
\sum_{s,t\in\mathcal S_N}\abs{E(s)\cap E(t)}\allowbreak
\bs\Psi_s(\bs\beta_0)\bs\Psi_t(\bs\beta_0)^\top\).
Thus, summing over \(e\in\mathcal E_N\) counts each pentad pair \((s,t)\)
a total of \(\abs{E(s)\cap E(t)}\) times, whereas each ordered pair appears
only once in the target covariance \eqref{eq:target_covariance}.

We correct the repeated counting by extending the grouping to shared
subgraphs other than dyads. For example, for each triangle, we sum the
moments of all pentads containing that triangle and multiply the sum by its
transpose. Summing these matrices over all triangles includes the moment
product \(\bs\Psi_s(\bs\beta_0)\bs\Psi_t(\bs\beta_0)^\top\) once for
each triangle common to the two pentad graphs. Thus, the number of times
\(\bs\Psi_s(\bs\beta_0)\bs\Psi_t(\bs\beta_0)^\top\)
appears in the sum over all triangles equals the number of triangles
shared by the pentad graphs of \(s\) and \(t\). We apply this construction
to other subgraph types. For each type, the matrix sum counts each pair
\((s,t)\) once for every copy of that type shared by the pentad graphs
of \(s\) and \(t\). All these counts are determined by the
graph formed by their shared nodes and dyads, which we call their
\emph{common graph}.
For \(s=(i,j,k,l,m)\in\mathcal S_N\), write \(V_s:=\{i,j,k,l,m\}\), so
\((V_s,E(s))\) is its pentad graph and
\((V_s\cap V_t,E(s)\cap E(t))\) is the common graph of \(s\) and \(t\).
The common graph retains all shared nodes, including those incident to
no shared dyad.
Adding or subtracting the matrix sums therefore assigns each pair a weight
equal to the same signed combination of subgraph counts in its common graph.

When the common graph of a pentad pair is a tree, subtracting its dyad
count from its node count gives exactly the desired weight of one.
When the common graph contains cycles, additional corrections are needed;
for pentad graphs, these corrections can be expressed through triangle,
four-cycle, and diamond counts. We therefore
use the following five graphs, collected in \(\mathcal G\):
\mbox{\(G_{\mathrm{node}}\)
(\taxgraph[baseline=-0.5ex]{\node[taxNode] at (0,0) {};})},
\mbox{\(G_{\mathrm{dyad}}\)
(\taxgraph[baseline=-0.5ex]{
  \node[taxNode] (a) at (-.6,0) {};
  \node[taxNode] (b) at (.6,0) {};
  \draw[taxEdge] (a)--(b);
})},
\mbox{\(G_{\mathrm{triangle}}\)
(\taxgraph[baseline=-0.5ex]{
  \node[taxNode] (a) at (-.55,-.36) {};
  \node[taxNode] (b) at (.55,-.36) {};
  \node[taxNode] (c) at (0,.45) {};
  \draw[taxEdge] (a)--(b)--(c)--(a);
})},
\mbox{\(G_{\text{4-cycle}}\)
(\taxgraph[baseline=-0.5ex]{
  \node[taxNode] (a) at (-.55,-.4) {};
  \node[taxNode] (b) at (-.55,.4) {};
  \node[taxNode] (c) at (.55,.4) {};
  \node[taxNode] (d) at (.55,-.4) {};
  \draw[taxEdge] (a)--(b)--(c)--(d)--(a);
})}, and
\mbox{\(G_{\mathrm{diamond}}\)
(\taxgraph[baseline=-0.5ex]{
  \node[taxNode] (a) at (-.7,0) {};
  \node[taxNode] (b) at (0,.48) {};
  \node[taxNode] (c) at (.7,0) {};
  \node[taxNode] (d) at (0,-.48) {};
  \draw[taxEdge] (a)--(b)--(c)--(d)--(a);
  \draw[taxEdge] (a)--(c);
})}.
For \(G\in\mathcal G\) and \(s,t\in\mathcal S_N\), let
\(\zeta_G(s,t)\) denote the number of labeled copies of \(G\) contained in
their common graph \((V_s\cap V_t,E(s)\cap E(t))\); e.g.,
\(\zeta_{G_{\mathrm{dyad}}}(s,t)=\abs{E(s)\cap E(t)}\).
For example, under the subgraph convention in
Appendix~\ref{app:graph_notation}, a
\mbox{diamond
(\taxgraph[baseline=-0.5ex]{
  \node[taxNode] (a) at (-.7,0) {};
  \node[taxNode] (b) at (0,.48) {};
  \node[taxNode] (c) at (.7,0) {};
  \node[taxNode] (d) at (0,-.48) {};
  \draw[taxEdge] (a)--(b)--(c)--(d)--(a);
  \draw[taxEdge] (a)--(c);
})} contains a
\mbox{four-cycle
(\taxgraph[baseline=-0.5ex]{
  \node[taxNode] (a) at (-.55,-.4) {};
  \node[taxNode] (b) at (-.55,.4) {};
  \node[taxNode] (c) at (.55,.4) {};
  \node[taxNode] (d) at (.55,-.4) {};
  \draw[taxEdge] (a)--(b)--(c)--(d)--(a);
})} despite its extra diagonal dyad.

We choose these five graphs because the signed sum of their
counts assigns total weight one to every pentad pair \((s,t)\in\mathcal S_N^2\) whose
common graph \((V_s\cap V_t,E(s)\cap E(t))\) is connected and contains a dyad, as shown in the proof of
Lemma~\ref{lem:graph_count_covariance_approximation}:
\[
\zeta_{G_{\mathrm{node}}}(s,t)
-\zeta_{G_{\mathrm{dyad}}}(s,t)
+\zeta_{G_{\mathrm{triangle}}}(s,t)
+\zeta_{G_{\text{4-cycle}}}(s,t)
-\zeta_{G_{\mathrm{diamond}}}(s,t)=1.
\]
Pairs with no common dyad have zero conditional covariance, even if they
share nodes. For pairs whose common graph is disconnected and contains
a dyad, the signed combination need not assign weight one, but the
resulting aggregate error in conditional mean is negligible under the
scaling for each regime in Table~\ref{tab:boundaryregimes}, as shown in
the proof of Lemma~\ref{lem:graph_count_covariance_approximation}.

To construct the covariance estimator, we use the subgraph counts
\(\zeta_G(s,t)\) to weight products of pentad moments. We center each
moment at the sample mean by setting
\(\widehat{\bs\Psi}_s(\bs\beta)
:=\bs\Psi_s(\bs\beta)-\bs h_N(\bs\beta)\).
For each \(G\in\mathcal G\), define the matrix
\begin{equation}
\label{eq:graph_count_component}
\widehat{\bf\Omega}_N(G;\bs\beta)
:=\frac{1}{\abs{\mathcal S_N}^2}
\sum_{s,t\in\mathcal S_N}
\zeta_G(s,t)
\widehat{\bs\Psi}_s(\bs\beta)
\widehat{\bs\Psi}_t(\bs\beta)^\top.
\end{equation}
By the definition of \(\zeta_G(s,t)\), the pairwise sum in
\eqref{eq:graph_count_component} can be rearranged exactly as
\begin{equation}
\label{eq:graph_count_gram}
\widehat{\bf\Omega}_N(G;\bs\beta)
=
\frac{1}{\abs{\mathcal S_N}^2}
\sum_{\text{labeled copies }H\text{ of }G}
\left(
\sum_{s\in\mathcal S_N:\,H\subseteq(V_s,E(s))}
\widehat{\bs\Psi}_s(\bs\beta)
\right)
\left(
\sum_{t\in\mathcal S_N:\,H\subseteq(V_t,E(t))}
\widehat{\bs\Psi}_t(\bs\beta)
\right)^\top.
\end{equation}
Thus, we evaluate \eqref{eq:graph_count_component} by summing the centered
pentad moments for each labeled copy \(H\) of \(G\), multiplying each sum
by its transpose, and summing over \(H\). This computes
\(\widehat{\bf\Omega}_N(G;\bs\beta)\) without enumerating ordered pentad
pairs and shows that \(\widehat{\bf\Omega}_N(G;\bs\beta)\) is positive semidefinite.

We combine these five matrices to define the tractable universal covariance
estimator
\begin{equation}
\label{eq:universalOmega}
\begin{aligned}
\bOmegaUniv(\bs\beta)
&:=\widehat{\bf\Omega}_N(G_{\mathrm{node}};\bs\beta)
-\widehat{\bf\Omega}_N(G_{\mathrm{dyad}};\bs\beta)
+\widehat{\bf\Omega}_N(G_{\mathrm{triangle}};\bs\beta)\\
&\quad+\widehat{\bf\Omega}_N(G_{\text{4-cycle}};\bs\beta)
-\widehat{\bf\Omega}_N(G_{\mathrm{diamond}};\bs\beta)
+\frac{10}{N}\sum_{G\in\mathcal G}
\widehat{\bf\Omega}_N(G;\bs\beta).
\end{aligned}
\end{equation}
In \eqref{eq:universalOmega}, the dyad and diamond matrices enter with
negative signs, and \(\bOmegaUniv(\bs\beta)\) need not be positive
semidefinite. The term
\((10/N)\sum_{G\in\mathcal G}\widehat{\bf\Omega}_N(G;\bs\beta)\)
is a finite-sample correction that is asymptotically negligible under each
regime-specific scaling, as shown in the proof of
Lemma~\ref{lem:graph_count_covariance_approximation}.
The following theorem establishes the asymptotic validity of inference
based on the covariance estimator in \eqref{eq:universalOmega}.

\begin{thm}[Pentad-GMM inference]
\label{thm:tractable_gmm}
Suppose Assumptions~\ref{ass:dyad_ind}, \ref{ass:bounded_design},
\ref{ass:sparse_envelope}, \ref{ass:regimeSigma}, and
\ref{ass:consistency} hold, and
\(N^5\rho_N^4\to\infty\). Let \((a_N,\bf\Sigma)\) be defined as in
Theorem~\ref{thm:gmm} for the regime under consideration. For every fixed vector
\(\bs{c}\in\mathbb R^q\) such that
\(\bs{c}^\top\bf\Sigma \bs{c}>0\),
\[
\frac{\bs{c}^\top\bs h_N(\bs\beta_0)}
{\sqrt{\bs{c}^\top\bOmegaUniv(\bs\beta_0)\bs{c}}}
\xrightarrow{D}N(0,1).
\]
Define \(\widehat{\mathbf{V}}_N\) by the formula in
Corollary~\ref{cor:generic_gmm_inference} with
\(\widehat{\bf\Omega}_N
=\bOmegaUniv(\widehat{\bs\beta}_{\mathrm{GMM}})\).
Then \(a_N^2\rho_N^8\widehat{\mathbf{V}}_N\xrightarrow{P}\mathbf{V}_0\). Consequently, for every
fixed \(\bs c\in\mathbb R^d\) with \(\bs c^\top \mathbf{V}_0\bs c>0\),
\[
\frac{\bs c^\top(\widehat{\bs\beta}_{\mathrm{GMM}}-\bs\beta_0)}
{\sqrt{\bs c^\top\widehat{\mathbf{V}}_N\bs c}}
\xrightarrow{D}N(0,1).
\]

\end{thm}
The proof of Theorem~\ref{thm:tractable_gmm} and the supporting lemmas
are given in Appendix~\ref{app:covariance_proofs}.
The corresponding asymptotic \(100(1-\alpha)\%\) confidence interval for
\(\bs c^\top\bs\beta_0\) is
\(\bs c^\top\widehat{\bs\beta}_{\mathrm{GMM}}
\pm z_{1-\alpha/2}\sqrt{\bs c^\top\widehat{\mathbf{V}}_N\bs c}\),
where \(z_p\) denotes the \(p\)th quantile of the standard normal distribution.
The estimator \(\bOmegaUniv(\widehat{\bs\beta}_{\mathrm{GMM}})\) is regime-adaptive and can be implemented
without prior knowledge of the sparsity regime, preliminary estimation of
\(\rho_N\), or selection of a regime-specific covariance estimator.
When a positive semidefinite estimate is required, one may project
\(\bOmegaUniv(\widehat{\bs\beta}_{\mathrm{GMM}})\) onto the positive semidefinite cone by replacing
its negative eigenvalues with zero \citep{cameron2011robust,politis2011higher}. Since the covariance limits in
Lemma~\ref{lem:cov_consistency} are positive semidefinite, this
projection preserves consistency.

Our implementation uses the
estimator
\(\widehat{\bs\beta}^{\dagger}\) defined in
Section~\ref{sec:implementation}.
Lemmas~\ref{lem:uniform-convergence-population},
\ref{lem:cov_equicont}, and~\ref{lem:pair_normalization} imply that the
conclusions of Theorem~\ref{thm:tractable_gmm} remain valid for
\(\widehat{\bs\beta}^{\dagger}\) when \(\widehat{\mathbf{V}}_N\) is constructed with the
universal covariance estimator
\(\bOmegaUniv(\widehat{\bs\beta}^{\dagger})\) and sample Jacobian
\(\dot{\bs h}_N(\widehat{\bs\beta}^{\dagger})\) in place of
\(\bOmegaUniv(\widehat{\bs\beta}_{\mathrm{GMM}})\) and
\(\dot{\bs h}_N(\widehat{\bs\beta}_{\mathrm{GMM}})\), respectively.

\section{Simulation}
\label{sec:simulation}

This section evaluates the finite-sample performance of the pentad-GMM
estimator. The designs vary network density, symmetry of the dyadic
covariates, and network size. We consider \(N\in\{100,200\}\) and use 500
Monte Carlo replications in each of the eight resulting designs. The table
reports mean and median bias, Monte Carlo standard deviation (SD), root mean
squared error (RMSE), and coverage probabilities. The nominal coverage level
is 95\%.
The simulations use the covariate-only instrument in
\eqref{eq:covariate_only_instrument} and the normalized criterion in
\eqref{eq:pair_normalized_criterion}. The confidence intervals use
the positive-semidefinite projection of the universal covariance estimator
in \eqref{eq:universalOmega}.

We use the same underlying data-generating process in all designs. For each
node \(i\), we draw \(X_{i1}\) and \(X_{i2}\) independently from a standard
normal distribution truncated to \([-2,2]\) and draw
\(\xi_i\sim \mathrm{Uniform}[-0.5,0.5]\). We set \(\bs\beta_0=(\beta_1,\beta_2)^\top=(1,1)^\top\) and
\(\Gamma_i=\mu_N+0.2X_{i2}+0.8\xi_i\), allowing the fixed effects to be
correlated with \(X_{i2}\). The parameter \(\mu_N\) will be determined by the  network density in each design.
The directed shocks \(\epsilon_{ij}\) are drawn
independently from the standard logistic distribution and are independent of
\((X_{i1},X_{i2},\Gamma_i)\). The observed links are then generated by the
mutual-consent model in \eqref{mod:mod2}.

In addition to network size, we vary covariate symmetry and network density.
First, within each density design, we compare symmetric and asymmetric dyadic
covariates. The asymmetric specification allows a given pair of covariates (\(\bs X_i,\bs X_j\)) to enter the two
unilateral consent indices differently. In the symmetric specification,
\(X_{ij,k}=X_{ji,k}=|X_{ik}-X_{jk}|\) for \(k=1,2\). In the asymmetric
specification, we set 
\(X_{ij,1}=2/3|X_{i1}-X_{j1}|\) if \(X_{i1}\leq X_{j1}\), \(X_{ij,1}=4/3|X_{i1}-X_{j1}|\) otherwise, and
\(X_{ij,2}=|X_{i2}-X_{j2}|\).  The
underlying distance between \(X_{i1}\) and \(X_{j1}\) is unchanged, but its contribution
to the two unilateral link indices differs by a factor of two. Second, we
generate dense and sparse networks by calibrating \(\mu_N\) while holding the
covariate design fixed. The dense design sets \(\mu_N=-1.185\), giving mean
link densities of 0.467 in the symmetric specification and 0.454 in the
asymmetric specification. The sparse design calibrates \(\mu_N\) to obtain
link density \(10/N\). At \(N=100\), \(\mu_N\) is approximately \(-3.289\) in the
symmetric specification and \(-3.217\) in the asymmetric specification. At
\(N=200\), the corresponding values are \(-3.944\) and \(-3.871\). Mean
realized density is 0.100 at \(N=100\) and 0.050 at \(N=200\), so average
degree remains close to ten.

Table~\ref{tab:simulation-results} reports the results for the dense and sparse
designs. In the dense designs, increasing \(N\) from 100
to 200 reduces RMSE from 0.112--0.158 to 0.062--0.089. Mean and median bias
decline in every setting; at \(N=200\), their largest absolute values are 0.019
and 0.023, respectively. Coverage probabilities range from 0.948 to 0.966
at \(N=100\) and from 0.922 to 0.948 at \(N=200\).
In the sparse designs, increasing \(N\) from 100 to 200 reduces RMSE
from 0.195--0.203 to 0.180--0.195. Coverage probabilities range from
0.958 to 0.992 at \(N=100\) and from 0.942 to 0.976 at \(N=200\).
As shown by the results, RMSE declines more slowly with \(N\) than in the dense designs.
To understand this pattern, note that
Lemma~\ref{lem:det-row-supports} shows that every monomial in
\(\psi_{ijklm}\) contains between four and seven link indicators.
For this reason, each term in the expansion of \(\psi_{ijklm}\) requires several links
to be present simultaneously.
In sparse networks, these informative pentad configurations occur less
frequently, and sampling variability is correspondingly larger, reducing
the precision of the pentad-GMM estimator.

\begin{table}[!t]
\centering
\caption{Simulation Results for the Pentad-GMM estimator in the dense and sparse designs}
\label{tab:simulation-results}
\footnotesize
\renewcommand{\arraystretch}{0.95}
\setlength{\tabcolsep}{3.5pt}
\begin{tabular*}{0.82\linewidth}{@{\extracolsep{\fill}}lcccccccc@{}}
\toprule
&\multicolumn{4}{c}{Dense designs}&\multicolumn{4}{c}{Sparse designs}\\
\cmidrule(lr){2-5}\cmidrule(lr){6-9}
&\multicolumn{2}{c}{Symmetric}&\multicolumn{2}{c}{Asymmetric}
&\multicolumn{2}{c}{Symmetric}&\multicolumn{2}{c}{Asymmetric}\\
\cmidrule(lr){2-3}\cmidrule(lr){4-5}\cmidrule(lr){6-7}\cmidrule(lr){8-9}
\(N=100\)&\(\beta_1\)&\(\beta_2\)&\(\beta_1\)&\(\beta_2\)
&\(\beta_1\)&\(\beta_2\)&\(\beta_1\)&\(\beta_2\)\\
\midrule
Mean Bias & 0.043 & 0.043 & 0.070 & 0.040 & 0.115 & 0.112 & 0.099 & 0.093 \\
Median Bias & 0.039 & 0.037 & 0.050 & 0.045 & 0.164 & 0.136 & 0.110 & 0.109 \\
SD & 0.123 & 0.125 & 0.141 & 0.105 & 0.168 & 0.159 & 0.169 & 0.175 \\
RMSE & 0.131 & 0.132 & 0.158 & 0.112 & 0.203 & 0.195 & 0.196 & 0.198 \\
Coverage Probability & 0.956 & 0.948 & 0.966 & 0.950 & 0.958 & 0.972 & 0.992 & 0.978 \\
\midrule
\(N=200\)&\(\beta_1\)&\(\beta_2\)&\(\beta_1\)&\(\beta_2\)
&\(\beta_1\)&\(\beta_2\)&\(\beta_1\)&\(\beta_2\)\\
\midrule
Mean Bias & 0.000 & -0.002 & 0.019 & 0.010 & 0.093 & 0.084 & 0.116 & 0.097 \\
Median Bias & 0.008 & 0.004 & 0.023 & 0.013 & 0.096 & 0.081 & 0.138 & 0.100 \\
SD & 0.089 & 0.088 & 0.086 & 0.061 & 0.160 & 0.159 & 0.157 & 0.160 \\
RMSE & 0.089 & 0.088 & 0.088 & 0.062 & 0.185 & 0.180 & 0.195 & 0.187 \\
Coverage Probability & 0.934 & 0.948 & 0.924 & 0.922 & 0.942 & 0.964 & 0.976 & 0.962 \\
\bottomrule
\end{tabular*}
\par\smallskip
\begin{minipage}{0.82\linewidth}
\footnotesize
\textit{Notes:} Mean and median bias are the componentwise mean and median of
\(\widehat{\bs\beta}-\bs\beta_0\) across replications. SD is the sample
standard deviation of the estimates, and RMSE is the square root of the
average squared estimation error. Coverage probability is the fraction of
replications in which the nominal 95\% confidence interval contains the true coefficient.
\end{minipage}
\end{table}
\FloatBarrier

\section{Construction and Properties of the Pentad Moment}
\label{app:finiteproofs}

\subsection{Determinant Construction and Proof of Theorem~\ref{thm:moment_restriction}}
\label{app:detconstruction}

We first prove the linear relation in Lemma~\ref{lem:lem1} and define
the feasible coefficients. We then prove
Theorem~\ref{thm:moment_restriction} by combining
\(\det(\mathbf M_{ijklm}(\bs{\beta}_{0}))=0\) with conditional dyad
independence. Finally, we record the monomial
supports used in the sparse-regime and covariance arguments.
Write \(\alpha_u:=e^{-\Gamma_u}\) and
\(P_{uv}:=P_{uv}(\bs{\beta}_{0})\) as in Section~\ref{sec:model}.
For any parameter value \(\bs\beta\), let
\(W_{uv}(\bs\beta):=\exp(-\bs X_{uv}^\top\bs\beta)\). For distinct nodes
\(i,j,r\), define
\begin{equation}
\label{eq:delta-app}
\Delta_{ijr}(\bs\beta)
:=W_{ir}(\bs\beta)W_{jr}(\bs\beta)
\{W_{rj}(\bs\beta)-W_{ri}(\bs\beta)\}.
\end{equation}
\begin{proof}[Proof of Lemma~\ref{lem:lem1}]
Fix distinct nodes \(i,j,r\). We eliminate \(\alpha_r\) using the
\((i,r)\) and \((j,r)\) identities in \eqref{eq:triangle-p-identities},
with \(k\) replaced by \(r\), and then eliminate
\(\alpha_i\alpha_j\) using the \((i,j)\) identity.
Solving the \((i,r)\) and \((j,r)\) identities for \(\alpha_r\) and
equating the resulting expressions gives
\[
\begin{aligned}
0={}&(1-P_{jr}-P_{jr}W_{jr}(\bs{\beta}_{0})\alpha_j)
\{P_{ir}W_{ir}(\bs{\beta}_{0})W_{ri}(\bs{\beta}_{0})\alpha_i
+P_{ir}W_{ri}(\bs{\beta}_{0})\}\\
&-(1-P_{ir}-P_{ir}W_{ir}(\bs{\beta}_{0})\alpha_i)
\{P_{jr}W_{jr}(\bs{\beta}_{0})W_{rj}(\bs{\beta}_{0})\alpha_j
+P_{jr}W_{rj}(\bs{\beta}_{0})\}.
\end{aligned}
\]
Expanding and collecting the terms in \(\alpha_i\), \(\alpha_j\),
\(\alpha_i\alpha_j\), and the constant term, with
\(\Delta_{ijr}(\bs{\beta}_{0})\) defined in \eqref{eq:delta-app}, gives
\begin{equation}
\label{eq:triangle-expanded-app}
\begin{aligned}
0={}&\{P_{ir}(1-P_{jr})W_{ir}(\bs{\beta}_{0})W_{ri}(\bs{\beta}_{0})
+P_{ir}P_{jr}W_{ir}(\bs{\beta}_{0})W_{rj}(\bs{\beta}_{0})\}\alpha_{i}\\
&-\{P_{jr}(1-P_{ir})W_{jr}(\bs{\beta}_{0})W_{rj}(\bs{\beta}_{0})
+P_{ir}P_{jr}W_{jr}(\bs{\beta}_{0})W_{ri}(\bs{\beta}_{0})\}\alpha_{j}\\
&+P_{ir}P_{jr}\Delta_{ijr}(\bs{\beta}_{0})\alpha_{i}\alpha_{j}+P_{ir}(1-P_{jr})W_{ri}(\bs{\beta}_{0})
-P_{jr}(1-P_{ir})W_{rj}(\bs{\beta}_{0}).
\end{aligned}
\end{equation}
The \((i,j)\) identity in \eqref{eq:triangle-p-identities} gives
\[
\alpha_i\alpha_j
=\frac{1-P_{ij}-P_{ij}W_{ij}(\bs{\beta}_{0})\alpha_i
-P_{ij}W_{ji}(\bs{\beta}_{0})\alpha_j}
{P_{ij}W_{ij}(\bs{\beta}_{0})W_{ji}(\bs{\beta}_{0})}.
\]
Substitute for \(\alpha_i\alpha_j\) in
\eqref{eq:triangle-expanded-app} and multiply by \(P_{ij}\).
Collecting the coefficients of \(P_{ij}\alpha_i\),
\(P_{ij}\alpha_j\), and the constant term gives
\begin{equation}
\label{eq:triangle-coefficients-app}
\begin{aligned}
b_{ijr}(\bs{\beta}_{0})
&:=P_{ir}\Bigg[(1-P_{jr})W_{ir}(\bs{\beta}_{0})W_{ri}(\bs{\beta}_{0})
 +P_{jr}\left\{W_{ir}(\bs{\beta}_{0})W_{rj}(\bs{\beta}_{0})
-\frac{\Delta_{ijr}(\bs{\beta}_{0})}{W_{ji}(\bs{\beta}_{0})}\right\}\Bigg],\\[1mm]
c_{ijr}(\bs{\beta}_{0})
&:=-P_{jr}\Bigg[(1-P_{ir})W_{jr}(\bs{\beta}_{0})W_{rj}(\bs{\beta}_{0})
 +P_{ir}\left\{W_{jr}(\bs{\beta}_{0})W_{ri}(\bs{\beta}_{0})
+\frac{\Delta_{ijr}(\bs{\beta}_{0})}{W_{ij}(\bs{\beta}_{0})}\right\}\Bigg],\\[1mm]
d_{ijr}(\bs{\beta}_{0})
&:=P_{ij}\{P_{ir}(1-P_{jr})W_{ri}(\bs{\beta}_{0})
-P_{jr}(1-P_{ir})W_{rj}(\bs{\beta}_{0})\}\\
&\quad +(1-P_{ij})P_{ir}P_{jr}
\frac{\Delta_{ijr}(\bs{\beta}_{0})}
{W_{ij}(\bs{\beta}_{0})W_{ji}(\bs{\beta}_{0})}.
\end{aligned}
\end{equation}
Thus \(b_{ijr}(\bs{\beta}_{0})P_{ij}\alpha_i
+c_{ijr}(\bs{\beta}_{0})P_{ij}\alpha_j
+d_{ijr}(\bs{\beta}_{0})=0\),
which proves the lemma.
\end{proof}

Writing the linear relation in Lemma~\ref{lem:lem1} in terms of
\(P_{ij}\alpha_i\) and
\(P_{ij}\alpha_j\) keeps \(b_{ijr}(\bs{\beta}_{0})\) and
\(c_{ijr}(\bs{\beta}_{0})\) free of the pivot probability
\(P_{ij}\), while \(d_{ijr}(\bs{\beta}_{0})\) is affine in \(P_{ij}\).
Consequently, replacing probabilities by link indicators puts the pivot
link \(L_{ij}\) only in the third column of the feasible matrix.

Fix distinct nodes \(i,j,k,l,m\). For a generic parameter value
\(\bs\beta\) and each \(r\in\{k,l,m\}\), define
\begin{equation}
\label{eq:feasible-coefficients-app}
\begin{aligned}
\widetilde b_{ijr}(\bs\beta)
&:=L_{ir}\Bigg[(1-L_{jr})W_{ir}(\bs\beta)W_{ri}(\bs\beta)
+L_{jr}\left\{W_{ir}(\bs\beta)W_{rj}(\bs\beta)
-\frac{\Delta_{ijr}(\bs\beta)}{W_{ji}(\bs\beta)}\right\}\Bigg],\\
\widetilde c_{ijr}(\bs\beta)
&:=-L_{jr}\Bigg[(1-L_{ir})W_{jr}(\bs\beta)W_{rj}(\bs\beta)
+L_{ir}\left\{W_{jr}(\bs\beta)W_{ri}(\bs\beta)
+\frac{\Delta_{ijr}(\bs\beta)}{W_{ij}(\bs\beta)}\right\}\Bigg],\\
\widetilde d_{ijr}(\bs\beta)
&:=L_{ij}\Big[(1-L_{jr})L_{ir}W_{ri}(\bs\beta)
-(1-L_{ir})L_{jr}W_{rj}(\bs\beta)\Big] +(1-L_{ij})
\frac{L_{ir}L_{jr}\Delta_{ijr}(\bs\beta)}
{W_{ij}(\bs\beta)W_{ji}(\bs\beta)}.
\end{aligned}
\end{equation}
Assemble these coefficients as
\begin{equation}
\label{eq:feasible-matrix-app}
\widetilde{\mathbf M}_{ijklm}(\bs\beta)
:=
\begin{pmatrix}
\widetilde b_{ijk}(\bs\beta) & \widetilde c_{ijk}(\bs\beta) & \widetilde d_{ijk}(\bs\beta)\\
\widetilde b_{ijl}(\bs\beta) & \widetilde c_{ijl}(\bs\beta) & \widetilde d_{ijl}(\bs\beta)\\
\widetilde b_{ijm}(\bs\beta) & \widetilde c_{ijm}(\bs\beta) & \widetilde d_{ijm}(\bs\beta)
\end{pmatrix}.
\end{equation}
At \(\bs{\beta}_{0}\), \eqref{eq:feasible-coefficients-app} replaces each
\(P_{uv}\) and \(1-P_{uv}\) in \eqref{eq:triangle-coefficients-app}
with \(L_{uv}\) and \(1-L_{uv}\), respectively.

\begin{proof}[Proof of Theorem~\ref{thm:moment_restriction}]
Applying Lemma~\ref{lem:lem1} to \(r=k,l,m\) gives
\begin{equation}
\underbrace{\begin{bmatrix}
b_{ijk}(\bs{\beta}_{0}) & c_{ijk}(\bs{\beta}_{0}) & d_{ijk}(\bs{\beta}_{0})\\
b_{ijl}(\bs{\beta}_{0}) & c_{ijl}(\bs{\beta}_{0}) & d_{ijl}(\bs{\beta}_{0})\\
b_{ijm}(\bs{\beta}_{0}) & c_{ijm}(\bs{\beta}_{0}) & d_{ijm}(\bs{\beta}_{0})
\end{bmatrix}}_{\mathbf M_{ijklm}(\bs{\beta}_{0})}
\begin{bmatrix}
P_{ij}\alpha_i\\[4pt]
P_{ij}\alpha_j\\[4pt]
1
\end{bmatrix}
=
\begin{bmatrix}
0\\[4pt]
0\\[4pt]
0
\end{bmatrix}.
\label{eq:coeff_matrix}
\end{equation}
Since \((P_{ij}\alpha_i,P_{ij}\alpha_j,1)^\top\) is nonzero,
\eqref{eq:coeff_matrix} implies
\(\det(\mathbf M_{ijklm}(\bs{\beta}_{0}))=0\).

To show that the feasible determinant has conditional mean
\(\det(\mathbf M_{ijklm}(\bs{\beta}_{0}))\), expand
\(\psi_{ijklm}(\bs{\beta}_{0})
=\det(\widetilde{\mathbf M}_{ijklm}(\bs{\beta}_{0}))\) along the third column:
\[
\begin{aligned}
\det(\widetilde{\mathbf M}_{ijklm}(\bs{\beta}_{0}))
&=\widetilde d_{ijk}(\bs{\beta}_{0})
(\widetilde b_{ijl}(\bs{\beta}_{0})\widetilde c_{ijm}(\bs{\beta}_{0})
-\widetilde b_{ijm}(\bs{\beta}_{0})\widetilde c_{ijl}(\bs{\beta}_{0}))\\
&\quad-\widetilde d_{ijl}(\bs{\beta}_{0})
(\widetilde b_{ijk}(\bs{\beta}_{0})\widetilde c_{ijm}(\bs{\beta}_{0})
-\widetilde b_{ijm}(\bs{\beta}_{0})\widetilde c_{ijk}(\bs{\beta}_{0}))\\
&\quad+\widetilde d_{ijm}(\bs{\beta}_{0})
(\widetilde b_{ijk}(\bs{\beta}_{0})\widetilde c_{ijl}(\bs{\beta}_{0})
-\widetilde b_{ijl}(\bs{\beta}_{0})\widetilde c_{ijk}(\bs{\beta}_{0})).
\end{aligned}
\]
Each of the six products selects one entry from each column and one
from each row. Conditional on \((\bf X,\bs\Gamma)\), the weights
\(W_{uv}(\bs{\beta}_{0})\) are fixed. For a permutation
\((r_B,r_C,r_D)\) of \((k,l,m)\),
\(\widetilde b_{ijr_B}(\bs{\beta}_{0})\) depends only on
\(L_{ir_B},L_{jr_B}\),
\(\widetilde c_{ijr_C}(\bs{\beta}_{0})\) depends only on
\(L_{ir_C},L_{jr_C}\), and
\(\widetilde d_{ijr_D}(\bs{\beta}_{0})\) depends only on
\(L_{ij},L_{ir_D},L_{jr_D}\).
The three dyad sets are disjoint, so Assumption~\ref{ass:dyad_ind}
makes the three selected entries independent conditional on
\((\bf X,\bs\Gamma)\).

Each entry in
\eqref{eq:feasible-coefficients-app} is multilinear in its link indicators.
Using \(\E[L_{uv}\mid\bf X,\bs\Gamma]=P_{uv}\) and
\(\E[1-L_{uv}\mid\bf X,\bs\Gamma]=1-P_{uv}\), conditional dyad
independence therefore gives
\[
\begin{aligned}
&\E\!\left[
\widetilde b_{ijr_B}(\bs{\beta}_{0})
\widetilde c_{ijr_C}(\bs{\beta}_{0})
\widetilde d_{ijr_D}(\bs{\beta}_{0})
\mid\bf X,\bs\Gamma
\right] =
b_{ijr_B}(\bs{\beta}_{0})
c_{ijr_C}(\bs{\beta}_{0})
d_{ijr_D}(\bs{\beta}_{0}).
\end{aligned}
\]
Summing over the six determinant products with their signs gives
\(\E[\psi_{ijklm}(\bs{\beta}_{0})\mid\bf X,\bs\Gamma]
=\det(\mathbf M_{ijklm}(\bs{\beta}_{0}))=0\),
where the zero determinant follows from \eqref{eq:coeff_matrix}.
Taking conditional expectation with respect to \(\bf X\) proves
both restrictions in \eqref{eq:moment_cond-free}.
\end{proof}

The following lemma records the possible supports of the determinant
monomials, their connectedness, and their numbers of dyads. These properties
are used in Lemma~\ref{lem:kerneltaxonomy} and in the covariance bounds.

\begin{lem}[Distinct-dyad supports]
\label{lem:det-row-supports}
Fix \(s=(i,j,k,l,m)\in\mathcal S_N\), with pivot dyad \((i,j)\) and
peripheral nodes \(k,l,m\). Label the columns of
\(\widetilde{\mathbf M}_{ijklm}(\bs{\beta}_{0})\) by \(B,C,D\), in order.
For \(r\in\{k,l,m\}\), the monomial supports of the three row-\(r\)
entries belong to \(\mathcal M^B_r\), \(\mathcal M^C_r\), and
\(\mathcal M^D_r\), respectively, where \((u,v)\) denotes an unordered dyad and
\[
\begin{aligned}
\mathcal M^B_r
&=\bigl\{\{(i,r)\},\{(i,r),(j,r)\}\bigr\}, \ 
\mathcal M^C_r
=\bigl\{\{(j,r)\},\{(i,r),(j,r)\}\bigr\},\\
\mathcal M^D_r
&=\bigl\{\{(i,j),(i,r)\},\{(i,j),(j,r)\},\{(i,r),(j,r)\},\{(i,j),(i,r),(j,r)\}\bigr\}.
\end{aligned}
\]
The kernel \(\psi_{ijklm}(\bs{\beta}_{0})\) is a polynomial in
\(L_{ij},L_{ik},L_{il},L_{im},L_{jk},L_{jl},L_{jm}\),
with the representation
\begin{equation}
\label{eq:det-support-expansion}
\psi_{ijklm}(\bs{\beta}_{0})
=\sum_{\nu=1}^{\nu_{\max}}c_{s,\nu}(\bf X)m_{s,\nu}({\bf L}).
\end{equation}
Here \(\nu_{\max}<\infty\) does not depend on \(N\), the scalar
coefficients \(c_{s,\nu}(\bf X)\) depend only on the weights
\(W_{uv}(\bs\beta_0)\) within \(s\) and may be zero, and
\(m_{s,\nu}({\bf L})=\prod_{e\in E_{s,\nu}}L_e\).
Each support has the form \(E_{s,\nu}=A_B\cup A_C\cup A_D\), where
\((r_B,r_C,r_D)\) is a permutation of \((k,l,m)\),
\(A_B\in\mathcal M^B_{r_B}\), \(A_C\in\mathcal M^C_{r_C}\), and
\(A_D\in\mathcal M^D_{r_D}\). The sets \(A_B,A_C,A_D\) are pairwise
disjoint. Each \(E_{s,\nu}\subseteq E(s)\) forms a connected graph
on all five nodes, and
\(\deg m_{s,\nu}=\abs{E_{s,\nu}}\in\{4,5,6,7\}\).
\end{lem}

\begin{proof}
The row-wise support sets follow by expanding
\(\widetilde b_{ijr}(\bs{\beta}_{0})\),
\(\widetilde c_{ijr}(\bs{\beta}_{0})\), and
\(\widetilde d_{ijr}(\bs{\beta}_{0})\) in
\eqref{eq:feasible-coefficients-app}. Specifically,
\(\widetilde b_{ijr}(\bs{\beta}_{0})\) contains
only \(L_{ir}\) and \(L_{ir}L_{jr}\);
\(\widetilde c_{ijr}(\bs{\beta}_{0})\) contains only \(L_{jr}\) and
\(L_{ir}L_{jr}\); and \(\widetilde d_{ijr}(\bs{\beta}_{0})\) contains
only \(L_{ij}L_{ir}\), \(L_{ij}L_{jr}\), \(L_{ir}L_{jr}\), and
\(L_{ij}L_{ir}L_{jr}\). The determinant expansion selects one entry from
each of the \(B\)-, \(C\)-, and \(D\)-columns, using three different
peripheral rows and the usual determinant sign. The pivot dyad \((i,j)\)
can occur only in \(A_D\), and all other dyads in \(A_B,A_C,A_D\) are
attached to distinct peripheral nodes. Thus \(A_B,A_C,A_D\) are pairwise
disjoint, so the product of the three selected monomials has support
\(A_B\cup A_C\cup A_D\), with no repeated dyad.
Every \(D\)-support connects \(i\), \(j\), and \(r_D\), either through
the pivot dyad \((i,j)\) or through the path \(i-r_D-j\).
Moreover, every \(B\)-support contains \((i,r_B)\), every \(C\)-support
contains \((j,r_C)\), and \(r_B,r_C,r_D\) are distinct. Hence
\(A_B\cup A_C\cup A_D\) is connected and contains all five nodes of the pentad.
Since \(\abs{A_B},\abs{A_C}\in\{1,2\}\) and
\(\abs{A_D}\in\{2,3\}\), pairwise disjointness of \(A_B,A_C,A_D\)
implies that the resulting monomial degree lies between four and seven.
Collecting terms with the same support gives
\eqref{eq:det-support-expansion}, with coefficients depending only on
the weights. The finite list of possible supports depends only on the
five node roles, so its size \(\nu_{\max}\) is independent of \(N\).
\end{proof}

\subsection{Proof of Theorem~\ref{thm:pentad_minimality}}

We first express conditional moments as vectors in the null space of a
coefficient matrix. We then establish rank properties for arbitrary
positive directed weights, before showing that the weights generated by
the covariates satisfy the required rank bounds almost everywhere.

Fix a finite simple graph \(G=(V,E)\), with the degree notation in
Appendix~\ref{app:graph_notation}. Associate two positive weights
\(W_{uv},W_{vu}\) with each dyad \((u,v)\in E\), and write
\(\mathbf W=(W_{uv},W_{vu})_{(u,v)\in E}\). For variables
\(\boldsymbol\alpha=(\alpha_u:u\in V)\), define
\(R_{uv}(\boldsymbol\alpha):=(1+W_{uv}\alpha_u)(1+W_{vu}\alpha_v)\)
and \(R_A(\boldsymbol\alpha):=\prod_{(u,v)\in A}R_{uv}(\boldsymbol\alpha)\)
for \(A\subseteq E\), with \(R_\varnothing=1\).
We abbreviate \(R_{uv}(\boldsymbol\alpha)\) and
\(R_A(\boldsymbol\alpha)\) as \(R_{uv}\) and \(R_A\), respectively.
Each \(R_A\) has degree at most \(d_E(u)\) in \(\alpha_u\).
Define \(\mathbf{C}_G(\mathbf W)\) to have columns indexed by \(A\subseteq E\)
and rows indexed by the monomials \(\prod_{u\in V}\alpha_u^{j_u}\),
where \(0\leq j_u\leq d_E(u)\).
The entry of \(\mathbf{C}_G(\mathbf W)\) in row \((j_u:u\in V)\) and column
\(A\) is the coefficient multiplying \(\prod_{u\in V}\alpha_u^{j_u}\)
when \(R_A(\boldsymbol\alpha)\) is expanded as a polynomial in
\((\alpha_u:u\in V)\).
Thus \(\mathbf{C}_G(\mathbf W)\) has \(\prod_{u\in V}(d_E(u)+1)\) rows and
\(2^{\lvert E\rvert}\) columns. Choose orders for the nodes, monomials,
and subsets once and keep them fixed when the weights vary.
The particular orders do not affect the rank of \(\mathbf{C}_G(\mathbf W)\),
since changing them only permutes its rows and columns.

For a real matrix \(\mathbf M\) with \(n\) columns,
\(\operatorname{rank}\mathbf M\) denotes the dimension of its column
space, equivalently the maximum number of linearly independent columns.
We write \(\operatorname{null}\mathbf M:=\{\mathbf v\in\mathbb R^n:
\mathbf M\mathbf v=\mathbf0\}\) for its \emph{null space}, and
\(\dim\) for the dimension of a vector space. The matrix has
\emph{full column rank} when \(\operatorname{rank}\mathbf M=n\),
equivalently when \(\operatorname{null}\mathbf M=\{\mathbf0\}\).

\begin{lem}\label{lem:local-moment-coefficients}
Under the assumptions of Theorem~\ref{thm:pentad_minimality}, fix
\(\bs\beta_0\neq\mathbf0\) and use the weights
\(W_{uv}=\exp\{-\bs\beta_0^\top w(\bs X_u,\bs X_v)\}\).
For a real-valued measurable \(\psi_G\), let
\(\bs c=(c_A(\mathbf X_G))_{A\subseteq E}\) be the coefficients
in the unique expansion
\(\psi_G(\mathbf L_G,\mathbf X_G;\bs\beta_0)
=\sum_{A\subseteq E}c_A(\mathbf X_G)
\prod_{(u,v)\in E\setminus A}L_{uv}\).
The map from
\(\bigl(\psi_G(\bs l,\mathbf X_G;\bs\beta_0)\bigr)_{
\bs l\in\{0,1\}^{\lvert E\rvert}}\) to \(\bs c\) is a linear
bijection. Moreover, \eqref{eq:graph-local-moment} holds if and only if
\(\bs c\in\operatorname{null}\mathbf{C}_G(\mathbf W)\) for almost every
\(\mathbf X_G=(\bs X_u:u\in V)\).
\end{lem}

\begin{proof}
To prove the bijection, fix \(\mathbf X_G\) and \(\bs\beta_0\).
For each \(B\subseteq E\), define the link pattern
\(\bs l^B=(l^B_{uv})_{(u,v)\in E}
=(\1\{(u,v)\notin B\})_{(u,v)\in E}\),
so that \(l^B_{uv}=0\) on \(B\) and \(l^B_{uv}=1\) on \(E\setminus B\).
Every link pattern has
exactly one such representation. Then the product
\(\prod_{(u,v)\in E\setminus A}l^B_{uv}\) equals one if
\(B\subseteq A\), and zero otherwise.
Given the coefficient vector \(\bs c\), evaluating the expansion
at each pattern \(\bs l^B\) determines the function values by
\begin{equation}\label{eq:link-pattern-values}
\psi_G(\bs l^B,\mathbf X_G;\bs\beta_0)
=\sum_{A:B\subseteq A\subseteq E}c_A(\mathbf X_G),
\qquad B\subseteq E.
\end{equation}
First, we prove that the linear map from coefficients to function
values is injective. By linearity, it suffices to show that zero
function values imply \(\bs c=\mathbf0\). Taking \(B=E\)
gives \(c_E(\mathbf X_G)=0\). Suppose that
\(c_A(\mathbf X_G)=0\) whenever \(\lvert A\rvert>k\).
For every \(B\subseteq E\) with \(\lvert B\rvert=k\),
\eqref{eq:link-pattern-values} and the induction hypothesis give
\(0=c_B(\mathbf X_G)
+\sum_{A:B\subsetneq A\subseteq E}c_A(\mathbf X_G)
=c_B(\mathbf X_G)\).
Descending induction on \(k\) gives \(\bs c=\mathbf0\),
proving injectivity.

Next, we prove surjectivity. Given arbitrary function values
\(\bigl(\psi_G(\bs l^B,\mathbf X_G;\bs\beta_0)\bigr)_{B\subseteq E}\),
define
\begin{equation}\label{eq:link-pattern-coefficients}
c_A(\mathbf X_G)
=\sum_{B:A\subseteq B\subseteq E}(-1)^{\lvert B\setminus A\rvert}
\psi_G(\bs l^B,\mathbf X_G;\bs\beta_0),
\qquad A\subseteq E.
\end{equation}
For each \(B\subseteq E\), substituting these coefficients into
the expansion evaluated at \(\bs l^B\) gives
\[
\begin{aligned}
\sum_{A:B\subseteq A\subseteq E}c_A(\mathbf X_G)
&=\sum_{A:B\subseteq A\subseteq E}
  \sum_{B':A\subseteq B'\subseteq E}(-1)^{\lvert B'\setminus A\rvert}
  \psi_G(\bs l^{B'},\mathbf X_G;\bs\beta_0)\\
&=\sum_{B':B\subseteq B'\subseteq E}
  \psi_G(\bs l^{B'},\mathbf X_G;\bs\beta_0)
  \sum_{A:B\subseteq A\subseteq B'}(-1)^{\lvert B'\setminus A\rvert}\\
&=\psi_G(\bs l^B,\mathbf X_G;\bs\beta_0).
\end{aligned}
\]
The inner sum
\(\sum_{A:B\subseteq A\subseteq B'}(-1)^{\lvert B'\setminus A\rvert}\)
equals one when \(B'=B\), since only \(A=B\) occurs.
If \(B'\neq B\), choose a dyad \(e\in B'\setminus B\) and pair
each \(A\) not containing \(e\) with \(A\cup\{e\}\).
The paired terms have opposite signs, so the inner sum is zero.
Thus the constructed coefficients reproduce every prescribed function
value, proving surjectivity.

Together, injectivity and surjectivity establish existence and uniqueness
of the expansion. Equation~\eqref{eq:link-pattern-coefficients}
gives the inverse map from function values to coefficients.
This inverse is linear, proving the claimed linear bijection.
The coefficients \((-1)^{\lvert B\setminus A\rvert}\) multiplying
the function values in \eqref{eq:link-pattern-coefficients} do not
depend on \(\mathbf X_G\) or \(\bs\beta_0\).
Since these function values are measurable in \(\mathbf X_G\),
their finite linear combination \(c_A(\mathbf X_G)\) is also measurable
in \(\mathbf X_G\).

We now prove the equivalence with \eqref{eq:graph-local-moment}.
Write \(P_{uv}=P_{uv}(\bs\beta_0)\) and
\(\alpha_u=e^{-\Gamma_u}\). Recall that
\(R_{uv}(\boldsymbol\alpha)=(1+W_{uv}\alpha_u)(1+W_{vu}\alpha_v)\).
The product-logit formula in Assumption~\ref{ass:dyad_ind} therefore gives
\(R_{uv}(\boldsymbol\alpha)=P_{uv}^{-1}\).
Since \(L_{uv}\) is binary,
\(\E[L_{uv}\mid\bf X,\bs\Gamma]=P_{uv}\).
Conditional independence of the dyads given \((\bf X,\bs\Gamma)\)
therefore gives
\(\E[\prod_{(u,v)\in E\setminus A}L_{uv}\mid\bf X,\bs\Gamma]
=\prod_{(u,v)\in E\setminus A}P_{uv}\) for each \(A\subseteq E\).
Each \(c_A(\mathbf X_G)\) is measurable with respect to \(\bf X\),
so taking conditional expectations in the expansion of \(\psi_G\) yields
\(\E[\psi_G(\mathbf L_G,\mathbf X_G;\bs\beta_0)\mid\bf X,\bs\Gamma]
=\allowbreak\sum_{A\subseteq E}c_A(\mathbf X_G)
\prod_{(u,v)\in E\setminus A}P_{uv}\).
Thus \eqref{eq:graph-local-moment} is equivalent to
\(\sum_{A\subseteq E}c_A(\mathbf X_G)
\prod_{(u,v)\in E\setminus A}P_{uv}=0\) almost surely.
Dividing this equality by \(\prod_{(u,v)\in E}P_{uv}>0\)
replaces each factor \(\prod_{(u,v)\in E\setminus A}P_{uv}\) by
\(\frac{\prod_{(u,v)\in E\setminus A}P_{uv}}
{\prod_{(u,v)\in E}P_{uv}}
=\prod_{(u,v)\in A}P_{uv}^{-1}=R_A(\boldsymbol\alpha)\).
Hence \eqref{eq:graph-local-moment} is equivalent to
\begin{equation}
\sum_{A\subseteq E}c_A(\mathbf X_G)R_A(\boldsymbol\alpha)=0
\quad\text{a.s.}
\label{eq:exact-local-polynomial}
\end{equation}
Taking the conditional probability of the equality given \(\bf X\)
yields
\begin{equation}
\Pr\!\left(\sum_{A\subseteq E}c_A(\mathbf X_G)R_A(\boldsymbol\alpha)=0
\;\middle|\;\bf X\right)=1\quad\text{a.s.}
\label{eq:conditional-local-polynomial}
\end{equation}
Independence across nodes and the joint density of \((\bs X_i,\Gamma_i)\)
give a joint density for \((\bf X,(\Gamma_u:u\in V))\). Hence
\((\Gamma_u:u\in V)\) has a conditional density given \(\bf X\) for
almost every \(\bf X\). The change of variables
\(\Gamma_u=-\log\alpha_u\) implies that the conditional distribution
of \((\alpha_u:u\in V)\) given \(\bf X\) is absolutely continuous
with respect to Lebesgue measure on \((0,\infty)^{\lvert V\rvert}\)
for almost every \(\bf X\).

Fix \(\bf X\) for which the conditional density of
\((\alpha_u:u\in V)\) exists and
\eqref{eq:conditional-local-polynomial} holds. For these fixed covariates,
the sum in \eqref{eq:exact-local-polynomial} is a polynomial in
\((\alpha_u:u\in V)\). If the polynomial in
\eqref{eq:exact-local-polynomial} were not identically zero, its zero
set would have Lebesgue measure zero by
\citet[Proposition~1]{mityagin2020zero}, since polynomials are real analytic.
Absolute continuity of the conditional distribution of
\((\alpha_u:u\in V)\) given \(\bf X\) would then make the conditional
probability in \eqref{eq:conditional-local-polynomial} equal to zero,
contradicting the required value of one.
The polynomial in \eqref{eq:exact-local-polynomial} must therefore
be identically zero. Recall that the column of \(\mathbf{C}_G(\mathbf W)\)
indexed by \(A\subseteq E\) contains the coefficients of
\(R_A(\boldsymbol\alpha)\) in the fixed order of the monomials
\(\prod_{u\in V}\alpha_u^{j_u}\), \(0\leq j_u\leq d_E(u)\).
The monomial coefficient vector of
\(\sum_{A\subseteq E}c_A(\mathbf X_G)R_A(\boldsymbol\alpha)\)
is therefore \(\mathbf{C}_G(\mathbf W)\bs c\), so
\(\mathbf{C}_G(\mathbf W)\bs c=\mathbf0\).
Thus \eqref{eq:graph-local-moment} implies
\(\bs c\in\operatorname{null}\mathbf{C}_G(\mathbf W)\) for almost every
\(\mathbf X_G\).
Conversely, suppose \(\mathbf{C}_G(\mathbf W)\bs c=\mathbf0\) for almost
every \(\mathbf X_G\). The polynomial
\(\sum_{A\subseteq E}c_A(\mathbf X_G)R_A(\boldsymbol\alpha)\)
is then identically zero, and in particular vanishes at
\(\alpha_u=e^{-\Gamma_u}\), \(u\in V\).
At these values,
\(R_A(\boldsymbol\alpha)=\prod_{(u,v)\in A}P_{uv}^{-1}\).
Multiplying \eqref{eq:exact-local-polynomial} by
\(\prod_{(u,v)\in E}P_{uv}\) therefore gives
\(\sum_{A\subseteq E}c_A(\mathbf X_G)
\prod_{(u,v)\in E\setminus A}P_{uv}=0\) almost surely.
The expansion of \(\psi_G\) and conditional independence of the dyads
given \((\bf X,\bs\Gamma)\) now give
\(\E[\psi_G(\mathbf L_G,\mathbf X_G;\bs\beta_0)\mid\bf X,\bs\Gamma]
=\allowbreak\sum_{A\subseteq E}c_A(\mathbf X_G)
\prod_{(u,v)\in E\setminus A}P_{uv}=0\) almost surely,
which proves \eqref{eq:graph-local-moment}.
Finally, \(\mathbf{C}_G(\mathbf W)\) has \(2^{\lvert E\rvert}\) columns,
so rank--nullity gives
\begin{equation}
\dim\operatorname{null}\mathbf{C}_G(\mathbf W)
=2^{\lvert E\rvert}-\operatorname{rank}\mathbf{C}_G(\mathbf W).
\label{eq:local-moment-nullity}
\end{equation}
\end{proof}

By Lemma~\ref{lem:local-moment-coefficients}, a function \(\psi_G\)
satisfying \eqref{eq:graph-local-moment} is trivial if and only if
\(\bs c=\mathbf0\) for almost every \(\mathbf X_G\).
Thus, to determine whether nontrivial functions satisfying
\eqref{eq:graph-local-moment} exist, we study the nonzero solutions of
\(\mathbf{C}_G(\mathbf W)\bs c=\mathbf0\).
By \eqref{eq:local-moment-nullity},
\(\mathbf{C}_G(\mathbf W)\bs c=\mathbf0\) has a nonzero solution at a fixed
\(\mathbf X_G\) if and only if
\(\operatorname{rank}\mathbf{C}_G(\mathbf W)<2^{\lvert E\rvert}\).
In particular, full column rank for almost every \(\mathbf X_G\)
implies that every \(\psi_G\) satisfying \eqref{eq:graph-local-moment}
is trivial.

We next establish the  rank properties of \(\mathbf{C}_G(\mathbf W)\).
In the following algebraic lemmas, \(\mathbf W\) consists of arbitrary
positive numbers, and we treat \((\alpha_u:u\in V)\) as polynomial
variables. We may therefore evaluate \(R_A(\boldsymbol\alpha)\)
at any real values of \((\alpha_u:u\in V)\).
The proof of Theorem~\ref{thm:pentad_minimality} will apply these rank
results to \(W_{uv}=\exp\{-\bs\beta_0^\top w(\bs X_u,\bs X_v)\}\).

\begin{lem}\label{lem:coefficient-matrix-properties}
The coefficient matrices \(\mathbf{C}_G(\mathbf W)\) have the following properties.
\begin{enumerate}[label=\textnormal{(\alph*)},leftmargin=2em]
\item Assign distinct new labels to the nodes, and replace each old
label by its new label everywhere in \(E\) and in the subscripts of
\(W_{uv}\), without changing the numerical weights.
The coefficient matrix of the relabeled graph has the same rank as
\(\mathbf{C}_G(\mathbf W)\).
\item For positive numbers \((s_u:u\in V)\), set
\(\widetilde W_{uv}=s_uW_{uv}\) for every neighbor \(v\) of \(u\).
Then
\begin{equation}
\operatorname{rank}\mathbf{C}_G(\widetilde{\mathbf W})
=\operatorname{rank}\mathbf{C}_G(\mathbf W).
\label{eq:outgoing-weight-invariance}
\end{equation}
\item If \(G\) is a subgraph of a graph whose coefficient matrix has
full column rank, then \(\mathbf{C}_G\), using the same weights \((W_{uv},W_{vu})\) on its dyads,
has full column rank.
\item If \(G\) is the disjoint union of \(G_1=(V_1,E_1)\) and
\(G_2=(V_2,E_2)\), then, up to row and column permutations,
\(\mathbf{C}_G=\mathbf{C}_{G_1}\otimes \mathbf{C}_{G_2}\). In particular,
\(\operatorname{rank}\mathbf{C}_G=\operatorname{rank}\mathbf{C}_{G_1}\operatorname{rank}\mathbf{C}_{G_2}\).
\end{enumerate}
\end{lem}

\begin{proof}
For \textnormal{(a)}, rename each variable \(\alpha_u\) using the new
label of node \(u\). Every polynomial \(R_A\) then has the same
coefficients as before, indexed by the renamed monomials and dyad subset.
Consequently, the coefficient matrix of the relabeled graph differs from
\(\mathbf{C}_G(\mathbf W)\) only by row and column permutations, which preserve rank.
For \textnormal{(b)}, replacing \(W_{uv}\) by \(s_uW_{uv}\) in every
\(R_A\) is exactly the substitution \(\alpha_u\mapsto s_u\alpha_u\).
The entire row indexed by \((j_u:u\in V)\) is multiplied by
\(\prod_{u\in V}s_u^{j_u}>0\), an invertible diagonal row transformation.

For \textnormal{(c)}, retain from the larger coefficient matrix the columns indexed by \(A\subseteq E\). Because the weights on the dyads of \(G\) are unchanged, these columns contain the coefficients of the same polynomials \(R_A\) that define \(\mathbf{C}_G\). Any monomial outside the row index set of \(\mathbf{C}_G\) has coefficient zero in every retained column. Deleting these extra zero rows and reordering rows and columns therefore gives \(\mathbf{C}_G\). The retained columns are linearly independent because the larger matrix has full column rank. Deleting zero rows and reordering rows and columns preserves this independence, so \(\mathbf{C}_G\) also has full column rank.

For \textnormal{(d)}, fix \(A\subseteq E\) and set \(A_j=A\cap E_j\)
for \(j=1,2\). Then \(A=A_1\cup A_2\) with
\(A_1\cap A_2=\varnothing\), so \(R_A=R_{A_1}R_{A_2}\).
The polynomial \(R_{A_j}\) involves only
\((\alpha_u:u\in V_j)\), and \(V_1\cap V_2=\varnothing\).
For each row index \((j_u:u\in V)\), the monomial
\(\prod_{u\in V}\alpha_u^{j_u}\) can therefore arise in the expansion
of \(R_{A_1}R_{A_2}\) only by multiplying
\(\prod_{u\in V_1}\alpha_u^{j_u}\) from \(R_{A_1}\) and
\(\prod_{u\in V_2}\alpha_u^{j_u}\) from \(R_{A_2}\).
The coefficient of \(\prod_{u\in V}\alpha_u^{j_u}\) in \(R_A\)
is therefore the product of the two component coefficients.
By the definition of the coefficient matrices,
\([\mathbf{C}_G]_{(j_u:u\in V),A}
=\allowbreak[\mathbf{C}_{G_1}]_{(j_u:u\in V_1),A_1}
[\mathbf{C}_{G_2}]_{(j_u:u\in V_2),A_2}\).
Thus, up to row and column permutations,
\(\mathbf{C}_G=\mathbf{C}_{G_1}\otimes \mathbf{C}_{G_2}\).
The rank of a Kronecker product is the product of the ranks, so
\(\operatorname{rank}\mathbf{C}_G
=\operatorname{rank}\mathbf{C}_{G_1}\operatorname{rank}\mathbf{C}_{G_2}\).
\end{proof}

The next lemma uses the following separation of weights with the same
first index:
\begin{equation}
W_{uv}\neq W_{uv'}
\quad\text{whenever }v\neq v'\text{ are neighbors of }u.
\label{eq:generic-outgoing-separation}
\end{equation}
We will prove that \eqref{eq:generic-outgoing-separation} holds for
almost every \((\bs X_u:u\in V)\) under the assumptions of the theorem.

\begin{lem}\label{lem:separated-weight-ranks}
Let \(G=(V,E)\) be a finite simple graph.
For positive weights (\(W_{uv},W_{vu}\)) satisfying \eqref{eq:generic-outgoing-separation}:
\textnormal{(a)} if every connected component of \(G\) is a tree or
has exactly one cycle, then \(\operatorname{rank}\mathbf{C}_G(\mathbf W)=2^{\lvert E\rvert}\);
\textnormal{(b)} for the pentad \(G=(V,E)\) with \(V=\{i,j,k,l,m\}\)
and \(E=\{(i,j),\allowbreak(i,k),\allowbreak(j,k),\allowbreak(i,l),\allowbreak(j,l),\allowbreak(i,m),\allowbreak(j,m)\}\),
\(\operatorname{rank}\mathbf{C}_G(\mathbf W)\leq127\).
Moreover, the moment function \(\psi_{ijklm}(\bs\beta_0)\)
defined in Theorem~\ref{thm:moment_restriction} is nontrivial.
\end{lem}

\begin{proof}
For \textnormal{(a)}, the definition of \(\mathbf{C}_G(\mathbf W)\) gives
\(\bigl(R_A(\boldsymbol\alpha)\bigr)_{A\subseteq E}
=\allowbreak\left(\prod_{u\in V}\alpha_u^{j_u}\right)_{(j_u:u\in V)}
\mathbf{C}_G(\mathbf W)\).
Both indexed vectors are row vectors, and \((j_u:u\in V)\) ranges
over the monomial indices \(0\leq j_u\leq d_E(u)\).
Choose \(2^{\lvert E\rvert}\) values of \(\boldsymbol\alpha\) and
fix an order for these choices. For each chosen value
\(\boldsymbol\alpha\), use
\(\bigl(R_A(\boldsymbol\alpha)\bigr)_{A\subseteq E}\) as the row of
\(\bs{\mathcal R}_G\) and
\(\left(\prod_{u\in V}\alpha_u^{j_u}\right)_{(j_u:u\in V)}\) as the
row of \(\bs{\mathcal M}_G\), placing both rows at the position assigned
to \(\boldsymbol\alpha\) in the chosen order.
For each chosen \(\boldsymbol\alpha\), the row of \(\bs{\mathcal R}_G\)
equals the row of \(\bs{\mathcal M}_G\) multiplied by \(\mathbf{C}_G(\mathbf W)\).
Since left multiplication cannot increase rank, we obtain
\begin{equation}\label{eq:polynomial-evaluation-matrix}
\bs{\mathcal R}_G=\bs{\mathcal M}_G \mathbf{C}_G(\mathbf W),
\qquad
\operatorname{rank}\bs{\mathcal R}_G
\leq\operatorname{rank}\mathbf{C}_G(\mathbf W)\leq2^{\lvert E\rvert}.
\end{equation}
Since \(\bs{\mathcal R}_G\) is \(2^{\lvert E\rvert}\times2^{\lvert E\rvert}\),
it suffices to choose the values of \(\boldsymbol\alpha\) so that
\(\det\bs{\mathcal R}_G\neq0\). We will make \(\bs{\mathcal R}_G\) triangular
with nonzero diagonal entries.
To obtain the required zero entries, we use
\(R_A=\prod_{(u,v)\in A}R_{uv}\): if \(R_{uv}=0\), then
\(R_A=0\) for every \(A\) containing \((u,v)\).
We will therefore choose which dyad factors \(R_{uv}\) vanish at
each evaluation, while keeping the remaining dyad factors nonzero.

The factorization
\(R_{uv}=(1+W_{uv}\alpha_u)(1+W_{vu}\alpha_v)\) allows us to make
\(R_{uv}=0\) by setting either \(\alpha_u=-1/W_{uv}\) or
\(\alpha_v=-1/W_{vu}\). We write \((u,v)\mapsto u\) when we assign
dyad \((u,v)\) to node \(u\) and select the substitution
\(\alpha_u=-1/W_{uv}\). If two distinct dyads \((u,v)\) and
\((u,v')\) were assigned to \(u\), the selected substitutions would
prescribe different values of \(\alpha_u\) by
\eqref{eq:generic-outgoing-separation}. We therefore first construct
an assignment in which each node receives at most one dyad, and
then use the assignment to choose values of \((\alpha_u:u\in V)\).

For each connected component of \(G\) that is a tree, choose a root
node. For every node \(u\) in this component other than the root,
let \(v\) be the next node on the unique path from \(u\) to the root,
and assign \((u,v)\mapsto u\).
Every dyad in this component joins a non-root node to its next node
toward the root,
so every dyad is assigned exactly once. Each non-root node \(u\)
receives only the dyad \((u,v)\) on its path to the root; the root
receives none.

For each connected component of \(G\) that has exactly one cycle,
list the distinct nodes of its unique cycle as \(u_1,\ldots,u_\ell\)
so that \((u_j,u_{j+1})\in E\) for \(j=1,\ldots,\ell\), where
\(u_{\ell+1}=u_1\).
Assign \((u_j,u_{j+1})\mapsto u_{j+1}\) for \(j=1,\ldots,\ell\),
so each cycle node receives exactly one cycle dyad.
For every node \(u\) in this component but outside the cycle, let
\(v\) be the next node on the unique path from \(u\) to the cycle,
and assign \((u,v)\mapsto u\). Every dyad of this component that is
not on the cycle is assigned exactly once by this rule. Each node outside the cycle receives only the
dyad \((u,v)\) on its path to the cycle, and no cycle node receives
a dyad outside the cycle. Thus every node in the component receives
exactly one dyad.

We now use the assignment to choose values of \((\alpha_u:u\in V)\).
For each \(B\subseteq E\), define the vector
\(\boldsymbol\alpha^B=(\alpha_u^B:u\in V)\) by
\[
\alpha_u^B=
\begin{cases}
-1/W_{uv},&\text{if }(u,v)\in E\setminus B\text{ is assigned to }u,\\
0,&\text{if no dyad in }E\setminus B\text{ is assigned to }u.
\end{cases}
\]
Because each node receives at most one dyad, \(\alpha_u^B\) is
uniquely specified for every \(u\in V\).
Using the vectors \(\boldsymbol\alpha^B\) as the chosen values of
\(\boldsymbol\alpha\), we obtain
\(\bs{\mathcal R}_G=\bigl(R_A(\boldsymbol\alpha^B)\bigr)_{B,A\subseteq E}\)
and
\(\bs{\mathcal M}_G=\left(\prod_{u\in V}(\alpha_u^B)^{j_u}\right)_{B,(j_u:u\in V)}\).
In particular, \([\bs{\mathcal R}_G]_{B,A}=R_A(\boldsymbol\alpha^B)\).
We now prove \(\det\bs{\mathcal R}_G\neq0\).

Fix \((u,v)\in E\). If \((u,v)\notin B\), assigning it to \(u\)
gives \(1+W_{uv}\alpha_u^B=0\), while assigning it to \(v\) gives
\(1+W_{vu}\alpha_v^B=0\). Thus \(R_{uv}(\boldsymbol\alpha^B)=0\)
in either case.
If \((u,v)\in B\), then \(1+W_{uv}\alpha_u^B=1\) when
\(\alpha_u^B=0\). If a dyad \((u,v')\in E\setminus B\) is assigned
to \(u\), then \(v'\neq v\) and
\(1+W_{uv}\alpha_u^B=1-W_{uv}/W_{uv'}\neq0\) by \eqref{eq:generic-outgoing-separation}.
At node \(v\), \(1+W_{vu}\alpha_v^B=1\) when \(\alpha_v^B=0\).
If a dyad \((v,u')\in E\setminus B\) is assigned to \(v\), then
\(u'\neq u\) and
\(1+W_{vu}\alpha_v^B=1-W_{vu}/W_{vu'}\neq0\) by \eqref{eq:generic-outgoing-separation}.
Thus \(R_{uv}(\boldsymbol\alpha^B)
=(1+W_{uv}\alpha_u^B)(1+W_{vu}\alpha_v^B)\neq0\) for \((u,v)\in B\).
We have proved \(R_{uv}(\boldsymbol\alpha^B)=0\iff(u,v)\notin B\).
Recall that \([\bs{\mathcal R}_G]_{B,A}\), the entry in row \(B\) and
column \(A\), equals
\(R_A(\boldsymbol\alpha^B)
=\prod_{(u,v)\in A}R_{uv}(\boldsymbol\alpha^B)\).
Hence \([\bs{\mathcal R}_G]_{B,A}\neq0\iff A\subseteq B\).

Use the same ordered list of subsets of \(E\) for the row index
\(B\) and the column index \(A\) of \(\bs{\mathcal R}_G\), placing subsets
with fewer dyads first. The condition
\([\bs{\mathcal R}_G]_{B,A}\neq0\iff A\subseteq B\) makes
\(\bs{\mathcal R}_G\) lower triangular, with
\([\bs{\mathcal R}_G]_{B,B}=R_B(\boldsymbol\alpha^B)\neq0\).
Hence \(\det\bs{\mathcal R}_G=\prod_{B\subseteq E}R_B(\boldsymbol\alpha^B)\neq0\),
so \(\operatorname{rank}\bs{\mathcal R}_G=2^{\lvert E\rvert}\).
Equation~\eqref{eq:polynomial-evaluation-matrix} now gives
\(2^{\lvert E\rvert}=\operatorname{rank}\bs{\mathcal R}_G
\leq\operatorname{rank}\mathbf{C}_G(\mathbf W)\leq2^{\lvert E\rvert}\),
proving \(\operatorname{rank}\mathbf{C}_G(\mathbf W)=2^{\lvert E\rvert}\).
An isolated node \(u\) belongs to no dyad, so no \(R_A\) depends on
\(\alpha_u\), and the construction sets \(\alpha_u^B=0\) for every
\(B\subseteq E\). Thus isolated nodes do not affect \(\bs{\mathcal R}_G\).
The conclusion \(\operatorname{rank}\mathbf{C}_G(\mathbf W)=2^{\lvert E\rvert}\)
therefore also holds when \(G\) contains isolated nodes.
If \(E=\varnothing\), its only subset is \(A=\varnothing\), and
\(R_\varnothing=1\). Hence \(\mathbf{C}_G(\mathbf W)=(1)\) has rank
\(1=2^{\lvert E\rvert}\), so \(\mathbf{C}_G(\mathbf W)\) also has full column
rank when the graph has no dyads.

For \textnormal{(b)}, consider the pentad \(G=(V,E)\) with
\(V=\{i,j,k,l,m\}\),
\(E=\{(i,j),\allowbreak(i,k),\allowbreak(j,k),\allowbreak(i,l),\allowbreak(j,l),\allowbreak(i,m),\allowbreak(j,m)\}\),
and pivot dyad \((i,j)\). Recall that Theorem~\ref{thm:moment_restriction}
defines \(\psi_{ijklm}(\bs\beta_0)
=\det(\widetilde{\mathbf M}_{ijklm}(\bs\beta_0))\), where the matrix
\(\widetilde{\mathbf M}_{ijklm}(\bs\beta_0)\) is given in
\eqref{eq:feasible-matrix-app}, with entries specified in
\eqref{eq:feasible-coefficients-app}.
In \eqref{eq:delta-app} and \eqref{eq:feasible-coefficients-app},
replace each \(W_{uv}(\bs\beta_0)\) by the given positive weight
\(W_{uv}\). We keep the notation \(\psi_{ijklm}(\bs\beta_0)\) for
the resulting determinant, which is a function of \(\mathbf L_G\)
and \(\mathbf W\).
To show that \(\psi_{ijklm}(\bs\beta_0)\) is nonzero, choose the link pattern
\(\bs l^*=(l^*_{uv})_{(u,v)\in E}\) with
\(l^*_{jk}=l^*_{il}=l^*_{im}=l^*_{jm}=1\) and
\(l^*_{ij}=l^*_{ik}=l^*_{jl}=0\).
Substituting these link values into
\eqref{eq:feasible-coefficients-app} and \eqref{eq:feasible-matrix-app}
gives
\[
\left.\psi_{ijklm}(\bs\beta_0)\right|_{\mathbf L_G=\bs l^*}
=-\frac{W_{il}W_{im}W_{jk}W_{jm}W_{kj}W_{li}}
{W_{ij}W_{ji}}(W_{mi}-W_{mj})\neq0.
\]
The last step follows from \eqref{eq:generic-outgoing-separation} at
node \(m\) and positivity of the other factors.

The monomial argument in Lemma~\ref{lem:det-row-supports} shows
that \(\psi_{ijklm}(\bs\beta_0)\) has degree at most one in each
\(L_{uv}\). For fixed \(\mathbf W\), write
\(\psi_{ijklm}(\bs\beta_0)
=\sum_{A\subseteq E}c_A^{\mathrm{pent}}
\prod_{(u,v)\in E\setminus A}L_{uv}\).
This expansion is a polynomial identity, valid for every
\(\mathbf L_G\in\mathbb R^7\). The coefficients are unique by the
algebraic argument in Lemma~\ref{lem:local-moment-coefficients}.
Write \(\bs c^{\mathrm{pent}}=(c_A^{\mathrm{pent}})_{A\subseteq E}\).
These coefficients depend only on \(\mathbf W\), not on
\(\boldsymbol\alpha\).

For any \(\boldsymbol\alpha\in(0,\infty)^5\), define the numbers
\(P_{uv}:=R_{uv}(\boldsymbol\alpha)^{-1}
=\{(1+W_{uv}\alpha_u)(1+W_{vu}\alpha_v)\}^{-1}\)
for \((u,v)\in E\).
Substituting \(L_{uv}=P_{uv}\) for every \((u,v)\in E\) into
\(\widetilde{\mathbf M}_{ijklm}(\bs\beta_0)\) gives
\(\mathbf M_{ijklm}(\bs\beta_0)\), with the entries in
\eqref{eq:triangle-coefficients-app} evaluated at
\(P_{uv}=R_{uv}(\boldsymbol\alpha)^{-1}\) and the given weights
\((W_{uv},W_{vu})\).
This follows by comparing \eqref{eq:feasible-coefficients-app}
with \eqref{eq:triangle-coefficients-app}.
The numbers \(P_{uv}\) satisfy \eqref{eq:triangle-p-identities},
so the algebraic calculation leading to \eqref{eq:coeff_matrix} gives
\(\mathbf M_{ijklm}(\bs\beta_0)
(P_{ij}\alpha_i,P_{ij}\alpha_j,1)^\top=\mathbf0\).
Since \((P_{ij}\alpha_i,P_{ij}\alpha_j,1)^\top\neq\mathbf0\),
we have \(\det(\mathbf M_{ijklm}(\bs\beta_0))=0\).
Substituting the same values \(L_{uv}=P_{uv}\) into the polynomial
expansion of \(\psi_{ijklm}(\bs\beta_0)\) therefore yields
\(\sum_{A\subseteq E}c_A^{\mathrm{pent}}
\prod_{(u,v)\in E\setminus A}P_{uv}
=\allowbreak\det(\mathbf M_{ijklm}(\bs\beta_0))=0\).
Dividing by \(\prod_{(u,v)\in E}P_{uv}>0\) and using
\(\prod_{(u,v)\in A}P_{uv}^{-1}=R_A(\boldsymbol\alpha)\) gives
\(\sum_{A\subseteq E}c_A^{\mathrm{pent}}R_A(\boldsymbol\alpha)=0\)
for every \(\boldsymbol\alpha\in(0,\infty)^5\).
Thus the polynomial
\(\sum_{A\subseteq E}c_A^{\mathrm{pent}}R_A(\boldsymbol\alpha)\)
vanishes on a nonempty open set and is identically zero.
By the definition of \(\mathbf{C}_G(\mathbf W)\), the monomial coefficient
vector of this polynomial is \(\mathbf{C}_G(\mathbf W)\bs c^{\mathrm{pent}}\).
Hence \(\mathbf{C}_G(\mathbf W)\bs c^{\mathrm{pent}}=\mathbf0\).

Recall that \(\psi_{ijklm}(\bs\beta_0)
=\sum_{A\subseteq E}c_A^{\mathrm{pent}}
\prod_{(u,v)\in E\setminus A}L_{uv}\).
Since \(\psi_{ijklm}(\bs\beta_0)\) is nonzero at
\(\mathbf L_G=\bs l^*\), at least one coefficient
\(c_A^{\mathrm{pent}}\) in this identity is nonzero.
Hence \(\bs c^{\mathrm{pent}}\neq\mathbf0\).
Consequently, \(\mathbf{C}_G(\mathbf W)\) has \(2^{\lvert E\rvert}=128\)
columns and a nonzero vector in its null space.
Equation~\eqref{eq:local-moment-nullity} therefore gives
\(\operatorname{rank}\mathbf{C}_G(\mathbf W)\leq127\).
\end{proof}

Lemma~\ref{lem:separated-weight-ranks}\textnormal{(b)} gives the upper
bound \(\operatorname{rank}\mathbf{C}_G(\mathbf W)\leq127\) for the pentad
under \eqref{eq:generic-outgoing-separation}.
We next construct, for each five-node, seven-dyad graph \(G\),
positive weights \(\mathbf W\) such that
\(\operatorname{rank}\mathbf{C}_G(\mathbf W)\geq127\) if \(G\) is the pentad
and \(\operatorname{rank}\mathbf{C}_G(\mathbf W)\geq128\) otherwise.
We begin by computing the rank when every directed weight equals
one, because the polynomials \(R_A\) then have a particularly simple
form. This computation will provide the starting point for studying
how the rank changes when the weights move away from one.

Fix \(G=(V,E)\) with \(\lvert V\rvert=5\) and \(\lvert E\rvert=7\).
Set \(W_{uv}=W_{vu}=1\) for every \((u,v)\in E\), denote this weight
vector by \(\mathbf1\), and let
\(r_G:=\operatorname{rank}\mathbf{C}_G(\mathbf1)\).
At \(\mathbf W=\mathbf1\), column \(A\) of \(\mathbf{C}_G(\mathbf1)\)
contains the coefficients of
\(R_A=\prod_{u\in V}(1+\alpha_u)^{d_A(u)}\), where \(d_A(u)\) is
the degree of \(u\) in \((V,A)\).
To compute \(r_G\), apply the substitution
\(\alpha_u\mapsto\alpha_u-1\) for every \(u\in V\).
This substitution acts by an invertible row transformation on the
monomial coefficients: its inverse is
\(\alpha_u\mapsto\alpha_u+1\). Neither substitution increases the
degree in any \(\alpha_u\) beyond \(d_E(u)\).
Under \(\alpha_u\mapsto\alpha_u-1\), column \(A\)
becomes the coefficient vector of the single monomial
\(\prod_{u\in V}\alpha_u^{d_A(u)}\).
Each column has a single entry equal to one in the row indexed
by its degree vector, with all other entries zero. Thus distinct
degree vectors give linearly independent columns, while equal
degree vectors give identical columns. Consequently,
\(r_G=\lvert\{(d_A(u):u\in V):A\subseteq E\}\rvert\).

A five-node graph with seven dyads has a complement with five nodes and
\(\binom{5}{2}-7=3\) dyads.
Up to relabeling, the complement is a
\mbox{three-dyad star
(\taxgraph[baseline=-0.5ex]{
    \node[taxNode] (a) at (0,0) {};
    \node[taxNode] (b) at (0,.65) {};
    \node[taxNode] (c) at (-.65,-.35) {};
    \node[taxNode] (d) at (.65,-.35) {};
    \node[taxNode] (e) at (1.4,.35) {};
    \draw[taxEdge] (a)--(b);
    \draw[taxEdge] (a)--(c);
    \draw[taxEdge] (a)--(d);
})}, a
\mbox{triangle
(\taxgraph[baseline=-0.5ex]{
    \node[taxNode] (a) at (-.55,-.36) {};
    \node[taxNode] (b) at (.55,-.36) {};
    \node[taxNode] (c) at (0,.45) {};
    \node[taxNode] (d) at (1.4,.4) {};
    \node[taxNode] (e) at (1.4,-.4) {};
    \draw[taxEdge] (a)--(b)--(c)--(a);
})}, a
\mbox{three-dyad path
(\taxgraph[baseline=-0.5ex]{
    \node[taxNode] (a) at (0,0) {};
    \node[taxNode] (b) at (.7,0) {};
    \node[taxNode] (c) at (1.4,0) {};
    \node[taxNode] (d) at (2.1,0) {};
    \node[taxNode] (e) at (3,0) {};
    \draw[taxEdge] (a)--(b)--(c)--(d);
})}, or a
\mbox{two-dyad path and a disjoint dyad
(\taxgraph[baseline=-0.5ex]{
    \node[taxNode] (a) at (-.55,.15) {};
    \node[taxNode] (b) at (0,.55) {};
    \node[taxNode] (c) at (.55,.15) {};
    \node[taxNode] (d) at (-.45,-.45) {};
    \node[taxNode] (e) at (.45,-.45) {};
    \draw[taxEdge] (a)--(b)--(c);
    \draw[taxEdge] (d)--(e);
})},
including isolated nodes where necessary. After deleting isolated
nodes, a connected three-dyad graph is a triangle or one of the two trees on four nodes;
a disconnected one consists of a two-dyad path and a disjoint dyad,
since three disjoint dyads would require six nodes.
Table~\ref{tab:five-node-ranks} specifies a representative of each
type by listing its missing dyads; \(E\) contains all other dyads
on the stated node set. The second representative is the pentad
with pivot \((i,j)\). We use the representatives specified in the
table throughout the calculations.

Since \(\sum_{u\in V} d_A(u)=2\lvert A\rvert\), subsets of different sizes
have different degree vectors. Also,
\(d_{E\setminus A}(u)=d_E(u)-d_A(u)\), so the numbers of distinct
vectors for sizes \(k\) and \(7-k\) agree. Enumerating subsets of
sizes \(0,1,2,3\) gives counts \((1,7,19,27)\) for the star
complement, \((1,7,18,28)\) for the triangle complement,
\((1,7,19,29)\) for the three-dyad path complement, and
\((1,7,18,28)\) for the two-dyad path with a disjoint dyad.
Doubling each total gives the ranks \(r_G\) in the table.
Our target is \(\operatorname{rank}\mathbf{C}_G(\mathbf W)\geq128\) for the
non-pentad graphs and \(\operatorname{rank}\mathbf{C}_G(\mathbf W)\geq127\)
for the pentad. Accordingly, set \(s_G=128-r_G\) for the former
and \(s_G=127-r_G\) for the latter. It then suffices to choose
positive weights \(\mathbf W\) such that
\(\operatorname{rank}\mathbf{C}_G(\mathbf W)\geq r_G+s_G\).
Table~\ref{tab:five-node-ranks} summarizes \(r_G\) and \(s_G\)
for the four graphs.

  \begin{table}[H]
  \centering
  \caption{Representative seven-dyad graphs on \(V=\{i,j,k,l,m\}\)
  and their ranks \(r_G\) and required rank increases \(s_G\).}
  \label{tab:five-node-ranks}
  {\renewcommand{\arraystretch}{1.35}%
  \setlength{\tabcolsep}{8pt}%
  \begin{tabular}{@{}ccccc@{}}
  \toprule
  Degrees & Complement & Missing dyads & \(r_G\) & \(s_G\)\\
  \midrule
  \((4,3,3,3,1)\) & \taxgraph[baseline=-0.5ex]{
    \node[taxNode] (a) at (0,0) {};
    \node[taxNode] (b) at (0,.65) {};
    \node[taxNode] (c) at (-.65,-.35) {};
    \node[taxNode] (d) at (.65,-.35) {};
    \node[taxNode] (e) at (1.4,.35) {};
    \draw[taxEdge] (a)--(b);
    \draw[taxEdge] (a)--(c);
    \draw[taxEdge] (a)--(d);
  } & \(\{(i,j),(i,k),(i,l)\}\) & 108 & 20\\
  \addlinespace[0.35ex]
  \((4,4,2,2,2)\) & \taxgraph[baseline=-0.5ex]{
    \node[taxNode] (a) at (-.55,-.36) {};
    \node[taxNode] (b) at (.55,-.36) {};
    \node[taxNode] (c) at (0,.45) {};
    \node[taxNode] (d) at (1.4,.4) {};
    \node[taxNode] (e) at (1.4,-.4) {};
    \draw[taxEdge] (a)--(b)--(c)--(a);
  } & \(\{(k,l),(k,m),(l,m)\}\) & 108 & 19\\
  \addlinespace[0.35ex]
  \((4,3,3,2,2)\) & \taxgraph[baseline=-0.5ex]{
    \node[taxNode] (a) at (0,0) {};
    \node[taxNode] (b) at (.7,0) {};
    \node[taxNode] (c) at (1.4,0) {};
    \node[taxNode] (d) at (2.1,0) {};
    \node[taxNode] (e) at (3,0) {};
    \draw[taxEdge] (a)--(b)--(c)--(d);
  } & \(\{(i,j),(j,k),(k,l)\}\) & 112 & 16\\
  \addlinespace[0.35ex]
  \((3,3,3,3,2)\) & \taxgraph[baseline=-0.5ex]{
    \node[taxNode] (a) at (-.55,.15) {};
    \node[taxNode] (b) at (0,.55) {};
    \node[taxNode] (c) at (.55,.15) {};
    \node[taxNode] (d) at (-.45,-.45) {};
    \node[taxNode] (e) at (.45,-.45) {};
    \draw[taxEdge] (a)--(b)--(c);
    \draw[taxEdge] (d)--(e);
  } & \(\{(i,j),(i,k),(l,m)\}\) & 108 & 20\\
  \bottomrule
  \end{tabular}
  }
  \end{table}

We now perturb the weights away from \(\mathbf1\) along a path
\(\mathbf W(t)=\mathbf1+t\mathbf{H}+o(t)\), where \(t\) is a scalar tending
to zero and \(\mathbf{H}=(H_{uv},H_{vu})_{(u,v)\in E}\) is a fixed real array.
Thus each directed weight satisfies \(W_{uv}(t)=1+tH_{uv}+o(t)\):
\(H_{uv}\) specifies its first-order change from one.
The aim is to increase the ranks \(r_G\) in
Table~\ref{tab:five-node-ranks} by choosing \(\mathbf{H}\) appropriately.
To study the resulting change in \(\mathbf{C}_G(\mathbf W(t))\), we use its
polynomial dependence on the weights:
expanding \(R_A\) shows that every entry of \(\mathbf{C}_G(\mathbf W)\)
is a polynomial in \(\mathbf W\) with integer coefficients.
Every minor of \(\mathbf{C}_G(\mathbf W)\) is also a polynomial in \(\mathbf W\), since it is a
determinant of these entries. Consequently, each minor is analytic
in the covariates wherever the weights are analytic in the covariates, a property we
will use later. The next lemma gives a sufficient condition
for \(\operatorname{rank}\mathbf{C}_G(\mathbf W(t))\geq r_G+s_G\).

\begin{lem}\label{lem:weight-rank-lifting}
Fix a graph \(G\), let \(r_G=\operatorname{rank}\mathbf{C}_G(\mathbf1)\), and
consider \(\mathbf W(t)=\mathbf1+t\mathbf{H}+o(t)\) as \(t\to0\). Write
\(\mathbf{C}_G(\mathbf W(t))=\mathbf{C}_G(\mathbf1)+t\dot{\mathbf{C}}_G+o(t)\), where
\(\dot{\mathbf{C}}_G:=\left.\frac{d}{dt}\mathbf{C}_G(\mathbf1+t\mathbf{H})\right|_{t=0}\)
is the derivative of \(\mathbf{C}_G\) at \(\mathbf1\) in direction \(\mathbf{H}\).
Choose fixed invertible rational matrices \(\mathbf U_G,\mathbf V_G\)
such that
\[
\mathbf U_G\mathbf{C}_G(\mathbf1)\mathbf V_G
=\begin{pmatrix}\mathbf{I}_{r_G}&0\\0&0\end{pmatrix},
\]
and let \(\mathbf{T}_G\) be the block of \(\mathbf U_G\dot{\mathbf{C}}_G\mathbf V_G\)
after its first \(r_G\) rows and columns. If
\(\operatorname{rank}\mathbf{T}_G\geq s_G\) for a nonnegative integer \(s_G\), then
\begin{equation}
\operatorname{rank}\mathbf{C}_G(\mathbf W(t))\geq r_G+s_G
\quad\text{for all nonzero }t\text{ sufficiently close to zero}.
\label{eq:rank-lifting}
\end{equation}
If the path is defined only for \(t>0\) and satisfies
\(\mathbf W(t)=\mathbf1+t\mathbf{H}+o(t)\) as \(t\downarrow0\), then
\eqref{eq:rank-lifting} holds for all sufficiently small \(t>0\).
\end{lem}

\begin{proof}
The entries of \(\mathbf{C}_G\) are polynomials in \(\mathbf W\), so the
expansion exists and \(\dot{\mathbf{C}}_G\) depends linearly on \(\mathbf{H}\).
Rational Gaussian elimination at the integer matrix \(\mathbf{C}_G(\mathbf1)\)
provides \(\mathbf U_G,\mathbf V_G\). Keep both matrices fixed as
\(t\) and \(\mathbf{H}\) vary. Then
\[
\mathbf U_G\mathbf{C}_G(\mathbf W(t))\mathbf V_G
=\begin{pmatrix}\mathbf{I}_{r_G}+O(t)&O(t)\\O(t)&t\mathbf{T}_G+o(t)\end{pmatrix}.
\]
For small \(t\), the upper-left block is invertible with bounded
inverse. Row operations eliminate the lower-left block and replace
the lower-right block by the Schur complement
\[
t\mathbf{T}_G+o(t)-O(t)(\mathbf{I}_{r_G}+O(t))^{-1}O(t)=t(\mathbf{T}_G+o(1)),
\]
since the subtracted product is \(O(t^2)=o(t)\).
Column operations then eliminate the upper-right block without
changing either diagonal block, giving
\[
\begin{pmatrix}\mathbf{I}_{r_G}+O(t)&0\\0&t(\mathbf{T}_G+o(1))\end{pmatrix}.
\]
These invertible row and column operations preserve rank. The
upper-left block has rank \(r_G\), and multiplication by \(t\neq0\)
does not change the rank of the lower-right factor. Hence
\(\operatorname{rank}\mathbf{C}_G(\mathbf W(t))=r_G+\operatorname{rank}(\mathbf{T}_G+o(1))\)
for all nonzero \(t\) sufficiently close to zero. If \(s_G\geq1\), a nonzero minor of \(\mathbf{T}_G\) of
size \(s_G\) remains nonzero in \(\mathbf{T}_G+o(1)\) by the continuity of determinant.
The case \(s_G=0\) is trivial.
For a path defined only for \(t>0\), the block elimination and
continuity argument apply using the expansion as \(t\downarrow0\).
\end{proof}

We now construct directions \(\mathbf{H}\) satisfying
\(\operatorname{rank}\mathbf{T}_G\geq s_G\), so that
Lemma~\ref{lem:weight-rank-lifting} yields the required rank bound for
\(\mathbf{C}_G(\mathbf W)\).

\begin{lem}\label{lem:five-node-directions}
For each representative graph \(G=(V,E)\) in
Table~\ref{tab:five-node-ranks}, with \(V=\{i,j,k,l,m\}\),
define \(\mathbf{T}_G\) from the direction \(\mathbf{H}\) as in
Lemma~\ref{lem:weight-rank-lifting}. Then:
\begin{enumerate}[label=\textnormal{(\alph*)},leftmargin=2em]
\item For \((z_i,z_j,z_k,z_l,z_m)=(0,1,3,7,12)\) and
\(H_{uv}=z_v\), we have \(\operatorname{rank}\mathbf{T}_G=s_G\).
\item Let \((z_i,z_j,z_k,z_l,z_m)\) be any permutation of
\((0,1,2,3,4)\), so that each node is assigned a different value.
For every \(\kappa\in\{1,2\}\) and every
\(a_+,a_-\in\mathbb R\) with \(a_+a_-\neq0\), the direction
\[
H_{uv}=\begin{cases}
a_+\lvert z_v-z_u\rvert^\kappa,&z_u<z_v,\\
a_-\lvert z_v-z_u\rvert^\kappa,&z_u>z_v
\end{cases}
\]
satisfies \(\operatorname{rank}\mathbf{T}_G\geq s_G\).
\end{enumerate}
\end{lem}

\begin{proof}
We describe the exact finite calculations establishing both assertions.
For the path \(W_{uv}(t)=1+tH_{uv}\), column \(A\) of \(\dot{\mathbf{C}}_G\)
is obtained by collecting the coefficients in
\[
\left.\frac{d}{dt}R_A(\bs \alpha)\right|_{t=0}
=\sum_{(u,v)\in A}
\{H_{uv}\alpha_u+H_{vu}\alpha_v+(H_{uv}+H_{vu})\alpha_u\alpha_v\}
\prod_{(a,b)\in A\setminus\{(u,v)\}}(1+\alpha_a)(1+\alpha_b).
\]
For \(A=\varnothing\), the derivative is zero. Deleting the first
\(r_G\) rows and columns of \(\mathbf U_G\dot{\mathbf{C}}_G\mathbf V_G\)
gives \(\mathbf{T}_G\), which depends linearly on \(\mathbf{H}\).

For \textnormal{(a)}, use the four graphs in
Table~\ref{tab:five-node-ranks} in order, with
\((z_i,z_j,z_k,z_l,z_m)=(0,1,3,7,12)\). These values determine
\(H_{uv}=z_v\), allowing every entry of \(\mathbf{T}_G\) to be computed
exactly as a rational number. Exact rational elimination then
verifies that \(\operatorname{rank}\mathbf{T}_G\) equals \(20,19,16,20\),
respectively, matching \(s_G\) for each graph.
The script \texttt{verify\_pentad\_minimality.py}
performs these exact rank calculations.\footnote{Computer algebra has also been used to derive and verify
moment conditions in functional differencing; see, e.g.,
\citet{honore2024moment}. Our computer-assisted calculations throughout
this proof use exact rational and algebraic arithmetic, so no numerical
approximation is involved. In particular, all matrix ranks are computed exactly.
The script is available at
\url{https://www.zeqiwu.com/research/NTU/verify_pentad_minimality.py}.}

For \textnormal{(b)}, fix a representative graph and
\(\kappa\in\{1,2\}\), and let \((z_i,z_j,z_k,z_l,z_m)\) be any
permutation of \((0,1,2,3,4)\). There are \(5!=120\) such assignments.
Keep these values fixed throughout the calculation. Let
\(\mathbf{T}_{G,+}\) and \(\mathbf{T}_{G,-}\) denote the matrices \(\mathbf{T}_G\) computed
from the direction \(\mathbf{H}\) in \textnormal{(b)} with
\((a_+,a_-)=(1,0)\) and \((0,1)\), respectively.
Both constructions use the same fixed
elimination matrices \(\mathbf U_G,\mathbf V_G\). Linearity in \(\mathbf{H}\) gives
\(\mathbf{T}_G=a_+\mathbf{T}_{G,+}+a_-\mathbf{T}_{G,-}\). With \(\rho=a_-/a_+\neq0\),
\(\operatorname{rank}\mathbf{T}_G=\operatorname{rank}(\mathbf{T}_{G,+}+\rho \mathbf{T}_{G,-})\).
Both blocks have rational entries because the node values are integers
and the elimination matrices are rational.

At \(\rho=1\), every entry of \(\mathbf{T}_{G,+}+\mathbf{T}_{G,-}\) can be computed
exactly as a rational number. Exact rational elimination verifies
\(\operatorname{rank}(\mathbf{T}_{G,+}+\mathbf{T}_{G,-})\geq s_G\).
For each graph and each \(\kappa\in\{1,2\}\), we carry out this
calculation for all \(120\) assignments of the node values.
Scan the columns of \(\mathbf{T}_{G,+}+\mathbf{T}_{G,-}\) in their fixed order and retain a column
whenever it is independent of those already retained, until \(s_G\)
columns have been selected. Within these columns, select \(s_G\)
independent rows by the same rule. Denote the selected row and column
sets by \(I,J\), and set
\(p(\rho)=\det\{(\mathbf{T}_{G,+})_{I,J}+\rho(\mathbf{T}_{G,-})_{I,J}\}\).
The polynomial \(p(\rho)\) has rational coefficients, degree at most \(s_G\),
and \(p(1)\neq0\). For any \(\rho\neq0\) with \(p(\rho)\neq0\),
the selected \(s_G\times s_G\) submatrix of
\(\mathbf{T}_{G,+}+\rho \mathbf{T}_{G,-}\) is nonsingular, so
\(\operatorname{rank}\mathbf{T}_G=\operatorname{rank}(\mathbf{T}_{G,+}+\rho \mathbf{T}_{G,-})\geq s_G\).

Factor \(p(\rho)\) into irreducible polynomials with rational
coefficients, retaining the constant factor and the multiplicity
of each factor.
In every calculated case the nonconstant factors have degree one or
two. Their real roots are determined exactly: a linear factor
\(a\rho+b\) has root \(-b/a\); a quadratic factor
\(a\rho^2+b\rho+c\) has real roots
\((-b\pm\sqrt{b^2-4ac})/(2a)\) precisely when \(b^2-4ac\geq0\).
We check each distinct nonzero real root of \(p\).
For any nonzero real \(\rho\) that is not a root of \(p\), we have
\(p(\rho)\neq0\), so the selected \(s_G\times s_G\) submatrix of
\(\mathbf{T}_{G,+}+\rho \mathbf{T}_{G,-}\) is nonsingular and
\(\operatorname{rank}\mathbf{T}_G=\operatorname{rank}(\mathbf{T}_{G,+}+\rho \mathbf{T}_{G,-})\geq s_G\).

For each nonzero real root \(\rho_*\), keep the graph, \(\kappa\),
and node values fixed. We compute the rank of
\(\mathbf{T}_{G,+}+\rho_*\mathbf{T}_{G,-}\) using exact arithmetic and verify that
it is at least \(s_G\). Rational roots use exact rational arithmetic;
irrational roots use exact quadratic algebraic arithmetic, retaining
the square root symbolically.
Together with the argument for \(p(\rho)\neq0\), these calculations
verify \(\operatorname{rank}\mathbf{T}_G
=\operatorname{rank}(\mathbf{T}_{G,+}+\rho \mathbf{T}_{G,-})\geq s_G\)
for every real \(\rho\neq0\).

The script \texttt{verify\_pentad\_minimality.py} checks all
\(120\) assignments of \((z_i,z_j,z_k,z_l,z_m)\) obtained by
permuting \((0,1,2,3,4)\), for each of the four graphs and each
\(\kappa\in\{1,2\}\), giving \(4\times120\times2=960\) cases.
It verifies the initial ranks at \(\rho=1\),
finds only linear and quadratic factors, and finds at most four
distinct nonzero real roots per case. Counting roots separately across
cases gives \(1392\) roots, at each of which a nonzero minor of size
\(s_G\) is obtained.
The selected indices, factorizations, and exact
nonzero determinants can be reproduced by the script.
For each real \(\rho\neq0\), either \(p(\rho)\neq0\)
or the exact calculation at \(\rho\) supplies another nonzero minor
of size \(s_G\).
Thus \(\operatorname{rank}\mathbf{T}_G\geq s_G\) for every nonzero \(\rho\),
with the same node values throughout each case.
This proves \textnormal{(b)} for every permutation
\((z_i,z_j,z_k,z_l,z_m)\) of \((0,1,2,3,4)\).
\end{proof}

\begin{proof}[Proof of Theorem~\ref{thm:pentad_minimality}]
Fix \(\bs\beta_0\neq\mathbf0\). From now on, the weights are the
functions \(W_{uv}=\exp\{-\bs\beta_0^\top w(\bs X_u,\bs X_v)\}\)
of the node covariates. By Lemma~\ref{lem:local-moment-coefficients},
it suffices to prove that, for almost every \((\bs X_u:u\in V)\),
\(\mathbf{C}_G(\mathbf W)\) has full column rank for every graph with at
most six dyads and for each non-pentad five-node, seven-dyad graph,
and has rank \(127\) for the pentad.
We first establish analyticity and weight separation
\eqref{eq:generic-outgoing-separation}, then construct
the five-node rank bounds from Lemma~\ref{lem:five-node-directions}.

Write \(x_{u,r}\) for coordinate \(r\) of a possible value of
\(\bs X_u\). For each coordinate, fix one of conditions
\textnormal{(i)} and \textnormal{(ii)} that it satisfies.
Under \textnormal{(i)}, \(w_{r'}\) is analytic on \(\mathbb R^2\).
Under \textnormal{(ii)}, the analytic expression for
\(w_{r'}(x_{u,r'},x_{v,r'})\) depends on the sign of
\(x_{v,r'}-x_{u,r'}\). We therefore divide the covariate space
into regions on which these signs are fixed.

Let \(\mathcal P(V)
=\{(\pi_1,\ldots,\pi_{|V|})\in V^{|V|}:
\pi_a\neq\pi_b\ \text{for all }1\leq a<b\leq |V|\}\)
be the set of all permutations of the node labels in \(V\).
For each coordinate \(r'\) treated under
\textnormal{(ii)}, choose a permutation
\(\pi^{(r')}=(\pi^{(r')}_1,\ldots,\pi^{(r')}_{|V|})
\in\mathcal P(V)\).
Each combination of the permutations \(\pi^{(r')}\) defines a
region that is a subset of the covariate space
\(\mathbb R^{d\lvert V\rvert}\). It consists of all vectors
\((x_{u,r}:u\in V,\ r=1,\ldots,d)\) satisfying
\[
x_{\pi^{(r')}_1,r'}<\cdots<x_{\pi^{(r')}_{|V|},r'}
\qquad
\text{for every coordinate \(r'\) treated under \textnormal{(ii)}}.
\]
Coordinates treated under \textnormal{(i)} are unrestricted.
Each region is open and convex, hence connected, in
\(\mathbb R^{d\lvert V\rvert}\).
We consider all combinations of the permutations \(\pi^{(r')}\),
allowing different coordinates to use different permutations.
The resulting finitely many regions cover the space except for the hyperplanes
\(x_{u,r'}=x_{v,r'}\), where \(u\neq v\) and \(r'\) is treated
under \textnormal{(ii)}. If no coordinate is treated under
\textnormal{(ii)}, the entire space is one region.

Within each region, a fixed analytic branch \(f_{r',+}\) or
\(f_{r',-}\) applies to each term
\(w_{r'}(x_{u,r'},x_{v,r'})\) treated under \textnormal{(ii)}.
It follows that, with \(\bs\beta_0\) fixed,
\(W_{uv}=\exp\{-\sum_{r'=1}^d\beta_{0r'}w_{r'}(x_{u,r'},x_{v,r'})\}\)
is real analytic in all \(d\lvert V\rvert\) coordinates
\((x_{u,r}:u\in V,\ r=1,\ldots,d)\) on that region: sums and
composition with the exponential preserve analyticity. The entries
and minors of \(\mathbf{C}_G(\mathbf W)\) are polynomials in the weights (\(W_{uv},W_{vu}:\ u,v\in V\)),
so they are real analytic in
\((x_{u,r}:u\in V,\ r=1,\ldots,d)\) on each region.
The weights, and hence the minors, are continuous across coordinate
ties because \(f_{r',+}(0)=f_{r',-}(0)\).

We first prove \eqref{eq:generic-outgoing-separation} for almost every
\((\bs X_u:u\in V)\). Fix distinct nodes \(u,v,v'\), with \(v,v'\)
neighbors of \(u\), and define
\(D_{u;v,v'}=\bs\beta_0^\top\{w(\bs X_u,\bs X_v)-w(\bs X_u,\bs X_{v'})\}\).
The equality \(W_{uv}=W_{uv'}\) is equivalent to \(D_{u;v,v'}=0\).
It therefore suffices to show that \(D_{u;v,v'}\neq0\) for almost
every \((\bs X_u:u\in V)\).
Fix one of the regions and a coordinate \(r\) with \(\beta_{0r}\neq0\).
The function \(D_{u;v,v'}\) is analytic on the region. If \(r\) is
treated under \textnormal{(i)}, set \(x_{u,r}=x_{v',r}=x_r^\circ\),
vary \(x_{v,r}\) near \(x_r^\circ\), and hold all remaining
coordinates fixed at values satisfying the inequalities defining
the chosen region. This path stays in the region because coordinate
\(r\) is unrestricted under \textnormal{(i)}. At \(x_{v,r}=x_r^\circ\),
\(\partial D_{u;v,v'}/\partial x_{v,r}
=\beta_{0r}\partial_{y_r}w_r(x_r^\circ,x_r^\circ)\neq0\).
If \(r\) is treated under \textnormal{(ii)}, then within the region
\[
\frac{\partial^{\kappa_r}D_{u;v,v'}}{\partial x_{v,r}^{\kappa_r}}
=\begin{cases}
\beta_{0r}f_{r,+}^{(\kappa_r)}(x_{v,r}-x_{u,r}),&x_{u,r}<x_{v,r},\\
(-1)^{\kappa_r}\beta_{0r}f_{r,-}^{(\kappa_r)}(x_{u,r}-x_{v,r}),&x_{u,r}>x_{v,r}.
\end{cases}
\]
The derivative is nonzero for sufficiently small
\(\lvert x_{v,r}-x_{u,r}\rvert\), by continuity and the nonzero
derivatives at zero in the theorem. For every node \(a\in V\),
replace \(x_{a,r}\) by \(\varepsilon x_{a,r}\), keeping all other
coordinates fixed. For \(\varepsilon>0\), these covariates remain
in the chosen region because all strict inequalities are preserved.
As \(\varepsilon\downarrow0\), the difference
\(\lvert\varepsilon x_{v,r}-\varepsilon x_{u,r}\rvert
=\varepsilon\lvert x_{v,r}-x_{u,r}\rvert\) tends to zero.
Hence \(\partial^{\kappa_r}D_{u;v,v'}/\partial x_{v,r}^{\kappa_r}\)
is nonzero at the scaled covariates for sufficiently small
\(\varepsilon>0\). Under either condition, \(D_{u;v,v'}\)
therefore has a nonzero derivative somewhere in the region and is
not identically zero. Its zero set has Lebesgue measure zero by
\citet[Proposition~1]{mityagin2020zero}. Each hyperplane
\(x_{a,r'}=x_{b,r'}\), where \(a,b\in V\), \(a\neq b\), and
\(r'\) is treated under \textnormal{(ii)}, also has Lebesgue measure
zero. These hyperplanes consist of covariate vectors for which two
nodes have equal values in coordinate \(r'\).
Taking the finite union of these hyperplanes and the zero sets over
all regions and all triples \((u,v,v')\) of distinct nodes with
\(v,v'\) both neighbors of \(u\) therefore proves
\eqref{eq:generic-outgoing-separation} almost everywhere.

Now fix one of the four labeled graphs in
Table~\ref{tab:five-node-ranks} and one region in
\(\mathbb R^{5d}\). We construct covariates in that region for which
\(\operatorname{rank}\mathbf{C}_G(\mathbf W)\geq r_G+s_G\).
Choose \(r\) with \(\beta_{0r}\neq0\). Initially set
\(x_{u,r'}=\bar x_{r'}\) for every node \(u\) and each \(r'\neq r\),
where \(\bar x_{r'}\) are arbitrary fixed real numbers. These
coordinates will be separated after constructing a nonzero minor.

If \(r\) is treated under \textnormal{(i)}, use the labels fixed
in Table~\ref{tab:five-node-ranks}, take
\((z_i,z_j,z_k,z_l,z_m)=(0,1,3,7,12)\), and set
\(x_{u,r}=x_r^\circ+tz_u\) for small \(t>0\).
The original weights along this covariate path are
\(W_{uv}(t)=\exp\{
-\sum_{r'\neq r}\beta_{0r'}w_{r'}(\bar x_{r'},\bar x_{r'})
-\beta_{0r}w_r(x_r^\circ+tz_u,x_r^\circ+tz_v)\}\).
Set
\(\widetilde W_{uv}(t)=s_uW_{uv}(t)\), where
\(s_u=\exp\{\sum_{r'\neq r}\beta_{0r'}w_{r'}(\bar x_{r'},\bar x_{r'})
+\beta_{0r}w_r(x_r^\circ+tz_u,x_r^\circ)\}\).
The diagonal row transformation in
\eqref{eq:outgoing-weight-invariance} preserves rank. Taylor expansion
in the second argument gives
\[
\begin{aligned}
\widetilde W_{uv}(t)
&=s_uW_{uv}(t)=\exp\!\left[-\beta_{0r}
\{w_r(x_r^\circ+tz_u,x_r^\circ+tz_v)-w_r(x_r^\circ+tz_u,x_r^\circ)\}\right]\\
&=1-t\beta_{0r}\frac{\partial w_r}{\partial y_r}(x_r^\circ,x_r^\circ)z_v+O(t^2).
\end{aligned}
\]
The direction \(H_{uv}=z_v\) is covered by
Lemma~\ref{lem:five-node-directions}\textnormal{(a)}. The expansion
multiplies that direction by
\(-\beta_{0r}\partial_{y_r}w_r(x_r^\circ,x_r^\circ)\neq0\).
Since \(\mathbf{T}_G\) depends linearly on \(\mathbf{H}\), multiplying \(\mathbf{H}\) by
\(-\beta_{0r}\partial_{y_r}w_r(x_r^\circ,x_r^\circ)\neq0\) multiplies \(\mathbf{T}_G\)
by that scalar and leaves its rank equal to \(s_G\).
Lemma~\ref{lem:weight-rank-lifting} gives
\(\operatorname{rank}\mathbf{C}_G(\widetilde{\mathbf W}(t))\geq r_G+s_G\).
By \eqref{eq:outgoing-weight-invariance},
\(\operatorname{rank}\mathbf{C}_G(\mathbf W(t))\geq r_G+s_G\)
for all sufficiently small \(t>0\).

If \(r\) is treated under \textnormal{(ii)}, assign \(0,1,2,3,4\)
to \((z_u:u\in V)\) in the precise node order required by the chosen
region in coordinate \(r\). For example, if the chosen region requires
\(x_{k,r}<x_{i,r}<x_{m,r}<x_{j,r}<x_{l,r}\), take
\((z_k,z_i,z_m,z_j,z_l)=(0,1,2,3,4)\).
Lemma~\ref{lem:five-node-directions}\textnormal{(b)}
applies to this assignment of the node values. Set \(x_{u,r}=c+tz_u\) for any fixed
\(c\in\mathbb R\) and small \(t>0\), and put
\(a_\pm=f_{r,\pm}^{(\kappa_r)}(0)/\kappa_r!\).
Both coefficients are nonzero. Define \(\mathbf{H}\) by
Lemma~\ref{lem:five-node-directions}\textnormal{(b)} with
\(\kappa=\kappa_r\). Use the  scaling
\(\widetilde W_{uv}(t)=s_uW_{uv}(t)\), now with
\(s_u=\exp\{\sum_{r'\neq r}\beta_{0r'}w_{r'}(\bar x_{r'},\bar x_{r'})
+\beta_{0r}f_{r,+}(0)\}\) for every \(u\).
Taylor expansion of \(f_{r,+}\) or \(f_{r,-}\) at zero,
according to the sign of \(z_v-z_u\), gives
\(w_r(c+tz_u,c+tz_v)-f_{r,+}(0)=t^{\kappa_r}H_{uv}+o(t^{\kappa_r})\).
When \(\kappa_r=2\), the linear term vanishes because
\(f'_{r,+}(0)=f'_{r,-}(0)=0\). Consequently,
\[
\widetilde W_{uv}(t)
=\exp\!\left[-\beta_{0r}\{w_r(c+tz_u,c+tz_v)-f_{r,+}(0)\}\right]
=1-t^{\kappa_r}\beta_{0r}H_{uv}+o(t^{\kappa_r}).
\]
Since \(\mathbf{T}_G\) depends linearly on \(\mathbf{H}\), multiplying \(\mathbf{H}\) by
\(-\beta_{0r}\neq0\) multiplies \(\mathbf{T}_G\) by that scalar
and leaves its rank at least \(s_G\).
Set \(\tau=t^{\kappa_r}\). Then
\(\widetilde{\mathbf W}(\tau^{1/\kappa_r})
=\mathbf1+\tau(-\beta_{0r}\mathbf{H})+o(\tau)\) as \(\tau\downarrow0\).
Lemma~\ref{lem:weight-rank-lifting}, applied with parameter \(\tau\)
and direction \(-\beta_{0r}\mathbf{H}\), gives
\(\operatorname{rank}\mathbf{C}_G(\widetilde{\mathbf W}(\tau^{1/\kappa_r}))
\geq r_G+s_G\) for all sufficiently small \(\tau>0\).
Since \(t=\tau^{1/\kappa_r}\), \eqref{eq:outgoing-weight-invariance}
then yields \(\operatorname{rank}\mathbf{C}_G(\mathbf W(t))\geq r_G+s_G\)
for all sufficiently small \(t>0\).

In either case, fix one such \(t>0\) and a nonzero minor of size
\(r_G+s_G\) of the original matrix \(\mathbf{C}_G(\mathbf W(t))\).
Recall that we have held \(x_{u,r'}=\bar x_{r'}\) fixed for every
node \(u\in V\) and every coordinate \(r'\neq r\).
We now adjust the coordinates \(r'\neq r\) treated under
\textnormal{(ii)} so that the full covariate vector lies in the
chosen region.
For each coordinate \(r'\neq r\) treated under
\textnormal{(ii)}, choose a value \(x_{u,r'}\) near \(\bar x_{r'}\)
for each of the five nodes \(u\in V\), with the five values distinct
and ordered as required by the chosen region in coordinate \(r'\).
Keep coordinate \(r\) and all coordinates
treated under \textnormal{(i)} fixed. Continuity of the weights and
determinants across ties ensures that the selected minor stays nonzero
for sufficiently small changes. We have therefore constructed values of
\((x_{u,r'}:u\in V,\ r'=1,\ldots,d)\) inside the chosen region
where \(\operatorname{rank}\mathbf{C}_G(\mathbf W)\geq r_G+s_G\).

On the chosen region, the selected minor of size \(r_G+s_G\) is
real analytic in \((x_{u,r}:u\in V,\ r=1,\ldots,d)\) and is
nonzero at the constructed covariate vector. Its zero set therefore
has Lebesgue measure zero by \citet[Proposition~1]{mityagin2020zero}.
Outside this zero set, the selected \((r_G+s_G)\times(r_G+s_G)\)
submatrix is nonsingular, so
\(\operatorname{rank}\mathbf{C}_G(\mathbf W)\geq r_G+s_G\).
The construction applies to every region. The finite union of the
zero sets of the minors selected in the respective regions, together
with the hyperplanes
\(x_{a,r'}=x_{b,r'}\) for distinct nodes \(a,b\in V\) and coordinates
\(r'\) treated under \textnormal{(ii)}, has Lebesgue measure zero.
Thus, for each of the four graphs \(G\) specified in
Table~\ref{tab:five-node-ranks},
\(\operatorname{rank}\mathbf{C}_G(\mathbf W)\geq r_G+s_G\) for almost every
\((\bs X_u:u\in V)\).
Every other labeled five-node, seven-dyad graph is obtained by
relabeling the nodes of one of the four graphs in
Table~\ref{tab:five-node-ranks}. Apply the same relabeling to the
node covariates.
The relabeling preserves rank by
Lemma~\ref{lem:coefficient-matrix-properties}\textnormal{(a)}.
The coordinate permutation also preserves sets of Lebesgue measure zero,
so the rank bound holds almost everywhere for the relabeled graph.
Hence every non-pentad five-node, seven-dyad graph has rank \(128\)
almost everywhere. For the pentad, the lower bound is \(127\);
Lemma~\ref{lem:separated-weight-ranks}\textnormal{(b)}
gives the matching upper bound.
The pentad therefore has rank \(127\) almost everywhere.

We next show that \(\mathbf{C}_G(\mathbf W)\) has full column rank almost
everywhere for every graph \(G\) with at most six dyads.
First, we show that every graph on at most five nodes with at most
six dyads is contained in a non-pentad seven-dyad graph on five nodes.
To see this, take any graph \(G=(V,E)\) with \(\lvert V\rvert\leq5\)
and \(\lvert E\rvert\leq6\). Add \(5-\lvert V\rvert\) isolated
nodes to \(G\), then add \(7-\lvert E\rvert\) dyads between
previously unconnected pairs of the five nodes. The resulting graph
contains \(G\) and has exactly five nodes and seven dyads.
If the resulting graph is a pentad, at least one of its dyads was
added. Replace one added dyad by a dyad among the three peripheral
nodes. The complement of the modified five-node graph consists of
the two dyads still missing among the three peripheral nodes and
the removed pentad dyad. Since the removed dyad is not
between two peripheral nodes, the complement cannot be a triangle.
The modified graph is non-pentad and still contains \(G\).
Since \(G\) is a subgraph of the resulting non-pentad five-node
graph, Lemma~\ref{lem:coefficient-matrix-properties}\textnormal{(c)}
implies that \(\mathbf{C}_G(\mathbf W)\) has full column rank for almost
every choice of all five nodes' covariates.

The full-column-rank conclusion for \(\mathbf{C}_G(\mathbf W)\) above is with respect to
the covariates of all five nodes. If \(\lvert V\rvert<5\), we now
deduce that it holds for almost every choice of the covariates of
the original nodes in \(V\). Let
\(\mathcal N_G\subseteq\mathbb R^{d\lvert V\rvert}\) be the set of
covariate vectors \((x_{u,r}:u\in V,\ r=1,\ldots,d)\) for which
\(\mathbf{C}_G(\mathbf W)\) does not have full column rank.
Since \(\mathbf{C}_G(\mathbf W)\) does not depend on the added nodes'
covariates, the set where it fails to have full column rank in
\(\mathbb R^{5d}\) is exactly
\(\mathcal N_G\times\mathbb R^{d(5-\lvert V\rvert)}\).
The full-column-rank result for all five nodes' covariates shows
that this product has Lebesgue measure zero. Its subset
\(\mathcal N_G\times[0,1]^{d(5-\lvert V\rvert)}\) therefore also
has measure zero. Since the unit cube \([0,1]^{d(5-\lvert V\rvert)}\) has volume one in \(\mathbb R^{d(5-\lvert V\rvert)}\), Fubini's
theorem implies that \(\mathcal N_G\) has Lebesgue measure zero.
Hence \(\mathbf{C}_G(\mathbf W)\) has full column rank for almost every
\((\bs X_u:u\in V)\).

For an arbitrary graph \(G=(V,E)\) with at most six dyads, consider each connected
component \(G_j=(V_j,E_j)\). If \(\lvert V_j\rvert\leq5\),
the result for graphs on at most five nodes gives full column rank
of \(\mathbf{C}_{G_j}(\mathbf W)\) for almost every
\((\bs X_u:u\in V_j)\). If a component
has \(n\geq6\) nodes and \(m\) dyads, connectedness and the total
dyad bound give \(n-1\leq m\leq6\leq n\). Thus \(m=n-1\) or
\(m=n\). Since the component is connected, it is a tree when
\(m=n-1\) and has exactly one cycle when \(m=n\).
Since we have established \eqref{eq:generic-outgoing-separation}
for almost every \((\bs X_u:u\in V_j)\),
Lemma~\ref{lem:separated-weight-ranks}\textnormal{(a)} implies that
\(\mathbf{C}_{G_j}(\mathbf W)\) has full column rank almost everywhere.
For each connected component \(G_j\), the set of covariate vectors
\((\bs X_u:u\in V_j)\) for which \(\mathbf{C}_{G_j}(\mathbf W)\) does not
have full column rank has Lebesgue measure zero. Allowing the
covariates of nodes in \(V\setminus V_j\) to vary freely gives a
set of Lebesgue measure zero in \(\mathbb R^{d\lvert V\rvert}\),
by Fubini's theorem. Since \(G\) has finitely many connected
components, the union of the sets in \(\mathbb R^{d\lvert V\rvert}\)
where \(\mathbf{C}_{G_j}(\mathbf W)\) fails to have full column rank has Lebesgue measure zero.
Thus, for almost every \((\bs X_u:u\in V)\),
the matrices \(\mathbf{C}_{G_j}(\mathbf W)\) have full column rank
simultaneously for all connected components \(G_j\) of \(G\).
Lemma~\ref{lem:coefficient-matrix-properties}\textnormal{(d)}
then gives \(\operatorname{rank}\mathbf{C}_G=\prod_j2^{\lvert E_j\rvert}
=2^{\lvert E\rvert}\). An isolated node has coefficient matrix
\((1)\), so it does not change the rank product.

We have shown that \(\mathbf{C}_G(\mathbf W)\) has full column rank
\(2^{\lvert E\rvert}\) for almost every \((\bs X_u:u\in V)\)
whenever \(G\) has at most six dyads or is a non-pentad five-node,
seven-dyad graph. For a pentad, \(\mathbf{C}_G(\mathbf W)\) has rank
\(127\) for almost every \((\bs X_u:u\in V)\).

Since the node covariates have a joint density under our assumptions,
the above rank results hold with probability one for each fixed graph
and labeling. For fixed \(N\), only finitely many graphs and
labelings are under consideration, so the statements hold
simultaneously with probability one.
For any \(\psi_G\) satisfying \eqref{eq:graph-local-moment},
Lemma~\ref{lem:local-moment-coefficients} gives
\(\bs c\in\operatorname{null}\mathbf{C}_G(\mathbf W)\) almost surely. Full column rank
forces \(\bs c=\mathbf0\), so
\(\psi_G(\bs l,\mathbf X_G;\bs\beta_0)
=\sum_{A\subseteq E}c_A(\mathbf X_G)
\prod_{(u,v)\in E\setminus A}l_{uv}=0\)
for every link pattern \(\bs l\in\{0,1\}^{\lvert E\rvert}\).
This proves part~\textnormal{(a)}, since
a simple graph with fewer than five nodes has at most six dyads.
For a non-pentad five-node, seven-dyad graph, full column rank
of \(\mathbf{C}_G(\mathbf W)\) forces \(\psi_G\) to be trivial.
Thus, among five-node, seven-dyad graphs, only a pentad can
admit a nontrivial moment function satisfying
\eqref{eq:graph-local-moment}.

For arbitrary distinct node labels \(i,j,k,l,m\), let
\(G=(V,E)\) be the pentad with \(V=\{i,j,k,l,m\}\)
and pivot \((i,j)\).
Recall that \(\bs c^{\mathrm{pent}}
=(c_A^{\mathrm{pent}})_{A\subseteq E}\) denotes the coefficients
of \(\prod_{(u,v)\in E\setminus A}L_{uv}\) in
\(\psi_{ijklm}(\bs\beta_0)\).
Since \eqref{eq:generic-outgoing-separation} holds almost surely,
the proof of Lemma~\ref{lem:separated-weight-ranks}\textnormal{(b)}
gives \(\bs c^{\mathrm{pent}}\in
\operatorname{null}\mathbf{C}_G(\mathbf W)\setminus\{\mathbf0\}\)
almost surely.
For any \(\psi_G\) satisfying \eqref{eq:graph-local-moment}, let
\(\bs c=(c_A(\mathbf X_G))_{A\subseteq E}\) be the coefficients
of \(\prod_{(u,v)\in E\setminus A}L_{uv}\) in
\(\psi_G(\mathbf L_G,\mathbf X_G;\bs\beta_0)\).
Lemma~\ref{lem:local-moment-coefficients} gives
\(\mathbf{C}_G(\mathbf W)\bs c=\mathbf0\) almost surely.
Since \(\mathbf{C}_G(\mathbf W)\) has \(128\) columns and rank \(127\)
almost surely, its null space is one-dimensional. Both \(\bs c\)
and the nonzero vector \(\bs c^{\mathrm{pent}}\) lie in this
null space, so \(\bs c\) is a scalar multiple of
\(\bs c^{\mathrm{pent}}\) almost surely.
By \eqref{eq:link-pattern-values}, this proportionality implies
proportionality of the function values for every link pattern. Define the factor
directly from those values by
\[
C(\mathbf X_G)
=\frac{\displaystyle\sum_{\bs l\in\{0,1\}^7}
\psi_G(\bs l,\mathbf X_G;\bs\beta_0)
\left(\left.\psi_{ijklm}(\bs\beta_0)\right|_{\mathbf L_G=\bs l}\right)}
{\displaystyle\sum_{\bs l\in\{0,1\}^7}
\left(\left.\psi_{ijklm}(\bs\beta_0)\right|_{\mathbf L_G=\bs l}\right)^2}
\]
when the denominator is positive, and set \(C(\mathbf X_G)=0\)
otherwise. Each summand is measurable in \(\mathbf X_G\), and the
denominator is positive almost surely because the pentad moment is
nonzero. Thus \(C\) is measurable, and for every link pattern,
\(\psi_G(\bs l,\mathbf X_G;\bs\beta_0)
=C(\mathbf X_G)\left.\psi_{ijklm}(\bs\beta_0)\right|_{\mathbf L_G=\bs l}\)
almost surely. Taking \(\bs l=\mathbf L_G\) gives
\(\psi_G(\mathbf L_G,\mathbf X_G;\bs\beta_0)
=C(\mathbf X_G)\psi_{ijklm}(\bs\beta_0)\) almost surely,
proving part~\textnormal{(b)}. The argument applies to every fixed
\(\bs\beta_0\neq\mathbf0\). The proof is completed.
\end{proof}
\subsection{Relabeling identities}

\begin{lem}
\label{lem:pivsym}
\label{lem:periphantisym}
With ordered dyadic covariates, the pentad moment function \(\psi_{ijklm}\) has the following relabeling
properties. \textnormal{(i)} It is invariant under interchange of the pivotal
nodes $i$ and $j$:
$\psi_{ijklm}(\bs\beta)=\psi_{jiklm}(\bs\beta)$.
\textnormal{(ii)} It is alternating under permutations of the peripheral nodes
$k,l,m$: for any permutation $\sigma$ of $\{k,l,m\}$,
\(\psi_{ij,\sigma(k,l,m)}(\bs\beta)
=\sgn(\sigma)\psi_{ijklm}(\bs\beta)\), where \(\sgn(\sigma)=1\) if
\(\sigma\) is even and \(\sgn(\sigma)=-1\) if \(\sigma\) is odd.
\end{lem}

\begin{proof}
For part~(i), the definitions in \eqref{eq:delta-app} and
\eqref{eq:feasible-coefficients-app} imply that interchanging \(i\) and \(j\) gives
\(\Delta_{jir}(\bs\beta)
=W_{jr}(\bs\beta)W_{ir}(\bs\beta)
\{W_{ri}(\bs\beta)-W_{rj}(\bs\beta)\}
=-\Delta_{ijr}(\bs\beta)\),
while the pivot weights \(W_{ij}(\bs\beta)\) and \(W_{ji}(\bs\beta)\)
exchange roles. Substituting
these identities in the entries of
\(\widetilde{\mathbf M}_{jiklm}(\bs\beta)\) gives
$\widetilde b_{jik}(\bs\beta)=-\widetilde c_{ijk}(\bs\beta)$,
$\widetilde c_{jik}(\bs\beta)=-\widetilde b_{ijk}(\bs\beta)$, and
$\widetilde d_{jik}(\bs\beta)=-\widetilde d_{ijk}(\bs\beta)$
for every peripheral nodes $k,l,m$.
Thus $\widetilde{\mathbf M}_{jiklm}(\bs\beta)$ is obtained from
$\widetilde{\mathbf M}_{ijklm}(\bs\beta)$
by swapping the first two columns and multiplying all three columns
by $-1$. The determinant multiplier is $(-1)(-1)^3=1$, so
$\psi_{jiklm}(\bs\beta)=\psi_{ijklm}(\bs\beta)$.
For part~(ii), permuting the peripheral nodes only permutes the rows of
$\widetilde{\mathbf M}_{ijklm}(\bs\beta)$ by the same permutation.
Therefore the determinant
is multiplied by the sign of that permutation:
$\psi_{ij,\sigma(k,l,m)}(\bs\beta)=\sgn(\sigma)\psi_{ijklm}(\bs\beta)$.
\end{proof}

\section{Auxiliary Probability Tools}
\label{app:probtools}

\begin{lem}[Conditional Lindeberg central limit theorem for independent arrays]
\label{lem:conditional_triangular_array_clt}
For each \(N\), let \(\mathcal F_N\) be a sigma-field and let
\(\{X_{i,N}:i\in\mathcal I_N\}\) be conditionally independent random
variables given \(\mathcal F_N\), where \(\mathcal I_N\) is finite,
\(\E[X_{i,N}\mid\mathcal F_N]=0\), and
\(\E[X_{i,N}^2\mid\mathcal F_N]<\infty\). Set
\(S_N:=\sum_{i\in\mathcal I_N}X_{i,N}\) and
\(V_N:=\sum_{i\in\mathcal I_N}\E[X_{i,N}^2\mid\mathcal F_N]\).
Suppose that, for some deterministic \(\sigma^2\ge0\),
\(V_N\xrightarrow{P}\sigma^2\) and, for every \(\varepsilon>0\),
\(\sum_{i\in\mathcal I_N}
\E\!\left[X_{i,N}^2\1\{\abs{X_{i,N}}>\varepsilon\}
\mid\mathcal F_N\right]\xrightarrow{P}0\). Then
\(S_N\xrightarrow{D}N(0,\sigma^2)\).
\end{lem}

\begin{proof}
Conditional independence and centering imply
\(\E[S_N\mid\mathcal F_N]=0\) and
\(\Var(S_N\mid\mathcal F_N)=V_N\).
When \(\sigma^2=0\), conditional Chebyshev's inequality gives
\(\Pr\{\abs{S_N}>\varepsilon\mid\mathcal F_N\}
\le V_N/\varepsilon^2\xrightarrow{P}0\) for every \(\varepsilon>0\).
Because these conditional probabilities are bounded by one, their
convergence to zero also holds in \(L^1\). Taking expectations gives
\(S_N\xrightarrow{P}0\).

Suppose \(\sigma^2>0\), and fix an arbitrary subsequence
\(N^{(0)}_1<N^{(0)}_2<\cdots\). For \(m\in\mathbb N\), write
\[
L_N(m):=\sum_{i\in\mathcal I_N}
\E\!\left[X_{i,N}^2\1\{\abs{X_{i,N}}>1/m\}
\mid\mathcal F_N\right].
\]
Convergence in probability implies that every subsequence has a further
subsequence converging almost surely. Hence one may first choose a
subsequence \(\{N^{(1)}_k\}_{k\ge1}\) of \(\{N^{(0)}_k\}_{k\ge1}\) such
that \(V_{N^{(1)}_k}\to\sigma^2\) almost surely. Recursively, after
\(\{N^{(m)}_k\}_{k\ge1}\) has been chosen, select a subsequence
\(\{N^{(m+1)}_k\}_{k\ge1}\) of it such that
\(L_{N^{(m+1)}_k}(m)\to0\) almost surely. Now define the diagonal
subsequence by \(\widetilde N_k:=N^{(k+1)}_k\). For every fixed
\(m\in\mathbb N\), the tail \(\{\widetilde N_k:k\ge m\}\) is a
subsequence of \(\{N^{(m+1)}_k\}_{k\ge1}\). Since the intersection of
countably many events of probability one also has probability one,
\(V_{\widetilde N_k}\to\sigma^2\) and
\(L_{\widetilde N_k}(m)\to0\) for every \(m\in\mathbb N\) hold
simultaneously almost surely. Relabelling
this diagonal subsequence by \(N\), we have, simultaneously for every
\(m\in\mathbb N\),
\[
\sum_{i\in\mathcal I_N}
\E\!\left[X_{i,N}^2\1\{\abs{X_{i,N}}>1/m\}
\mid\mathcal F_N\right]\to0
\qquad\text{almost surely}.
\]
For each \(N\), the finite real-valued vector
\((X_{i,N}:i\in\mathcal I_N)\) admits a regular conditional distribution
given \(\mathcal F_N\); see
\citet[Theorem~4.1.17]{durrett2019probability}. Take all conditional
expectations in the lemma to be the versions obtained by integration
with respect to this conditional distribution. Conditional independence
and centering imply that, outside a single set of probability zero,
the coordinates are independent and centered under these conditional distributions, and
their total variance is \(V_N\). On the same probability-one set, for
each \(\varepsilon>0\), choose \(m\) such that \(1/m<\varepsilon\). Then
\[
\sum_{i\in\mathcal I_N}
\E\!\left[X_{i,N}^2\1\{\abs{X_{i,N}}>\varepsilon\}
\mid\mathcal F_N\right]
\le
\sum_{i\in\mathcal I_N}
\E\!\left[X_{i,N}^2\1\{\abs{X_{i,N}}>1/m\}
\mid\mathcal F_N\right]\to0.
\]
The Lindeberg--Feller central limit theorem applied to these conditional laws
therefore gives, for every \(t\in\mathbb R\),
\(
\E[\exp(\mathrm{i}tS_N)\mid\mathcal F_N]
\to\exp(-t^2\sigma^2/2)
\) almost surely.
Conditional characteristic functions are bounded by one, so dominated
convergence and the law of iterated expectations give
\(\E[\exp(\mathrm{i}tS_N)]\to\exp(-t^2\sigma^2/2)\) along the further
subsequence for every fixed \(t\in\mathbb R\). To recover convergence
along the full sequence, fix \(t\in\mathbb R\). If
\(\E[\exp(\mathrm{i}tS_N)]\) did not converge to
\(\exp(-t^2\sigma^2/2)\), then there would exist \(\delta>0\) and a
subsequence along which
\(
\left|\E[\exp(\mathrm{i}tS_N)]-\exp(-t^2\sigma^2/2)\right|
\ge\delta
\).
Because the subsequence fixed at the beginning of the proof was arbitrary,
this subsequence contains a further subsequence along which
\(\E[\exp(\mathrm{i}tS_N)]\to\exp(-t^2\sigma^2/2)\). The absolute
difference between \(\E[\exp(\mathrm{i}tS_N)]\) and
\(\exp(-t^2\sigma^2/2)\) then converges to zero along the further
subsequence, contradicting its lower bound \(\delta\). Hence
\(\E[\exp(\mathrm{i}tS_N)]\to\exp(-t^2\sigma^2/2)\) for every
\(t\in\mathbb R\).
L\'evy's continuity theorem gives
\(S_N\xrightarrow{D}N(0,\sigma^2)\).
\end{proof}

\begin{lem}[Mixed-order multilinear-form central limit theorem under conditional independence]
\label{lem:dejong_multilinear_form}
Let \(R<\infty\) be fixed. For each \(N\), let
\(\mathcal F_N\) be a sigma-field, let
\(\{Y_{i,N}:i\in\mathcal I_N\}\) be conditionally independent random
variables given \(\mathcal F_N\), where each \(\mathcal I_N\) is finite and
\(\abs{\mathcal I_N}\to\infty\), and suppose
\(\E[Y_{i,N}\mid\mathcal F_N]=0\) and
\(\E[Y_{i,N}^2\mid\mathcal F_N]<\infty\) almost surely. Let
\(\mathcal A_{r,N}\) be a collection of subsets
\(A\subseteq\mathcal I_N\) with \(\abs A=r\), let \(a_{A,N}\) be
\(\mathcal F_N\)-measurable, and define
\(W_{r,N}:=\sum_{A\in\mathcal A_{r,N}}
a_{A,N}\prod_{i\in A}Y_{i,N}\) and
\(W_N:=\sum_{r=1}^{R}W_{r,N}\).
Write \(\sigma_{r,N}^2:=\Var(W_{r,N}\mid\mathcal F_N)\) and
\(\sigma_N^2:=\Var(W_N\mid\mathcal F_N)\).
Suppose that \(\E[W_{r,N}^4\mid\mathcal F_N]<\infty\) almost surely for
every \(r=1,\ldots,R\), \(\sigma_N^2\xrightarrow{P}\sigma^2\) for some deterministic
\(\sigma^2\ge0\), and, for every \(r=1,\ldots,R\),
\begin{equation}
\label{eq:mf-clt-fourth-condition}
\abs{\E[W_{r,N}^4\mid\mathcal F_N]-3\sigma_{r,N}^4}
\xrightarrow{P}0
\end{equation}
and
\begin{equation}
\label{eq:mf-clt-max-condition}
\max_{i\in\mathcal I_N}
\sum_{A\in\mathcal A_{r,N}:\,i\in A}
a_{A,N}^2\prod_{j\in A}\E[Y_{j,N}^2\mid\mathcal F_N]
\xrightarrow{P}0.
\end{equation}
Then \(W_N\xrightarrow{D}N(0,\sigma^2)\).
\end{lem}

\begin{proof}
We first show that
\(\sigma_N^2=\sum_{r=1}^R\sigma_{r,N}^2\), which allows us to analyze
each $W_{r,N}$ separately. If
\(A\in\mathcal A_{r,N}\) and \(B\in\mathcal A_{s,N}\) are distinct, then
the symmetric difference
\(A\triangle B:=(A\setminus B)\cup(B\setminus A)\) is nonempty.
Conditional independence and \(\E[Y_{i,N}\mid\mathcal F_N]=0\) give
\[
\E\!\left[
\prod_{i\in A}Y_{i,N}\prod_{j\in B}Y_{j,N}
\;\middle|\;\mathcal F_N
\right]
=
\prod_{k\in A\cap B}\E[Y_{k,N}^2\mid\mathcal F_N]
\prod_{k\in A\triangle B}\E[Y_{k,N}\mid\mathcal F_N]
=0.
\]
Each product is also conditionally centered. Hence, for \(A\ne B\), the
conditional covariance satisfies
\(\Cov(\prod_{i\in A}Y_{i,N},\prod_{j\in B}Y_{j,N}\mid\mathcal F_N)=0\).
Since the coefficients \(a_{A,N}\) are \(\mathcal F_N\)-measurable,
\(\E[W_{r,N}\mid\mathcal F_N]=0\),
\(\Cov(W_{r,N},W_{s,N}\mid\mathcal F_N)=0\) for \(r\ne s\), and all
cross terms in \(\Var(W_{r,N}\mid\mathcal F_N)\) vanish. Hence
\begin{equation}
\label{eq:mf-clt-variance-decomposition}
\sigma_{r,N}^2
=\sum_{A\in\mathcal A_{r,N}}a_{A,N}^2
\prod_{j\in A}\E[Y_{j,N}^2\mid\mathcal F_N],
\qquad
\sigma_N^2=\sum_{r=1}^{R}\sigma_{r,N}^2.
\end{equation}

If \(\sigma^2=0\), conditional Chebyshev's inequality gives
\(\Pr\{\abs{W_N}>\delta\mid\mathcal F_N\}\le
\sigma_N^2/\delta^2\xrightarrow{P}0\) for every \(\delta>0\). These
conditional probabilities are bounded by one, so they converge to zero
in \(L^1\). Taking expectations gives \(W_N\xrightarrow{P}0\), which
proves the result when \(\sigma^2=0\).

Suppose \(\sigma^2>0\). Fix \(t\in\mathbb R\). We first show that
\(\E[\exp(\mathrm{i}tW_N)\mid\mathcal F_N]\xrightarrow{P}
\exp(-t^2\sigma^2/2)\). If this convergence fails, there exist
\(\varepsilon,\eta>0\) and a subsequence \(\{N_k\}_{k\ge1}\) such that,
for every \(k\),
\begin{equation}
\label{eq:mf-clt-bad-event}
\Pr\!\left\{
\left|\E[\exp(\mathrm{i}tW_{N_k})\mid\mathcal F_{N_k}]
-\exp(-t^2\sigma^2/2)\right|>\varepsilon
\right\}>\eta.
\end{equation}
The variance and fourth-moment conditions in the lemma, together with
\eqref{eq:mf-clt-max-condition}, continue to hold in probability along the subsequence
\(\{N_k\}_{k\ge1}\). Hence there is a deterministic sequence
\(\delta_k\downarrow0\) such that the following three inequalities hold
simultaneously with probability approaching one:
\begin{equation}
\label{eq:mf-clt-good-event}
\begin{gathered}
\abs{\sigma_{N_k}^2-\sigma^2}\le\delta_k,
\qquad
\max_{1\le r\le R}
\abs{\E[W_{r,N_k}^4\mid\mathcal F_{N_k}]-3\sigma_{r,N_k}^4}
\le\delta_k,\\
\max_{1\le r\le R}\max_{i\in\mathcal I_{N_k}}
\sum_{A\in\mathcal A_{r,N_k}:i\in A}
a_{A,N_k}^2\prod_{j\in A}\E[Y_{j,N_k}^2\mid\mathcal F_{N_k}]
\le\delta_k.
\end{gathered}
\end{equation}
For each \(k\), the finite real-valued vector
\((Y_{i,N_k}:i\in\mathcal I_{N_k})\) admits a regular conditional
distribution given \(\mathcal F_{N_k}\); see
\citet[Theorem~4.1.17]{durrett2019probability}. We take each conditional
moment in \eqref{eq:mf-clt-good-event} to be the version obtained by
integration with respect to this conditional distribution. By the
conditional-independence and centering assumptions in the lemma, there is
a probability-one set on which the variables
\(\{Y_{i,N_k}:i\in\mathcal I_{N_k}\}\) are independent and centered under
the conditional distribution. By the union bound,
\eqref{eq:mf-clt-bad-event}, and \eqref{eq:mf-clt-good-event}, for all
sufficiently large \(k\),
\[
\begin{aligned}
&\Pr\!\left\{
\begin{gathered}
\left|\E[\exp(\mathrm{i}tW_{N_k})\mid\mathcal F_{N_k}]
-\exp(-t^2\sigma^2/2)\right|>\varepsilon,\\
\abs{\sigma_{N_k}^2-\sigma^2}\le\delta_k,\quad
\max_{1\le r\le R}
\abs{\E[W_{r,N_k}^4\mid\mathcal F_{N_k}]-3\sigma_{r,N_k}^4}
\le\delta_k,\\
\max_{1\le r\le R}\max_{i\in\mathcal I_{N_k}}
\sum_{A\in\mathcal A_{r,N_k}:i\in A}
a_{A,N_k}^2\prod_{j\in A}\E[Y_{j,N_k}^2\mid\mathcal F_{N_k}]
\le\delta_k
\end{gathered}
\right\}\\
&\quad\ge \eta
-\Pr\!\left\{\abs{\sigma_{N_k}^2-\sigma^2}>\delta_k\right\}
-\Pr\!\left\{
\max_{1\le r\le R}
\abs{\E[W_{r,N_k}^4\mid\mathcal F_{N_k}]-3\sigma_{r,N_k}^4}
>\delta_k\right\}\\
&\qquad
-\Pr\!\left\{
\max_{1\le r\le R}\max_{i\in\mathcal I_{N_k}}
\sum_{A\in\mathcal A_{r,N_k}:i\in A}
a_{A,N_k}^2\prod_{j\in A}\E[Y_{j,N_k}^2\mid\mathcal F_{N_k}]
>\delta_k\right\}
=\eta-o(1)>\frac{\eta}{2}.
\end{aligned}
\]
Intersecting this event with the probability-one set on which
the variables \(\{Y_{i,N_k}:i\in\mathcal I_{N_k}\}\) are independent and
centered under the conditional distribution still leaves positive
probability. Hence, for every sufficiently large \(k\), we can choose
\(\omega_k\) such that \eqref{eq:mf-clt-good-event} holds, the absolute
difference in \eqref{eq:mf-clt-bad-event} exceeds \(\varepsilon\), and the
regular conditional distribution at \(\omega_k\) makes the variables
\(\{Y_{i,N_k}:i\in\mathcal I_{N_k}\}\) independent and centered.

Under the regular conditional distribution at \(\omega_k\), the coefficients
\(a_{A,N_k}(\omega_k)\) are deterministic. By
\eqref{eq:mf-clt-variance-decomposition},
\(\sum_{r=1}^R \sigma_{r,N_k}^2(\omega_k)=\sigma_{N_k}^2(\omega_k)\),
which converges to \(\sigma^2\) by \eqref{eq:mf-clt-good-event}. Because each
coordinate is nonnegative, the sequence
\((\sigma_{1,N_k}^2(\omega_k),\ldots,\sigma_{R,N_k}^2(\omega_k))\) is
bounded in \(\mathbb R^R\). The Bolzano--Weierstrass theorem therefore
provides a further subsequence of \(k\), without relabelling it, such that
\(\sigma_{r,N_k}^2(\omega_k)\to v_r\) for every \(r=1,\ldots,R\). Then
\(v_r\ge0\) and \(\sum_{r=1}^R v_r=\sigma^2\).

Fix \(r\) with \(v_r>0\). Under the selected conditional law,
\(W_{r,N_k}/\sqrt{\sigma_{r,N_k}^2(\omega_k)}\) has variance one for all
sufficiently large \(k\). For every \(A\in\mathcal A_{r,N_k}\) and every
\(B\subsetneq A\), independence and centering under the selected conditional
law give \(\E[\prod_{i\in A}Y_{i,N_k}\mid
\sigma(Y_{j,N_k}:j\in B)]=0\). Consequently,
\(W_{r,N_k}/\sqrt{\sigma_{r,N_k}^2(\omega_k)}\) is a completely degenerate,
possibly nonsymmetric U-statistic of fixed order \(r\). By
\eqref{eq:mf-clt-good-event},
\begin{equation}
\label{eq:mf-clt-standardized-conditions}
\begin{aligned}
&\abs{
\frac{\E[W_{r,N_k}^4\mid\mathcal F_{N_k}](\omega_k)}
{[\sigma_{r,N_k}^2(\omega_k)]^2}-3}=
\frac{\abs{\E[W_{r,N_k}^4\mid\mathcal F_{N_k}](\omega_k)
-3[\sigma_{r,N_k}^2(\omega_k)]^2}}
{[\sigma_{r,N_k}^2(\omega_k)]^2}\le \frac{\delta_k}{[\sigma_{r,N_k}^2(\omega_k)]^2}
\to0,\\
&\frac{1}{\sigma_{r,N_k}^2(\omega_k)}
\max_{i\in\mathcal I_{N_k}}
\sum_{A\in\mathcal A_{r,N_k}:i\in A}
a_{A,N_k}^2(\omega_k)
\prod_{j\in A}\E[Y_{j,N_k}^2\mid\mathcal F_{N_k}](\omega_k)\le \frac{\delta_k}{\sigma_{r,N_k}^2(\omega_k)}
\to0.
\end{aligned}
\end{equation}
By \eqref{eq:mf-clt-good-event},
\(\sigma_{r,N_k}^2(\omega_k)\to v_r>0\), and the arbitrariness of \(r\),
the two limits in \eqref{eq:mf-clt-standardized-conditions} hold for every
\(r\) satisfying \(v_r>0\).

If exactly one \(v_r\) is positive,
the variance-one and complete-degeneracy properties of
\(W_{r,N_k}/\sqrt{\sigma_{r,N_k}^2(\omega_k)}\),
\(\abs{\mathcal I_{N_k}}\to\infty\), and
\eqref{eq:mf-clt-standardized-conditions} verify the conditions of
\citet[Theorem~1.2]{dobler2017quantitative}; see also
\citet[Theorem~1]{dejong1990central}. Hence
\(W_{r,N_k}/\sqrt{\sigma_{r,N_k}^2(\omega_k)}\) converges in distribution
to a standard normal random variable.
Suppose instead that at least two \(v_r\)'s are positive. Since
\(\abs{\mathcal I_{N_k}}\to\infty\),
the condition \(n_m\to\infty\) in
\citet[Theorem~1.7]{dobler2017quantitative} holds after reindexing the
variables in each row. Consider the vector with coordinates \(W_{r,N_k}/\sqrt{\sigma_{r,N_k}^2(\omega_k)}\)
with \(v_r>0\). Each coordinate is a completely degenerate, possibly
nonsymmetric U-statistic of order \(r\), and
\eqref{eq:mf-clt-standardized-conditions} verifies conditions (ii) and
(iii) in \citet[Theorem~1.7]{dobler2017quantitative}. Since the
coordinates are standardized and
\(\Cov(W_{r,N_k},W_{s,N_k}\mid\mathcal F_{N_k})(\omega_k)=0\)
for \(r\ne s\), their
covariance matrix is the identity matrix, which verifies condition (i).
There is at most one coordinate of each order, so condition (iv) trivially holds.
\citet[Theorem~1.7]{dobler2017quantitative} therefore gives joint
convergence to independent standard normal variables. Hence, in either
case, the standardized components with \(v_r>0\) converge jointly to
independent standard normal variables. Since
\[
\sum_{\{r:v_r>0\}}W_{r,N_k}
=\sum_{\{r:v_r>0\}}\sqrt{\sigma_{r,N_k}^2(\omega_k)}
\frac{W_{r,N_k}}{\sqrt{\sigma_{r,N_k}^2(\omega_k)}},
\]
Slutsky's theorem gives convergence of this sum to
\(N(0,\sum_{\{r:v_r>0\}}v_r)\).

For the  degrees such that \(v_r=0\), conditional orthogonality gives
\[
\Var\!\left(\sum_{\{r:v_r=0\}}W_{r,N_k}
\middle|\mathcal F_{N_k}\right)(\omega_k)
=\sum_{\{r:v_r=0\}}\sigma_{r,N_k}^2(\omega_k)\to0.
\]
Conditional Chebyshev's inequality therefore gives, for every \(\delta>0\),
\[
\Pr\!\left\{\left|\sum_{\{r:v_r=0\}}W_{r,N_k}\right|>\delta
\,\middle|\,\mathcal F_{N_k}\right\}(\omega_k)\to0.
\]
Hence \(W_{N_k}\) converges under the regular conditional distributions
at \(\omega_k\) to \(N(0,\sum_{r=1}^R v_r)=N(0,\sigma^2)\). This
forces the conditional characteristic-function difference in
\eqref{eq:mf-clt-bad-event}, evaluated at \(\omega_k\), to converge to
zero, contradicting its being larger than \(\varepsilon\). Thus, for
every \(t\in\mathbb R\),
\(\E[\exp(\mathrm{i}tW_N)\mid\mathcal F_N]\xrightarrow{P}
\exp(-t^2\sigma^2/2)\). Because conditional characteristic functions
are bounded by one, this convergence also holds in \(L^1\). Iterated
expectations and L\'evy's continuity theorem therefore give
\(W_N\xrightarrow{D}N(0,\sigma^2)\).
\end{proof}

\section{Asymptotic Theory}\label{sec:regime_complete_asymptotic}
This appendix establishes the asymptotic properties of the sample moment and the pentad-GMM estimator.
It also proves the asymptotic equivalence of \(\widehat{\bs\beta}^{\dagger}\) and \(\widehat{\bs\beta}_{\mathrm{GMM}}\), as well as the consistency of the covariance estimator.
Recall that
\(\bs{\Psi}_s(\bs{\beta}):= \bs Z_s\psi_s(\bs{\beta})\) and
\(\bs h_N(\bs{\beta}):=\abs{\mathcal S_N}^{-1}\sum_{s\in\mathcal S_N}\bs{\Psi}_s(\bs{\beta})\).
Throughout the paper, \(\norm{\bs x}:=(\bs x^\top\bs x)^{1/2}\)
denotes the Euclidean norm of a vector \(\bs x\), and
\(\norm{\bf A}:=\sup_{\norm{\bs x}=1}\norm{\bf A\bs x}\)
denotes the operator norm of a matrix \(\bf A\).

\subsection{Setup and auxiliary lemmas}
We observe the undirected adjacency matrix \(\mathbf L\). The observed covariate
array is \(\bf X=(\bs X_1,\dots,\bs X_N)^\top\).
The individual fixed effects
\(\bs\Gamma=(\Gamma_1,\dots,\Gamma_N)^\top\) are unobserved.
Collect node \(i\)'s characteristics in \(\bs{U}_i:=(\bs X_i,\Gamma_i)\).

Let \(\bs X_{ij}:=w(\bs X_i,\bs X_j)\) and
\(W_{ij}(\bs{\beta}):=\exp(-\bs X_{ij}^\top\bs{\beta})\),
where $w$ is a measurable map that need not be symmetric. 
Here \(\mathcal S_N:=\{(i,j,k,l,m)\in[N]^5:i,j,k,l,m\text{ are all distinct}\}\).
For $s=(i,j,k,l,m)$, let
\(E(s):=\{(i,j),(i,k),(i,l),(i,m),(j,k),(j,l),(j,m)\}\) be the set of the
seven dyads entering configuration $s$.

Define the sigma-field \(\mathcal O_N:=\sigma(\bs{U}_1,\dots,\bs{U}_N)\).
For an unordered dyad \(e=(u,v)\) with \(u<v\), write
\(\mathcal E_N:=\{(u,v):1\le u<v\le N\}\),
\(P_e(\bs{\beta}_{0}):=\E[L_e\mid \mathcal O_N]\),
and \(\xi_e:=L_e-P_e(\bs{\beta}_{0})\).

The next two lemmas compare conditional link probabilities with
\(\rho_N\Lambda_{i,N}\Lambda_{j,N}\) and bound sums of products of
conditional link probabilities. These bounds are used in  the sparse
and ultra-sparse parts of the proof of
Lemma~\ref{lem:sparse_dyadset_fourth}.

\begin{lem}[Bounds for conditional link probabilities]
\label{lem:conditional_degree_comparison}
Under Assumptions~\ref{ass:dyad_ind}, \ref{ass:bounded_design}, and
\ref{ass:sparse_envelope}, for every \(i\ne j\),
\begin{equation}
\label{eq:conditional_degree_comparison}
e^{-8C_X\norm{\bs\beta_0}}\rho_N\Lambda_{i,N}\Lambda_{j,N}
\le P_{ij}(\bs\beta_0)\le
e^{8C_X\norm{\bs\beta_0}}\rho_N\Lambda_{i,N}\Lambda_{j,N}
\quad\text{a.s.}
\end{equation}
Moreover, for every integer \(m\ge2\),
\begin{equation}
\label{eq:centered_link_weighted_bound}
\E[\abs{\xi_{ij}}^m\mid\mathcal O_N]
\le P_{ij}(\bs\beta_0)
\le e^{8C_X\norm{\bs\beta_0}}
\rho_N\Lambda_{i,N}\Lambda_{j,N}
\quad\text{a.s.}
\end{equation}
\end{lem}

\begin{proof}
For the logistic cdf, \(e^{-a}F(t)\le F(t+x)\le e^aF(t)\) whenever
\(\abs{x}\le a\). Assumption~\ref{ass:bounded_design} therefore gives
\(
e^{-2C_X\norm{\bs\beta_0}}F(\Gamma_i)F(\Gamma_j)
\le P_{ij}(\bs\beta_0)\le
e^{2C_X\norm{\bs\beta_0}}F(\Gamma_i)F(\Gamma_j)
\).
Since \(\bs{U}_i\) and \(\bs{U}_j\) are i.i.d. under Assumption~\ref{ass:dyad_ind},
unconditional expectations give the bounds for \(\rho_N\), whereas
conditional expectations given \(\bs{U}_i\) give the bounds for
\(\rho_N\Lambda_{i,N}\):
\[
\begin{aligned}
e^{-2C_X\norm{\bs\beta_0}}\{\E F(\Gamma_i)\}^2
&\le\rho_N\le
e^{2C_X\norm{\bs\beta_0}}\{\E F(\Gamma_i)\}^2,\\
e^{-2C_X\norm{\bs\beta_0}}F(\Gamma_i)\E F(\Gamma_i)
&\le\rho_N\Lambda_{i,N}\le
e^{2C_X\norm{\bs\beta_0}}F(\Gamma_i)\E F(\Gamma_i).
\end{aligned}
\]
In particular, the preceding inequalities give
\[
\begin{aligned}
P_{ij}(\bs\beta_0)
&\le e^{2C_X\norm{\bs\beta_0}}F(\Gamma_i)F(\Gamma_j)\le
e^{6C_X\norm{\bs\beta_0}}
\frac{\rho_N^2\Lambda_{i,N}\Lambda_{j,N}}
{\{\E F(\Gamma_i)\}^2}
\le e^{8C_X\norm{\bs\beta_0}}
\rho_N\Lambda_{i,N}\Lambda_{j,N}.
\end{aligned}
\]
The reverse inequalities give
\[
\begin{aligned}
P_{ij}(\bs\beta_0)
&\ge e^{-2C_X\norm{\bs\beta_0}}F(\Gamma_i)F(\Gamma_j)\ge
e^{-6C_X\norm{\bs\beta_0}}
\frac{\rho_N^2\Lambda_{i,N}\Lambda_{j,N}}
{\{\E F(\Gamma_i)\}^2}
\ge e^{-8C_X\norm{\bs\beta_0}}
\rho_N\Lambda_{i,N}\Lambda_{j,N}.
\end{aligned}
\]
This proves \eqref{eq:conditional_degree_comparison}. Finally,
\(\abs{\xi_{ij}}\le1\)
and the conditional Bernoulli variance imply
\(\E[\abs{\xi_{ij}}^m\mid\mathcal O_N]\le
\E[\xi_{ij}^2\mid\mathcal O_N]
=P_{ij}(\bs\beta_0)\{1-P_{ij}(\bs\beta_0)\}\le
P_{ij}(\bs\beta_0)\), which proves
\eqref{eq:centered_link_weighted_bound}.
\end{proof}

For a finite graph \((V,E)\), consider the product of link probabilities
\(\prod_{(a,b)\in E}P_{ab}(\bs\beta_0)^{r_{ab}}\),
where each \(r_{ab}\ge1\) is an integer.
We call \(\sum_{(a,b)\in E}r_{ab}\) the \emph{total exponent} of the
product and define the \emph{incidence multiplicity} at node \(v\in V\) as
\(\kappa_v:=\sum_{(a,b)\in E:\,v\in\{a,b\}}r_{ab}\).
The total exponent counts all factors with multiplicity; \(\kappa_v\)
counts only those incident to node \(v\).
For distinct nodes \(i,j,k\), the product
\(P_{ij}(\bs\beta_0)^2P_{jk}(\bs\beta_0)\) has total
exponent \(\sum_{(a,b)\in E}r_{ab}=3\) and incidence multiplicities
\(\kappa_i=2\), \(\kappa_j=3\), and \(\kappa_k=1\), while its support
graph is the path \(i-j-k\) with two distinct dyads.
In the following lemma, the total exponent determines the power of
\(\rho_N\), and the incidence multiplicities determine the required
moment orders of \(\Lambda_{i,N}\).

\begin{lem}[Bounds for products of conditional link probabilities]
\label{lem:weighted_probability_sums}
Suppose Assumptions~\ref{ass:dyad_ind}, \ref{ass:bounded_design}, and
\ref{ass:sparse_envelope} hold. Let \((V,E)\) be a fixed finite graph with
node set \(V=\{1,\ldots,m\}\), where \(m\le N\), whose dyads join distinct nodes and are written as
\((u,v)\). For each \((u,v)\in E\), let
\(r_{uv}\ge1\) be a fixed integer exponent, and let \(\kappa_v\) denote
the incidence multiplicity at each node \(v\in V\).
Suppose \(\kappa_v\le16\) for every \(v\in V\).
Then, for every ordered tuple \((i_1,\ldots,i_m)\) of distinct elements
of \([N]\),
\begin{equation}
\label{eq:weighted_fixed_label_bound}
\E\!\left[
\prod_{(u,v)\in E}P_{i_ui_v}(\bs\beta_0)^{r_{uv}}
\right]
\le C\rho_N^{\sum_{(u,v)\in E}r_{uv}},
\end{equation}
and
\begin{equation}
\label{eq:weighted_unrooted_sum}
\E\!\left[
\sum_{\substack{i_1,\ldots,i_m\in[N]\\\mathrm{distinct}}}
\prod_{(u,v)\in E}P_{i_ui_v}(\bs\beta_0)^{r_{uv}}
\right]
\le C N^{\abs V}\rho_N^{\sum_{(u,v)\in E}r_{uv}}.
\end{equation}
For every nonempty fixed subset \(V_1\subseteq V\) such that
\(\kappa_v\le8\) for every \(v\in V_1\), relabel \(V\) so that
\(V_1=\{1,\ldots,\ell\}\), where \(\ell=\abs{V_1}\). Then
\begin{equation}
\label{eq:weighted_rooted_sum}
\max_{\substack{i_1,\ldots,i_\ell\in[N]\\\mathrm{distinct}}}
\sum_{\substack{i_{\ell+1},\ldots,i_m\in[N]\\i_1,\ldots,i_m\ \mathrm{distinct}}}
\prod_{(u,v)\in E}P_{i_ui_v}(\bs\beta_0)^{r_{uv}}
=o_p\!\left(N^{\abs V}\rho_N^{\sum_{(u,v)\in E}r_{uv}}\right).
\end{equation}
\end{lem}

\begin{proof}
Applying \eqref{eq:conditional_degree_comparison} to every link-probability
factor gives
\[
\prod_{(u,v)\in E}P_{i_ui_v}(\bs\beta_0)^{r_{uv}}
\le C\rho_N^{\sum_{(u,v)\in E}r_{uv}}
\prod_{v\in V}\Lambda_{i_v,N}^{\kappa_v}.
\]
The distinct labels \(i_1,\ldots,i_m\) correspond to independent node
characteristics. Since \(\kappa_v\le16\) for every \(v\in V\),
Assumption~\ref{ass:sparse_envelope} gives
\(
\E\!\left[
\prod_{v\in V}\Lambda_{i_v,N}^{\kappa_v}
\right]
=\prod_{v\in V}
\E\!\left[\Lambda_{i,N}^{\kappa_v}\right]
\le C
\).
This proves \eqref{eq:weighted_fixed_label_bound}. Summing over at most
\(N^{\abs V}\) ordered tuples of distinct labels proves
\eqref{eq:weighted_unrooted_sum}.

For every \(0<\kappa\le8\), \(\sup_N\E[\Lambda_{i,N}^{16}]<\infty\) and the
union bound give, for every \(\varepsilon>0\),
\[
\Pr\!\left\{\max_{i\le N}\Lambda_{i,N}^{\kappa}>\varepsilon N\right\}
\le N\Pr\!\left\{\Lambda_{i,N}^{\kappa}>\varepsilon N\right\}
\le\frac{N\E[\Lambda_{i,N}^{16}]}{(\varepsilon N)^{16/\kappa}}
\le C\varepsilon^{-16/\kappa}N^{1-16/\kappa}\to0,
\]
while \(N^{-1}\sum_{i=1}^N\Lambda_{i,N}^{\kappa}=O_p(1)\) for every
\(0<\kappa\le16\) by Markov's inequality. Applying
\eqref{eq:conditional_degree_comparison} and separating the factors
indexed by \(V_1\) gives
\begin{align*}
&\max_{\substack{i_1,\ldots,i_\ell\in[N]\\\mathrm{distinct}}}
\sum_{\substack{i_{\ell+1},\ldots,i_m\in[N]\\i_1,\ldots,i_m\ \mathrm{distinct}}}
\prod_{(u,v)\in E}P_{i_ui_v}(\bs\beta_0)^{r_{uv}}\\
&\quad\le C\rho_N^{\sum_{(u,v)\in E}r_{uv}}
\max_{\substack{i_1,\ldots,i_\ell\in[N]\\\mathrm{distinct}}}
\left\{
\prod_{v\in V_1}\Lambda_{i_v,N}^{\kappa_v}
\sum_{\substack{i_{\ell+1},\ldots,i_m\in[N]\\i_1,\ldots,i_m\ \mathrm{distinct}}}
\prod_{v\in V\setminus V_1}\Lambda_{i_v,N}^{\kappa_v}
\right\}.
\end{align*}
For each fixed tuple \((i_1,\ldots,i_\ell)\), all summands are nonnegative. Allowing the nodes
in \(V\setminus V_1\) to take arbitrary labels in \([N]\), including
repeated labels and the fixed labels \(i_1,\ldots,i_\ell\), therefore gives
\begin{align*}
\sum_{\substack{i_{\ell+1},\ldots,i_m\in[N]\\i_1,\ldots,i_m\ \mathrm{distinct}}}
\prod_{v\in V\setminus V_1}\Lambda_{i_v,N}^{\kappa_v}
&\le
\sum_{i_{\ell+1},\ldots,i_m\in[N]}
\prod_{v\in V\setminus V_1}\Lambda_{i_v,N}^{\kappa_v}
=\prod_{v\in V\setminus V_1}
\sum_{i=1}^N\Lambda_{i,N}^{\kappa_v},\\
\max_{\substack{i_1,\ldots,i_\ell\in[N]\\\mathrm{distinct}}}
\prod_{v\in V_1}\Lambda_{i_v,N}^{\kappa_v}
&\le\prod_{v\in V_1}\max_{i\le N}\Lambda_{i,N}^{\kappa_v}.
\end{align*}
Substituting these bounds into the left-hand side of
\eqref{eq:weighted_rooted_sum}, dividing by
\(N^{\abs V}\rho_N^{\sum_{(u,v)\in E}r_{uv}}\), and using
\(N^{\abs V}=N^{\abs{V_1}}N^{\abs{V\setminus V_1}}\) yields
\begin{align}
&\frac{1}{N^{\abs V}\rho_N^{\sum_{(u,v)\in E}r_{uv}}}
\max_{\substack{i_1,\ldots,i_\ell\in[N]\\\mathrm{distinct}}}
\sum_{\substack{i_{\ell+1},\ldots,i_m\in[N]\\i_1,\ldots,i_m\ \mathrm{distinct}}}
\prod_{(u,v)\in E}P_{i_ui_v}(\bs\beta_0)^{r_{uv}}\notag\\
&\quad\le C
\prod_{v\in V_1}\left(
\frac{1}{N}\max_{i\le N}\Lambda_{i,N}^{\kappa_v}\right)
\prod_{v\in V\setminus V_1}\left(
\frac{1}{N}\sum_{i=1}^N\Lambda_{i,N}^{\kappa_v}\right).
\label{eq:weighted_rooted_factorization}
\end{align}
Moreover,
\begin{align*}
\frac{1}{N}\max_{i\le N}\Lambda_{i,N}^{\kappa_v}
&=\begin{cases}
o_p(1),&0<\kappa_v\le8,\\
N^{-1}=o(1),&\kappa_v=0,
\end{cases}
&& v\in V_1,\\
\frac{1}{N}\sum_{i=1}^N\Lambda_{i,N}^{\kappa_v}
&=\begin{cases}
O_p(1),&0<\kappa_v\le16,\\
1,&\kappa_v=0,
\end{cases}
&& v\in V\setminus V_1.
\end{align*}
Since \(V_1\ne\varnothing\), the first product on the right-hand side of
\eqref{eq:weighted_rooted_factorization} is \(o_p(1)\), whereas the second
is \(O_p(1)\). This proves
\eqref{eq:weighted_rooted_sum}.
\end{proof}

For a multi-index \(\alpha=(\alpha_1,\ldots,\alpha_d)\in\mathbb N_0^d\),
where \(\mathbb N_0=\{0,1,2,\ldots\}\), write
\(\abs{\alpha}:=\sum_{j=1}^d\alpha_j\) and
\(\partial_{\bs\beta}^{\alpha}
:=\partial_{\beta_1}^{\alpha_1}\cdots\partial_{\beta_d}^{\alpha_d}\),
with zero-order differentiation leaving the function unchanged.

\begin{lem}[Bounded  polynomial kernel]
\label{lem:poly}
Let \(\mathcal B\subset\mathbb R^d\) be any fixed compact set.
Under Assumption~\ref{ass:bounded_design}, there are constants
$0<\underline W<\overline W<\infty$ such that
\(\underline W\le W_{ij}(\bs{\beta})=\exp(-\bs X_{ij}^\top\bs{\beta})
\le \overline W\)
uniformly over $i\neq j$ and $\bs{\beta}\in\mathcal B$.
Moreover, there exists a finite deterministic constant
$\nu_{\max}$, not depending on $N$, such that for every
$s\in\mathcal S_N$ and every \(\bs\beta\in\mathcal B\),
\[
\bs{\Psi}_s(\bs{\beta})
=\sum_{\nu=1}^{\nu_{\max}}\bs c_{s,\nu}(\bs\beta,\mathcal O_N)m_{s,\nu}({\bf L}),
\]
where each $\bs c_{s,\nu}(\bs\beta,\mathcal O_N)\in\mathbb{R}^{q}$ is
\(\sigma(\bf X_s)\)-measurable,
and each \(m_{s,\nu}({\bf L})\) has the square-free
form \(\prod_{e\in E_{s,\nu}}L_e\) for some \(E_{s,\nu}\subseteq E(s)\),
and every scalar monomial appearing in \(\bs{\Psi}_s(\bs\beta)\) has
degree between $4$ and $7$. Finally, for every fixed
integer \(r\), there is a finite constant \(C_r\), not depending on
\(N\), such that
\begin{equation}
\label{eq:poly-coefficient-derivative-bound}
\sup_N\sup_{s\in\mathcal S_N}\sup_{\bs\beta\in\mathcal B}
\sum_{\nu=1}^{\nu_{\max}}
\sum_{\abs{\alpha}\le r}
\norm{\partial_{\bs\beta}^{\alpha}
\bs c_{s,\nu}(\bs\beta,\mathcal O_N)}
\le C_r
\quad\text{a.s.}
\end{equation}
\end{lem}

\begin{proof}
The boundedness of \(W_{ij}(\bs\beta)\) follows from the exponential
form \(W_{ij}(\bs\beta)=\exp(-\bs X_{ij}^{\top}\bs\beta)\), bounded
dyadic covariates, and compactness of \(\mathcal B\).
The row-wise expansion in Lemma~\ref{lem:det-row-supports} applies
to every \(\bs\beta\): replacing \(\bs\beta_0\) by \(\bs\beta\)
changes only the coefficients in \eqref{eq:det-support-expansion}.
Multiplying by \(\bs Z_s\) gives the stated expansion of
\(\bs\Psi_s(\bs\beta)\), with the same monomials
\(m_{s,\nu}({\bf L})\) and supports satisfying
\(4\le\abs{E_{s,\nu}}\le7\).

It remains to establish the uniform bound on
\(\norm{\partial_{\bs\beta}^{\alpha}\bs c_{s,\nu}(\bs\beta,\mathcal O_N)}\)
for \(\abs\alpha\le r\) in \eqref{eq:poly-coefficient-derivative-bound}.
Each coefficient in the expansion of
\(\bs\Psi_s(\bs\beta)\) is
\(\bs Z_s=z(\bf X_s)\) multiplied by a finite linear combination of
products of terms \(W_{uv}(\bs\beta)\) and
\(W_{uv}(\bs\beta)^{-1}\), where \(u,v\in V_s\), with the number of
factors bounded by an absolute constant. Hence every coefficient is
\(\sigma(\bf X_s)\)-measurable. For any multi-index \(\alpha\),
\(\partial_{\bs\beta}^{\alpha}W_{uv}(\bs\beta)
=(-1)^{\abs{\alpha}}\bs X_{uv}^{\alpha}W_{uv}(\bs\beta)\)
and
\(\partial_{\bs\beta}^{\alpha}W_{uv}(\bs\beta)^{-1}
=\bs X_{uv}^{\alpha}W_{uv}(\bs\beta)^{-1}\).
Assumption~\ref{ass:bounded_design} and compactness of \(\mathcal B\)
therefore bound these derivatives uniformly over \(u,v,N\), and
\(\bs\beta\in\mathcal B\), for each fixed \(\alpha\). Applying the
Leibniz rule to the finitely many products in each coefficient and
summing over \(\nu\le\nu_{\max}\) and \(\abs{\alpha}\le r\) gives
\eqref{eq:poly-coefficient-derivative-bound} for every fixed integer \(r\).
\end{proof}

\subsection{Hoeffding decomposition over dyads}
\label{app:hoeffding_decomposition_over_dyads}
Define  
\(\mathcal O_N^{B}:=\sigma\big(\mathcal O_N,\{L_e:e\in B\}\big)\). For every nonempty finite dyad set $A\subseteq\mathcal E_N$, define
\begin{equation}
\label{eq:projdef}
\begin{split}
\bs\phi_{A,N}
&:=\sum_{B\subseteq A}(-1)^{\abs{A}-\abs{B}}
\E[\bs h_N(\bs{\beta}_{0})\mid \mathcal O_N^{B}].
\end{split}
\end{equation}
Set
\(\bs\phi_{\varnothing,N}:=\E[\bs h_N(\bs{\beta}_{0})\mid \mathcal O_N]\).
The conditional moment restriction gives \(\bs\phi_{\varnothing,N}=0\).

\begin{lem}[Exact representation for a fixed dyad set]
\label{lem:exact_representation}
Under Assumptions~\ref{ass:dyad_ind} and \ref{ass:bounded_design}, fix a nonempty dyad set
$A\subseteq\mathcal E_N$ that is contained in at least one configuration.
Then
\begin{equation}
\label{eq:phi_exact_rep}
\bs\phi_{A,N}=\frac{1}{\abs{\mathcal S_N}}\bs C_{A,N}(\mathcal O_N)\prod_{e\in A}\xi_e,
\end{equation}
where
\begin{equation}
\label{eq:projection_coeff_exact}
\bs C_{A,N}(\mathcal O_N):=\sum_{s\in\mathcal S_N:\,A\subseteq E(s)}\sum_{\nu:A\subseteq E_{s,\nu}}
 \bs c_{s,\nu}(\bs\beta_0,\mathcal O_N)\prod_{f\in E_{s,\nu}\setminus A} P_f(\bs{\beta}_{0}),
\end{equation}
Recall that \(E_{s,\nu}\) is the dyad support of
\(m_{s,\nu}({\bf L})\).
\end{lem}

\begin{proof}
Fix a nonempty dyad set \(A\subseteq\mathcal E_N\). By
Lemma~\ref{lem:poly}, the coefficients in the finite polynomial
expansion of \(\bs\Psi_s(\bs\beta_0)\) are \(\mathcal O_N\)-measurable.
It therefore suffices to apply \eqref{eq:projdef} to a single monomial
\(m_{s,\nu}(L)=\prod_{f\in E_{s,\nu}}L_f\). For every \(B\subseteq A\),
conditional dyadic independence given \(\mathcal O_N\) gives
\begin{equation}
\label{eq:monomial_conditional_mean}
\E[m_{s,\nu}(L)\mid\mathcal O_N^B]
=
\left(\prod_{f\in E_{s,\nu}\cap B}L_f\right)
\left(\prod_{f\in E_{s,\nu}\setminus B}P_f(\bs\beta_0)\right).
\end{equation}
If \(A\not\subseteq E_{s,\nu}\), choose
\(e\in A\setminus E_{s,\nu}\). For each
\(B\subseteq A\setminus\{e\}\), the conditional expectations in
\eqref{eq:monomial_conditional_mean} for \(B\) and \(B\cup\{e\}\) are
equal, while their inclusion-exclusion signs in \eqref{eq:projdef} are opposite. Pairing
\(B\) with \(B\cup\{e\}\) makes the projection of \(m_{s,\nu}(L)\) zero.
If \(A\subseteq E_{s,\nu}\), then for every \(B\subseteq A\),
\(E_{s,\nu}\cap B=B\) and
\(E_{s,\nu}\setminus B=(A\setminus B)\cup(E_{s,\nu}\setminus A)\),
where the union is disjoint. Substituting
\eqref{eq:monomial_conditional_mean} and taking the factor
\(\prod_{f\in E_{s,\nu}\setminus A}P_f(\bs\beta_0)\), which does not
depend on \(B\), outside the sum yields
\[
\begin{aligned}
&\sum_{B\subseteq A}(-1)^{\abs A-\abs B}
\E[m_{s,\nu}(L)\mid\mathcal O_N^B]\\
&=\left[
\sum_{B\subseteq A}(-1)^{\abs A-\abs B}
\left(\prod_{e\in B}L_e\right)
\left(\prod_{e\in A\setminus B}P_e(\bs\beta_0)\right)
\right]
\prod_{f\in E_{s,\nu}\setminus A}P_f(\bs\beta_0)\\
&=\left(\prod_{e\in A}\xi_e\right)
\prod_{f\in E_{s,\nu}\setminus A}P_f(\bs\beta_0).
\end{aligned}
\]
For the second equality, expand
\(\prod_{e\in A}\{L_e-P_e(\bs\beta_0)\}\): choosing \(L_e\) for
\(e\in B\) and \(-P_e(\bs\beta_0)\) for \(e\in A\setminus B\)
produces the sign \((-1)^{\abs A-\abs B}\). Summing over all
\(B\subseteq A\) and using \(\xi_e=L_e-P_e(\bs\beta_0)\) gives
\(\prod_{e\in A}\xi_e\).
Multiplying by \(\bs c_{s,\nu}(\bs\beta_0,\mathcal O_N)\), summing over
\(s\in\mathcal S_N\) and \(\nu\le\nu_{\max}\), and dividing by
\(\abs{\mathcal S_N}\) proves \eqref{eq:phi_exact_rep} with the
coefficient in \eqref{eq:projection_coeff_exact}.
\end{proof}

\begin{lem}[Hoeffding decomposition over dyads and conditional orthogonality]
\label{lem:dyadic_hoeffding_orthogonality}
Suppose Assumptions~\ref{ass:dyad_ind} and
\ref{ass:bounded_design} hold. Then the following hold:
\begin{enumerate}[label=\textnormal{(\roman*)},leftmargin=2em]
\item if \(A\subseteq \mathcal E_N\) is not contained in \(E(s)\) for any
\(s\in\mathcal S_N\), or if \(\abs{A}>7\), then \(\bs\phi_{A,N}=0\);
\item the exact decomposition
\begin{equation}
\label{eq:exactdecomp}
\bs h_N(\bs{\beta}_{0})=\sum_{A\subseteq\mathcal E_N:1\le\abs{A}\le7}\bs\phi_{A,N}
\end{equation}
holds almost surely;
\item if $A\neq A'$, then
\(\E[\bs\phi_{A,N}\bs\phi_{A',N}^\top\mid \mathcal O_N]=0\).
Hence the family
\(\{\bs\phi_{A,N}:A\subseteq\mathcal E_N,1\le\abs{A}\le 7\}\) is
conditionally orthogonal.
\end{enumerate}
\end{lem}

\begin{proof}
For part (i), if \(A\) is not contained in any \(E(s)\), or if
\(\abs A>7\), then no monomial support \(E_{s,\nu}\subseteq E(s)\)
contains \(A\). The cancellation for \(A\not\subseteq E_{s,\nu}\)
established in the proof of Lemma~\ref{lem:exact_representation}
therefore makes the projection of each monomial zero. Since the coefficients
\(\bs c_{s,\nu}(\bs\beta_0,\mathcal O_N)\) are
\(\mathcal O_N\)-measurable, applying \eqref{eq:projdef} term by term
to the finite polynomial expansion of \(\bs h_N(\bs\beta_0)\) gives
\(\bs\phi_{A,N}=0\).

For part (ii), substituting \(L_f=P_f(\bs\beta_0)+\xi_f\) in each
monomial gives
\[
\prod_{f\in E_{s,\nu}}L_f
=
\sum_{A\subseteq E_{s,\nu}}
\left(\prod_{e\in A}\xi_e\right)
\prod_{f\in E_{s,\nu}\setminus A}P_f(\bs\beta_0).
\]
Multiply by \(\bs c_{s,\nu}(\bs\beta_0,\mathcal O_N)\), sum over
\(s\in\mathcal S_N\) and \(\nu\le\nu_{\max}\), and divide by
\(\abs{\mathcal S_N}\). Equations~\eqref{eq:phi_exact_rep} and
\eqref{eq:projection_coeff_exact} identify the terms with nonempty
\(A\) as \(\bs\phi_{A,N}\). The empty-set term is
\(\bs\phi_{\varnothing,N}=\E[\bs h_N(\bs\beta_0)\mid\mathcal O_N]=0\)
by the conditional moment restriction. Removing the empty-set term and
using \(\abs{E_{s,\nu}}\le7\) proves \eqref{eq:exactdecomp}.

For part (iii), the empty-set projection and the projections in part (i)
are zero. For all other \(A\ne A'\), the coefficients
\(\bs C_{A,N}(\mathcal O_N)\) and
\(\bs C_{A',N}(\mathcal O_N)\) are \(\mathcal O_N\)-measurable, and
the shocks \(\{\xi_e:e\in\mathcal E_N\}\) are conditionally independent
and centered given \(\mathcal O_N\). Hence \eqref{eq:phi_exact_rep} gives
\[
\begin{aligned}
&\E[\bs\phi_{A,N}\bs\phi_{A',N}^{\top}\mid\mathcal O_N]=
\frac{\bs C_{A,N}(\mathcal O_N)\bs C_{A',N}(\mathcal O_N)^{\top}}
{\abs{\mathcal S_N}^{2}}
\prod_{e\in A\cap A'}\E[\xi_e^2\mid\mathcal O_N]
\prod_{e\in A\triangle A'}\E[\xi_e\mid\mathcal O_N]
=0.
\end{aligned}
\]
The product over \(A\triangle A'\) is zero because
\(A\triangle A'\ne\varnothing\) and \(\E[\xi_e\mid\mathcal O_N]=0\).
\end{proof}

\subsection{Proof of Theorem~\ref{thm:regimeCLT}}
\label{app:projectionclasses}
Recall that, for a graph class \(g\) arising from dyad subsets
\(A\subseteq E_{s,\nu}\subseteq E(s)\), \(r(g)\) and \(v(g)\) denote
the common numbers of dyads and distinct nodes, respectively, and
\(m(g)\) is the minimal dyadic degree of a kernel monomial whose support
contains a labeled dyad set \(A\) with \([G_A]\in g\).
A graph class may contain multiple isomorphism classes sharing these
values of \(r(g)\), \(v(g)\), and \(m(g)\). The corresponding projection
sum is \(\bs\Pi_N(g)=\sum_{A\in\mathscr C_N(g)}\bs\phi_{A,N}\).

We first verify the graph classification in Table~\ref{tab:regimemap}.

\begin{lem}
\label{lem:kerneltaxonomy}
Fix \(s=(i,j,k,l,m)\in\mathcal S_N\), with pivot dyad \((i,j)\) and
peripheral nodes \(k,l,m\). Let \(E_{s,\nu}\) range over the determinant
monomial supports in \eqref{eq:det-support-expansion}, and let \(A\)
range over all nonempty subsets of these supports.
Table~\ref{tab:regimemap} lists every unlabeled graph shape \([G_A]\) that
can arise and groups these shapes into the graph classes \(g\) defined by
their common values of \(r(g)\), \(v(g)\), and \(m(g)\). Every shape
displayed in the table is attained by at least one such \(A\).
\end{lem}

\begin{proof}
By Lemma~\ref{lem:det-row-supports}, each \(E_{s,\nu}\) is generated by
one row-wise support from each of the \(B\)-, \(C\)-, and \(D\)-columns
and contains between four and seven dyads.
Equation~\eqref{eq:projection_coeff_exact} shows that a nonzero projection
coefficient can occur only if \(A\) is a nonempty subset of at least one
such support. Enumerating all nonempty subsets \(A\subseteq E_{s,\nu}\) up
to graph isomorphism and grouping the resulting shapes by
\((r(g),v(g),m(g))\) gives Table~\ref{tab:regimemap}.
\end{proof}

We next verify the variance scales reported in Table~\ref{tab:regimemap}.

\begin{prop}
\label{prop:classbound}
Under Assumptions \ref{ass:dyad_ind}, \ref{ass:bounded_design}, and
\ref{ass:sparse_envelope}, every graph class \(g\) arising from dyad
subsets \(A\subseteq E_{s,\nu}\subseteq E(s)\) satisfies the following
bound:
\begin{equation}
\label{eq:classbound}
\E\big[\norm{\bs\Pi_N(g)}^2\big]=O\big(N^{-v(g)}\rho_N^{2m(g)-r(g)}\big).
\end{equation}
\end{prop}

\begin{proof}
Fix one labeled dyad set \(A\in\mathscr C_N(g)\). Then
\(\abs{A}=r(g)\), and the nodes in \(A\) have cardinality
\(v(g)\). By Lemma~\ref{lem:exact_representation},
\(\bs\phi_{A,N}
=\abs{\mathcal S_N}^{-1}\bs C_{A,N}(\mathcal O_N)
\prod_{e\in A}\xi_e.\)
Hence
\[
\norm{\bs\phi_{A,N}}^2
=
\abs{\mathcal S_N}^{-2}
\norm{\bs C_{A,N}(\mathcal O_N)}^2
\prod_{e\in A}\xi_e^2 .
\]
Taking conditional expectation given \(\mathcal O_N\), and using that
\(\bs C_{A,N}(\mathcal O_N)\) is \(\mathcal O_N\)-measurable,
gives
\[
\E[\norm{\bs\phi_{A,N}}^2\mid\mathcal O_N]
=
\abs{\mathcal S_N}^{-2}
\norm{\bs C_{A,N}(\mathcal O_N)}^2
\E\!\left[\prod_{e\in A}\xi_e^2\mid\mathcal O_N\right].
\]
Under Assumption~\ref{ass:dyad_ind}, the dyads are conditionally
independent given \(\mathcal O_N\). Since
\(\xi_e=L_e-P_e(\bs\beta_0)\) is a measurable function of \(L_e\) and
\(\mathcal O_N\), \(\{\xi_e:e\in A\}\) are conditionally
independent given \(\mathcal O_N\). Therefore
\(\E\!\left[\prod_{e\in A}\xi_e^2\mid\mathcal O_N\right]
=
\prod_{e\in A}\E[\xi_e^2\mid\mathcal O_N].\)
For each dyad \(e\), \(L_e\mid\mathcal O_N\) is Bernoulli with success
probability \(P_e(\bs\beta_0)\), so
\[
\E[\xi_e^2\mid\mathcal O_N]
=
\E[(L_e-P_e(\bs\beta_0))^2\mid\mathcal O_N]
=
P_e(\bs\beta_0)\{1-P_e(\bs\beta_0)\}.
\]
Thus
\[
\E\!\left[\prod_{e\in A}\xi_e^2\mid\mathcal O_N\right]
=
\prod_{e\in A}P_e(\bs\beta_0)\{1-P_e(\bs\beta_0)\}.
\]
Taking expectations once more yields
\[
\E[\norm{\bs\phi_{A,N}}^2]
=
\abs{\mathcal S_N}^{-2}
\E\!\left[
\norm{\bs C_{A,N}(\mathcal O_N)}^2
\prod_{e\in A}P_e(\bs\beta_0)\{1-P_e(\bs\beta_0)\}
\right].
\]

By Lemma~\ref{lem:exact_representation},
\[
\bs C_{A,N}(\mathcal O_N)
=
\sum_{s\in\mathcal S_N:\,A\subseteq E(s)}
\sum_{\nu:A\subseteq E_{s,\nu}}
\bs c_{s,\nu}(\bs\beta_0,\mathcal O_N)
\prod_{f\in E_{s,\nu}\setminus A}P_f(\bs\beta_0).
\]
Expanding the square produces finitely many cross products. Each is
indexed by \((s,\nu,t,\mu)\), with \(s,t\in\mathcal S_N\),
\(A\subseteq E(s)\cap E(t)\), \(A\subseteq E_{s,\nu}\), and
\(A\subseteq E_{t,\mu}\), and has the form
\[
\left\langle
\bs c_{s,\nu}(\bs\beta_0,\mathcal O_N),
\bs c_{t,\mu}(\bs\beta_0,\mathcal O_N)
\right\rangle
\prod_{f\in E_{s,\nu}\setminus A}P_f(\bs\beta_0)
\prod_{h\in E_{t,\mu}\setminus A}P_h(\bs\beta_0).
\]
The number of monomials in each kernel is uniformly bounded, and
Lemma~\ref{lem:poly} gives a uniform bound on
\(\norm{\bs c_{s,\nu}(\bs\beta_0,\mathcal O_N)}\)
over \(s\in\mathcal S_N\) and \(\nu\le\nu_{\max}\).
After multiplying by
\(\prod_{e\in A}P_e(\bs\beta_0)\{1-P_e(\bs\beta_0)\}\) and using
\(1-P_e(\bs\beta_0)\le1\), the contribution of each cross product is
bounded up to a universal constant by
\[
\prod_{f\in E_{s,\nu}\setminus A}P_f(\bs\beta_0)
\prod_{h\in E_{t,\mu}\setminus A}P_h(\bs\beta_0)
\prod_{e\in A}P_e(\bs\beta_0).
\]
A dyad that appears in both \(E_{s,\nu}\setminus A\) and
\(E_{t,\mu}\setminus A\) contributes two to the total exponent of this
product. Since \([G_A]\in g\),
\(\abs{A}=r(g)\). Also, the two surviving monomial supports both
contain \(A\), and thus the definition of the
minimal dyadic degree gives \(\abs{E_{s,\nu}}\ge m(g)\) and
\(\abs{E_{t,\mu}}\ge m(g)\). Moreover,
\[
\abs{E_{s,\nu}\setminus A}=\abs{E_{s,\nu}}-r(g),
\qquad
\abs{E_{t,\mu}\setminus A}=\abs{E_{t,\mu}}-r(g),
\qquad
\abs{A}=r(g).
\]
Therefore the total exponent of the product of link probabilities is
\(
\{\abs{E_{s,\nu}}-r(g)\}
+\{\abs{E_{t,\mu}}-r(g)\}
+r(g)
=
\abs{E_{s,\nu}}+\abs{E_{t,\mu}}-r(g).
\)
The link-probability factors are drawn from the two monomial supports
\(E_{s,\nu}\) and \(E_{t,\mu}\), so their incidence multiplicity at
every node is at most eight. For the graph formed by
\(E_{s,\nu}\cup E_{t,\mu}\), use its distinct node labels as
\(i_1,\ldots,i_{\abs V}\) and let \(r_{uv}\) be the multiplicity of the
corresponding link-probability factor in this product. Applying
\eqref{eq:weighted_fixed_label_bound} in
Lemma~\ref{lem:weighted_probability_sums} gives
\[
\E\!\left[
\prod_{f\in E_{s,\nu}\setminus A}P_f(\bs\beta_0)
\prod_{h\in E_{t,\mu}\setminus A}P_h(\bs\beta_0)
\prod_{e\in A}P_e(\bs\beta_0)
\right]
\le C\rho_N^{\abs{E_{s,\nu}}+\abs{E_{t,\mu}}-r(g)}.
\]
Since
\(\abs{E_{s,\nu}}+\abs{E_{t,\mu}}-r(g)\ge 2m(g)-r(g)\) and
\(\rho_N\le1\), the expectation is bounded by a constant times
\(\rho_N^{2m(g)-r(g)}\).

It remains to count how many such cross products appear for the fixed
labeled dyad set \(A\). The set \(A\) contains exactly \(v(g)\) distinct
nodes. A configuration \(s=(i,j,k,l,m)\) satisfying \(A\subseteq E(s)\)
must place these \(v(g)\) fixed node labels into the five roles of the
configuration. The number of admissible role assignments depends only on
the graph shape represented by \(A\), not on \(N\). Once these roles
are fixed,
the remaining \(5-v(g)\) nodes can be filled by distinct labels outside
the node set of \(A\), giving \(O(N^{5-v(g)})\) possible configurations
\(s\). Hence the number of ordered pairs \((s,t)\) satisfying
\(A\subseteq E(s)\cap E(t)\) is \(O(N^{2(5-v(g))})\). The number of
monomial pairs \((\nu,\mu)\) is uniformly bounded by Lemma~\ref{lem:poly}.
Since \(\abs{\mathcal S_N}\asymp N^5\), we obtain
\[
\E[\norm{\bs\phi_{A,N}}^2]
=
O\!\left(
N^{-10}N^{2(5-v(g))}
\rho_N^{2m(g)-r(g)}
\right)
=
O\!\left(
N^{-2v(g)}
\rho_N^{2m(g)-r(g)}
\right).
\]

Finally, recall that \(\bs\Pi_N(g)=\sum_{A\in\mathscr C_N(g)}\bs\phi_{A,N}\). Therefore
\[
\norm{\bs\Pi_N(g)}^2
=
\sum_{A\in\mathscr C_N(g)}\norm{\bs\phi_{A,N}}^2
+
\sum_{A,A'\in\mathscr C_N(g):\,A\neq A'}
\bs\phi_{A,N}^{\top}\bs\phi_{A',N}.
\]
For \(A\neq A'\), Lemma~\ref{lem:dyadic_hoeffding_orthogonality} gives
\(\E[\bs\phi_{A,N}\bs\phi_{A',N}^{\top}\mid\mathcal O_N]=0.\)
Thus, \(\E[\norm{\bs\Pi_N(g)}^2]=\sum_{A\in\mathscr C_N(g)}\E[\norm{\bs\phi_{A,N}}^2]\).
There are \(O(N^{v(g)})\) labeled copies \(A\in\mathscr C_N(g)\).
Therefore
\[
\E[\norm{\bs\Pi_N(g)}^2]
=
O\!\left(
N^{v(g)}
N^{-2v(g)}
\rho_N^{2m(g)-r(g)}
\right)
=
O\!\left(
N^{-v(g)}
\rho_N^{2m(g)-r(g)}
\right),
\]
which is exactly \eqref{eq:classbound}. This proves the bound.
\end{proof}

To establish asymptotic normality in the sparse and ultra-sparse 
regimes, we apply Lemma~\ref{lem:dejong_multilinear_form}, and we verify
its condition~\eqref{eq:mf-clt-fourth-condition}
in the following lemma.
\begin{lem}
\label{lem:sparse_dyadset_fourth}
Suppose Assumptions~\ref{ass:dyad_ind}, \ref{ass:bounded_design}, and
\ref{ass:sparse_envelope} hold, and fix \(\bs{c}\in\mathbb R^q\). In the
sparse regime \(\rho_N\asymp N^{-1}\), let \(g\) be any one of
\(g_{\mathrm{dyad}},g_{\text{2-star}},g_{\mathrm{3tree}}\), and
\(g_{\mathrm{4tree}}\), set \(\mathcal A_N:=\mathscr C_N(g)\), and set
\(b_{A,N}:=N^{9/2}\abs{\mathcal S_N}^{-1}
\bs{c}^\top\bs C_{A,N}(\mathcal O_N)\). In the ultra-sparse regime
\(N\rho_N\to0\) and \(N^5\rho_N^4\to\infty\), set
\(\mathcal A_N:=\mathscr C_N(g_{\mathrm{4tree}})\) and
\(b_{A,N}:=N^{5/2}\rho_N^{-2}\abs{\mathcal S_N}^{-1}
\bs{c}^\top\bs C_{A,N}(\mathcal O_N)\). In either case, let
\(Q_N:=\sum_{A\in\mathcal A_N}b_{A,N}\prod_{e\in A}\xi_e\) and
\(V_N:=\Var(Q_N\mid\mathcal O_N)\). Then
\(\E[Q_N^4\mid\mathcal O_N]=3V_N^2+o_p(1)\).
\end{lem}

\begin{figure}[!ht]
\centering
\begin{subfigure}[t]{0.30\textwidth}
\centering
\begin{tikzpicture}[
  x=0.50cm,y=0.50cm,
  every node/.style={font=\scriptsize},
  dot/.style={circle,fill=black,inner sep=1.45pt},
  edge/.style={line width=0.75pt}
]
  \foreach \x/\lab in {0/1,1/2,2/3,3/4,4/5,5/6} {
    \node[dot,label=below:{$\lab$}] (c\lab) at (\x,0) {};
  }
  \draw[edge] (c1)--(c2);
  \draw[edge] (c2)--(c3);
  \draw[edge] (c3)--(c4);
  \draw[edge] (c4)--(c5);
  \draw[edge] (c5)--(c6);
\end{tikzpicture}
\caption{Connected union}
\end{subfigure}
\hfill
\begin{subfigure}[t]{0.34\textwidth}
\centering
\begin{tikzpicture}[
  x=0.50cm,y=0.50cm,
  every node/.style={font=\scriptsize},
  dot/.style={circle,fill=black,inner sep=1.45pt},
  edge/.style={line width=0.75pt}
]
  \foreach \x/\lab in {0/1,1/2,2/3,3/4,4/5} {
    \node[dot,label=above:{$\lab$}] (d\lab) at (\x,0) {};
  }
  \node[dot,label=below:{$6$}] (d6) at (0,-1.35) {};
  \node[dot,label=below:{$7$}] (d7) at (1,-1.35) {};
  \node[dot,label=below:{$8$}] (d8) at (2,-1.35) {};
  \node[dot,label=below:{$9$}] (d9) at (3,-1.35) {};
  \node[dot,label=below:{$10$}] (d10) at (4,-1.35) {};
  \draw[edge] (d1)--(d2);
  \draw[edge] (d2)--(d3);
  \draw[edge] (d3)--(d4);
  \draw[edge] (d4)--(d5);
  \draw[edge] (d6)--(d7);
  \draw[edge] (d7)--(d8);
  \draw[edge] (d8)--(d9);
  \draw[edge] (d9)--(d10);
\end{tikzpicture}
\caption{Disconnected union}
\end{subfigure}
\hfill
\begin{subfigure}[t]{0.30\textwidth}
\centering
\begin{tikzpicture}[
  x=0.45cm,y=0.50cm,
  every node/.style={font=\scriptsize},
  dot/.style={circle,fill=black,inner sep=1.45pt},
  edge/.style={line width=0.75pt}
]
  \node[dot,label=above:{$1$}] (b1) at (0,0) {};
  \node[dot,label=above:{$2$}] (b2) at (1,0) {};
  \node[dot,label=above:{$3$}] (b3) at (2,0) {};
  \node[dot,label=above:{$4$}] (b4) at (3,0) {};
  \node[dot,label=above:{$5$}] (b5) at (4,0) {};
  \node[dot,label=below:{$6$}] (b6) at (3,-1.35) {};
  \node[dot,label=below:{$7$}] (b7) at (4,-1.35) {};
  \node[dot,label=below:{$8$}] (b8) at (5,-1.35) {};
  \node[dot,label=below:{$9$}] (b9) at (6,-1.35) {};
  \draw[edge] (b1)--(b2);
  \draw[edge] (b2)--(b3);
  \draw[edge] (b3)--(b4);
  \draw[edge] (b4)--(b5);
  \draw[edge] (b4)--(b6);
  \draw[edge] (b6)--(b7);
  \draw[edge] (b7)--(b8);
  \draw[edge] (b8)--(b9);
\end{tikzpicture}
\caption{Connected union}
\end{subfigure}
\caption[Connectedness of dyad-set unions]{Connectedness of the union of
four dyad sets, using copies of the four-edge path on five nodes from
Table~\ref{tab:regimemap}. Here \(A_j\) denotes the \(j\)th dyad set in
the quadruple \((A_1,A_2,A_3,A_4)\), and \(12\) denotes the dyad between
nodes \(1\) and \(2\). In panel (a), take
\(A_1=A_3=\{12,23,34,45\}\) and
\(A_2=A_4=\{23,34,45,56\}\); each \(A_j\) is a four-edge path on five
nodes, and the union is connected. In panel (b), take
\(A_1=A_2=\{12,23,34,45\}\) and
\(A_3=A_4=\{67,78,89,9\,10\}\); the two paths use disjoint node sets,
so the union has two connected components. In panel (c), take
\(A_1=A_2=\{12,23,34,45\}\) and
\(A_3=A_4=\{46,67,78,89\}\). The dyad \(46\) attaches the second
four-edge path to node \(4\) in the first path, so the union is
connected.}
\label{fig:sparse-fourth-unions}
\end{figure}
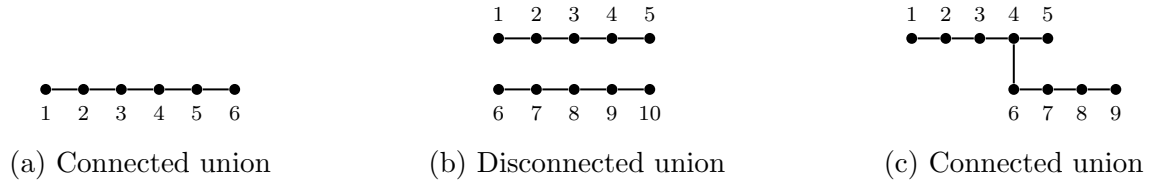

\begin{proof}
Set \(Q_{A,N}:=b_{A,N}\prod_{e\in A}\xi_e\), so
\(Q_N=\sum_{A\in\mathcal A_N}Q_{A,N}\).
Expand \(\E[Q_N^4\mid\mathcal O_N]\) as a sum over
\((A_1,A_2,A_3,A_4)\in\mathcal A_N^4\) with
\(Q_N^4=\sum_{(A_1,A_2,A_3,A_4)\in\mathcal A_N^4}
\prod_{j=1}^4 Q_{A_j,N}\).
For a fixed quadruple, write \(k=\abs{A_1\cup A_2\cup A_3\cup A_4}\) and
let \(n_e:=\abs{\{j:e\in A_j\}}\) be the number of times dyad \(e\)
appears in the quadruple. Since the coefficients \(b_{A,N}\) are
\(\mathcal O_N\)-measurable,
\[
\E\!\left[
\prod_{j=1}^4 Q_{A_j,N}
\mid \mathcal O_N
\right]
=
\left(\prod_{j=1}^4 b_{A_j,N}\right)
\E\!\left[
\prod_{j=1}^4\prod_{e\in A_j}\xi_e
\mid \mathcal O_N
\right].
\]
Since each dyad \(e\) in \(A_1\cup A_2\cup A_3\cup A_4\) appears
\(n_e\) times across the four sets,
\(\prod_{j=1}^4\prod_{e\in A_j}\xi_e
=\prod_{e\in A_1\cup A_2\cup A_3\cup A_4}\xi_e^{n_e}\).
The conditional independence of the dyad shocks given \(\mathcal O_N\)
therefore yields
\begin{equation}
\label{eq:sparse-fourth-prod-moment}
\E\!\left[
\prod_{j=1}^4 Q_{A_j,N}
\mid \mathcal O_N
\right]
=
\left(\prod_{j=1}^4 b_{A_j,N}\right)
\prod_{e\in A_1\cup A_2\cup A_3\cup A_4}\E[\xi_e^{n_e}\mid\mathcal O_N].
\end{equation}

If some dyad appears exactly once \(n_e =1\), then one factor is
\(\E[\xi_e\mid\mathcal O_N]=0\), and the whole term in \eqref{eq:sparse-fourth-prod-moment} vanishes.
Hence every nonzero quadruple must satisfy
\begin{equation}
\label{eq:sparse-fourth-no-singleton}
n_e\ge 2
\qquad\text{for all }e\in A_1\cup A_2\cup A_3\cup A_4.
\end{equation}
Moreover, since \(\abs{\xi_e}\le 1\), for every \(m\ge 2\) we have
\(\abs{\xi_e}^m\le \xi_e^2\). Lemma
\ref{lem:conditional_degree_comparison} therefore gives
\begin{equation}
\label{eq:sparse-fourth-moment-envelope}
\E[\abs{\xi_e}^m\mid\mathcal O_N]\le
\E[\xi_e^2\mid\mathcal O_N]\le P_e(\bs\beta_0).
\end{equation}
For every nonzero quadruple with \(k\) distinct dyads,
\eqref{eq:sparse-fourth-no-singleton} and
\eqref{eq:sparse-fourth-moment-envelope} bound the product of conditional
moments in \eqref{eq:sparse-fourth-prod-moment} as follows:
\begin{equation}
\label{eq:sparse-fourth-moment-bound}
\left|
\prod_{e\in A_1\cup A_2\cup A_3\cup A_4}\E[\xi_e^{n_e}\mid\mathcal O_N]
\right|
\le
\prod_{e\in A_1\cup A_2\cup A_3\cup A_4}P_e(\bs\beta_0).
\end{equation}

Next, we compute the conditional variance. If \(A\ne B\), then
\(A\triangle B\ne\varnothing\), and
\[
\begin{aligned}
&\E[Q_{A,N}Q_{B,N}\mid\mathcal O_N]
=b_{A,N}b_{B,N}
\E\!\left[\prod_{e\in A\cap B}\xi_e^2
\prod_{e\in A\triangle B}\xi_e\mid\mathcal O_N\right]\\
&=
b_{A,N}b_{B,N}
\prod_{e\in A\cap B}\E[\xi_e^2\mid\mathcal O_N]
\prod_{e\in A\triangle B}\E[\xi_e\mid\mathcal O_N]=0.
\end{aligned}
\]
Therefore \(Q_{A,N}\) are conditionally orthogonal, and
\begin{equation}
\label{eq:sparse-fourth-vdecomp}
V_N=\sum_{A\in\mathcal A_N}v_{A,N},\qquad
v_{A,N}:=\E[Q_{A,N}^2\mid\mathcal O_N]
=\abs{b_{A,N}}^2\E\!\left[\prod_{e\in A}\xi_e^2\mid\mathcal O_N\right].
\end{equation}

We now separate the fourth-moment expansion according to whether the
four dyad sets form two identical pairs. The leading terms have one of the
three forms \((A,A,B,B)\), \((A,B,A,B)\), or \((A,B,B,A)\), with
\(A\cap B=\varnothing\); see Figure~\ref{fig:sparse-fourth-unions} for examples.
Using \(v_{A,N}\) from
\eqref{eq:sparse-fourth-vdecomp}, the contribution of the ordered
quadruple \((A,A,B,B)\) is
\[
\begin{aligned}
&\E[Q_{A,N}Q_{A,N}Q_{B,N}Q_{B,N}\mid\mathcal O_N]
=b_{A,N}^2b_{B,N}^2
\E\!\left[\prod_{e\in A}\xi_e^2
\prod_{f\in B}\xi_f^2\mid\mathcal O_N\right]\\
&=
\abs{b_{A,N}}^2
\E\!\left[\prod_{e\in A}\xi_e^2\mid\mathcal O_N\right]
\abs{b_{B,N}}^2
\E\!\left[\prod_{f\in B}\xi_f^2\mid\mathcal O_N\right]
=v_{A,N}v_{B,N},
\end{aligned}
\]
where the second equality uses \(A\cap B=\varnothing\) and conditional
independence of the dyad shocks. The same value is obtained from the
other two placements, \((A,B,A,B)\) and \((A,B,B,A)\). Hence the
disjoint two-pair terms contribute
\(3\sum_{A,B\in\mathcal A_N:\,A\cap B=\varnothing}
v_{A,N}v_{B,N}\).
By \eqref{eq:sparse-fourth-vdecomp}, this is the disjoint part of
\(3V_N^2=3\sum_{A,B\in\mathcal A_N}v_{A,N}v_{B,N}\).
The remaining, overlapping part of the target \(3V_N^2\) is
\begin{equation}
\label{eq:sparse-fourth-overlap-pairs}
3\sum_{A,B\in\mathcal A_N:\,A\cap B\ne\varnothing}
v_{A,N}v_{B,N}.
\end{equation}
This target term need not equal the contribution of the corresponding
overlapping quadruples to the actual fourth-moment expansion. We therefore
show separately that \eqref{eq:sparse-fourth-overlap-pairs} and the sum of
all actual nonzero quadruple terms outside the disjoint two-pair class are
both \(o_p(1)\).

We first claim that every actual nonzero quadruple outside the disjoint
two-pair class has connected
union. Suppose to the contrary that some remaining nonzero quadruple has
disconnected union. Fix such a quadruple and write it as
\((A_1,A_2,A_3,A_4)\). Write
\(E:=A_1\cup A_2\cup A_3\cup A_4\) for the edge set of its
union graph. Then \(E\) can be decomposed into the edge sets of the
connected components of that graph, \(E=E^{(1)}\cup\cdots\cup E^{(q)}\)
with \(q\ge 2\).
Since every \(A_j\in\mathcal A_N\) is connected, each \(A_j\) must be
contained in exactly one component. For each component define
\(J_a:=\{j\in\{1,2,3,4\}: A_j\subseteq E^{(a)}\}\).
Each nonempty \(J_a\) must contain at least two indices. Indeed, if
\(J_a=\{j\}\),
then every dyad \(e\in A_j\subseteq E^{(a)}\) appears only in \(A_j\),
so \(n_e=1\), contradicting \eqref{eq:sparse-fourth-no-singleton}.
Since there are only four indices, this forces \(q=2\), and after
relabeling \(J_1=\{i,j\}\) and \(J_2=\{k,\ell\}\).

Now take any dyad \(e\in A_i\cup A_j\). Because no dyad of \(E\) joins
two different connected components, \(e\) cannot belong to \(A_k\) or
\(A_\ell\). Hence \(e\) appears only through \(A_i\) and \(A_j\). Since
the quadruple is nonzero, \(n_e\ge2\), so \(e\) must belong to both
\(A_i\) and \(A_j\). Therefore \(A_i=A_j\). Similarly, \(A_k=A_\ell\).
Since \(E^{(1)}\) and \(E^{(2)}\) are different connected components, we
also have \(A_i\cap A_k=\varnothing\). Thus the quadruple is exactly a
disjoint two-pair leading term, contradicting that it belongs to the
remainder. This proves the claim.

Using the definitions of \(b_{A_j,N}\), we now expand each
\(\bs C_{A_j,N}(\mathcal O_N)\), \(j=1,\ldots,4\), by
\eqref{eq:projection_coeff_exact}. For each \(j=1,\ldots,4\), an
expanded coefficient term selects an ordered pentad \(s_j\) and a
monomial support \(E_{s_j,\nu_j}\) containing \(A_j\). After bounding
\(\norm{\bs c_{s_j,\nu_j}(\bs\beta_0,\mathcal O_N)}\) by
Lemma~\ref{lem:poly} and applying \eqref{eq:sparse-fourth-moment-bound},
the probability factors in each expanded term are
\begin{equation}
\label{eq:sparse-fourth-probability-product}
\left\{\prod_{j=1}^4
\prod_{e\in E_{s_j,\nu_j}\setminus A_j}P_e(\bs\beta_0)\right\}
\prod_{e\in A_1\cup\cdots\cup A_4}P_e(\bs\beta_0),
\end{equation}
where \(p:=\sum_{j=1}^4\abs{E_{s_j,\nu_j}\setminus A_j}
+\abs{A_1\cup\cdots\cup A_4}\) is the total exponent.
The support graph has edge set \(\bigcup_{j=1}^4 E_{s_j,\nu_j}\).
Each \(E_{s_j,\nu_j}\) is connected by
Lemma~\ref{lem:det-row-supports} and contains the nonempty set \(A_j\).
Since \(A_1\cup\cdots\cup A_4\) is connected, the union of the four
monomial supports is connected. Each distinct dyad contributes at least
one factor to \eqref{eq:sparse-fourth-probability-product}, so the union
has at most \(p\) dyads and hence at most \(p+1\) nodes.
Each pentad support has node degree at most four, so the incidence
multiplicity in \eqref{eq:sparse-fourth-probability-product} is at most
\(4\cdot4=16\) at each node. For each fixed pattern
of node overlaps and monomial supports, list the distinct nodes of the
four pentads as \(i_1,\ldots,i_{\abs V}\). The sum for that pattern is
bounded, up to a constant independent of \(N\), by the sum over
distinct labels in \eqref{eq:weighted_unrooted_sum}, with \(r_{uv}\) equal to
the multiplicity of the corresponding probability factor in
\eqref{eq:sparse-fourth-probability-product}, total exponent \(p\), and
\(\abs V\le p+1\). Thus \eqref{eq:weighted_unrooted_sum}
bounds the expected sum for each pattern by \(C N^{p+1}\rho_N^p\).

In the sparse regime, the product of the four coefficient
normalizations is
\(
\left(\frac{N^{9/2}}{\abs{\mathcal S_N}}\right)^4=O(N^{-2}).
\)
Thus, for each fixed pattern of node overlaps and monomial supports
in the remainder, the expectation of the sum of absolute contributions
is at most
\(
C N^{-2}N^{p+1}\rho_N^p
=C N^{-1}(N\rho_N)^p=O(N^{-1}),
\)
because \(\rho_N\asymp N^{-1}\). Only finitely many unlabeled joint
patterns arise for \(A_1,\ldots,A_4\), the overlaps among their containing
pentads \(s_1,\ldots,s_4\), and the selected monomial supports. Markov's
inequality therefore gives
\[
\sum_{(A_1,\ldots,A_4)\in\mathcal A_N^4}^{*}
\left|\E\!\left[\prod_{j=1}^4Q_{A_j,N}\mid\mathcal O_N\right]\right|
=o_p(1),
\]
where the star excludes \((A,A,B,B)\), \((A,B,A,B)\), and
\((A,B,B,A)\) with \(A,B\in\mathcal A_N\) and
\(A\cap B=\varnothing\).

We use the coefficient expansion in \eqref{eq:projection_coeff_exact}
to bound \eqref{eq:sparse-fourth-overlap-pairs}. If \(A\cap B\ne\varnothing\),
then \(A\cup B\) is connected. Substitute the definitions of
\(b_{A,N}\) and \(b_{B,N}\) into \(v_{A,N}v_{B,N}\), expand
\(\bs C_{A,N}(\mathcal O_N)\) and \(\bs C_{B,N}(\mathcal O_N)\), and use
\(P_e(\bs\beta_0)\{1-P_e(\bs\beta_0)\}\le P_e(\bs\beta_0)\) for each
variance factor and \(P_e(\bs\beta_0)^2\le P_e(\bs\beta_0)\) for
\(e\in A\cap B\). The product of the variance factors over \(A\) and
\(B\) is then bounded by \(\prod_{e\in A\cup B}P_e(\bs\beta_0)\).
With \(A_1=A_2=A\) and \(A_3=A_4=B\), the \(O(N^{-1})\) bound for
connected unions of four monomial supports therefore applies to every
overlap pattern in
\eqref{eq:sparse-fourth-overlap-pairs}. Hence
\(
\E\!\left[
\left|\E[Q_N^4\mid\mathcal O_N]-3V_N^2\right|
\right]
\le C N^{-1}\to 0
\).
Markov's inequality gives
\begin{equation}
\label{eq:sparse-fourth-fourthmoment-regular}
\E[Q_N^4\mid\mathcal O_N]=3V_N^2+o_p(1).
\end{equation}

In the ultra-sparse regime, \(\mathcal A_N=\mathscr C_N(g_{\mathrm{4tree}})\),
so every \(A\in\mathcal A_N\) satisfies \(\abs A=4\) and \(\abs{V(A)}=5\).
The product of the four coefficient normalizations is
\(
\left(
\frac{N^{5/2}\rho_N^{-2}}{\abs{\mathcal S_N}}
\right)^4
=O(N^{-10}\rho_N^{-8}).
\)
For an expanded connected pattern, let
\(p:=\sum_{j=1}^4\abs{E_{s_j,\nu_j}\setminus A_j}
+\abs{A_1\cup\cdots\cup A_4}\) denote the total exponent of the product
of link probabilities used to bound the terms with that pattern.
Applying \eqref{eq:weighted_unrooted_sum} to the graphs formed by the
four monomial supports, with \(\abs V\le p+1\),
total exponent \(p\), and
\(\kappa_v\le16\), and multiplying by the coefficient normalization
therefore gives the bound
\begin{equation}
\label{eq:sparse-fourth-remainder-ultrasparse}
C N^{-10}\rho_N^{-8}N^{p+1}\rho_N^p
=C N^{p-9}\rho_N^{p-8}.
\end{equation}
A nonzero quadruple contains at least the four distinct dyads of
\(A_1\), so \(p\ge4\). Since \(N\rho_N\to0\), the bound in
\eqref{eq:sparse-fourth-remainder-ultrasparse} equals
\(C(N\rho_N)^{p-4}/(N^5\rho_N^4)\le C/(N^5\rho_N^4)=o(1)\)
for all sufficiently large \(N\).
The bound in \eqref{eq:sparse-fourth-remainder-ultrasparse} also applies
to every pair \(A,B\) summed in \eqref{eq:sparse-fourth-overlap-pairs},
after using \(P_e(\bs\beta_0)^2\le P_e(\bs\beta_0)\) on \(A\cap B\) again.
Since \(\abs A=\abs B=4\), their union contains at least four distinct
dyads, so \(p\ge4\) also holds for each expanded term.
Summing the expectation bounds \(C/(N^5\rho_N^4)\) for the
absolute contributions of the finitely many patterns gives
\(
\E\!\left[
\left|\E[Q_N^4\mid\mathcal O_N]-3V_N^2\right|
\right]
\to0.
\)
Markov's inequality now gives
\(
\E[Q_N^4\mid\mathcal O_N]=3V_N^2+o_p(1),
\) which, combined with
equation~\eqref{eq:sparse-fourth-fourthmoment-regular}, proves the lemma.
\end{proof}

Now, we can prove Theorem~\ref{thm:regimeCLT}.
\begin{proof}[Proof of Theorem~\ref{thm:regimeCLT}]
Set
\(\bs R_N^{\mathrm D}:=\bs h_N(\bs\beta_0)-\bs\Pi_N^{\mathrm D}\),
\(\bs R_N^{\mathrm S}:=\bs h_N(\bs\beta_0)-\bs\Pi_N^{\mathrm S}\), and
\(\bs R_N^{\mathrm{US}}:=\bs h_N(\bs\beta_0)-\bs\Pi_N^{\mathrm{US}}\).
By \eqref{eq:exactdecomp}, each remainder is the sum of the
projection sums in Table~\ref{tab:regimemap} that are not
included in the corresponding leading term. Lemma~\ref{lem:dyadic_hoeffding_orthogonality}
removes the cross terms, and Proposition~\ref{prop:classbound}, with the
values of \((r(g),v(g),m(g))\) in Table~\ref{tab:regimemap} and
\(\rho_N\le1\), gives
\begin{equation}
\label{eq:regime_remainder_second_moments}
\begin{aligned}
N^2\rho_N^{-7}\E[\norm{\bs R_N^{\mathrm{D}}}^2]
&=O((N\rho_N)^{-1})+O((N\rho_N)^{-2})
  +O((N\rho_N)^{-3})=o(1),\\
N^9\E[\norm{\bs R_N^{\mathrm{S}}}^2]
&=O((N\rho_N)^5\rho_N)+O((N\rho_N)^6\rho_N)
  +O((N\rho_N)^4\rho_N)=o(1),\\
N^5\rho_N^{-4}\E[\norm{\bs R_N^{\mathrm{US}}}^2]
&=O((N\rho_N)^3)+O((N\rho_N)^2)+O(N\rho_N)+O(\rho_N)=o(1).
\end{aligned}
\end{equation}
The three lines use \(N\rho_N\to\infty\),
\(\rho_N\asymp N^{-1}\), and \(N\rho_N\to0\), respectively.
Chebyshev's inequality therefore gives
\(N\rho_N^{-7/2}\bs R_N^{\mathrm D}=o_p(1)\),
\(N^{9/2}\bs R_N^{\mathrm S}=o_p(1)\), and
\(N^{5/2}\rho_N^{-2}\bs R_N^{\mathrm{US}}=o_p(1)\).
It remains to prove the Gaussian limits for the leading projection terms in
the three cases.
Fix \(\bs{c}\in\mathbb R^q\). We prove convergence for each fixed \(\bs{c}\) and then invoke
the Cram\'er--Wold device. Let \(a_N=N\rho_N^{-7/2}\), \(N^{9/2}\), or
\(N^{5/2}\rho_N^{-2}\) in the dense, sparse, or ultra-sparse regime,
respectively, and set
\(b_{A,N}(\bs{c};\mathcal O_N):=
a_N\abs{\mathcal S_N}^{-1}\bs{c}^\top\bs C_{A,N}(\mathcal O_N)\).
Lemma~\ref{lem:exact_representation} gives, for each graph class \(g\)
in Table~\ref{tab:regimemap},
\[
a_N\bs{c}^\top\bs\Pi_N(g)
=\sum_{A\in\mathscr C_N(g)}
b_{A,N}(\bs{c};\mathcal O_N)\prod_{e\in A}\xi_e.
\]
The coefficients \(b_{A,N}(\bs{c};\mathcal O_N)\) are
\(\mathcal O_N\)-measurable. Under Assumption~\ref{ass:dyad_ind}, the
shocks \(\{\xi_e:e\in\mathcal E_N\}\) are conditionally independent
given \(\mathcal O_N\), with \(\E[\xi_e\mid\mathcal O_N]=0\),
\(\abs{\xi_e}\le1\), and
\(\E[\xi_e^2\mid\mathcal O_N]=P_e(\bs\beta_0)\{1-P_e(\bs\beta_0)\}\).

In the dense or mildly sparse regime, the leading term is
\(\bs\Pi_N^{\mathrm{D}}=\bs\Pi_N(g_{\mathrm{dyad}})\).
The normalized projection \(N\rho_N^{-7/2}\bs{c}^\top\bs\Pi_N^{\mathrm D}\)
is therefore \(\sum_{e\in\mathcal E_N}b_{\{e\},N}(\bs{c};\mathcal O_N)\xi_e\).
The summands \(b_{\{e\},N}(\bs{c};\mathcal O_N)\xi_e\) are therefore
conditionally independent and centered given \(\mathcal O_N\).
To verify the conditional Lindeberg condition in
Lemma~\ref{lem:conditional_triangular_array_clt}, fix \(\varepsilon>0\) and use
\begin{equation}
\label{eq:regime_dense_lindeberg_bound}
\begin{aligned}
&\sum_{e\in\mathcal E_N}
\E\!\left[
b_{\{e\},N}(\bs{c};\mathcal O_N)^2\xi_e^2
\1\{|b_{\{e\},N}(\bs{c};\mathcal O_N)\xi_e|>\varepsilon\}
\mid\mathcal O_N\right]\le
\varepsilon^{-2}
\sum_{e\in\mathcal E_N}
b_{\{e\},N}(\bs{c};\mathcal O_N)^4\E[\xi_e^4\mid\mathcal O_N].
\end{aligned}
\end{equation}
It therefore suffices to prove
\(\sum_{e\in\mathcal E_N}b_{\{e\},N}(\bs{c};\mathcal O_N)^4
\E[\xi_e^4\mid\mathcal O_N]=o_p(1)\).
For a fixed dyad \(e\), substitute \eqref{eq:projection_coeff_exact}
into \(b_{\{e\},N}(\bs{c};\mathcal O_N)\) and expand its fourth power.
Each term selects \(s_j\in\mathcal S_N\) and a
monomial index \(\nu_j\in\{1,\ldots,\nu_{\max}\}\) from each of the
four factors, with \(e\in E_{s_j,\nu_j}\) for \(j=1,\ldots,4\).
 Thus
\[
\begin{aligned}
&b_{\{e\},N}(\bs{c};\mathcal O_N)^4=\left(\frac{N\rho_N^{-7/2}}{\abs{\mathcal S_N}}\right)^4
\sum_{s_1,\ldots,s_4\in\mathcal S_N}
\sum_{\substack{\nu_1,\ldots,\nu_4\\
e\in\bigcap_{j=1}^4 E_{s_j,\nu_j}}}
\prod_{j=1}^4
\bigl\{\bs{c}^\top\bs c_{s_j,\nu_j}(\bs\beta_0,\mathcal O_N)\bigr\}\\
&\qquad\qquad\qquad\qquad\times
\prod_{j=1}^4\prod_{f\in E_{s_j,\nu_j}\setminus\{e\}}
P_f(\bs\beta_0).
\end{aligned}
\]
By Lemma~\ref{lem:poly} and \(\abs{\xi_e}\le1\),
\[
\begin{aligned}
\prod_{j=1}^4\abs{\bs{c}^\top
\bs c_{s_j,\nu_j}(\bs\beta_0,\mathcal O_N)}
&\le \norm{\bs c}^4
\prod_{j=1}^4\norm{\bs c_{s_j,\nu_j}(\bs\beta_0,\mathcal O_N)}
\le C,\\
\E[\xi_e^4\mid\mathcal O_N]
&\le \E[\xi_e^2\mid\mathcal O_N]
=P_e(\bs\beta_0)\{1-P_e(\bs\beta_0)\}
\le P_e(\bs\beta_0).
\end{aligned}
\]
The  bound on \(\prod_{j=1}^4\abs{\bs{c}^\top
\bs c_{s_j,\nu_j}(\bs\beta_0,\mathcal O_N)}\) holds almost surely, uniformly over \(N\) and \(s_j,\nu_j\); \(C\) may depend on the fixed vector \(\bs c\).
Thus
\[
\begin{aligned}
&b_{\{e\},N}(\bs c;\mathcal O_N)^4
\E[\xi_e^4\mid\mathcal O_N]\\
&\quad=\left(\frac{N\rho_N^{-7/2}}{\abs{\mathcal S_N}}\right)^4
\sum_{s_1,\ldots,s_4\in\mathcal S_N}
\sum_{\substack{\nu_1,\ldots,\nu_4\\
e\in\bigcap_{j=1}^4E_{s_j,\nu_j}}}
\prod_{j=1}^4
\{\bs c^\top\bs c_{s_j,\nu_j}(\bs\beta_0,\mathcal O_N)\}
\E[\xi_e^4\mid\mathcal O_N]
\prod_{j=1}^4\prod_{f\in E_{s_j,\nu_j}\setminus\{e\}}P_f(\bs\beta_0)\\
&\quad\le\left(\frac{N\rho_N^{-7/2}}{\abs{\mathcal S_N}}\right)^4
\sum_{s_1,\ldots,s_4\in\mathcal S_N}
\sum_{\substack{\nu_1,\ldots,\nu_4\\
e\in\bigcap_{j=1}^4E_{s_j,\nu_j}}}
\left|\prod_{j=1}^4
\{\bs c^\top\bs c_{s_j,\nu_j}(\bs\beta_0,\mathcal O_N)\}\right|
\E[\xi_e^4\mid\mathcal O_N]
\prod_{j=1}^4\prod_{f\in E_{s_j,\nu_j}\setminus\{e\}}P_f(\bs\beta_0)\\
&\quad\le C\left(\frac{N\rho_N^{-7/2}}{\abs{\mathcal S_N}}\right)^4
\sum_{s_1,\ldots,s_4\in\mathcal S_N}
\sum_{\substack{\nu_1,\ldots,\nu_4\\
e\in\bigcap_{j=1}^4E_{s_j,\nu_j}}}
P_e(\bs\beta_0)
\prod_{j=1}^4\prod_{f\in E_{s_j,\nu_j}\setminus\{e\}}P_f(\bs\beta_0).
\end{aligned}
\]
For each selected \(e,s_1,\ldots,s_4,\nu_1,\ldots,\nu_4\), the
four supports $E_{s_j,\nu_j}$ share the two endpoints of \(e\), and each support has
five nodes. Each support therefore adds at most three nodes outside
\(e\), so their union has at most \(2+4(5-2)=14\) nodes.
By Lemma~\ref{lem:det-row-supports},
\(4\le\abs{E_{s_j,\nu_j}}\le7\), so the probability product has
total exponent
\(13\le1+\sum_{j=1}^4(\abs{E_{s_j,\nu_j}}-1)\le25\).
In the probability product, \(P_e(\bs\beta_0)\) has exponent one.
For any other dyad \(f\) in the union graph, the exponent of
\(P_f(\bs\beta_0)\) equals the number of supports
\(E_{s_j,\nu_j}\), \(j=1,\ldots,4\), containing \(f\).
Since each of the four graphs has node degrees at most four,
the incidence multiplicity at each node is at most \(4\cdot4=16\).
Applying \eqref{eq:weighted_fixed_label_bound} to the union graph,
with these exponents,
and using \(\rho_N\le1\), we obtain
\[
\begin{aligned}
&
\E\!\left[
P_e(\bs\beta_0)
\prod_{j=1}^4\prod_{f\in E_{s_j,\nu_j}\setminus\{e\}}
P_f(\bs\beta_0)
\right]
\le
C\rho_N^{\,1+\sum_{j=1}^4(\abs{E_{s_j,\nu_j}}-1)}
\le C\rho_N^{13}.
\end{aligned}
\]
The bound holds uniformly over \(N\) and all selected indices
\(e,s_1,\ldots,s_4,\nu_1,\ldots,\nu_4\).
For a fixed dyad \(e\), its endpoints can occupy any of
the seven dyad positions in \(E(s)\), in either order; the other
three positions receive distinct labels outside \(e\). Hence
\[
\abs{\{s\in\mathcal S_N:e\in E(s)\}}
=2\cdot7\,(N-2)(N-3)(N-4)\le14N^3.
\]
Since \(E_{s,\nu}\subseteq E(s)\) and each pentad has at most
\(\nu_{\max}\) monomial indices,
\[
\begin{aligned}
&\sum_{e\in\mathcal E_N}\sum_{s_1,\ldots,s_4\in\mathcal S_N}
\sum_{\substack{\nu_1,\ldots,\nu_4\\
e\in\bigcap_{j=1}^4E_{s_j,\nu_j}}}1\le
\binom N2\{14(N-2)(N-3)(N-4)\nu_{\max}\}^4
\le C N^{14}.
\end{aligned}
\]
Using
\(\{N\rho_N^{-7/2}/\abs{\mathcal S_N}\}^4
=O(N^{-16}\rho_N^{-14})\), we therefore obtain
\[
\begin{aligned}
&\E\!\left[
\sum_{e\in\mathcal E_N}
b_{\{e\},N}(\bs{c};\mathcal O_N)^4
\E[\xi_e^4\mid\mathcal O_N]
\right]\\
&\quad\le C\left(\frac{N\rho_N^{-7/2}}{\abs{\mathcal S_N}}\right)^4
\sum_{e\in\mathcal E_N}\sum_{s_1,\ldots,s_4\in\mathcal S_N}
\sum_{\substack{\nu_1,\ldots,\nu_4\\
e\in\bigcap_{j=1}^4E_{s_j,\nu_j}}}
\E\!\left[
P_e(\bs\beta_0)
\prod_{j=1}^4\prod_{f\in E_{s_j,\nu_j}\setminus\{e\}}
P_f(\bs\beta_0)\right]\\
&\quad\le C N^{-16}\rho_N^{-14}N^{14}\rho_N^{13}
=\frac{C}{N^2\rho_N}=o(1),
\end{aligned}
\]
where \(N^2\rho_N=N(N\rho_N)\to\infty\). Markov's inequality and
\eqref{eq:regime_dense_lindeberg_bound} verify the conditional Lindeberg condition in Lemma~\ref{lem:conditional_triangular_array_clt}.
Assumption~\ref{ass:regimeSigma}(i) gives
\(\Var(N\rho_N^{-7/2}\bs{c}^\top\bs\Pi_N^{\mathrm D}\mid\mathcal O_N)
\xrightarrow{P}\bs{c}^\top\bf\Sigma_{\mathrm D}\bs{c}\).
Lemma~\ref{lem:conditional_triangular_array_clt} therefore gives
\(N\rho_N^{-7/2}\bs{c}^\top\bs\Pi_N^{\mathrm D}
\xrightarrow{D}N(0,\bs{c}^\top\bf\Sigma_{\mathrm D}\bs{c})\).

In the sparse regime, \(\rho_N\asymp N^{-1}\) and \(a_N=N^{9/2}\).
For \(r=1,\ldots,4\), let \(\mathcal A_{r,N}\) denote
\(\mathscr C_N(g_{\mathrm{dyad}})\),
\(\mathscr C_N(g_{\text{2-star}})\),
\(\mathscr C_N(g_{\mathrm{3tree}})\), and
\(\mathscr C_N(g_{\mathrm{4tree}})\), respectively. Define
\[
Q_{r,N}(\bs{c}):=\sum_{A\in\mathcal A_{r,N}}
b_{A,N}(\bs{c};\mathcal O_N)\prod_{e\in A}\xi_e,
\]
and set \(V_{r,N}(\bs{c}):=\Var(Q_{r,N}(\bs{c})\mid\mathcal O_N)\).
Then \(N^{9/2}\bs{c}^\top\bs\Pi_N^{\mathrm{S}}=\sum_{r=1}^{4}Q_{r,N}(\bs{c})\).
To apply Lemma~\ref{lem:dejong_multilinear_form}, take \(R=4\),
\(\mathcal F_N=\mathcal O_N\), \(\mathcal I_N=\mathcal E_N\),
\(Y_{e,N}=\xi_e\), and \(a_{A,N}=b_{A,N}(\bs{c};\mathcal O_N)\).

To verify \eqref{eq:mf-clt-max-condition}, it suffices to show that
\begin{equation}
\label{eq:regime_sparse_maximal_influence}
\max_{e\in\mathcal E_N}\sum_{r=1}^4
\sum_{A\in\mathcal A_{r,N}:\,e\in A}
\abs{b_{A,N}(\bs{c};\mathcal O_N)}^2
\E\!\left[\prod_{f\in A}\xi_f^2\mid\mathcal O_N\right]
=o_p(1).
\end{equation}
We now establish \eqref{eq:regime_sparse_maximal_influence}. We first bound
the contribution of each dyad set \(A\). Expanding the square of
\(\bs{c}^\top\bs C_{A,N}(\mathcal O_N)\) in
\eqref{eq:projection_coeff_exact} gives a sum over all ordered pairs of summands,
indexed by \((s,\nu)\) and \((t,\mu)\), where \(s,t\in\mathcal S_N\)
and \(\nu,\mu\in\{1,\ldots,\nu_{\max}\}\) satisfy
\(A\subseteq E_{s,\nu}\cap E_{t,\mu}\).
Then we have 
\[
\begin{aligned}
&\bigl\{\bs{c}^\top\bs C_{A,N}(\mathcal O_N)\bigr\}^2=\sum_{s,t\in\mathcal S_N}
\sum_{\substack{\nu,\mu\\A\subseteq E_{s,\nu}\cap E_{t,\mu}}}
\bigl\{\bs{c}^\top\bs c_{s,\nu}(\bs\beta_0,\mathcal O_N)\bigr\}
\bigl\{\bs{c}^\top\bs c_{t,\mu}(\bs\beta_0,\mathcal O_N)\bigr\}\\
&\qquad\qquad\times
\prod_{f\in E_{s,\nu}\setminus A}P_f(\bs\beta_0)
\prod_{f\in E_{t,\mu}\setminus A}P_f(\bs\beta_0).
\end{aligned}
\]
By conditional independence and
\eqref{eq:poly-coefficient-derivative-bound},
\[
\begin{aligned}
\E\!\left[\prod_{f\in A}\xi_f^2\mid\mathcal O_N\right]
&=\prod_{f\in A}P_f(\bs\beta_0)\{1-P_f(\bs\beta_0)\}
\le\prod_{f\in A}P_f(\bs\beta_0),\\
\left|\{\bs c^\top\bs c_{s,\nu}(\bs\beta_0,\mathcal O_N)\}
\{\bs c^\top\bs c_{t,\mu}(\bs\beta_0,\mathcal O_N)\}\right|
&\le \norm{\bs c}^2
\norm{\bs c_{s,\nu}(\bs\beta_0,\mathcal O_N)}
\norm{\bs c_{t,\mu}(\bs\beta_0,\mathcal O_N)}
\le C.
\end{aligned}
\]
Hence
\begin{equation}
\label{eq:regime_influence_expansion}
\begin{aligned}
&\abs{b_{A,N}(\bs{c};\mathcal O_N)}^2
\E\!\left[\prod_{f\in A}\xi_f^2\mid\mathcal O_N\right]\\
&\quad\le \left(\frac{a_N}{\abs{\mathcal S_N}}\right)^2
\sum_{s,t\in\mathcal S_N}
\sum_{\substack{\nu,\mu\\A\subseteq E_{s,\nu}\cap E_{t,\mu}}}
\left|\{\bs c^\top\bs c_{s,\nu}(\bs\beta_0,\mathcal O_N)\}
\{\bs c^\top\bs c_{t,\mu}(\bs\beta_0,\mathcal O_N)\}\right|\\
&\qquad\qquad\times
\E\!\left[\prod_{f\in A}\xi_f^2\mid\mathcal O_N\right]
\prod_{f\in E_{s,\nu}\setminus A}P_f(\bs\beta_0)
\prod_{f\in E_{t,\mu}\setminus A}P_f(\bs\beta_0)\\
&\quad\le C\left(\frac{a_N}{\abs{\mathcal S_N}}\right)^2
\sum_{\substack{s,t\in\mathcal S_N\\A\subseteq E(s)\cap E(t)}}
\sum_{\substack{\nu,\mu\\A\subseteq E_{s,\nu}\cap E_{t,\mu}}}
\prod_{f\in E_{s,\nu}\setminus A}P_f(\bs\beta_0)
\prod_{f\in E_{t,\mu}\setminus A}P_f(\bs\beta_0)
\prod_{f\in A}P_f(\bs\beta_0).
\end{aligned}
\end{equation}
For \(A\subseteq E_{s,\nu}\cap E_{t,\mu}\),
\[
\begin{aligned}
&\prod_{f\in E_{s,\nu}\setminus A}P_f(\bs\beta_0)
\prod_{f\in E_{t,\mu}\setminus A}P_f(\bs\beta_0)
\prod_{f\in A}P_f(\bs\beta_0)\\
&\quad=
\prod_{f\in(E_{s,\nu}\cap E_{t,\mu})\setminus A}P_f(\bs\beta_0)^2
\prod_{f\in E_{s,\nu}\setminus E_{t,\mu}}P_f(\bs\beta_0)
\prod_{f\in E_{t,\mu}\setminus E_{s,\nu}}P_f(\bs\beta_0)
\prod_{f\in A}P_f(\bs\beta_0).
\end{aligned}
\]
For each \(A\subseteq E_{s,\nu}\cap E_{t,\mu}\), let
\(p:=\abs{E_{s,\nu}\setminus A}+\abs{E_{t,\mu}\setminus A}+\abs A\)
be the total exponent of the probability product.
Let \(V:=V_s\cup V_t\) and \(E:=E_{s,\nu}\cup E_{t,\mu}\).
For each \((u,v)\in E\), let \(r_{uv}\) be the exponent
of \(P_{uv}(\bs\beta_0)\) in the final bound in
\eqref{eq:regime_influence_expansion}.
Then \(r_{uv}=2\) on \((E_{s,\nu}\cap E_{t,\mu})\setminus A\)
and \(r_{uv}=1\) on the remaining dyads in \(E\).
By Lemma~\ref{lem:det-row-supports}, the graphs
\((V_s,E_{s,\nu})\) and \((V_t,E_{t,\mu})\) are connected, each on five nodes.
Since \(A\ne\varnothing\) is contained in both edge sets, their union
\((V,E)\) is connected, so \(\abs E\ge\abs V-1\). Thus
\(p=\abs E+\abs{(E_{s,\nu}\cap E_{t,\mu})\setminus A}\ge\abs E\),
which gives \(\abs V\le p+1\). Also, \(\sum_{(u,v)\in E}r_{uv}=p\).
Each of the graphs \((V_s,E_{s,\nu})\) and \((V_t,E_{t,\mu})\)
has five nodes, so a node can be joined to at most four others within
each graph.
Hence the node degrees \(d_{E_{s,\nu}}(v)\) and \(d_{E_{t,\mu}}(v)\)
are at most four for every \(v\in V\). Since \(\kappa_v\) sums the probability exponents on
dyads incident to \(v\),
\(\kappa_v=d_{E_{s,\nu}}(v)+d_{E_{t,\mu}}(v)-d_A(v)\le4+4=8\),
where \(d_A(v)\) counts the dyads in \(A\) incident to \(v\).
Relabel the nodes of \((V,E)\) as \(\{1,\ldots,\abs V\}\), with
\(V_1=\{1,2\}\) representing the endpoints of \(e\). Applying
\eqref{eq:weighted_rooted_sum} gives
\[
\begin{aligned}
&\left(\frac{N^{9/2}}{\abs{\mathcal S_N}}\right)^2
\max_{i_1\ne i_2}
\sum_{\substack{i_3,\ldots,i_{\abs V}\in[N]\\i_1,\ldots,i_{\abs V}\ \mathrm{distinct}}}
\prod_{(u,v)\in E}P_{i_ui_v}(\bs\beta_0)^{r_{uv}}\\
&\quad=O(N^{-1})\,o_p(N^{\abs V}\rho_N^p)
=o_p(N^{\abs V-1}\rho_N^p)=o_p(1),
\end{aligned}
\]
where \(N^{\abs V-1}\rho_N^p\le(N\rho_N)^p=O(1)\).
Up to relabeling, the number of union graphs \((V,E)\) together with
assignments of the edge exponents \(r_{uv}\in\{1,2\}\) is bounded
independently of \(N\), since \(\abs V\le10\).
For each fixed graph and assignment of its edge exponents, any assignment
of distinct node labels admits at most a constant number of possible tuples
\((A,s,t,\nu,\mu)\).
Summing \eqref{eq:regime_influence_expansion} over
\(A\in\mathcal A_{r,N}\) containing \(e\), for \(r=1,\ldots,4\),
and taking the maximum over \(e\) therefore proves
\eqref{eq:regime_sparse_maximal_influence}. Since all
summands are nonnegative, \eqref{eq:regime_sparse_maximal_influence} implies
\eqref{eq:mf-clt-max-condition} for every \(r=1,\ldots,4\).

Lemma~\ref{lem:sparse_dyadset_fourth} gives
\(\E[Q_{r,N}(\bs{c})^4\mid\mathcal O_N]=3V_{r,N}(\bs{c})^2+o_p(1)\)
for \(r=1,\ldots,4\), verifying \eqref{eq:mf-clt-fourth-condition} in
Lemma~\ref{lem:dejong_multilinear_form}.
Lemma~\ref{lem:dyadic_hoeffding_orthogonality} and
Assumption~\ref{ass:regimeSigma}(ii) give
\[
\Var(N^{9/2}\bs{c}^\top\bs\Pi_N^{\mathrm S}\mid\mathcal O_N)
=\sum_{r=1}^4V_{r,N}(\bs{c})
\xrightarrow{P}\bs{c}^\top\bf\Sigma_{\mathrm S}\bs{c}.
\]
Lemma~\ref{lem:dejong_multilinear_form} therefore gives
\(N^{9/2}\bs{c}^\top\bs\Pi_N^{\mathrm S}
\xrightarrow{D}N(0,\bs{c}^\top\bf\Sigma_{\mathrm S}\bs{c})\).

In the ultra-sparse regime, \(N\rho_N\to0\) and
\(N^5\rho_N^4\to\infty\). Take \(a_N=N^{5/2}\rho_N^{-2}\).
The leading term is \(\bs\Pi_N^{\mathrm{US}}=\bs\Pi_N(g_{\mathrm{4tree}})\).
Let \(\mathcal A^U_{4,N}:=\mathscr C_N(g_{\mathrm{4tree}})\), and set
\[
Q^U_{4,N}(\bs{c}):=\sum_{A\in\mathcal A^U_{4,N}}
b_{A,N}(\bs{c};\mathcal O_N)\prod_{e\in A}\xi_e.
\]
Write \(V^U_{4,N}(\bs{c}):=\Var(Q^U_{4,N}(\bs{c})\mid\mathcal O_N)\).
Then \(Q^U_{4,N}(\bs{c})=N^{5/2}\rho_N^{-2}\bs{c}^\top\bs\Pi_N^{\mathrm{US}}\).
To apply Lemma~\ref{lem:dejong_multilinear_form}, take \(R=4\),
\(\mathcal F_N=\mathcal O_N\), \(\mathcal I_N=\mathcal E_N\),
\(Y_{e,N}=\xi_e\), \(\mathcal A_{r,N}=\varnothing\) for \(r<4\),
\(\mathcal A_{4,N}=\mathcal A^U_{4,N}\), and
\(a_{A,N}=b_{A,N}(\bs{c};\mathcal O_N)\).
To verify \eqref{eq:mf-clt-max-condition}, it suffices to show that
\begin{equation}
\label{eq:regime_ultrasparse_maximal_influence}
\max_{e\in\mathcal E_N}
\sum_{A\in\mathscr C_N(g_{\mathrm{4tree}}):\,e\in A}
\abs{b_{A,N}(\bs{c};\mathcal O_N)}^2
\E\!\left[\prod_{f\in A}\xi_f^2\mid\mathcal O_N\right]
=o_p(1).
\end{equation}
We adapt the argument for \eqref{eq:regime_sparse_maximal_influence}
to establish \eqref{eq:regime_ultrasparse_maximal_influence}.
In \eqref{eq:regime_influence_expansion}, we now take
\(a_N=N^{5/2}\rho_N^{-2}\) and restrict \(A\) to sets in
\(\mathcal A^U_{4,N}\) containing \(e\).
Here \(A\) is a four-edge tree, so \(\abs A=4\) and \(\abs{V(A)}=5\).
The inclusion \(A\subseteq E_{s,\nu}\cap E_{t,\mu}\) implies
\(V(A)\subseteq V_s\cap V_t\). Since \(V_s\) and \(V_t\) each
contain exactly five nodes, \(V_s=V_t=V(A)\), so their union
\(V=V_s\cup V_t\) satisfies \(\abs V=5\).
The probability product in \eqref{eq:regime_influence_expansion} has total exponent
\(p=\abs{E_{s,\nu}\setminus A}+\abs{E_{t,\mu}\setminus A}+4\ge4\).
The incidence bound \(\kappa_v\le8\) established for \(E_{s,\nu}\) and \(E_{t,\mu}\)
also applies here.
Taking \(V_1=\{1,2\}\) for the endpoints of \(e\),
\eqref{eq:weighted_rooted_sum} yields
\[
\begin{aligned}
&\left(\frac{N^{5/2}\rho_N^{-2}}{\abs{\mathcal S_N}}\right)^2
\max_{i_1\ne i_2}
\sum_{\substack{i_3,i_4,i_5\in[N]\\i_1,\ldots,i_5\ \mathrm{distinct}}}
\prod_{(u,v)\in E}P_{i_ui_v}(\bs\beta_0)^{r_{uv}}\\
&\quad=O(N^{-5}\rho_N^{-4})\,o_p(N^5\rho_N^p)
=o_p(\rho_N^{p-4})=o_p(1),
\end{aligned}
\]
where \(\rho_N^{p-4}\le1\). The argument in the sparse regime again gives
finitely many union graphs \((V,E)\) and assignments of their edge exponents
\(r_{uv}\), up to relabeling, and a uniform bound on the number of possible tuples
\((A,s,t,\nu,\mu)\) for each fixed graph, assignment of edge exponents,
and assignment of distinct node labels.
Summing \eqref{eq:regime_influence_expansion} over
\(A\in\mathcal A^U_{4,N}\) containing \(e\) and maximizing over \(e\)
therefore proves \eqref{eq:regime_ultrasparse_maximal_influence}, which verifies
\eqref{eq:mf-clt-max-condition} for \(r=4\); for \(r<4\), the sums in
\eqref{eq:mf-clt-max-condition} are zero because \(\mathcal A_{r,N}\) is empty.
Lemma~\ref{lem:sparse_dyadset_fourth} gives
\(\E[Q^U_{4,N}(\bs{c})^4\mid\mathcal O_N]=3V^U_{4,N}(\bs{c})^2+o_p(1)\),
verifying \eqref{eq:mf-clt-fourth-condition} for \(r=4\).
Assumption~\ref{ass:regimeSigma}(iii) gives
\(V^U_{4,N}(\bs{c})\xrightarrow{P}\bs{c}^\top\bf\Sigma_{\mathrm{US}}\bs{c}\).
Lemma~\ref{lem:dejong_multilinear_form} therefore gives
\(N^{5/2}\rho_N^{-2}\bs{c}^\top\bs\Pi_N^{\mathrm{US}}
\xrightarrow{D}N(0,\bs{c}^\top\bf\Sigma_{\mathrm{US}}\bs{c})\).
Both Lemmas~\ref{lem:conditional_triangular_array_clt} and
\ref{lem:dejong_multilinear_form} allow a zero limiting variance, so
the limits hold for every \(\bs c\in\mathbb R^q\).
The Cram\'er--Wold device gives the multivariate Gaussian limits for the
three leading projection sums. The remainder bounds in
\eqref{eq:regime_remainder_second_moments} and Slutsky's theorem prove
the three assertions of Theorem~\ref{thm:regimeCLT}.
\end{proof}

\subsection{Proof of Theorem~\ref{thm:gmm} and Corollary~\ref{cor:generic_gmm_inference}}

We first state an analytic lemma on uniform convergence of functions
and their first derivatives.

\begin{lem}
\label{lem:uniform-convergence-differentiation}
Let \(\mathcal K\subset\mathbb R^d\) be nonempty and compact, and let
\(\bs f:\mathcal K\to\mathbb R^q\) be deterministic and continuous,
with \(d\) and \(q\) fixed.

\textnormal{(i)} Suppose random maps
\(\bs f_N:\mathcal K\to\mathbb R^q\) satisfy
\(\bs f_N(\bs x)\xrightarrow{P}\bs f(\bs x)\) for every fixed
\(\bs x\in\mathcal K\). If
\(\norm{\bs f_N(\bs x)-\bs f_N(\bs y)}
\le L_N\norm{\bs x-\bs y}\) holds almost surely for all
\(\bs x,\bs y\in\mathcal K\), with nonnegative \(L_N=O_p(1)\),
then \(\sup_{\bs x\in\mathcal K}
\norm{\bs f_N(\bs x)-\bs f(\bs x)}\xrightarrow{P}0\).

\textnormal{(ii)} Suppose \(\mathcal K\) is convex
with nonempty interior. Let \(\bs f_N\) be deterministic twice continuously differentiable maps
on an open neighborhood of \(\mathcal K\), converging pointwise to
\(\bs f\) on \(\mathcal K\). If all first and second partial
derivatives of \(\bs f_N\) are uniformly bounded on \(\mathcal K\),
independently of \(N\), then \(\bs f\) is continuously differentiable
on \(\operatorname{int}(\mathcal K)\), the set of interior points of
\(\mathcal K\). For every compact
\(\mathcal K_0\subset\operatorname{int}(\mathcal K)\) and every
\(j=1,\ldots,d\), \(\sup_{\bs x\in\mathcal K_0}
\left\|\frac{\partial\bs f_N(\bs x)}{\partial x_j}
-\frac{\partial\bs f(\bs x)}{\partial x_j}\right\|\to 0\).
\end{lem}

\begin{proof}
For (i), fix \(\varepsilon,\eta>0\) and choose \(M>0\) such that
\(\limsup_N\Pr(L_N>M)\le\eta\). Uniform continuity of \(\bs f\)
on \(\mathcal K\) allows us to choose \(t>0\) with
\(tM\le\varepsilon/3\) and
\(\norm{\bs f(\bs x)-\bs f(\bs y)}\le\varepsilon/3\) whenever
\(\bs x,\bs y\in\mathcal K\) and \(\norm{\bs x-\bs y}\le t\).
Cover \(\mathcal K\) by finitely many balls of radius \(t\), with
centers \(\bs x_1,\ldots,\bs x_m\in\mathcal K\). The triangle
inequality gives
\[
\sup_{\bs x\in\mathcal K}\norm{\bs f_N(\bs x)-\bs f(\bs x)}
\le\max_{1\le\ell\le m}\norm{\bs f_N(\bs x_\ell)-\bs f(\bs x_\ell)}
+tL_N+\varepsilon/3.
\]
For fixed \(t\), pointwise convergence and a finite union bound make \(\max_{1\le\ell\le m}\norm{\bs f_N(\bs x_\ell)-\bs f(\bs x_\ell)}\) converge to zero in probability.
On \(\{L_N\le M\}\),   \(\sup_{\bs x\in\mathcal K}\norm{\bs f_N(\bs x)-\bs f(\bs x)}\) can exceed \(\varepsilon\) only
if \(\max_{1\le\ell\le m}\norm{\bs f_N(\bs x_\ell)-\bs f(\bs x_\ell)}\) exceeds \(\varepsilon/3\). Consequently,
\(\limsup_N\Pr(\sup_{\bs x\in\mathcal K}
\norm{\bs f_N(\bs x)-\bs f(\bs x)}>\varepsilon)\le\eta\).
Letting \(\eta\downarrow0\) proves (i).

For (ii), the first derivative bounds and convexity of \(\mathcal K\)
give \(\bs f_N\) a common Lipschitz constant. Part (i), applied to
deterministic maps, therefore gives
\(\sup_{\bs x\in\mathcal K}\norm{\bs f_N(\bs x)-\bs f(\bs x)}\to0\).
Fix a compact
\(\mathcal K_0\subset\operatorname{int}(\mathcal K)\), and let
\(\bs e_j\) be the \(j\)th coordinate vector. For all sufficiently
small \(t>0\), the segment from \(\bs x-t\bs e_j\) to
\(\bs x+t\bs e_j\) lies in \(\mathcal K\) for every
\(\bs x\in\mathcal K_0\). The integral remainder in the first-order
Taylor formula and the second derivative bound give
\[
\sup_{\bs x\in\mathcal K_0}
\left\|\frac{\bs f_N(\bs x+t\bs e_j)-\bs f_N(\bs x-t\bs e_j)}{2t}
-\frac{\partial\bs f_N(\bs x)}{\partial x_j}\right\|\le Ct,
\]
where \(C\) is independent of \(N\) and \(t\). The triangle
inequality gives, for any \(N,M\),
\[
\begin{aligned}
&\sup_{\bs x\in\mathcal K_0}
\left\|\frac{\partial\bs f_N(\bs x)}{\partial x_j}
-\frac{\partial\bs f_M(\bs x)}{\partial x_j}\right\|\\
&\quad\le
\sup_{\bs x\in\mathcal K_0}
\left\|
\frac{\bs f_N(\bs x+t\bs e_j)-\bs f_N(\bs x-t\bs e_j)}{2t}
-\frac{\bs f_M(\bs x+t\bs e_j)-\bs f_M(\bs x-t\bs e_j)}{2t}
\right\|+2Ct\\
&\quad\le t^{-1}\sup_{\bs x\in\mathcal K}
\norm{\bs f_N(\bs x)-\bs f_M(\bs x)}+2Ct.
\end{aligned}
\]
The triangle inequality and
\(\sup_{\bs x\in\mathcal K}\norm{\bs f_N(\bs x)-\bs f(\bs x)}\to0\) give
\(\sup_{\bs x\in\mathcal K}\norm{\bs f_N(\bs x)-\bs f_M(\bs x)}\to0\)
as \(N,M\to\infty\). Taking \(N,M\to\infty\) and then
\(t\downarrow0\) shows that
\(\partial\bs f_N/\partial x_j\) is uniformly Cauchy on
\(\mathcal K_0\), and hence converges uniformly there.
For any \(\bs x\in\operatorname{int}(\mathcal K)\), choose \(t>0\)
small enough that \(\bs x+s\bs e_j\in\operatorname{int}(\mathcal K)\)
for all \(\abs s\le t\). On \([-t,t]\), the maps
\(s\mapsto\bs f_N(\bs x+s\bs e_j)\) converge pointwise to
\(s\mapsto\bs f(\bs x+s\bs e_j)\), and their derivatives
\(\partial\bs f_N(\bs x+s\bs e_j)/\partial x_j\) converge uniformly,
since \(\mathcal K_0\) was arbitrary.
By \citet[Theorem~7.17]{rudin1976principles}, applied componentwise
on \([-t,t]\),
\(\partial\bs f(\bs x)/\partial x_j
=\lim_{N\to\infty}\partial\bs f_N(\bs x)/\partial x_j\).
Since each \(\partial\bs f_N/\partial x_j\) is continuous,
\citet[Theorem~7.12]{rudin1976principles}, applied componentwise on
closed balls contained in \(\operatorname{int}(\mathcal K)\), implies
that \(\partial\bs f/\partial x_j\) is continuous on
\(\operatorname{int}(\mathcal K)\). This proves (ii) for every
\(j=1,\ldots,d\).
\end{proof}

We now establish the uniform convergence of \(\rho_N^{-4}\bs h_N(\bs\beta)\) and  its Jacobian.

\begin{lem}
\label{lem:uniform-convergence-population}
Suppose Assumptions~\ref{ass:dyad_ind}, \ref{ass:bounded_design},
\ref{ass:sparse_envelope}, and \ref{ass:consistency}(i) hold.
Let \(\mathcal B\) be a compact subset of \(\mathbb R^d\),
with \(\bs\beta_0\) in its interior, and suppose
\(N^5\rho_N^4\to\infty\).
Then \(\bs h_0\) is continuously differentiable on a neighborhood of
\(\bs\beta_0\). Let
\(
\dot{\bs h}_0:=
\frac{\partial\bs h_0(\bs\beta_0)}{\partial\bs\beta^\top}
\)
denote the Jacobian of \(\bs h_0\) evaluated at \(\bs\beta_0\).
Then
\(\sup_{\bs\beta\in\mathcal B}\allowbreak\left\|
\rho_N^{-4}\bs h_N(\bs\beta)-\bs h_0(\bs\beta)
\right\|\xrightarrow{P}0\). Moreover, for any deterministic sequence
\(\delta_N\downarrow0\),
\(\sup_{\bs\beta\in\mathcal B:\norm{\bs\beta-\bs\beta_0}\le\delta_N}\allowbreak
\left\|\rho_N^{-4}\dot{\bs h}_N(\bs\beta)-\dot{\bs h}_0
\right\|\xrightarrow{P}0\).
\end{lem}

\begin{proof}
Write
\begin{equation}
\label{eq:moment-population-decomposition}
\rho_N^{-4}\bs h_N(\bs\beta)-\bs h_0(\bs\beta)
=
\rho_N^{-4}\{\bs h_N(\bs\beta)-\E[\bs h_N(\bs\beta)]\}
+\rho_N^{-4}\E[\bs h_N(\bs\beta)]-\bs h_0(\bs\beta).
\end{equation}
We first control
\(\rho_N^{-4}\{\bs h_N(\bs\beta)-\E[\bs h_N(\bs\beta)]\}\)
and its first derivatives with respect to \(\bs\beta\), uniformly on
\(\mathcal B\). We then prove uniform convergence of
\(\rho_N^{-4}\E[\bs h_N]\) to \(\bs h_0\). Applying
Lemma~\ref{lem:uniform-convergence-differentiation}(ii) to
\(\rho_N^{-4}\E[\bs h_N]\) will show that \(\bs h_0\) is continuously
differentiable near \(\bs\beta_0\) and that
\(\rho_N^{-4}\E[\dot{\bs h}_N(\bs\beta_0)]\to\dot{\bs h}_0\).
We then show that
\(\rho_N^{-4}\{\dot{\bs h}_N(\bs\beta)
-\E[\dot{\bs h}_N(\bs\beta)]\}\xrightarrow{P}0\) and
\(\rho_N^{-4}\E[\dot{\bs h}_N(\bs\beta)]-\dot{\bs h}_0\to0\),
both uniformly over \(\bs\beta\in\mathcal B\) with
\(\norm{\bs\beta-\bs\beta_0}\le\delta_N\).

Choose a compact rectangle \(\mathcal R\) whose interior contains
\(\mathcal B\). The polynomial expansion in Lemma~\ref{lem:poly},
applied on \(\mathcal R\), gives, for \(\abs\alpha\le2\),
\[
\partial_{\bs\beta}^{\alpha}\bs h_N(\bs\beta)=\frac{1}{\abs{\mathcal S_N}}
\sum_{s\in\mathcal S_N} \partial_{\bs\beta}^{\alpha}\bs\Psi_s(\bs\beta)
=\frac{1}{\abs{\mathcal S_N}}
\sum_{s\in\mathcal S_N}\sum_{\nu=1}^{\nu_{\max}}
\partial_{\bs\beta}^{\alpha}
\bs c_{s,\nu}(\bs\beta,\mathcal O_N)
\prod_{e\in E_{s,\nu}}L_e.
\]
Differentiation acts only on the coefficients. Their derivatives are
uniformly bounded by \eqref{eq:poly-coefficient-derivative-bound}, and
the supports remain connected on the five nodes of \(s\), with
\(4\le\abs{E_{s,\nu}}\le7\), by Lemma~\ref{lem:det-row-supports}.
To bound the expectation of each product, first use iterated
expectations and conditional independence of the dyads given
\(\mathcal O_N=\sigma(\bs{U}_1,\ldots,\bs{U}_N)\):
\[
\E\!\left[\prod_{e\in E_{s,\nu}}L_e\right]
=\E\!\left[\E\!\left[\prod_{e\in E_{s,\nu}}L_e
\,\middle|\,\mathcal O_N\right]\right]
=\E\!\left[\prod_{e\in E_{s,\nu}}P_e(\bs\beta_0)\right]
\le C\rho_N^{\abs{E_{s,\nu}}}\le C\rho_N^4.
\]
Here \eqref{eq:weighted_fixed_label_bound} applies to
\((V_s,E_{s,\nu})\), with the five distinct labels in \(s\), exponent
\(r_{uv}=1\) on each dyad, and incidence multiplicity at most four.
The bound on
\(\norm{\partial_{\bs\beta}^{\alpha}\bs c_{s,\nu}(\bs\beta,\mathcal O_N)}\)
in \eqref{eq:poly-coefficient-derivative-bound} and the finite sum over
\(\nu\) therefore give
\begin{equation}
\label{eq:moment-derivative-expectation-bound}
\sup_{s\in\mathcal S_N}
\E\!\left[\sup_{\bs\beta\in\mathcal R}
\norm{\partial_{\bs\beta}^{\alpha}\bs\Psi_s(\bs\beta)}\right]
\le C\rho_N^4,
\qquad
\E\!\left[\sup_{\bs\beta\in\mathcal R}
\norm{\partial_{\bs\beta}^{\alpha}\bs h_N(\bs\beta)}\right]
\le C\rho_N^4,
\quad \abs\alpha\le2.
\end{equation}
For each \(N\), smoothness of the coefficients and the integrable
derivative envelopes in \eqref{eq:moment-derivative-expectation-bound}
make \(\E[\bs h_N]\) twice continuously differentiable on
\(\operatorname{int}(\mathcal R)\), with
\(\partial_{\bs\beta}^{\alpha}\E[\bs h_N]
=\E[\partial_{\bs\beta}^{\alpha}\bs h_N]\)
for \(\abs\alpha\le2\). The bound
\(\norm{\E[\partial_{\bs\beta}^{\alpha}\bs h_N]}
\le\E[\norm{\partial_{\bs\beta}^{\alpha}\bs h_N}]\)
and Markov's inequality applied to
\eqref{eq:moment-derivative-expectation-bound} yield,
with \(C\) independent of \(N\),
\begin{equation}
\label{eq:cov_equicont_moment_derivative_bound}
\sup_{\bs\beta\in\mathcal R}
\rho_N^{-4}\norm{\E[\partial_{\bs\beta}^{\alpha}\bs h_N(\bs\beta)]}\le C,
\qquad
\sup_{\bs\beta\in\mathcal R}
\norm{\partial_{\bs\beta}^{\alpha}\bs h_N(\bs\beta)}
=O_p(\rho_N^4),
\quad \abs\alpha\le2.
\end{equation}

We now bound the second moment in
\eqref{eq:moment-centered-covariance-expansion}.
For fixed \(\bs\beta\in\mathcal R\) and \(\abs\alpha\le1\), expanding
the squared norm gives
\begin{equation}
\label{eq:moment-centered-covariance-expansion}
\begin{aligned}
&\E\!\left[\left\|
\rho_N^{-4}\{\partial_{\bs\beta}^{\alpha}\bs h_N(\bs\beta)
-\E[\partial_{\bs\beta}^{\alpha}\bs h_N(\bs\beta)]\}
\right\|^2\right]=\frac{\rho_N^{-8}}{\abs{\mathcal S_N}^2}
\sum_{s,t\in\mathcal S_N}
\operatorname{tr}\operatorname{Cov}\!\left(
\partial_{\bs\beta}^{\alpha}\bs\Psi_s(\bs\beta),
\partial_{\bs\beta}^{\alpha}\bs\Psi_t(\bs\beta)\right).
\end{aligned}
\end{equation}
Each coefficient of \(\partial_{\bs\beta}^{\alpha}\bs\Psi_s\) is
measurable in \(\bf X_s\), and the conditional law of the dyads in
\(E(s)\) depends only on \(\{\bs{U}_i:i\in V_s\}\). Thus the i.i.d. node
sampling and conditional dyadic independence in
Assumption~\ref{ass:dyad_ind} make
\(\partial_{\bs\beta}^{\alpha}\bs\Psi_s(\bs\beta)\) and
\(\partial_{\bs\beta}^{\alpha}\bs\Psi_t(\bs\beta)\) independent when
\(V_s\cap V_t=\varnothing\); their covariance in
\eqref{eq:moment-centered-covariance-expansion} is zero.

For \(v:=\abs{V_s\cap V_t}\ge1\), fix any
\(\nu,\mu\in\{1,\ldots,\nu_{\max}\}\), indexing the monomials in the
polynomial expansions of \(\bs\Psi_s(\bs\beta)\) and
\(\bs\Psi_t(\bs\beta)\), respectively. By
Lemma~\ref{lem:det-row-supports}, the corresponding dyad supports
\(E_{s,\nu}\) and \(E_{t,\mu}\) are connected on \(V_s\) and \(V_t\),
respectively. Since \(V_s\) and \(V_t\) share at least one node, the
union \(E_{s,\nu}\cup E_{t,\mu}\) is connected on \(10-v\) nodes and
therefore contains at least \(9-v\) distinct dyads. Using \(L_e^2=L_e\), iterated
expectations, and conditional dyadic independence gives
\[
\E\!\left[
\prod_{e\in E_{s,\nu}}L_e\prod_{f\in E_{t,\mu}}L_f\right]
=\E\!\left[\prod_{e\in E_{s,\nu}\cup E_{t,\mu}}
P_e(\bs\beta_0)\right]
\le C\rho_N^{\abs{E_{s,\nu}\cup E_{t,\mu}}}
\le C\rho_N^{9-v},
\]
where we apply \eqref{eq:weighted_fixed_label_bound} to
\((V_s\cup V_t,E_{s,\nu}\cup E_{t,\mu})\). Each distinct dyad contributes
one probability factor, so \(r_{uv}=1\). Each node is incident to at most
four dyads from each support, so its incidence multiplicity is at most
\(4+4=8\). By Lemma~\ref{lem:poly},
\[
\begin{aligned}
\partial_{\bs\beta}^{\alpha}\bs\Psi_s(\bs\beta)
&=\sum_{\nu=1}^{\nu_{\max}}
\partial_{\bs\beta}^{\alpha}\bs c_{s,\nu}(\bs\beta,\mathcal O_N)
\prod_{e\in E_{s,\nu}}L_e, \ \ 
\partial_{\bs\beta}^{\alpha}\bs\Psi_t(\bs\beta)
&=\sum_{\mu=1}^{\nu_{\max}}
\partial_{\bs\beta}^{\alpha}\bs c_{t,\mu}(\bs\beta,\mathcal O_N)
\prod_{f\in E_{t,\mu}}L_f.
\end{aligned}
\]
The triangle inequality, \eqref{eq:poly-coefficient-derivative-bound},
and \(L_e^2=L_e\) therefore give
\[
\begin{aligned}
&\E\!\left[\norm{\partial_{\bs\beta}^{\alpha}\bs\Psi_s(\bs\beta)}
\norm{\partial_{\bs\beta}^{\alpha}\bs\Psi_t(\bs\beta)}\right]\le C\sum_{\nu,\mu=1}^{\nu_{\max}}
\E\!\left[\prod_{e\in E_{s,\nu}\cup E_{t,\mu}}L_e\right]
\le C\rho_N^{9-v},
\end{aligned}
\]
where the final inequality uses \eqref{eq:weighted_fixed_label_bound},
\(\abs{E_{s,\nu}\cup E_{t,\mu}}\ge9-v\), and the fact that
\(\nu_{\max}\) does not depend on \(N\).
Moreover, \eqref{eq:moment-derivative-expectation-bound} gives
\(\E[\norm{\partial_{\bs\beta}^{\alpha}\bs\Psi_s(\bs\beta)}]
\E[\norm{\partial_{\bs\beta}^{\alpha}\bs\Psi_t(\bs\beta)}]
\le C\rho_N^8\). Consequently,
\[
\left|\operatorname{tr}\operatorname{Cov}\!\left(
\partial_{\bs\beta}^{\alpha}\bs\Psi_s(\bs\beta),
\partial_{\bs\beta}^{\alpha}\bs\Psi_t(\bs\beta)\right)\right|
\le C(\rho_N^{9-v}+\rho_N^8)
\le C\rho_N^{9-v}.
\]
For each ordered pentad \(s\), choose the \(v\) shared nodes from
\(V_s\), the other \(5-v\) nodes from its complement, and then order
the five nodes of \(t\). Hence
\[
\#\{(s,t)\in\mathcal S_N^2:\abs{V_s\cap V_t}=v\}
=\abs{\mathcal S_N}\binom5v\binom{N-5}{5-v}5!
\le C N^{10-v}.
\]
Substituting the covariance bound and the pair count into
\eqref{eq:moment-centered-covariance-expansion}, with
\(\abs{\mathcal S_N}\asymp N^5\), proves
\begin{equation}
\label{eq:moment_centered_second_moment}
\begin{aligned}
&\sup_{\bs\beta\in\mathcal R}
\E\!\left[\left\|
\rho_N^{-4}\{\partial_{\bs\beta}^{\alpha}\bs h_N(\bs\beta)
-\E[\partial_{\bs\beta}^{\alpha}\bs h_N(\bs\beta)]\}
\right\|^2\right]\\
&\qquad\le C\rho_N^{-8}N^{-10}
\sum_{v=1}^5N^{10-v}\rho_N^{9-v}
=C\sum_{v=1}^5N^{-v}\rho_N^{1-v},
\qquad \abs\alpha\le1.
\end{aligned}
\end{equation}
The \(v=1\) summand is \(N^{-1}\). For \(2\le v\le5\),
\(N^{-v}\rho_N^{1-v}
=N^{(v-5)/4}(N^5\rho_N^4)^{-(v-1)/4}=o(1)\)
under  \(N^5\rho_N^4\to \infty\). Chebyshev's inequality thus gives
\(\rho_N^{-4}\{\partial_{\bs\beta}^{\alpha}\bs h_N(\bs\beta)
-\E[\partial_{\bs\beta}^{\alpha}\bs h_N(\bs\beta)]\}
\xrightarrow{P}0\) for every fixed \(\bs\beta\in\mathcal R\) and
\(\abs\alpha\le1\).

By the triangle inequality, the cases \(\abs\alpha=1\) and
\(\abs\alpha=2\) of \eqref{eq:cov_equicont_moment_derivative_bound}
bound all first and second partial derivatives of
\(\rho_N^{-4}\{\bs h_N-\E[\bs h_N]\}\) by \(O_p(1)\), uniformly
on \(\mathcal R\). Since \(\mathcal R\) is convex and \(d,q\) are
fixed, applying the mean-value theorem componentwise along line
segments in \(\mathcal R\) gives Lipschitz constants of order
\(O_p(1)\) on \(\mathcal B\) for
\(\rho_N^{-4}\{\bs h_N-\E[\bs h_N]\}\) and for each column of
\(\rho_N^{-4}\{\dot{\bs h}_N-\E[\dot{\bs h}_N]\}\).
Using \eqref{eq:moment_centered_second_moment} and Chebyshev's
inequality for pointwise convergence, we apply
Lemma~\ref{lem:uniform-convergence-differentiation}(i) to
\(\rho_N^{-4}\{\bs h_N-\E[\bs h_N]\}\) and to each column of
\(\rho_N^{-4}\{\dot{\bs h}_N-\E[\dot{\bs h}_N]\}\).
Since \(d\) is fixed, summing the column norms gives
\begin{equation}
\label{eq:moment-centered-uniform-convergence}
\begin{aligned}
\sup_{\bs\beta\in\mathcal B}
\left\|\rho_N^{-4}
\{\bs h_N(\bs\beta)-\E[\bs h_N(\bs\beta)]\}\right\|
&\xrightarrow{P}0, \ \ 
\sup_{\bs\beta\in\mathcal B}
\left\|\rho_N^{-4}
\{\dot{\bs h}_N(\bs\beta)-\E[\dot{\bs h}_N(\bs\beta)]\}\right\|
&\xrightarrow{P}0.
\end{aligned}
\end{equation}

We next control
\(\rho_N^{-4}\E[\bs h_N(\bs\beta)]-\bs h_0(\bs\beta)\)
in \eqref{eq:moment-population-decomposition}.
Assumption~\ref{ass:consistency}(i) gives pointwise convergence of
\(\rho_N^{-4}\E[\bs h_N]\) to the continuous map \(\bs h_0\) on
\(\mathcal B\). Applying the mean-value theorem componentwise to
\(\rho_N^{-4}\E[\bs h_N]\) along line segments in \(\mathcal R\)
gives a Lipschitz constant independent of \(N\), by the cases
\(\abs\alpha=1\) of \eqref{eq:cov_equicont_moment_derivative_bound}.
Lemma~\ref{lem:uniform-convergence-differentiation}(i),
applied to \(\rho_N^{-4}\E[\bs h_N]\) on \(\mathcal B\), therefore gives
\begin{equation}
\label{eq:population-uniform-convergence}
\sup_{\bs\beta\in\mathcal B}
\left\|\rho_N^{-4}\E[\bs h_N(\bs\beta)]-\bs h_0(\bs\beta)\right\|
\to0.
\end{equation}
Substituting \eqref{eq:moment-centered-uniform-convergence} and
\eqref{eq:population-uniform-convergence} into
\eqref{eq:moment-population-decomposition} gives
\(\sup_{\bs\beta\in\mathcal B}
\norm{\rho_N^{-4}\bs h_N(\bs\beta)-\bs h_0(\bs\beta)}\xrightarrow{P}0\).

We now prove that \(\bs h_0\) is continuously differentiable near
\(\bs\beta_0\) and that
\(\rho_N^{-4}\E[\dot{\bs h}_N(\bs\beta_0)]\to\dot{\bs h}_0\).
Choose a closed ball
\(\mathcal K\) centered at \(\bs\beta_0\), with positive radius and
contained in \(\operatorname{int}(\mathcal B)\). The maps
\(\rho_N^{-4}\E[\bs h_N]\) are twice continuously differentiable on a neighborhood of
\(\mathcal K\) by \eqref{eq:moment-derivative-expectation-bound}.
Equation~\eqref{eq:cov_equicont_moment_derivative_bound} bounds the
first and second partial derivatives of \(\rho_N^{-4}\E[\bs h_N]\)
uniformly on \(\mathcal K\), and Assumption~\ref{ass:consistency}(i)
gives \(\rho_N^{-4}\E[\bs h_N(\bs\beta)]\to\bs h_0(\bs\beta)\)
for every \(\bs\beta\in\mathcal K\). Applying
Lemma~\ref{lem:uniform-convergence-differentiation}(ii) to
\(\rho_N^{-4}\E[\bs h_N]\) on \(\mathcal K\) shows that
\(\bs h_0\) is continuously differentiable on the interior
\(\operatorname{int}(\mathcal K)\) of \(\mathcal K\), and in particular that
\(\rho_N^{-4}\E[\dot{\bs h}_N(\bs\beta_0)]\to\dot{\bs h}_0\).

Finally, for \(\bs\beta\in\mathcal B\) with
\(\norm{\bs\beta-\bs\beta_0}\le\delta_N\), write
\[
\begin{aligned}
\rho_N^{-4}\dot{\bs h}_N(\bs\beta)-\dot{\bs h}_0
&=\rho_N^{-4}\{\dot{\bs h}_N(\bs\beta)
-\E[\dot{\bs h}_N(\bs\beta)]\}\\
&\quad+\rho_N^{-4}\{\E[\dot{\bs h}_N(\bs\beta)]
-\E[\dot{\bs h}_N(\bs\beta_0)]\}+\rho_N^{-4}\E[\dot{\bs h}_N(\bs\beta_0)]-\dot{\bs h}_0.
\end{aligned}
\]
The line segment from \(\bs\beta_0\) to \(\bs\beta\) lies in
\(\mathcal R\). Applying the mean-value theorem componentwise to
\(\rho_N^{-4}\E[\dot{\bs h}_N]\) along this segment, with the bounds
for \(\abs\alpha=2\) in \eqref{eq:cov_equicont_moment_derivative_bound},
gives \(\norm{\rho_N^{-4}\{\E[\dot{\bs h}_N(\bs\beta)]
-\E[\dot{\bs h}_N(\bs\beta_0)]\}}
\le C\norm{\bs\beta-\bs\beta_0}\). Hence
\[
\begin{aligned}
&\sup_{\bs\beta\in\mathcal B:\norm{\bs\beta-\bs\beta_0}\le\delta_N}
\norm{\rho_N^{-4}\dot{\bs h}_N(\bs\beta)-\dot{\bs h}_0}\le\sup_{\bs\beta\in\mathcal B}
\left\|\rho_N^{-4}\{\dot{\bs h}_N(\bs\beta)
-\E[\dot{\bs h}_N(\bs\beta)]\}\right\|+C\delta_N\\
&\qquad\qquad\qquad\qquad\qquad\qquad\qquad\qquad+\norm{\rho_N^{-4}\E[\dot{\bs h}_N(\bs\beta_0)]-\dot{\bs h}_0}.
\end{aligned}
\]
Equation~\eqref{eq:moment-centered-uniform-convergence} gives
\(\rho_N^{-4}\{\dot{\bs h}_N(\bs\beta)
-\E[\dot{\bs h}_N(\bs\beta)]\}\xrightarrow{P}0\),
uniformly over \(\bs\beta\in\mathcal B\).
Also, \(C\delta_N\to0\) and
\(\norm{\rho_N^{-4}\E[\dot{\bs h}_N(\bs\beta_0)]-\dot{\bs h}_0}\to0\).
\end{proof}

Now we can prove Theorem~\ref{thm:gmm}.

\begin{proof}[Proof of Theorem~\ref{thm:gmm}]
The condition \(N^5\rho_N^4\to\infty\) in Theorem~\ref{thm:gmm}, together
with Assumptions~\ref{ass:dyad_ind}, \ref{ass:bounded_design}, and
\ref{ass:sparse_envelope}, allows us to apply
Lemma~\ref{lem:uniform-convergence-population}. Together with
Assumption~\ref{ass:consistency}(i), it gives
\(\sup_{\bs\beta\in\mathcal B}
\norm{\rho_N^{-4}\bs h_N(\bs\beta)-\bs h_0(\bs\beta)}
\xrightarrow{P}0\)
and, for every deterministic sequence \(\delta_N\downarrow0\),
\(\sup_{\bs\beta\in\mathcal B:\norm{\bs\beta-\bs\beta_0}\le\delta_N}
\norm{\rho_N^{-4}\dot{\bs h}_N(\bs\beta)-\dot{\bs h}_0}
\xrightarrow{P}0\).
Combining \(\sup_{\bs\beta\in\mathcal B}
\norm{\rho_N^{-4}\bs h_N(\bs\beta)-\bs h_0(\bs\beta)}\xrightarrow{P}0\)
with Assumption~\ref{ass:consistency}(iii) gives
\[
\sup_{\bs\beta\in\mathcal B}
\left|\rho_N^{-8}\bs h_N(\bs\beta)^\top\widehat{\bf W}_N\bs h_N(\bs\beta)
-\bs h_0(\bs\beta)^\top\bf W\bs h_0(\bs\beta)\right|
\xrightarrow{P}0.
\]
Multiplication by \(\rho_N^{-8}>0\) does not change the minimizers in
\eqref{eq:beta_gmm}. By Assumption~\ref{ass:consistency}(i) and positive
definiteness of \(\bf W\), the continuous limiting criterion
\(\bs h_0(\bs\beta)^\top\bf W\bs h_0(\bs\beta)\) has its unique minimum
at \(\bs\beta_0\). Compactness of \(\mathcal B\) and the argmin theorem
of \citet[Theorem~5.7, p.~45]{vandervaart1998asymptotic} therefore give
\(\widehat{\bs\beta}_{\mathrm{GMM}}\xrightarrow{P}\bs\beta_0\). Since
\(\bs{\beta}_{0}\in\operatorname{int}(\mathcal B)\), consistency implies that
\(\widehat{\bs{\beta}}_{\mathrm{GMM}}\) is an interior point of \(\mathcal B\) with probability
approaching one. On that event the first-order condition is
\begin{equation}
\label{eq:gmm_first_order_condition}
\dot{\bs h}_N(\widehat{\bs{\beta}}_{\mathrm{GMM}})^\top
\widehat{\bf W}_N\bs h_N(\widehat{\bs{\beta}}_{\mathrm{GMM}})=0.
\end{equation}
For the mean-value expansion, define
\(\bar{\mathbf{D}}_N:=\int_0^1
\dot{\bs h}_N\!\left(\bs\beta_0
+t(\widehat{\bs{\beta}}_{\mathrm{GMM}}-\bs\beta_0)\right)dt\).
Then
\begin{equation}
\label{eq:gmm_integral_expansion}
\bs h_N(\widehat{\bs{\beta}}_{\mathrm{GMM}})
=
\bs h_N(\bs{\beta}_{0})+\bar{\mathbf{D}}_N(\widehat{\bs{\beta}}_{\mathrm{GMM}}-\bs{\beta}_{0}).
\end{equation}
Lemma~\ref{lem:uniform-convergence-population} gives
\(\sup_{\bs\beta\in\mathcal B:\norm{\bs\beta-\bs\beta_0}\le\delta_N}\allowbreak
\left\|\rho_N^{-4}\dot{\bs h}_N(\bs\beta)-\dot{\bs h}_0\right\|
\xrightarrow{P}0\) for every deterministic sequence \(\delta_N\downarrow0\).
Since \(\widehat{\bs\beta}_{\mathrm{GMM}}\xrightarrow{P}\bs\beta_0\),
we obtain
\begin{equation}
\label{eq:gmm_jacobian_limits}
\rho_N^{-4}\dot{\bs h}_N(\widehat{\bs{\beta}}_{\mathrm{GMM}})
\xrightarrow{P}\dot{\bs h}_0,
\qquad
\rho_N^{-4}\bar{\mathbf{D}}_N\xrightarrow{P}\dot{\bs h}_0.
\end{equation}
Equation~\eqref{eq:gmm_jacobian_limits} and
Assumption~\ref{ass:consistency}(ii)--(iii) give
\(\rho_N^{-4}\dot{\bs h}_N(\widehat{\bs\beta}_{\mathrm{GMM}})^\top
\widehat{\bf W}_N\rho_N^{-4}\bar{\mathbf{D}}_N
\xrightarrow{P}\dot{\bs h}_0^\top\bf W\dot{\bs h}_0\).
Since \(\dot{\bs h}_0^\top\bf W\dot{\bs h}_0\) is positive definite,
\(\rho_N^{-4}\dot{\bs h}_N(\widehat{\bs\beta}_{\mathrm{GMM}})^\top
\widehat{\bf W}_N\rho_N^{-4}\bar{\mathbf{D}}_N\) is nonsingular with probability approaching one.
Combining \eqref{eq:gmm_integral_expansion}
with \eqref{eq:gmm_first_order_condition} gives, with probability approaching one,
\begin{equation}
\label{eq:gmm_linearization}
a_N\rho_N^4(\widehat{\bs{\beta}}_{\mathrm{GMM}}-\bs{\beta}_{0})
 =
-\Big[\rho_N^{-4}\dot{\bs h}_N(\widehat{\bs{\beta}}_{\mathrm{GMM}})^\top
\widehat{\bf W}_N\rho_N^{-4}\bar{\mathbf{D}}_N\Big]^{-1}
\Big[\rho_N^{-4}\dot{\bs h}_N(\widehat{\bs{\beta}}_{\mathrm{GMM}})^\top\widehat{\bf W}_N\Big]
\big[a_N \bs h_N(\bs{\beta}_{0})\big].
\end{equation}
With \(a_N\) and \(\bf\Sigma\) corresponding to the regime in
Theorem~\ref{thm:gmm}, Theorem~\ref{thm:regimeCLT} gives
\(a_N\bs h_N(\bs\beta_0)\allowbreak\xrightarrow{D}N(0,\bf\Sigma)\), and in
particular \(a_N \bs h_N(\bs{\beta}_{0})=O_p(1)\).

Equation~\eqref{eq:gmm_jacobian_limits} and
\(\widehat{\bf W}_N\xrightarrow{P}\bf W\), together with
Assumption~\ref{ass:consistency}(ii)--(iii), imply that the inverse
matrix in \eqref{eq:gmm_linearization} converges in probability to
\((\dot{\bs h}_0^\top\bf W\dot{\bs h}_0)^{-1}\) and is \(O_p(1)\).
Substituting into \eqref{eq:gmm_linearization}, using
\eqref{eq:gmm_jacobian_limits} and
\(a_N\bs h_N(\bs\beta_0)=O_p(1)\), gives the rate
\begin{equation}
\label{eq:gmm_rate_for_cov_equicont}
a_N\rho_N^4
\norm{\widehat{\bs{\beta}}_{\mathrm{GMM}}-\bs\beta_0}
=O_p(1).
\end{equation}
Finally, applying Slutsky's theorem to \eqref{eq:gmm_linearization},
using Theorem~\ref{thm:regimeCLT}, \eqref{eq:gmm_jacobian_limits}, and
\(\widehat{\bf W}_N\xrightarrow{P}\bf W\), gives
\[
a_N\rho_N^4(\widehat{\bs{\beta}}_{\mathrm{GMM}}-\bs\beta_0)
\xrightarrow{D}
-\big(\dot{\bs h}_0^\top\bf W\dot{\bs h}_0\big)^{-1}
\dot{\bs h}_0^\top\bf W \bs{Z},
\qquad
\bs{Z}\sim N(0,\bf\Sigma),
\]
with covariance matrix \(\mathbf{V}_0\).
\end{proof}

\begin{proof}[Proof of Corollary~\ref{cor:generic_gmm_inference}]
Equation~\eqref{eq:gmm_jacobian_limits} gives
\(
\rho_N^{-4}\widehat{\dot{\bs h}}_N
\xrightarrow{P}\dot{\bs h}_0\).
Under Assumption~\ref{ass:consistency}(ii)--(iii),
\(\rho_N^{-4}\widehat{\dot{\bs h}}_N\xrightarrow{P}\dot{\bs h}_0\)
and \(a_N^2\widehat{\bf\Omega}_N\xrightarrow{P}\bf\Sigma\)
imply \(a_N^2\rho_N^8\widehat{\mathbf{V}}_N\xrightarrow{P}\mathbf{V}_0\) by the definition of
\(\widehat{\mathbf{V}}_N\) in Corollary~\ref{cor:generic_gmm_inference}
and the continuous mapping theorem.
Theorem~\ref{thm:gmm} and Slutsky's theorem give
\(\bs c^\top(\widehat{\bs\beta}_{\mathrm{GMM}}-\bs\beta_0)
/\sqrt{\bs c^\top\widehat{\mathbf{V}}_N\bs c}
\xrightarrow{D}N(0,1)\).
\end{proof}

\subsection{\texorpdfstring
  {Asymptotic equivalence of \(\widehat{\bs\beta}^{\dagger}\)
   and \(\widehat{\bs\beta}_{\mathrm{GMM}}\)}
  {Asymptotic equivalence of the pair-normalized and GMM estimators}}

In this section, we show that \(\widehat{\bs\beta}^{\dagger}\) and \(\widehat{\bs\beta}_{\mathrm{GMM}}\) are asymptotically equivalent.

\begin{lem}
\label{lem:pair_normalization}
Suppose the conditions of Theorem~\ref{thm:gmm} hold. If
\(\widehat{\bs\beta}^{\dagger}\in\arg\min_{\bs\beta\in\mathcal B}
Q_N^\dagger(\bs\beta)\), then, under each of the three regimes in
Theorem~\ref{thm:gmm},
\(\widehat{\bs\beta}^{\dagger}\xrightarrow{P}\bs\beta_0\) and
\(a_N\rho_N^4
(\widehat{\bs\beta}^{\dagger}-\widehat{\bs\beta}_{\mathrm{GMM}})
\xrightarrow{P}0\).
\end{lem}

\begin{proof}
Assumption~\ref{ass:consistency}(iii) gives
\(\norm{\widehat{\bf W}_N}=O_p(1)\) and
\(\lambda_{\min}(\widehat{\bf W}_N)^{-1}=O_p(1)\).
Write
\[
D_N(\bs\beta):=\binom{N}{2}^{-1}\sum_{i<j}
\bs m_{ij,N}(\bs\beta)^\top\widehat{\bf W}_N
\bs m_{ij,N}(\bs\beta),
\qquad
c_N:=\rho_N^7+N^{-1}\rho_N^6+N^{-2}\rho_N^5+N^{-3}\rho_N^4.
\]
For multi-indices \(\alpha\) and
\(\gamma=(\gamma_1,\ldots,\gamma_d)\in\mathbb N_0^d\),
\(\gamma\le\alpha\) means \(0\le\gamma_j\le\alpha_j\) for every
\(j=1,\ldots,d\), and
\(\binom{\alpha}{\gamma}:=\prod_{j=1}^d\binom{\alpha_j}{\gamma_j}\).
For \(\abs{\alpha}\le2\), the product rule and
\eqref{eq:pivotal_pair_moment} give
\[
\begin{aligned}
\partial_{\bs\beta}^{\alpha}D_N(\bs\beta)
&=\frac{1}{\binom{N}{2}\{(N-2)(N-3)(N-4)\}^2}
\sum_{i<j}
\sum_{k,l,m\in[N]\setminus\{i,j\}:\,k,l,m\ \mathrm{distinct}}\\
&\quad{}\times
\sum_{k',l',m'\in[N]\setminus\{i,j\}:\,k',l',m'\ \mathrm{distinct}}
\sum_{\gamma\le\alpha}\binom{\alpha}{\gamma}\\[-1mm]
&\qquad\times
\{\partial_{\bs\beta}^{\gamma}
\bs\Psi_{(i,j,k,l,m)}(\bs\beta)\}^{\top}
\widehat{\bf W}_N
\{\partial_{\bs\beta}^{\alpha-\gamma}
\bs\Psi_{(i,j,k',l',m')}(\bs\beta)\}.
\end{aligned}
\]
The multi-indices \(\gamma\) and \(\alpha-\gamma\) specify the
coordinate derivatives applied to the first and second kernel factors,
respectively; the sum includes both \(\gamma=0\) and \(\gamma=\alpha\).
We will apply \eqref{eq:weighted_unrooted_sum} to control
\(\partial_{\bs\beta}^{\alpha}D_N(\bs\beta)\).
To do so, we identify the graphs generated by pairs of monomials and
determine their node counts, dyad counts, and maximum degrees.

Consider the two pentads \((i,j,k,l,m)\) and \((i,j,k',l',m')\),
which share the pivot nodes \(i,j\), and suppose their peripheral
triples share exactly \(r\in\{0,1,2,3\}\) nodes.
Take one monomial from the polynomial expansion of
\(\partial_{\bs\beta}^{\gamma}\bs\Psi_{(i,j,k,l,m)}(\bs\beta)\)
and one from that of
\(\partial_{\bs\beta}^{\alpha-\gamma}\bs\Psi_{(i,j,k',l',m')}(\bs\beta)\).
The dyad support of each monomial is the set of dyads \((u,v)\) whose
link indicators \(L_{uv}\) appear in that monomial.
Differentiation with respect to \(\bs\beta\) acts on the coefficients
and leaves each product of \(L_{uv}\) unchanged.
Equation~\eqref{eq:poly-coefficient-derivative-bound} gives uniform
bounds on
\(\norm{\partial_{\bs\beta}^{\gamma}\bs c_{s,\nu}(\bs\beta,\mathcal O_N)}\)
for \(\abs\gamma\le2\). By Lemma~\ref{lem:det-row-supports},
each monomial's dyad support forms a connected graph containing all
five nodes of its pentad. Since the two monomial supports share
\(i,j\) and exactly \(r\) peripheral nodes, the graph formed by their
union is connected on \(5+5-(2+r)=8-r\) nodes and contains at least
\((8-r)-1=7-r\) distinct dyads.
Using \(L_{uv}^2=L_{uv}\) and Assumption~\ref{ass:dyad_ind}, the
conditional expectation of the product of the two monomials given
\((\bf X,\bs\Gamma)\) is a product of link probabilities over the
distinct dyads in their union. Each link-probability factor has exponent
one, and the graph has maximum degree at most \(7-r\le7\), satisfying
the bound of \(16\) required for \eqref{eq:weighted_unrooted_sum}.

For each \(r\), only finitely many ways of matching the shared peripheral
nodes and pairs of monomial supports arise up to relabeling,
independently of \(N\). Fix one such matching and support pair, and list
the \(8-r\) distinct nodes of the two pentads in a fixed order as
\(i_1,\ldots,i_{8-r}\). Dropping the restriction \(i<j\) bounds the
corresponding nonnegative sum by the sum over all distinct node labels
in \eqref{eq:weighted_unrooted_sum}. Applying \eqref{eq:weighted_unrooted_sum} with
\(\abs V=8-r\), \(r_{uv}=1\), and \(\abs E\ge7-r\), using
\(\rho_N\le1\), and summing over the finitely many matchings and support
pairs therefore bounds the expected sum of link-indicator products by
\(C N^{8-r}\rho_N^{7-r}\).
The normalization in \(\partial_{\bs\beta}^{\alpha}D_N(\bs\beta)\)
is \(O(N^{-8})\). Using the bounds on
\(\norm{\partial_{\bs\beta}^{\gamma}\bs c_{s,\nu}(\bs\beta,\mathcal O_N)}\)
for \(\abs\gamma\le2\) in \eqref{eq:poly-coefficient-derivative-bound},
\(\norm{\widehat{\bf W}_N}=O_p(1)\), and Markov's inequality, the
contribution of pairs with exactly \(r\) shared peripheral nodes is
\(O_p(N^{-8}N^{8-r}\rho_N^{7-r})=O_p(N^{-r}\rho_N^{7-r})\),
uniformly over \(\bs\beta\in\mathcal B\) and \(\abs{\alpha}\le2\).
The case \(\alpha=0\) controls \(D_N\), while
\(1\le\abs{\alpha}\le2\) controls its first and second derivatives.
Summing over \(r=0,1,2,3\) gives
\[
\sup_{\bs\beta\in\mathcal B}D_N(\bs\beta)
+\max_{1\le k\le2}\sup_{\bs\beta\in\mathcal B}
\norm{\nabla_{\bs\beta}^kD_N(\bs\beta)}
=O_p\!\left(\sum_{r=0}^3N^{-r}\rho_N^{7-r}\right)=O_p(c_N).
\]

We next show that \(\widehat c_N/c_N\xrightarrow{P}1\). Link indicators for dyads
that do not share a node are independent. For dyads sharing node \(1\),
we have
\[
\begin{aligned}
\E[L_{12}L_{13}\mid\bs X_1,\Gamma_1]
&=\E[P_{12}(\bs\beta_0)P_{13}(\bs\beta_0)\mid\bs X_1,\Gamma_1]=\E[L_{12}\mid\bs X_1,\Gamma_1]\,
\E[L_{13}\mid\bs X_1,\Gamma_1],
\end{aligned}
\]
where the first equality uses iterated expectations and the independence
of \(L_{12}\) and \(L_{13}\) conditional on \((\bf X,\bs\Gamma)\) in
Assumption~\ref{ass:dyad_ind}. The second equality uses the fact that,
conditional on \((\bs X_1,\Gamma_1)\),
\(P_{12}(\bs\beta_0)\) and \(P_{13}(\bs\beta_0)\) are functions of
the independent node characteristics \((\bs X_2,\Gamma_2)\) and
\((\bs X_3,\Gamma_3)\), respectively, together with
\(\E[P_{1j}(\bs\beta_0)\mid\bs X_1,\Gamma_1]
=\E[L_{1j}\mid\bs X_1,\Gamma_1]\) for \(j\in\{2,3\}\)
by iterated expectations.
Thus \(\Cov(L_{12},L_{13}\mid\bs X_1,\Gamma_1)=0\).
Since \((\bs X_2,\Gamma_2)\) and \((\bs X_3,\Gamma_3)\) are identically
distributed, iterated expectations and the definition of \(\Lambda_{1,N}\) give
\(\E[L_{1j}\mid\bs X_1,\Gamma_1]=\rho_N\Lambda_{1,N}\) for
\(j\in\{2,3\}\). The law of total covariance and the moment bound in
Assumption~\ref{ass:sparse_envelope} therefore yield
\[
\begin{aligned}
&\Cov(L_{12},L_{13})
=\E[\Cov(L_{12},L_{13}\mid\bs X_1,\Gamma_1)]+\Cov\!\left(\E[L_{12}\mid\bs X_1,\Gamma_1],
\E[L_{13}\mid\bs X_1,\Gamma_1]\right)\\
&=\Var\{\E[L_{12}\mid\bs X_1,\Gamma_1]\}
=\rho_N^2\Var(\Lambda_{1,N})\le\rho_N^2\E[\Lambda_{1,N}^2]
\le\rho_N^2\bigl(\E[\Lambda_{1,N}^{16}]\bigr)^{1/8}
=O(\rho_N^2).
\end{aligned}
\]
Recall that \(\widehat{\rho}_N=\binom{N}{2}^{-1}\sum_{i<j}L_{ij}+N^{-2}\).
Since \(\E[L_{ij}]=\rho_N\), expanding the variance gives
\[
\begin{aligned}
&\E\!\left[
\left\{\frac{\widehat{\rho}_N-N^{-2}}{\rho_N}-1\right\}^2
\right]
=\frac{\Var\!\left(\sum_{i<j}L_{ij}\right)}
{\binom{N}{2}^{2}\rho_N^2}=\frac{\binom{N}{2}\Var(L_{12})
+6\binom{N}{3}\Cov(L_{12},L_{13})}
{\binom{N}{2}^{2}\rho_N^2}\\
&\le\frac{\binom{N}{2}\rho_N+6\binom{N}{3}C\rho_N^2}
{\binom{N}{2}^{2}\rho_N^2}
\le C\left(\frac{1}{N^2\rho_N}+\frac1N\right)=o(1),
\end{aligned}
\]
where the second equality counts \(\binom{N}{2}\) individual dyads and
\(3\binom{N}{3}\) unordered pairs of distinct dyads sharing one node;
each covariance appears twice in the variance expansion. Distinct dyads
with disjoint endpoints are independent and contribute zero covariance.
The first inequality uses \(\Var(L_{12})=\rho_N(1-\rho_N)\le\rho_N\)
and \(\Cov(L_{12},L_{13})=O(\rho_N^2)\).
The last step follows from \(N^5\rho_N^4\to\infty\), which implies
\(N^2\rho_N=N^{3/4}(N^5\rho_N^4)^{1/4}\to\infty\). Hence \(\widehat{\rho}_N/\rho_N\xrightarrow{P}1\), and
\[
\left|\frac{\widehat c_N}{c_N}-1\right|
\le\max_{4\le j\le7}
\left|\left(\frac{\widehat{\rho}_N}{\rho_N}\right)^j-1\right|
\xrightarrow{P}0.
\]
Since \(D_N(\bs\beta)\ge0\) and \(\widehat c_N/c_N\xrightarrow{P}1\),
the bound \(\sup_{\bs\beta\in\mathcal B}D_N(\bs\beta)=O_p(c_N)\) gives
\[
\begin{aligned}
&\sup_{\bs\beta\in\mathcal B}\{D_N(\bs\beta)+\widehat c_N\}
=O_p(c_N), \ \sup_{\bs\beta\in\mathcal B}\{D_N(\bs\beta)+\widehat c_N\}^{-1}
\le\widehat c_N^{-1}=O_p(c_N^{-1}).
\end{aligned}
\]
For every fixed \(\varepsilon>0\), Lemma~\ref{lem:uniform-convergence-population}
and Assumption~\ref{ass:consistency}(i) give
\[
\inf_{\bs\beta\in\mathcal B:\,\norm{\bs\beta-\bs\beta_0}\ge\varepsilon}
\rho_N^{-8}\bs h_N(\bs\beta)^\top\widehat{\bf W}_N
\bs h_N(\bs\beta)
\xrightarrow{P}
\inf_{\bs\beta\in\mathcal B:\,\norm{\bs\beta-\bs\beta_0}\ge\varepsilon}
\bs h_0(\bs\beta)^\top\bf W\bs h_0(\bs\beta)>0.
\]
Since \(D_N(\bs\beta)\ge0\), \eqref{eq:pair_normalized_criterion} gives
\[
Q_N^\dagger(\bs\beta_0)
=\frac{\bs h_N(\bs\beta_0)^\top\widehat{\bf W}_N\bs h_N(\bs\beta_0)}
{D_N(\bs\beta_0)+\widehat c_N}
\le\frac{\bs h_N(\bs\beta_0)^\top\widehat{\bf W}_N\bs h_N(\bs\beta_0)}
{\widehat c_N},
\]
and
\[
\inf_{\substack{\bs\beta\in\mathcal B\\
\norm{\bs\beta-\bs\beta_0}\ge\varepsilon}}
Q_N^\dagger(\bs\beta)
\ge
\frac{\displaystyle\rho_N^8
\inf_{\substack{\bs\beta\in\mathcal B\\
\norm{\bs\beta-\bs\beta_0}\ge\varepsilon}}
\rho_N^{-8}\bs h_N(\bs\beta)^\top\widehat{\bf W}_N\bs h_N(\bs\beta)}
{\displaystyle\sup_{\bs\beta\in\mathcal B}
\{D_N(\bs\beta)+\widehat c_N\}}.
\]
By Theorem~\ref{thm:regimeCLT},
\(a_N\bs h_N(\bs\beta_0)=O_p(1)\), so
\(\norm{\bs h_N(\bs\beta_0)}=O_p(a_N^{-1})\). Consequently,
\[
\begin{aligned}
&\frac{Q_N^\dagger(\bs\beta_0)}
{\displaystyle\inf_{\bs\beta\in\mathcal B:\,\norm{\bs\beta-\bs\beta_0}\ge\varepsilon}
Q_N^\dagger(\bs\beta)}\le
\frac{\bs h_N(\bs\beta_0)^\top\widehat{\bf W}_N
\bs h_N(\bs\beta_0)
\displaystyle\sup_{\bs\beta\in\mathcal B}
\{D_N(\bs\beta)+\widehat c_N\}}
{\widehat c_N\rho_N^8
\displaystyle\inf_{\bs\beta\in\mathcal B:\,\norm{\bs\beta-\bs\beta_0}\ge\varepsilon}
\rho_N^{-8}\bs h_N(\bs\beta)^\top\widehat{\bf W}_N
\bs h_N(\bs\beta)}
=O_p\!\left(\frac{1}{a_N^2\rho_N^8}\right)=o_p(1),
\end{aligned}
\]
because \(a_N^2\rho_N^8\to\infty\) by the definition of \(a_N\).
Since \(Q_N^\dagger(\widehat{\bs\beta}^{\dagger})
\le Q_N^\dagger(\bs\beta_0)\),
\[
\Pr\{\norm{\widehat{\bs\beta}^{\dagger}-\bs\beta_0}\ge\varepsilon\}
\le \Pr\!\left\{
  \frac{Q_N^\dagger(\widehat{\bs\beta}^{\dagger})}
  {\displaystyle\inf_{\bs\beta\in\mathcal B:\,\norm{\bs\beta-\bs\beta_0}\ge\varepsilon}
  Q_N^\dagger(\bs\beta)}\ge1\right\} \le
\Pr\!\left\{
\frac{Q_N^\dagger(\bs\beta_0)}
{\displaystyle\inf_{\bs\beta\in\mathcal B:\,\norm{\bs\beta-\bs\beta_0}\ge\varepsilon}
Q_N^\dagger(\bs\beta)}\ge1\right\}
\to 0.
\]
The minimizing property
\(Q_N^\dagger(\widehat{\bs\beta}^{\dagger})
\le Q_N^\dagger(\bs\beta_0)\) also gives
\[
\begin{aligned}
\norm{\bs h_N(\widehat{\bs\beta}^{\dagger})}^2
&\le
\frac{\lambda_{\min}(\widehat{\bf W}_N)^{-1}
\sup_{\bs\beta\in\mathcal B}\{D_N(\bs\beta)+\widehat c_N\}}
{D_N(\bs\beta_0)+\widehat c_N}
\bs h_N(\bs\beta_0)^\top\widehat{\bf W}_N\bs h_N(\bs\beta_0)
=O_p(a_N^{-2}).
\end{aligned}
\]
Here the order follows from
\(\lambda_{\min}(\widehat{\bf W}_N)^{-1}=O_p(1)\),
\(\norm{\widehat{\bf W}_N}=O_p(1)\),
\(\sup_{\bs\beta\in\mathcal B}\{D_N(\bs\beta)+\widehat c_N\}=O_p(c_N)\),
\(\{D_N(\bs\beta_0)+\widehat c_N\}^{-1}=O_p(c_N^{-1})\), and
\(\norm{\bs h_N(\bs\beta_0)}=O_p(a_N^{-1})\).

Since \(\bs\beta_0\in\operatorname{int}(\mathcal B)\) and
\(\widehat{\bs\beta}^{\dagger}\xrightarrow{P}\bs\beta_0\), choose a
deterministic sequence \(\delta_N\downarrow0\) such that
\(\{\bs\beta:\norm{\bs\beta-\bs\beta_0}\le\delta_N\}
\subset\operatorname{int}(\mathcal B)\) and
\(\Pr\{\norm{\widehat{\bs\beta}^{\dagger}-\bs\beta_0}\le\delta_N\}\to1\).
On the event
\(\norm{\widehat{\bs\beta}^{\dagger}-\bs\beta_0}\le\delta_N\),
the first-order condition for \(Q_N^\dagger\), using that
\(\widehat c_N\) does not depend on \(\bs\beta\), gives
\begin{equation}
\label{eq:pair_normalized_foc}
\dot{\bs h}_N(\widehat{\bs\beta}^{\dagger})^\top
\widehat{\bf W}_N\bs h_N(\widehat{\bs\beta}^{\dagger})
=\frac{\bs h_N(\widehat{\bs\beta}^{\dagger})^\top
\widehat{\bf W}_N\bs h_N(\widehat{\bs\beta}^{\dagger})}
{2\{D_N(\widehat{\bs\beta}^{\dagger})+\widehat c_N\}}
\nabla_{\bs\beta}D_N(\widehat{\bs\beta}^{\dagger})
=O_p(a_N^{-2}),
\end{equation}
using \(\norm{\bs h_N(\widehat{\bs\beta}^{\dagger})}=O_p(a_N^{-1})\),
\(\sup_{\bs\beta\in\mathcal B}
\{D_N(\bs\beta)+\widehat c_N\}^{-1}=O_p(c_N^{-1})\), and
\(\sup_{\bs\beta\in\mathcal B}
\norm{\nabla_{\bs\beta}D_N(\bs\beta)}=O_p(c_N)\).
Thus \(a_N\rho_N^{-4}\dot{\bs h}_N(\widehat{\bs\beta}^{\dagger})^\top
\widehat{\bf W}_N\bs h_N(\widehat{\bs\beta}^{\dagger})
=O_p((a_N\rho_N^4)^{-1})=o_p(1)\).

The integral mean-value expansion along the segment from \(\bs\beta_0\)
to \(\widehat{\bs\beta}^{\dagger}\) is
\begin{equation}
\label{eq:pair_normalized_moment_expansion}
\begin{aligned}
\bs h_N(\widehat{\bs\beta}^{\dagger})
&=\bs h_N(\bs\beta_0)+\left[\int_0^1\dot{\bs h}_N\!\left(
\bs\beta_0+t(\widehat{\bs\beta}^{\dagger}-\bs\beta_0)\right)dt\right]
(\widehat{\bs\beta}^{\dagger}-\bs\beta_0).
\end{aligned}
\end{equation}
Lemma~\ref{lem:uniform-convergence-population} and the choice of
\(\delta_N\) imply
\(\sup_{0\le t\le1}
\norm{\rho_N^{-4}\dot{\bs h}_N\!\left(
\bs\beta_0+t(\widehat{\bs\beta}^{\dagger}-\bs\beta_0)\right)
-\dot{\bs h}_0}\xrightarrow{P}0\).
In particular, \(\rho_N^{-4}\dot{\bs h}_N(\widehat{\bs\beta}^{\dagger})
\xrightarrow{P}\dot{\bs h}_0\). The triangle inequality and
\(\int_0^1dt=1\) give
\begin{equation}
\label{eq:pair_normalized_jacobian_convergence}
\begin{aligned}
&\left\|\rho_N^{-4}\int_0^1\dot{\bs h}_N\!\left(
\bs\beta_0+t(\widehat{\bs\beta}^{\dagger}-\bs\beta_0)\right)\,dt
-\dot{\bs h}_0\right\|\\
&\quad\le\sup_{0\le t\le1}
\norm{\rho_N^{-4}\dot{\bs h}_N\!\left(
\bs\beta_0+t(\widehat{\bs\beta}^{\dagger}-\bs\beta_0)\right)
-\dot{\bs h}_0}\xrightarrow{P}0.
\end{aligned}
\end{equation}
Substituting the integral expansion~\eqref{eq:pair_normalized_moment_expansion}
into the first-order condition~\eqref{eq:pair_normalized_foc},
multiplying by \(a_N\rho_N^{-4}\), and using
\eqref{eq:pair_normalized_jacobian_convergence} together with
\(\widehat{\bf W}_N\xrightarrow{P}\bf W\), we obtain
\[
\{\dot{\bs h}_0^\top\bf W\dot{\bs h}_0+o_p(1)\}
a_N\rho_N^4(\widehat{\bs\beta}^{\dagger}-\bs\beta_0)
=-\{\dot{\bs h}_0^\top\bf W+o_p(1)\}
\{a_N\bs h_N(\bs\beta_0)\}+o_p(1).
\]
Assumption~\ref{ass:consistency}(ii)--(iii) makes
\(\dot{\bs h}_0^\top\bf W\dot{\bs h}_0\) nonsingular.
Since \(a_N\bs h_N(\bs\beta_0)=O_p(1)\), inverting
\(\dot{\bs h}_0^\top\bf W\dot{\bs h}_0+o_p(1)\) gives
\[
a_N\rho_N^4(\widehat{\bs\beta}^{\dagger}-\bs\beta_0)
=-(\dot{\bs h}_0^\top\bf W\dot{\bs h}_0)^{-1}
\dot{\bs h}_0^\top\bf W
\{a_N\bs h_N(\bs\beta_0)\}+o_p(1).
\]
Equations~\eqref{eq:gmm_linearization} and
\eqref{eq:gmm_jacobian_limits}, together with
\(\widehat{\bf W}_N\xrightarrow{P}\bf W\) and
\(a_N\bs h_N(\bs\beta_0)=O_p(1)\), give
\[
a_N\rho_N^4(\widehat{\bs\beta}_{\mathrm{GMM}}-\bs\beta_0)
=-(\dot{\bs h}_0^\top\bf W\dot{\bs h}_0)^{-1}
\dot{\bs h}_0^\top\bf W
\{a_N\bs h_N(\bs\beta_0)\}+o_p(1).
\]
Subtracting the two expansions yields
\(a_N\rho_N^4(\widehat{\bs\beta}^{\dagger}
-\widehat{\bs\beta}_{\mathrm{GMM}})\xrightarrow{P}0\).
\end{proof}

\subsection{Proofs for Theorem~\ref{thm:tractable_gmm}}
\label{app:covariance_proofs}

This section establishes consistency of the covariance estimator in
\eqref{eq:universalOmega} and proves Theorem~\ref{thm:tractable_gmm}. 
We begin with some lemmas.

\begin{lem}
\label{lem:graph_count_covariance_approximation}
Suppose Assumptions~\ref{ass:dyad_ind}, \ref{ass:bounded_design},
and \ref{ass:sparse_envelope} hold.
Let \(a_N\) be defined as in Theorem~\ref{thm:gmm} under each of the three regimes. Then
\(\norm{\bOmegaUniv(\bs\beta_0)
-\Var\{\bs h_N(\bs\beta_0)\mid\mathcal O_N\}}
=o_p(a_N^{-2})\).
\end{lem}

\begin{proof}
Unless an argument is displayed,
pentad moments, sample moments, and covariance estimators in this proof
are evaluated at \(\bs\beta_0\). Theorem~\ref{thm:moment_restriction} and the
\(\mathcal O_N\)-measurability of \(\bs Z_s\) give
\(\E[\bs\Psi_s\mid\mathcal O_N]=\boldsymbol0\).
Conditional dyadic independence therefore implies
\(\E[\bs\Psi_s\bs\Psi_t^\top\mid\mathcal O_N]=\mathbf0\)
whenever \(E(s)\cap E(t)=\varnothing\).
Recall that \(\zeta_G(s,t)\) counts the labeled copies of
\(G\in\mathcal G\) contained in the common graph
\((V_s\cap V_t,E(s)\cap E(t))\).
Each copy is counted once by its node and dyad sets, regardless of node
ordering, and additional dyads among the selected nodes are allowed.
Define the signed pair weight
\begin{equation}
\label{eq:pair_weight}
w_{st}:=
\zeta_{G_{\mathrm{node}}}(s,t)
-\zeta_{G_{\mathrm{dyad}}}(s,t)
+\zeta_{G_{\mathrm{triangle}}}(s,t)
+\zeta_{G_{\text{4-cycle}}}(s,t)
-\zeta_{G_{\mathrm{diamond}}}(s,t)
\end{equation}
and the uncentered matrix
\(\widehat{\bf\Omega}_N^\circ:=\abs{\mathcal S_N}^{-2}
\sum_{s,t\in\mathcal S_N}w_{st}\bs\Psi_s\bs\Psi_t^\top\).
Each count in \eqref{eq:pair_weight} is uniformly bounded, so
\(\abs{w_{st}}\le C\); also, \(w_{st}=0\) when
\(V_s\cap V_t=\varnothing\).
The sum \(\sum_{t\in\mathcal S_N}w_{st}\) is unchanged by
relabeling the nodes and therefore has the same value,
denoted by \(\varsigma_N\), for every \(s\in\mathcal S_N\).
Since \(w_{st}=w_{ts}\),
\(\sum_{s\in\mathcal S_N}w_{st}=\varsigma_N\) for every
\(t\in\mathcal S_N\).
By \eqref{eq:graph_count_component} and \eqref{eq:pair_weight}, the
first five terms of \eqref{eq:universalOmega} equal
\(\abs{\mathcal S_N}^{-2}\sum_{s,t\in\mathcal S_N}
w_{st}\widehat{\bs\Psi}_s\widehat{\bs\Psi}_t^\top\).
Expand \(\widehat{\bs\Psi}_s=\bs\Psi_s-\bs h_N\) in this sum and use
\(\sum_{s\in\mathcal S_N}\bs\Psi_s=\abs{\mathcal S_N}\bs h_N\).
Each of the two cross products sums to
\(-\varsigma_N\bs h_N\bs h_N^\top/\abs{\mathcal S_N}\), and the product of the
sample means sums to \(\varsigma_N\bs h_N\bs h_N^\top/\abs{\mathcal S_N}\).
Thus
\begin{equation}
\label{eq:signed_centering_identity}
\bOmegaUniv-\frac{10}{N}\sum_{G\in\mathcal G}\widehat{\bf\Omega}_N(G)
=\widehat{\bf\Omega}_N^\circ-\frac{\varsigma_N}{\abs{\mathcal S_N}}\bs h_N\bs h_N^\top.
\end{equation}
Subtracting the target covariance from
\eqref{eq:signed_centering_identity} gives
\begin{equation}
\label{eq:universal_four_term_decomposition}
\begin{aligned}
\bOmegaUniv-\Var\{\bs h_N\mid\mathcal O_N\}
&=\{\E[\widehat{\bf\Omega}_N^\circ\mid\mathcal O_N]
-\Var\{\bs h_N\mid\mathcal O_N\}\}
+\{\widehat{\bf\Omega}_N^\circ-\E[\widehat{\bf\Omega}_N^\circ\mid\mathcal O_N]\}\\
&\quad-\frac{\varsigma_N}{\abs{\mathcal S_N}}\bs h_N\bs h_N^\top
+\frac{10}{N}\sum_{G\in\mathcal G}\widehat{\bf\Omega}_N(G).
\end{aligned}
\end{equation}
We bound the four terms in \eqref{eq:universal_four_term_decomposition}
in their order of appearance.

For \(\E[\widehat{\bf\Omega}_N^\circ\mid\mathcal O_N]
-\Var\{\bs h_N\mid\mathcal O_N\}\), equation~\eqref{eq:target_covariance}
gives
\[
\E[\widehat{\bf\Omega}_N^\circ\mid\mathcal O_N]-\Var\{\bs h_N\mid\mathcal O_N\}
=\frac{1}{\abs{\mathcal S_N}^2}\sum_{s,t\in\mathcal S_N}
(w_{st}-1)\E[\bs\Psi_s\bs\Psi_t^\top\mid\mathcal O_N].
\]
Pairs with \(E(s)\cap E(t)=\varnothing\) contribute zero.
We next verify that \(w_{st}=1\) whenever the common graph
\((V_s\cap V_t,E(s)\cap E(t))\) is connected and contains a dyad.
Write \(s=(i,j,k,l,m)\), \(t=(i',j',k',l',m')\), and
\(v:=\abs{V_s\cap V_t}\).
For the following argument, we temporarily identify each unordered
dyad \((a,b)\) with its endpoint set \(\{a,b\}\).
By definition, \(E(s)\) consists of
the pivot dyad \(\{i,j\}\) and the six dyads joining \(i\) or \(j\)
to \(k,l,m\). Thus, for any dyad \(e\subseteq V_s\),
\(e\in E(s)\) if and only if \(e\cap\{i,j\}\ne\varnothing\).
For any dyad \(e\subseteq V_t\), the definition of \(E(t)\) likewise
gives \(e\in E(t)\) if and only if
\(e\cap\{i',j'\}\ne\varnothing\).
For a dyad \(e\subseteq V_s\cap V_t\), both characterizations yield:
\begin{equation}
\label{eq:common_dyad_membership}
e\in E(s)\cap E(t)
\quad\Longleftrightarrow\quad
e\cap\{i,j\}\ne\varnothing
\ \text{and}\ e\cap\{i',j'\}\ne\varnothing.
\end{equation}

First suppose that the two pivot dyads coincide. Without loss of generality, assume that \(i=i'\) and \(j=j'\).
By \eqref{eq:common_dyad_membership}, the common graph consists of
the pivot dyad \(\{i,j\}\) and \(2(v-2)\) dyads connecting the
common peripheral nodes to the pivots: each peripheral node is joined
to both \(i\) and \(j\). There are no dyads between peripheral nodes.
Figure~\ref{fig:common_pivot_graphs}(a) shows an
example with \(v=4\).
There are therefore \(1+2(v-2)=2v-3\) dyads and \(v-2\) triangles,
one triangle for each common peripheral node together with \(i,j\).
Every four-cycle or diamond must contain both pivot nodes and two
common peripheral nodes, since each peripheral node is adjacent only
to \(i\) and \(j\). Each pair of common peripheral nodes gives exactly
one \mbox{four-cycle
(\taxgraph[baseline=-0.5ex]{
  \node[taxNode] (a) at (-.55,-.4) {};
  \node[taxNode] (b) at (-.55,.4) {};
  \node[taxNode] (c) at (.55,.4) {};
  \node[taxNode] (d) at (.55,-.4) {};
  \draw[taxEdge] (a)--(b)--(c)--(d)--(a);
})}, using the four dyads joining the peripheral nodes to \(i,j\),
and one diamond, which additionally includes the pivot dyad \(\{i,j\}\).
Thus \(\zeta_{G_{\text{4-cycle}}}(s,t)
=\zeta_{G_{\mathrm{diamond}}}(s,t)=\binom{v-2}{2}\).
The four-cycle and diamond counts cancel in \eqref{eq:pair_weight},
leaving \(w_{st}=v-(2v-3)+(v-2)=1\).

Next suppose that the two pivot dyads share exactly one endpoint,
and relabel their endpoints so that \(i=i'\) and \(j\ne j'\).
By \eqref{eq:common_dyad_membership}, \(i\) is joined to every other
node in the common graph \((V_s\cap V_t,E(s)\cap E(t))\).
A dyad in this common graph not containing \(i\) must contain both
\(j\) and \(j'\) to have an endpoint in each pivot dyad, so the only
candidate is \(\{j,j'\}\).
If \(j,j'\in V_s\cap V_t\), then \(\{j,j'\}\) is a dyad among
the common nodes, with endpoint \(j\) in \(\{i,j\}\) and endpoint
\(j'\) in \(\{i,j'\}\); hence
\(\{j,j'\}\in E(s)\cap E(t)\) by \eqref{eq:common_dyad_membership}.
If either endpoint is not a common node, \(\{j,j'\}\) cannot belong
to the common graph.
Thus the common graph is a star centered at \(i\), possibly
with the additional dyad \(\{j,j'\}\).
Figure~\ref{fig:common_pivot_graphs}(b) shows an example with
\(\{j,j'\}\), and panel (c) shows an example without it.
Without \(\{j,j'\}\), the graph has \(v-1\) dyads and no cycles.
Adding \(\{j,j'\}\) creates exactly one triangle, on \(i,j,j'\),
and no other cycle: every common node outside \(\{i,j,j'\}\)
is joined only to \(i\) and cannot lie on a cycle.
Consequently, there are no four-cycles or diamonds, and
\eqref{eq:pair_weight} gives \(w_{st}=v-(v-1)=1\) without \(\{j,j'\}\), and \(w_{st}=v-v+1=1\) with \(\{j,j'\}\).

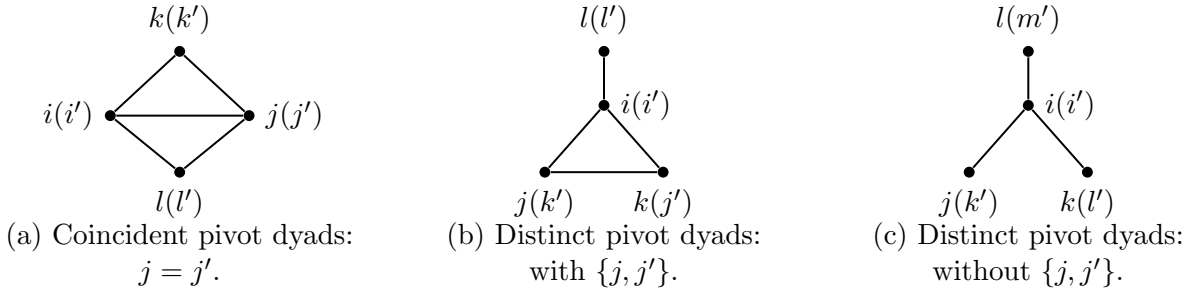
\begin{figure}[!htbp]
\centering
\captionsetup[subfigure]{justification=centering}
\begin{subfigure}[t]{0.32\textwidth}
\centering
\begin{tikzpicture}[
  x=0.68cm,y=0.68cm,
  every node/.style={font=\small},
  dot/.style={circle,fill=black,inner sep=1.45pt},
  edge/.style={draw=black,line width=0.75pt}
]
  \path[use as bounding box] (-2.2,-1.7) rectangle (2.2,1.8);
  \node[dot,label=left:{$i(i')$}] (a) at (-1.35,0) {};
  \node[dot,label=right:{$j(j')$}] (b) at (1.35,0) {};
  \node[dot,label=above:{$k(k')$}] (c) at (0,1.25) {};
  \node[dot,label=below:{$l(l')$}] (d) at (0,-1.10) {};
  \draw[edge] (a)--(b);
  \draw[edge] (a)--(c);
  \draw[edge] (a)--(d);
  \draw[edge] (b)--(c);
  \draw[edge] (b)--(d);
\end{tikzpicture}
\caption{Coincident pivot dyads:\\ \(j=j'\).}
\end{subfigure}
\hfill
\begin{subfigure}[t]{0.32\textwidth}
\centering
\begin{tikzpicture}[
  x=0.68cm,y=0.68cm,
  every node/.style={font=\small},
  dot/.style={circle,fill=black,inner sep=1.45pt},
  edge/.style={draw=black,line width=0.75pt}
]
  \path[use as bounding box] (-2.2,-1.7) rectangle (2.2,1.8);
  \node[dot,label=right:{$i(i')$}] (a) at (0,0.20) {};
  \node[dot,label=below:{$j(k')$}] (b) at (-1.15,-1.10) {};
  \node[dot,label=below:{$k(j')$}] (c) at (1.15,-1.10) {};
  \node[dot,label=above:{$l(l')$}] (d) at (0,1.25) {};
  \draw[edge] (a)--(b);
  \draw[edge] (a)--(c);
  \draw[edge] (a)--(d);
  \draw[edge] (b)--(c);
\end{tikzpicture}
\caption{Distinct pivot dyads:\\ with \(\{j,j'\}\).}
\end{subfigure}
\hfill
\begin{subfigure}[t]{0.32\textwidth}
\centering
\begin{tikzpicture}[
  x=0.68cm,y=0.68cm,
  every node/.style={font=\small},
  dot/.style={circle,fill=black,inner sep=1.45pt},
  edge/.style={draw=black,line width=0.75pt}
]
  \path[use as bounding box] (-2.2,-1.7) rectangle (2.2,1.8);
  \node[dot,label=right:{$i(i')$}] (a) at (0,0.20) {};
  \node[dot,label=below:{$j(k')$}] (b) at (-1.15,-1.10) {};
  \node[dot,label=below:{$k(l')$}] (c) at (1.15,-1.10) {};
  \node[dot,label=above:{$l(m')$}] (d) at (0,1.25) {};
  \draw[edge] (a)--(b);
  \draw[edge] (a)--(c);
  \draw[edge] (a)--(d);
\end{tikzpicture}
\caption{Distinct pivot dyads:\\ without \(\{j,j'\}\).}
\end{subfigure}
\caption{Examples of the common graph
\((V_s\cap V_t,E(s)\cap E(t))\), with four common nodes in each panel.
Parentheses give the node labels in \(t\).
In (a), the common graph has five dyads, two triangles, one four-cycle,
and one diamond. In (b), \(k=j'\), and the common graph is a star
centered at \(i=i'\) with the additional dyad \(\{j,j'\}\);
it has four dyads, one triangle, and no four-cycles or diamonds.
In (c), \(j'\notin V_s\), so the common graph is a star centered
at \(i=i'\), with three dyads and no cycles.}
\label{fig:common_pivot_graphs}
\end{figure}

Finally, suppose that the two pivot dyads \(\{i,j\}\) and
\(\{i',j'\}\) have no common endpoint.
By \eqref{eq:common_dyad_membership}, every dyad in the common graph
has one endpoint in \(\{i,j\}\) and the other in \(\{i',j'\}\).
A node in \(V_s\cap V_t\) but outside \(\{i,j,i',j'\}\) cannot
belong to any common dyad, because the other endpoint of such a
dyad would have to belong to both \(\{i,j\}\) and \(\{i',j'\}\),
which are disjoint.
Such a node would therefore be isolated in the common graph.
Since the common graph is connected and contains a dyad, every
common node must instead belong to \(\{i,j,i',j'\}\).

The four pivot endpoints need not all be common nodes: for example,
\(j\) belongs to \(V_s\), but may not belong to \(V_t\).
However, the common graph contains at least one dyad, with one
endpoint in \(\{i,j\}\) and the other in \(\{i',j'\}\).
Thus each pivot dyad has at least one endpoint in \(V_s\cap V_t\),
and may have both endpoints there.
Any two common nodes, one from each pivot dyad, are joined in the
common graph by \eqref{eq:common_dyad_membership}, because the dyad
between them has an endpoint in each pivot dyad.
Two common nodes from the same pivot dyad are not joined in the
common graph, because the dyad between them has no endpoint in the
other pivot dyad.

There are therefore three possibilities, illustrated in
Figure~\ref{fig:disjoint_pivot_graphs}.
If each pivot dyad has exactly one common endpoint, the common graph
consists of these two nodes and the single dyad joining them, as in
panel (a).
If one pivot dyad has one common endpoint and the other has two,
the single endpoint is joined to each of the other two nodes,
which are not joined to each other. The common graph is therefore
a two-star, as in panel (b); exchanging the roles of the two pivot
dyads gives the same graph shape.
If both endpoints of each pivot dyad are common nodes, the common
graph has all four nodes and exactly the dyads
\(\{i,i'\}\), \(\{i,j'\}\), \(\{j,i'\}\), and \(\{j,j'\}\).
These four dyads form a four-cycle, as in panel (c).
Equation~\eqref{eq:pair_weight} gives
\(w_{st}=2-1=1\), \(w_{st}=3-2=1\), or \(w_{st}=4-4+1=1\),
respectively.

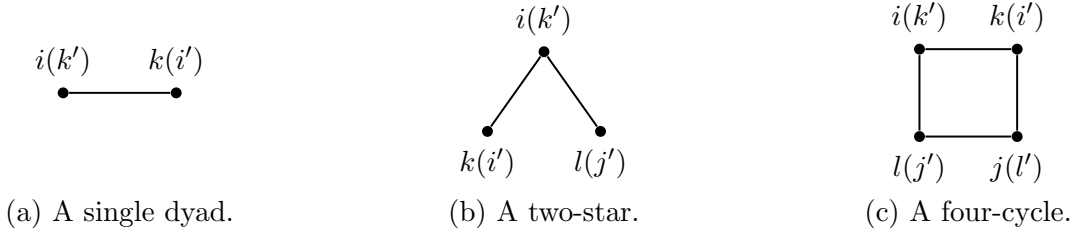
\begin{figure}[!htbp]
\centering
\captionsetup[subfigure]{justification=centering}
\begin{subfigure}[t]{0.32\textwidth}
\centering
\begin{tikzpicture}[
  x=0.68cm,y=0.68cm,
  every node/.style={font=\small},
  dot/.style={circle,fill=black,inner sep=1.45pt},
  edge/.style={draw=black,line width=0.75pt}
]
  \path[use as bounding box] (-2.2,-1.7) rectangle (2.2,1.8);
  \node[dot,label=above:{$i(k')$}] (a) at (-1.10,0) {};
  \node[dot,label=above:{$k(i')$}] (b) at (1.10,0) {};
  \draw[edge] (a)--(b);
\end{tikzpicture}
\caption{A single dyad.}
\end{subfigure}
\hfill
\begin{subfigure}[t]{0.32\textwidth}
\centering
\begin{tikzpicture}[
  x=0.68cm,y=0.68cm,
  every node/.style={font=\small},
  dot/.style={circle,fill=black,inner sep=1.45pt},
  edge/.style={draw=black,line width=0.75pt}
]
  \path[use as bounding box] (-2.2,-1.7) rectangle (2.2,1.8);
  \node[dot,label=above:{$i(k')$}] (a) at (0,0.80) {};
  \node[dot,label=below:{$k(i')$}] (b) at (-1.10,-0.75) {};
  \node[dot,label=below:{$l(j')$}] (c) at (1.10,-0.75) {};
  \draw[edge] (a)--(b);
  \draw[edge] (a)--(c);
\end{tikzpicture}
\caption{A two-star.}
\end{subfigure}
\hfill
\begin{subfigure}[t]{0.32\textwidth}
\centering
\begin{tikzpicture}[
  x=0.68cm,y=0.68cm,
  every node/.style={font=\small},
  dot/.style={circle,fill=black,inner sep=1.45pt},
  edge/.style={draw=black,line width=0.75pt}
]
  \path[use as bounding box] (-2.2,-1.7) rectangle (2.2,1.8);
  \node[dot,label=above:{$i(k')$}] (a) at (-0.95,0.85) {};
  \node[dot,label=above:{$k(i')$}] (b) at (0.95,0.85) {};
  \node[dot,label=below:{$j(l')$}] (c) at (0.95,-0.85) {};
  \node[dot,label=below:{$l(j')$}] (d) at (-0.95,-0.85) {};
  \draw[edge] (a)--(b)--(c)--(d)--(a);
\end{tikzpicture}
\caption{A four-cycle.}
\end{subfigure}
\caption{The three possible connected common graphs
\((V_s\cap V_t,E(s)\cap E(t))\) containing a dyad when the
pivot dyads \(\{i,j\}\) and \(\{i',j'\}\) have no common endpoint.
Only common nodes are shown; parentheses give their labels in \(t\).
In (a), \(V_s\cap V_t=\{i,k\}\), so each pivot dyad has one common
endpoint. In (b), \(V_s\cap V_t=\{i,k,l\}\), so \(\{i,j\}\) has
one common endpoint and \(\{i',j'\}\) has two.
In (c), \(V_s\cap V_t=\{i,j,k,l\}\), so both endpoints of each
pivot dyad are common nodes.}
\label{fig:disjoint_pivot_graphs}
\end{figure}

Thus only pairs whose common graph is disconnected and contains a
dyad contribute to \(\E[\widehat{\bf\Omega}_N^\circ\mid\mathcal O_N]-\Var\{\bs h_N\mid\mathcal O_N\}\), which yields:
\begin{equation}
\label{eq:conditional_mean_disconnected}
\begin{aligned}
&\E[\widehat{\bf\Omega}_N^\circ\mid\mathcal O_N]-\Var\{\bs h_N\mid\mathcal O_N\}=\frac{1}{\abs{\mathcal S_N}^2}
\sum_{\substack{s,t\in\mathcal S_N:\ E(s)\cap E(t)\ne\varnothing,\\
(V_s\cap V_t,E(s)\cap E(t))\ \mathrm{disconnected}}}
(w_{st}-1)\E[\bs\Psi_s\bs\Psi_t^\top\mid\mathcal O_N].
\end{aligned}
\end{equation}

For \(\E[\bs\Psi_s\bs\Psi_t^\top\mid\mathcal O_N]\) in
\eqref{eq:conditional_mean_disconnected}, write
\(\bs\Psi_s=\sum_{\nu=1}^{\nu_{\max}}
\bs c_{s,\nu}(\bs\beta_0,\mathcal O_N)
\prod_{e\in E_{s,\nu}}L_e\) by Lemma~\ref{lem:poly}.
The coefficient vectors are \(\mathcal O_N\)-measurable.
Using \(L_e^2=L_e\) and conditional dyadic independence gives
\begin{equation}
\label{eq:cov_pair_monomial_expansion}
\begin{aligned}
\E[\bs\Psi_s\bs\Psi_t^\top\mid\mathcal O_N]
&=\sum_{\nu=1}^{\nu_{\max}}\sum_{\mu=1}^{\nu_{\max}}
\bs c_{s,\nu}(\bs\beta_0,\mathcal O_N)
\bs c_{t,\mu}(\bs\beta_0,\mathcal O_N)^\top
\prod_{e\in E_{s,\nu}\cup E_{t,\mu}}P_e(\bs\beta_0).
\end{aligned}
\end{equation}
To bound the expectations of the probability products in
\eqref{eq:cov_pair_monomial_expansion} using
\eqref{eq:weighted_fixed_label_bound}, we first obtain a lower bound
for \(\abs{E_{s,\nu}\cup E_{t,\mu}}\), the number of dyads in each
product.
Fix a pair \((s,t)\) in \eqref{eq:conditional_mean_disconnected} and
write \(v:=\abs{V_s\cap V_t}\). Since the common graph is disconnected
and contains a dyad, \(3\le v\le5\).
By Lemma~\ref{lem:det-row-supports}, both
\((V_s,E_{s,\nu})\) and \((V_t,E_{t,\mu})\) are connected.
They share the \(v\) common nodes, so their union is connected on
\(10-v\) nodes. We show that the union contains a cycle and hence
has at least \(10-v\) dyads.
Suppose the graph \((V_s\cup V_t,E_{s,\nu}\cup E_{t,\mu})\)
were a tree. For any two distinct nodes in \(V_s\cap V_t\),
connectedness provides a simple path between them in
\((V_s,E_{s,\nu})\) and another in \((V_t,E_{t,\mu})\).
Both paths lie in \((V_s\cup V_t,E_{s,\nu}\cup E_{t,\mu})\),
so they must coincide because a tree has a unique simple path
between any two distinct nodes. The common path therefore lies in
\((V_s\cap V_t,E_{s,\nu}\cap E_{t,\mu})\).
Hence \((V_s\cap V_t,E_{s,\nu}\cap E_{t,\mu})\) would be connected.
Since \(E_{s,\nu}\cap E_{t,\mu}\subseteq E(s)\cap E(t)\),
every dyad on a path in \((V_s\cap V_t,E_{s,\nu}\cap E_{t,\mu})\)
also belongs to \(E(s)\cap E(t)\).
Thus any two distinct nodes in \(V_s\cap V_t\) would remain joined
by that path in \((V_s\cap V_t,E(s)\cap E(t))\).
The graph \((V_s\cap V_t,E(s)\cap E(t))\) would therefore be connected,
contrary to the restriction in \eqref{eq:conditional_mean_disconnected}.
Thus \((V_s\cup V_t,E_{s,\nu}\cup E_{t,\mu})\) contains a cycle.
Since it is connected on
\(\abs{V_s\cup V_t}=5+5-v=10-v\) nodes,
it has at least \(10-v\) dyads.
Equation~\eqref{eq:weighted_fixed_label_bound} therefore gives, for every pair
in \eqref{eq:conditional_mean_disconnected},
\begin{equation}
\label{eq:disconnected_pair_support_bound}
\E\!\left[\prod_{e\in E_{s,\nu}\cup E_{t,\mu}}P_e(\bs\beta_0)\right]
\le C\rho_N^{\abs{E_{s,\nu}\cup E_{t,\mu}}}
\le C\rho_N^{10-v}.
\end{equation}
For each fixed \(s\) and \(v\in\{1,\ldots,5\}\), there are
\(\binom5v\binom{N-5}{5-v}\) choices for \(V_t\) with
\(\abs{V_s\cap V_t}=v\), and each choice admits \(5!\)
orderings of \(t\). Summing over \(s\in\mathcal S_N\) gives
\begin{equation}
\label{eq:cov_pentad_pair_count}
\#\{(s,t)\in\mathcal S_N^2:\abs{V_s\cap V_t}=v\}
=\abs{\mathcal S_N}\binom5v\binom{N-5}{5-v}5!
\le C N^{10-v}.
\end{equation}
Lemma~\ref{lem:poly} and the monomial expansion in
\eqref{eq:cov_pair_monomial_expansion}, applied to
\eqref{eq:conditional_mean_disconnected}, give
\[
\begin{aligned}
&\E\norm{\E[\widehat{\bf\Omega}_N^\circ\mid\mathcal O_N]
-\Var\{\bs h_N\mid\mathcal O_N\}}
\\
&\le\frac{C}{\abs{\mathcal S_N}^2}
\sum_{v=3}^5
\sum_{\substack{s,t\in\mathcal S_N:\ \abs{V_s\cap V_t}=v,\\
E(s)\cap E(t)\ne\varnothing,\\
(V_s\cap V_t,E(s)\cap E(t))\ \mathrm{disconnected}}}
\sum_{\nu=1}^{\nu_{\max}}\sum_{\mu=1}^{\nu_{\max}}
\E\!\left[\prod_{e\in E_{s,\nu}\cup E_{t,\mu}}P_e(\bs\beta_0)\right]\\
&\le\frac{C}{\abs{\mathcal S_N}^2}
\sum_{v=3}^5
\abs{\mathcal S_N}\binom5v\binom{N-5}{5-v}5!\,\rho_N^{10-v}\le C\sum_{v=3}^5N^{-v}\rho_N^{10-v}
=\frac{C}{N}\sum_{v=2}^4N^{-v}\rho_N^{9-v}.
\end{aligned}
\]
Here the first inequality uses
\(\norm{\bs c_{s,\nu}(\bs\beta_0,\mathcal O_N)}\le C\) almost surely,
uniformly in \(N,s,\nu\), by
\eqref{eq:poly-coefficient-derivative-bound}, and
\(\abs{w_{st}-1}\le C\); the second uses
\eqref{eq:disconnected_pair_support_bound}, the finite number of
monomials, and \eqref{eq:cov_pentad_pair_count}.
The factor \(N^{-v}\) comes from dividing the pair count by
\(\abs{\mathcal S_N}^2\asymp N^{10}\).
By the definition of \(a_N\) in Theorem~\ref{thm:gmm},
\begin{equation}
\label{eq:cov_pair_scale_bound}
\sum_{v=2}^5N^{-v}\rho_N^{9-v}=O(a_N^{-2}).
\end{equation}
Since \(\sum_{v=2}^4N^{-v}\rho_N^{9-v}\) is bounded by the sum in
\eqref{eq:cov_pair_scale_bound}, we obtain
\[
\E\norm{\E[\widehat{\bf\Omega}_N^\circ\mid\mathcal O_N]
-\Var\{\bs h_N\mid\mathcal O_N\}}=O(N^{-1}a_N^{-2}).
\]
Markov's inequality then yields
\begin{equation}
\label{eq:signed_sum_conditional_mean_orders}
\norm{\E[\widehat{\bf\Omega}_N^\circ\mid\mathcal O_N]
-\Var\{\bs h_N\mid\mathcal O_N\}}
=O_p(N^{-1}a_N^{-2})=o_p(a_N^{-2}).
\end{equation}

We next bound the term
\(\widehat{\bf\Omega}_N^\circ-\E[\widehat{\bf\Omega}_N^\circ\mid\mathcal O_N]\)
in \eqref{eq:universal_four_term_decomposition}
by expanding its conditional variance entry by entry.
For \(a,b\in\{1,\ldots,q\}\), write \(\Psi_{s,a}\) for coordinate
\(a\) of \(\bs\Psi_s\). Then
\begin{equation}
\label{eq:signed_sum_entry_variance}
\begin{aligned}
\Var\{(\widehat{\bf\Omega}_N^\circ)_{ab}\mid\mathcal O_N\}
&=\frac{1}{\abs{\mathcal S_N}^4}
\sum_{s,t,s',t'\in\mathcal S_N}w_{st}w_{s't'}\operatorname{Cov}
(\Psi_{s,a}\Psi_{t,b},\Psi_{s',a}\Psi_{t',b}\mid\mathcal O_N).
\end{aligned}
\end{equation}
Consider a summand with \(w_{st}w_{s't'}\ne0\), and set
\(p:=\abs{V_s\cup V_t\cup V_{s'}\cup V_{t'}}\).
Because \(w_{st}\ne0\), the pentads \(s\) and \(t\) share at least
one node; because \(w_{s't'}\ne0\), the pentads \(s'\) and \(t'\)
also share at least one node. If
\((E(s)\cup E(t))\cap(E(s')\cup E(t'))=\varnothing\),
the two products in \eqref{eq:signed_sum_entry_variance} are
conditionally independent given \(\mathcal O_N\), so their conditional
covariance is zero. Otherwise, there is a dyad belonging to both
\(E(s)\cup E(t)\) and \(E(s')\cup E(t')\).
The two distinct endpoints of that dyad belong to both
\(V_s\cup V_t\) and \(V_{s'}\cup V_{t'}\), so
\(\abs{(V_s\cup V_t)\cap(V_{s'}\cup V_{t'})}\ge2\).
By the inclusion--exclusion principle,
\[
\begin{aligned}
p&=20-\abs{V_s\cap V_t}-\abs{V_{s'}\cap V_{t'}}
-\abs{(V_s\cup V_t)\cap(V_{s'}\cup V_{t'})}\\
&\le20-1-1-2=16.
\end{aligned}
\]
If \(p=16\), then \(\abs{V_s\cap V_t}=1\) and
\(\abs{V_{s'}\cap V_{t'}}=1\), and the two pairs share exactly two
nodes. Hence \(E(s)\cap E(t)=\varnothing\)
and \(E(s')\cap E(t')=\varnothing\).
The sets \(E(s)\cup E(t)\) and \(E(s')\cup E(t')\) share exactly
one dyad, because they have a common dyad and their node sets
intersect in only two nodes.
Since \(E(s)\cap E(t)=\varnothing\), the shared dyad cannot belong
to both \(E(s)\) and \(E(t)\).
If it belongs to \(E(s)\), then \(E(t)\) is disjoint from
\(E(s)\), \(E(s')\), and \(E(t')\).
If it belongs to \(E(t)\), then \(E(s)\) is disjoint from
\(E(t)\), \(E(s')\), and \(E(t')\).
Conditional dyadic independence therefore makes \(\bs\Psi_t\)
independent of the joint collection
\((\bs\Psi_s,\bs\Psi_{s'},\bs\Psi_{t'})\) in the first case,
and makes \(\bs\Psi_s\) independent of
\((\bs\Psi_t,\bs\Psi_{s'},\bs\Psi_{t'})\) in the second case,
conditional on \(\mathcal O_N\).
Since both \(\bs\Psi_s\) and \(\bs\Psi_t\) have conditional mean
zero given \(\mathcal O_N\), in either case
\(\E[\Psi_{s,a}\Psi_{t,b}\Psi_{s',a}\Psi_{t',b}\mid\mathcal O_N]=0\).
Also \(\E[\Psi_{s,a}\Psi_{t,b}\mid\mathcal O_N]=0\), because
\(E(s)\cap E(t)=\varnothing\). Hence
\(\operatorname{Cov}(\Psi_{s,a}\Psi_{t,b},\Psi_{s',a}\Psi_{t',b}
\mid\mathcal O_N)=0\) when \(p=16\).
Thus every nonzero summand in \eqref{eq:signed_sum_entry_variance}
must have \(p\le15\).
Also, \(p\ge\abs{V_s}=5\), since the union
\(V_s\cup V_t\cup V_{s'}\cup V_{t'}\) contains all five nodes of
the pentad \(s\). Only quadruples with \(5\le p\le15\) can
therefore contribute.

For each contributing quadruple \((s,t,s',t')\), expand the four moments using
Lemma~\ref{lem:poly} and fix monomial indices \(\nu,\mu,\nu',\mu'\), respectively.
Each selected support is connected and spans its pentad's five nodes.
The supports \(E_{s,\nu},E_{t,\mu}\) intersect in nodes, as do
\(E_{s',\nu'},E_{t',\mu'}\). The two unions
\(E_{s,\nu}\cup E_{t,\mu}\) and \(E_{s',\nu'}\cup E_{t',\mu'}\)
also intersect in nodes, so the union of the four supports is
connected on \(p\) nodes. In particular,
\[
\begin{aligned}
\abs{E_{s,\nu}\cup E_{t,\mu}\cup E_{s',\nu'}\cup E_{t',\mu'}}
&\ge p-1,\ \ 
\abs{E_{s,\nu}\cup E_{t,\mu}}+
\abs{E_{s',\nu'}\cup E_{t',\mu'}}\ge p-1.
\end{aligned}
\]
Using \(L_e^2=L_e\) and conditional dyadic independence, each term in
\(\E[\Psi_{s,a}\Psi_{t,b}\Psi_{s',a}\Psi_{t',b}\mid\mathcal O_N]\)
contains one factor \(P_e(\bs\beta_0)\) per dyad in
\(E_{s,\nu}\cup E_{t,\mu}\cup E_{s',\nu'}\cup E_{t',\mu'}\).
Each term in
\(\E[\Psi_{s,a}\Psi_{t,b}\mid\mathcal O_N]
\E[\Psi_{s',a}\Psi_{t',b}\mid\mathcal O_N]\)
contains one factor per dyad in \(E_{s,\nu}\cup E_{t,\mu}\) and one
per dyad in \(E_{s',\nu'}\cup E_{t',\mu'}\), counting common dyads
twice. Both probability products have total exponent at least \(p-1\).
Their incidence multiplicity at each node is at most \(16\),
because each of the four pentad supports contributes at most
four dyads incident to that node.
Equation~\eqref{eq:weighted_fixed_label_bound}, the uniform bound
\(\norm{\bs c_{s,\nu}(\bs\beta_0,\mathcal O_N)}\le C_0\)
from \eqref{eq:poly-coefficient-derivative-bound}, and the finite number
of monomials therefore give
\[
\E\!\left[\left|
\operatorname{Cov}(\Psi_{s,a}\Psi_{t,b},
\Psi_{s',a}\Psi_{t',b}\mid\mathcal O_N)\right|\right]
\le C\rho_N^{p-1}.
\]
For each \(5\le p\le15\), there are at most
\(\binom{N}{p}p^{20}=O(N^p)\) quadruples \((s,t,s',t')\) with
\(p\) distinct nodes, because their set of nodes can be chosen in
\(\binom{N}{p}\) ways and, for each chosen set, each of the twenty
positions in the four pentads has at most \(p\) choices of node.
For a real matrix \(\bf A=(A_{ab})\), the Cauchy--Schwarz inequality gives
\(\norm{\bf A}^2\le\sum_{a,b}A_{ab}^2\).
Taking expectations in \eqref{eq:signed_sum_entry_variance},
using \(\abs{w_{st}w_{s't'}}\le C\) and
\(\abs{\mathcal S_N}^{-4}=O(N^{-20})\), and summing over the fixed
number \(q^2\) of matrix entries yields
\begin{equation}
\label{eq:signed_sum_concentration}
\E\norm{\widehat{\bf\Omega}_N^\circ-
\E[\widehat{\bf\Omega}_N^\circ\mid\mathcal O_N]}^2
\le C\sum_{p=5}^{15}N^{p-20}\rho_N^{p-1}.
\end{equation}
Note that \[
 a_N^4\sum_{p=5}^{15}N^{p-20}\rho_N^{p-1}
 =\begin{cases}
 N^{-1}\sum_{p=5}^{15}(N\rho_N)^{p-15}\le C/N\to 0,
 &N\rho_N\to\infty,\\
 N^{-1}\sum_{p=5}^{15}(N\rho_N)^{p-1}\le C/N\to 0,
 &\rho_N\asymp N^{-1},\\
 N^{-1}\sum_{p=5}^{15}(N\rho_N)^{p-9}
 \le C/(N^5\rho_N^4)\to 0,&N\rho_N\to0.
 \end{cases}
\]
Chebyshev's inequality applied to
\eqref{eq:signed_sum_concentration} thus gives
\begin{equation}
\label{eq:signed_sum_stochastic_orders}
\norm{\widehat{\bf\Omega}_N^\circ-
\E[\widehat{\bf\Omega}_N^\circ\mid\mathcal O_N]}
=o_p(a_N^{-2}).
\end{equation}

Now we control \(-\frac{\varsigma_N}{\abs{\mathcal S_N}}\bs h_N\bs h_N^\top\) in
\eqref{eq:universal_four_term_decomposition}. Only pentads \(t\)
sharing a node with a fixed \(s\) can contribute to \(\varsigma_N\).
For each fixed \(s\), there are \(O(N^4)\) such pentads \(t\).
Since \(\abs{w_{st}}\le C\), we have
\(\abs{\varsigma_N}/\abs{\mathcal S_N}=O(N^{-1})\).
To bound \(\E\norm{\bs h_N}^2\), we now consider arbitrary pairs (\(s,t\))
with \(v:=\abs{V_s\cap V_t}\ge1\).
Fix monomial indices \(\nu,\mu\in\{1,\ldots,\nu_{\max}\}\).
By Lemma~\ref{lem:det-row-supports}, the supports \(E_{s,\nu}\) and
\(E_{t,\mu}\) form connected graphs on \(V_s\) and \(V_t\),
respectively. Since \(v\ge1\), their union is connected
on \(10-v\) nodes and has at least \(9-v\) dyads.
Conditional dyadic independence and
\eqref{eq:weighted_fixed_label_bound} therefore give
\begin{equation}
\label{eq:cov_pair_support_bound}
\E\!\left[\prod_{e\in E_{s,\nu}\cup E_{t,\mu}}L_e\right]
=\E\!\left[\prod_{e\in E_{s,\nu}\cup E_{t,\mu}}P_e(\bs\beta_0)\right]
\le C\rho_N^{\abs{E_{s,\nu}\cup E_{t,\mu}}}
\le C\rho_N^{9-v}.
\end{equation}
We obtain
\begin{equation}
\label{eq:h_second_moment_orders}
\E\norm{\bs h_N}^2
=\frac{1}{\abs{\mathcal S_N}^2}
\sum_{\substack{s,t\in\mathcal S_N:\ E(s)\cap E(t)\ne\varnothing}}
\E[\bs\Psi_t^\top\bs\Psi_s]
\le C\sum_{v=2}^5N^{-v}\rho_N^{9-v}=O(a_N^{-2}),
\end{equation}
where the first equality follows by iterated expectations from
\(\E[\bs\Psi_t^\top\bs\Psi_s\mid\mathcal O_N]=0\) whenever
\(E(s)\cap E(t)=\varnothing\).
For the inequality, we take traces and expectations in
\eqref{eq:cov_pair_monomial_expansion} and use the  bound
\(\norm{\bs c_{s,\nu}(\bs\beta_0,\mathcal O_N)}\le C_0\)
from \eqref{eq:poly-coefficient-derivative-bound},
together with \eqref{eq:cov_pair_support_bound} and
\eqref{eq:cov_pentad_pair_count}.
The final order follows from \eqref{eq:cov_pair_scale_bound}.
By Markov's inequality in \eqref{eq:h_second_moment_orders} and
\(\norm{\bs h_N\bs h_N^\top}=\norm{\bs h_N}^2\),
\begin{equation}
\label{eq:sample_centering_orders}
\norm{\frac{\varsigma_N}{\abs{\mathcal S_N}}\bs h_N\bs h_N^\top}
=O_p(N^{-1}a_N^{-2})=o_p(a_N^{-2}).
\end{equation}

To bound \((10/N)\sum_{G\in\mathcal G}\widehat{\bf\Omega}_N(G)\) in
\eqref{eq:universal_four_term_decomposition},
fix \(G\in\mathcal G\). The matrix \([\zeta_G(s,t)]_{s,t\in\mathcal S_N}\)
is symmetric, has nonnegative entries, and has a constant row sum by
symmetry across node labels. The similar calculation in
\eqref{eq:signed_centering_identity}, with \(w_{st}\) replaced by
\(\zeta_G(s,t)\), applies to the moments at every \(\bs\beta\in\mathcal B\)
and gives
\begin{equation}
\label{eq:graph_count_centering_identity}
\widehat{\bf\Omega}_N(G;\bs\beta)
=\frac{1}{\abs{\mathcal S_N}^2}\sum_{s,t\in\mathcal S_N}
\zeta_G(s,t)\left\{
\bs\Psi_s(\bs\beta)\bs\Psi_t(\bs\beta)^\top
-\bs h_N(\bs\beta)\bs h_N(\bs\beta)^\top\right\}.
\end{equation}
The Gram representation \eqref{eq:graph_count_gram} implies
\(\widehat{\bf\Omega}_N(G;\bs\beta)\succeq0\).
The subtracted matrix in \eqref{eq:graph_count_centering_identity} is
also positive semidefinite because \(\sum_{s,t}\zeta_G(s,t)\ge0\).
Taking traces in \eqref{eq:graph_count_centering_identity}, dropping
the nonnegative trace of the subtracted matrix, and using
\(\E[\bs\Psi_t^\top\bs\Psi_s\mid\mathcal O_N]=0\) whenever
\(E(s)\cap E(t)=\varnothing\), together with
\eqref{eq:cov_pentad_pair_count}, \eqref{eq:cov_pair_scale_bound},
and \eqref{eq:cov_pair_support_bound} gives
\begin{equation}
\label{eq:graph_count_trace_bound}
\begin{aligned}
\E\operatorname{tr}\{\widehat{\bf\Omega}_N(G;\bs\beta_0)\}
&\le\frac{1}{\abs{\mathcal S_N}^2}\sum_{s,t\in\mathcal S_N}
\zeta_G(s,t)\left|\E[\bs\Psi_t^\top\bs\Psi_s]\right|\le C\sum_{v=2}^5N^{-v}\rho_N^{9-v}=O(a_N^{-2}).
\end{aligned}
\end{equation}
Since the operator norm of a positive semidefinite matrix is at most
its trace, Markov's inequality and the five graphs in \(\mathcal G\)
yield
\begin{equation}
\label{eq:ten_over_N_adjustment}
a_N^2\norm{\frac{10}{N}\sum_{G\in\mathcal G}
\widehat{\bf\Omega}_N(G;\bs\beta_0)}
\le\frac{10}{N}\sum_{G\in\mathcal G}
a_N^2\operatorname{tr}\{\widehat{\bf\Omega}_N(G;\bs\beta_0)\}
=O_p(N^{-1})=o_p(1).
\end{equation}
Substituting \eqref{eq:signed_sum_conditional_mean_orders},
\eqref{eq:signed_sum_stochastic_orders},
\eqref{eq:sample_centering_orders}, and
\eqref{eq:ten_over_N_adjustment} into
\eqref{eq:universal_four_term_decomposition} and applying the triangle
inequality gives
\(a_N^2\norm{\bOmegaUniv(\bs\beta_0)
-\Var\{\bs h_N(\bs\beta_0)\mid\mathcal O_N\}}=o_p(1)\),
which is the asserted approximation in each regime.
\end{proof}

\begin{lem}
\label{lem:cov_consistency}
Under Assumptions~\ref{ass:dyad_ind}, \ref{ass:bounded_design},
\ref{ass:sparse_envelope}, and \ref{ass:regimeSigma},
the following hold: \textnormal{(i)} if \(N\rho_N\to\infty\), then
\(N^2\rho_N^{-7}\bOmegaUniv(\bs\beta_0)
\xrightarrow{P}\bf\Sigma_{\mathrm{D}}\); \textnormal{(ii)} if
\(\rho_N\asymp N^{-1}\), then
\(N^9\bOmegaUniv(\bs\beta_0)
\xrightarrow{P}\bf\Sigma_{\mathrm{S}}\); and \textnormal{(iii)} if
\(N\rho_N\to0\) and \(N^5\rho_N^4\to\infty\), then
\(N^5\rho_N^{-4}\bOmegaUniv(\bs\beta_0)
\xrightarrow{P}\bf\Sigma_{\mathrm{US}}\).
\end{lem}

\begin{proof}[Proof of Lemma~\ref{lem:cov_consistency}]
Let \((a_N,\bf\Sigma)\) correspond to the regime in
Theorem~\ref{thm:gmm}, and let
\(\mathrm r\in\{\mathrm D,\mathrm S,\mathrm{US}\}\) denote that regime.
The exact decomposition and conditional orthogonality in
Lemma~\ref{lem:dyadic_hoeffding_orthogonality} imply that
\(\bs\Pi_N^{\mathrm r}\) and
\(\bs h_N(\bs\beta_0)-\bs\Pi_N^{\mathrm r}\) are conditionally
centered and orthogonal given \(\mathcal O_N\). Hence
\begin{equation}
\label{eq:cov_target_orthogonal_decomposition}
\Var\{\bs h_N(\bs\beta_0)\mid\mathcal O_N\}
=\Var\{\bs\Pi_N^{\mathrm r}\mid\mathcal O_N\}
+\Var\{\bs h_N(\bs\beta_0)-\bs\Pi_N^{\mathrm r}\mid\mathcal O_N\}.
\end{equation}
The trace bound for positive semidefinite matrices and the conditional
centering of \(\bs h_N(\bs\beta_0)-\bs\Pi_N^{\mathrm r}\) give
\[
\norm{\Var\{\bs h_N(\bs\beta_0)-\bs\Pi_N^{\mathrm r}\mid\mathcal O_N\}}
\le\E[\norm{\bs h_N(\bs\beta_0)-\bs\Pi_N^{\mathrm r}}^2\mid\mathcal O_N].
\]
Equation~\eqref{eq:regime_remainder_second_moments} gives
\(a_N^2\E\norm{\bs h_N(\bs\beta_0)-\bs\Pi_N^{\mathrm r}}^2
\allowbreak=o(1)\).
For every \(\varepsilon>0\), Markov's inequality therefore gives
\[
\Pr\!\left\{a_N^2
\E[\norm{\bs h_N(\bs\beta_0)-\bs\Pi_N^{\mathrm r}}^2
\mid\mathcal O_N]>\varepsilon\right\}
\le\frac{a_N^2}{\varepsilon}
\E\norm{\bs h_N(\bs\beta_0)-\bs\Pi_N^{\mathrm r}}^2=o(1).
\]
Substituting into \eqref{eq:cov_target_orthogonal_decomposition} and applying
Assumption~\ref{ass:regimeSigma} yields
\begin{equation}
\label{eq:cov_target_three_regimes}
a_N^2\Var\{\bs h_N(\bs\beta_0)\mid\mathcal O_N\}
=a_N^2\Var\{\bs\Pi_N^{\mathrm r}\mid\mathcal O_N\}+o_p(1)
\xrightarrow{P}\bf\Sigma.
\end{equation}
Lemma~\ref{lem:graph_count_covariance_approximation} and
\eqref{eq:cov_target_three_regimes} now give
\(a_N^2\bOmegaUniv(\bs\beta_0)
=a_N^2\Var\{\bs h_N(\bs\beta_0)\mid\mathcal O_N\}+o_p(1)
\xrightarrow{P}\bf\Sigma\).
\end{proof}

We next record the covariance equicontinuity step used to replace
\(\bs\beta_0\) by \(\widehat{\bs{\beta}}_{\mathrm{GMM}}\) in the
tractable covariance estimator.

\begin{lem}
\label{lem:cov_equicont}
Suppose Assumptions~\ref{ass:dyad_ind}, \ref{ass:bounded_design}, and
\ref{ass:sparse_envelope} hold. Let \(\bOmegaUniv(\bs\beta)\) be defined
as in \eqref{eq:universalOmega}, and suppose \(N^5\rho_N^4\to\infty\).
Let \(a_N\) be defined as in Theorem~\ref{thm:gmm} under each of the three regimes.
If a random sequence \(\widetilde{\bs\beta}_N\in\mathcal B\)
satisfies
\(a_N\rho_N^4\norm{\widetilde{\bs\beta}_N-\bs\beta_0}=O_p(1)\),
then \(a_N^2\norm{\bOmegaUniv(\widetilde{\bs\beta}_N)
-\bOmegaUniv(\bs\beta_0)}\xrightarrow{P}0\).
\end{lem}

\begin{proof}
We first treat each
\(\widehat{\bf\Omega}_N(G;\bs\beta)\), \(G\in\mathcal G\), separately.
Fix \(G\). All sums over \(H\) range over its labeled copies, and
all sums over \(s,t\) range over \(\mathcal S_N\).
For every \(\bs\beta\in\mathcal B\), subtracting the
representations \eqref{eq:graph_count_gram} at \(\bs\beta\) and
\(\bs\beta_0\) and expanding gives
\begin{equation}
\label{eq:cov_equicont_gram_difference_expansion}
\begin{aligned}
&\widehat{\bf\Omega}_N(G;\bs\beta)
-\widehat{\bf\Omega}_N(G;\bs\beta_0)\\
&=\frac{1}{\abs{\mathcal S_N}^2}\sum_H\Biggl[
\left(\sum_{s:\,H\subseteq(V_s,E(s))}
\widehat{\bs\Psi}_s(\bs\beta_0)\right)
\left(\sum_{t:\,H\subseteq(V_t,E(t))}
\{\widehat{\bs\Psi}_t(\bs\beta)
-\widehat{\bs\Psi}_t(\bs\beta_0)\}\right)^\top\\
&\qquad+
\left(\sum_{s:\,H\subseteq(V_s,E(s))}
\{\widehat{\bs\Psi}_s(\bs\beta)
-\widehat{\bs\Psi}_s(\bs\beta_0)\}\right)
\left(\sum_{t:\,H\subseteq(V_t,E(t))}
\widehat{\bs\Psi}_t(\bs\beta_0)\right)^\top\\
&\qquad+
\left(\sum_{s:\,H\subseteq(V_s,E(s))}
\{\widehat{\bs\Psi}_s(\bs\beta)
-\widehat{\bs\Psi}_s(\bs\beta_0)\}\right)
\left(\sum_{t:\,H\subseteq(V_t,E(t))}
\{\widehat{\bs\Psi}_t(\bs\beta)
-\widehat{\bs\Psi}_t(\bs\beta_0)\}\right)^\top
\Biggr].
\end{aligned}
\end{equation}
For the first term in \eqref{eq:cov_equicont_gram_difference_expansion},
the triangle inequality and the Cauchy--Schwarz inequality give
\begin{equation}
\label{eq:cov_equicont_cross_product_bound}
\begin{aligned}
&\left\|
\frac{1}{\abs{\mathcal S_N}^2}\sum_H
\left(\sum_{s:\,H\subseteq(V_s,E(s))}
\widehat{\bs\Psi}_s(\bs\beta_0)\right)
\left(\sum_{t:\,H\subseteq(V_t,E(t))}
\{\widehat{\bs\Psi}_t(\bs\beta)
-\widehat{\bs\Psi}_t(\bs\beta_0)\}\right)^\top
\right\|\\
&\quad\le\frac{1}{\abs{\mathcal S_N}^2}\sum_H
\left\|\sum_{s:\,H\subseteq(V_s,E(s))}
\widehat{\bs\Psi}_s(\bs\beta_0)\right\|
\left\|\sum_{t:\,H\subseteq(V_t,E(t))}
\{\widehat{\bs\Psi}_t(\bs\beta)
-\widehat{\bs\Psi}_t(\bs\beta_0)\}\right\|\\
&\quad\le\frac{1}{\abs{\mathcal S_N}^2}
\left[\sum_H
\left\|\sum_{s:\,H\subseteq(V_s,E(s))}
\widehat{\bs\Psi}_s(\bs\beta_0)\right\|^2\right]^{1/2}
\left[\sum_H
\left\|\sum_{t:\,H\subseteq(V_t,E(t))}
\{\widehat{\bs\Psi}_t(\bs\beta)
-\widehat{\bs\Psi}_t(\bs\beta_0)\}\right\|^2\right]^{1/2}\\
&\quad=
\big[\operatorname{tr}\{\widehat{\bf\Omega}_N(G;\bs\beta_0)\}\big]^{1/2}
\left[
\frac{1}{\abs{\mathcal S_N}^2}\sum_H
\left\|\sum_{t:\,H\subseteq(V_t,E(t))}
\{\widehat{\bs\Psi}_t(\bs\beta)
-\widehat{\bs\Psi}_t(\bs\beta_0)\}\right\|^2
\right]^{1/2},
\end{aligned}
\end{equation}
where the equality follows from \eqref{eq:graph_count_gram}.
The second term in \eqref{eq:cov_equicont_gram_difference_expansion}
is the transpose of the first, so it has the same operator norm.
The third term is positive semidefinite, so its operator norm is
at most its trace:
\[
\begin{aligned}
&\left\|
\frac{1}{\abs{\mathcal S_N}^2}\sum_H
\left(\sum_{s:\,H\subseteq(V_s,E(s))}
\{\widehat{\bs\Psi}_s(\bs\beta)
-\widehat{\bs\Psi}_s(\bs\beta_0)\}\right)
\left(\sum_{t:\,H\subseteq(V_t,E(t))}
\{\widehat{\bs\Psi}_t(\bs\beta)
-\widehat{\bs\Psi}_t(\bs\beta_0)\}\right)^\top
\right\|\\
&\quad\le
\frac{1}{\abs{\mathcal S_N}^2}\sum_H
\left\|\sum_{s:\,H\subseteq(V_s,E(s))}
\{\widehat{\bs\Psi}_s(\bs\beta)
-\widehat{\bs\Psi}_s(\bs\beta_0)\}\right\|^2.
\end{aligned}
\]
Applying the triangle inequality to the three terms in
\eqref{eq:cov_equicont_gram_difference_expansion} and using
\eqref{eq:cov_equicont_cross_product_bound} for the first two gives
\begin{equation}
\label{eq:cov_equicont_gram_difference_bound}
\begin{aligned}
&\norm{\widehat{\bf\Omega}_N(G;\bs\beta)
-\widehat{\bf\Omega}_N(G;\bs\beta_0)}\\
&\quad\le
2\big[\operatorname{tr}\{\widehat{\bf\Omega}_N(G;\bs\beta_0)\}\big]^{1/2}
\left[
\frac{1}{\abs{\mathcal S_N}^2}\sum_H
\left\|\sum_{s:\,H\subseteq(V_s,E(s))}
\{\widehat{\bs\Psi}_s(\bs\beta)
-\widehat{\bs\Psi}_s(\bs\beta_0)\}\right\|^2
\right]^{1/2}\\
&\qquad+
\frac{1}{\abs{\mathcal S_N}^2}\sum_H
\left\|\sum_{s:\,H\subseteq(V_s,E(s))}
\{\widehat{\bs\Psi}_s(\bs\beta)
-\widehat{\bs\Psi}_s(\bs\beta_0)\}\right\|^2.
\end{aligned}
\end{equation}
Also,
\(a_N^2\operatorname{tr}\{\widehat{\bf\Omega}_N(G;\bs\beta_0)\}=O_p(1)\)
by \eqref{eq:graph_count_trace_bound} and Markov's inequality.
For the sum of squared norms in
\eqref{eq:cov_equicont_gram_difference_bound}, we have
\begin{equation}
\label{eq:cov_equicont_centering_identity}
\begin{aligned}
&\frac{1}{\abs{\mathcal S_N}^2}\sum_H
\left\|\sum_{s:\,H\subseteq(V_s,E(s))}
\{\widehat{\bs\Psi}_s(\bs\beta)-\widehat{\bs\Psi}_s(\bs\beta_0)\}\right\|^2\\
&\quad=\frac{1}{\abs{\mathcal S_N}^2}\sum_{s,t}\zeta_G(s,t)
\{\bs\Psi_s(\bs\beta)-\bs\Psi_s(\bs\beta_0)\}^\top
\{\bs\Psi_t(\bs\beta)-\bs\Psi_t(\bs\beta_0)\}\\
&\qquad-
\frac{\sum_{s,t}\zeta_G(s,t)}{\abs{\mathcal S_N}^2}
\norm{\bs h_N(\bs\beta)-\bs h_N(\bs\beta_0)}^2\\
&\quad\le\frac{1}{\abs{\mathcal S_N}^2}\sum_{s,t}\zeta_G(s,t)
\norm{\bs\Psi_s(\bs\beta)-\bs\Psi_s(\bs\beta_0)}
\norm{\bs\Psi_t(\bs\beta)-\bs\Psi_t(\bs\beta_0)},
\end{aligned}
\end{equation}
where the equality follows by applying
\eqref{eq:graph_count_centering_identity} to
\(\bs\Psi_s(\bs\beta)-\bs\Psi_s(\bs\beta_0)\), whose sample mean
is \(\bs h_N(\bs\beta)-\bs h_N(\bs\beta_0)\), and taking traces.
Since \(\zeta_G(s,t)\ge0\), the term subtracted in
\eqref{eq:cov_equicont_centering_identity} is nonnegative.
The inequality follows by omitting that subtraction and applying
the Cauchy--Schwarz inequality to each inner product.

To control the right-hand side of
\eqref{eq:cov_equicont_centering_identity}, we bound
\(\norm{\bs\Psi_s(\bs\beta)-\bs\Psi_s(\bs\beta_0)}\)
in terms of \(\norm{\bs\beta-\bs\beta_0}\). Fix a compact rectangle \(\mathcal R\) whose interior contains
\(\mathcal B\). The line segment between \(\bs\beta_0\) and any
\(\bs\beta\in\mathcal B\) lies in \(\mathcal R\). Applying the
bound on
\(\norm{\partial_{\bs\beta}^{\alpha}\bs c_{s,\nu}(\bs\beta,\mathcal O_N)}\)
for \(\abs\alpha=1\) in \eqref{eq:poly-coefficient-derivative-bound}
on \(\mathcal R\) to the representation
\(\bs\Psi_s(\bs\beta)=\sum_{\nu=1}^{\nu_{\max}}
\bs c_{s,\nu}(\bs\beta,\mathcal O_N)\prod_{e\in E_{s,\nu}}L_e\)
in Lemma~\ref{lem:poly} gives
\[
\norm{\bs\Psi_s(\bs\beta)-\bs\Psi_s(\bs\beta_0)}
\le C\norm{\bs\beta-\bs\beta_0}
\sum_{\nu=1}^{\nu_{\max}}\prod_{e\in E_{s,\nu}}L_e,
\qquad \bs\beta\in\mathcal B.
\]
Substituting into \eqref{eq:cov_equicont_centering_identity} and using
\(L_e^2=L_e\) gives the bound
\begin{equation}
\label{eq:cov_equicont_increment_bound}
\begin{aligned}
&\frac{1}{\abs{\mathcal S_N}^2}\sum_H
\left\|\sum_{s:\,H\subseteq(V_s,E(s))}
\{\widehat{\bs\Psi}_s(\bs\beta)-\widehat{\bs\Psi}_s(\bs\beta_0)\}\right\|^2\le\frac{C\norm{\bs\beta-\bs\beta_0}^2}{\abs{\mathcal S_N}^2}
\sum_{s,t}\zeta_G(s,t)
\sum_{\nu,\mu=1}^{\nu_{\max}}\prod_{e\in E_{s,\nu}\cup E_{t,\mu}}L_e.
\end{aligned}
\end{equation}
The sum over \((s,t,\nu,\mu)\) is nonnegative and does not depend on
\(\bs\beta\). Since \(\zeta_G(s,t)=0\) for disjoint node sets,
only pairs \((s,t)\) for which \(v:=\abs{V_s\cap V_t}\ge1\) contribute.
Equations~\eqref{eq:cov_pair_support_bound} and
\eqref{eq:cov_pentad_pair_count}, the uniformly bounded counts
\(\zeta_G(s,t)\), and the finite number of monomials give
\[
\begin{aligned}
\E\!\left[\frac{1}{\abs{\mathcal S_N}^2}\sum_{s,t}\zeta_G(s,t)
\sum_{\nu,\mu=1}^{\nu_{\max}}
\prod_{e\in E_{s,\nu}\cup E_{t,\mu}}L_e\right]
&\le C\sum_{v=1}^5N^{-v}\rho_N^{9-v}=O(a_N^{-2}+N^{-1}\rho_N^8),
\end{aligned}
\]
where the equality uses \eqref{eq:cov_pair_scale_bound} for
\(2\le v\le5\), retaining the \(v=1\) term explicitly.
Markov's inequality consequently gives
\begin{equation}
\label{eq:cov_equicont_pair_sum_orders}
\frac{1}{\abs{\mathcal S_N}^2}\sum_{s,t}\zeta_G(s,t)
\sum_{\nu,\mu=1}^{\nu_{\max}}\prod_{e\in E_{s,\nu}\cup E_{t,\mu}}L_e
=O_p(a_N^{-2}+N^{-1}\rho_N^8).
\end{equation}
The bound \eqref{eq:cov_equicont_increment_bound} holds simultaneously for
every \(\bs\beta\in\mathcal B\), and the sum in
\eqref{eq:cov_equicont_pair_sum_orders} does not depend on \(\bs\beta\).
We may therefore substitute
\(\bs\beta=\widetilde{\bs\beta}_N\). Moreover, by the definition of \(a_N\), we have 
\(a_N\rho_N^4\to\infty\). Using
\(a_N\rho_N^4\norm{\widetilde{\bs\beta}_N-\bs\beta_0}=O_p(1)\) gives
\begin{equation}
\label{eq:cov_equicont_gram_increment}
\begin{aligned}
&\frac{a_N^2}{\abs{\mathcal S_N}^2}\sum_H
\left\|\sum_{s:\,H\subseteq(V_s,E(s))}
\{\widehat{\bs\Psi}_s(\widetilde{\bs\beta}_N)
-\widehat{\bs\Psi}_s(\bs\beta_0)\}\right\|^2\\
&\quad\le
a_N^2\norm{\widetilde{\bs\beta}_N-\bs\beta_0}^2
O_p(a_N^{-2}+N^{-1}\rho_N^8)
=O_p\!\left((a_N\rho_N^4)^{-2}+N^{-1}\right)=o_p(1).
\end{aligned}
\end{equation}
Multiplying \eqref{eq:cov_equicont_gram_difference_bound} by \(a_N^2\)
and using \eqref{eq:cov_equicont_gram_increment} together with
\(a_N^2\operatorname{tr}\{\widehat{\bf\Omega}_N(G;\bs\beta_0)\}=O_p(1)\)
gives
\(a_N^2\norm{\widehat{\bf\Omega}_N(G;\widetilde{\bs\beta}_N)
-\widehat{\bf\Omega}_N(G;\bs\beta_0)}
=o_p(1)\).
There are five graphs in \(\mathcal G\), and each coefficient in
\eqref{eq:universalOmega} has absolute value at most \(1+10/N\).
The triangle inequality thus gives
\[
\begin{aligned}
&a_N^2\norm{\bOmegaUniv(\widetilde{\bs\beta}_N)
-\bOmegaUniv(\bs\beta_0)}\le\left(1+\frac{10}{N}\right)
\sum_{G\in\mathcal G}a_N^2
\norm{\widehat{\bf\Omega}_N(G;\widetilde{\bs\beta}_N)
-\widehat{\bf\Omega}_N(G;\bs\beta_0)}=o_p(1).
\end{aligned}
\]
The proof is completed.
\end{proof}

Now, we can prove Theorem~\ref{thm:tractable_gmm}.

\begin{proof}[Proof of Theorem~\ref{thm:tractable_gmm}]
Lemma~\ref{lem:cov_consistency} gives
\(a_N^2\bOmegaUniv(\bs\beta_0)\xrightarrow{P}\bf\Sigma\).
Together with Theorem~\ref{thm:regimeCLT} and Slutsky's theorem, this proves
\(\bs{c}^\top\bs h_N(\bs\beta_0)
/\sqrt{\bs{c}^\top\bOmegaUniv(\bs\beta_0)\bs{c}}
\xrightarrow{D}N(0,1)\).
The rate \eqref{eq:gmm_rate_for_cov_equicont} in the proof of
Theorem~\ref{thm:gmm} allows us to apply
Lemma~\ref{lem:cov_equicont} with
\(\widetilde{\bs\beta}_N=\widehat{\bs\beta}_{\mathrm{GMM}}\), which gives
\(a_N^2\norm{\bOmegaUniv(\widehat{\bs\beta}_{\mathrm{GMM}})
-\bOmegaUniv(\bs\beta_0)}
\xrightarrow{P}0\).
Hence \(a_N^2\bOmegaUniv(\widehat{\bs\beta}_{\mathrm{GMM}})
\xrightarrow{P}\bf\Sigma\).
Corollary~\ref{cor:generic_gmm_inference} then gives
\(a_N^2\rho_N^8\widehat{\mathbf{V}}_N\xrightarrow{P}\mathbf{V}_0\) and
\(\bs c^\top(\widehat{\bs\beta}_{\mathrm{GMM}}-\bs\beta_0)
/\sqrt{\bs c^\top\widehat{\mathbf{V}}_N\bs c}
\xrightarrow{D}N(0,1)\).
\end{proof}

\bibliographystyle{ecta}
\bibliography{sample}

@unpublished{dano2023transition,
  title={Transition Probabilities and Moment Restrictions in Dynamic Fixed Effects Logit Models},
  author={Dano, Kevin},
  note={Working paper, revised November 5, 2025},
  url={https://kevindano.github.io/assets/files/JMP_KevinDano.pdf},
  year={2025}
}

@article{honore2024moment,
  title={Moment Conditions for Dynamic Panel Logit Models with Fixed Effects},
  author={Honor{\'e}, Bo E and Weidner, Martin},
  journal={The Review of Economic Studies},
  volume={92},
  number={5},
  pages={3112--3137},
  year={2025},
  doi={10.1093/restud/rdae097},
  publisher={Oxford University Press}
}

@article{bonhomme2023functional,
  title={Functional Differencing in Networks},
  author={Bonhomme, St{\'e}phane and Dano, Kevin},
  journal={Revue {\'e}conomique},
  volume={75},
  number={1},
  pages={147--175},
  doi={10.3917/e.reco.751.0147},
  year={2024}
}

@techreport{bonhomme2025feedback,
  title={Moment Restrictions for Nonlinear Panel Data Models with Feedback},
  author={Bonhomme, St{\'e}phane and Dano, Kevin and Graham, Bryan S.},
  year={2025},
  institution={National Bureau of Economic Research},
  type={Working Paper},
  number={33966},
  doi={10.3386/w33966},
  url={https://kevindano.github.io/assets/files/BDG_feedback.pdf},
  note={Revised September 1, 2026}
}

@article{gao2020robust,
  title={Identification of Semiparametric Panel Multinomial Choice Models with Infinite-Dimensional Fixed Effects},
  author={Gao, Wayne Yuan and Li, Ming},
  journal={The Review of Economics and Statistics},
  note={Advance online publication, February 4, 2026},
  doi={10.1162/rest.a.1708},
  year={2026}
}

@article{gao2023logical,
  title={Logical Differencing in Dyadic Network Formation Models with Nontransferable Utilities},
  author={Gao, Wayne Yuan and Li, Ming and Xu, Sheng},
  journal={Journal of Econometrics},
  volume={235},
  number={1},
  pages={302--324},
  year={2023},
  doi={10.1016/j.jeconom.2022.03.008},
  publisher={Elsevier}
}

@article{bonhomme2012functional,
  title={Functional Differencing},
  author={Bonhomme, St{\'e}phane},
  journal={Econometrica},
  volume={80},
  number={4},
  pages={1337--1385},
  year={2012},
  publisher={Wiley Online Library}
}

@unpublished{toth2017semiparametric,
  title={Semiparametric Estimation in Network Formation Models with Homophily and Degree Heterogeneity},
  author={Toth, Peter},
  note={Available at SSRN 2988698},
  year={2017}
}

@incollection{honore2019panel,
  title={Panel Vector Autoregressions with Binary Data},
  author={Honor{\'e}, Bo E and Kyriazidou, Ekaterini},
  booktitle={Panel Data Econometrics: Theory},
  pages={197--223},
  year={2019},
  publisher={Elsevier}
}

@article{honore2021identification,
  title={Identification in Simple Binary Outcome Panel Data Models},
  author={Honor{\'e}, Bo E and De Paula, {\'A}ureo},
  journal={The Econometrics Journal},
  volume={24},
  number={2},
  pages={C78--C93},
  year={2021},
  publisher={Oxford University Press}
}

@unpublished{dobronyi2021identification,
  title={Identification of Dynamic Panel Logit Models with Fixed Effects},
  author={Dobronyi, Christopher and Gu, Jiaying and Kim, Kyoo il and Russell, Thomas M.},
  note={{arXiv} preprint arXiv:2104.04590v4},
  url={https://arxiv.org/abs/2104.04590v4},
  year={2026}
}

@article{chatterjee2011random,
  title={Random Graphs with a Given Degree Sequence},
  author={Chatterjee, Sourav and Diaconis, Persi and Sly, Allan},
  journal={The Annals of Applied Probability},
  volume={21},
  number={4},
  pages={1400--1435},
  year={2011},
  doi={10.1214/10-AAP728},
  publisher={Institute of Mathematical Statistics}
}

@article{yan2013central,
  title={A Central Limit Theorem in the {$\beta$}-Model for Undirected Random Graphs with a Diverging Number of Vertices},
  author={Yan, Ting and Xu, Jinfeng},
  journal={Biometrika},
  volume={100},
  number={2},
  pages={519--524},
  year={2013},
  doi={10.1093/biomet/ass084},
  publisher={Oxford University Press}
}

@unpublished{zeleneev2020identification,
  title={Identification and Estimation of Network Models with Nonparametric Unobserved Heterogeneity},
  author={Zeleneev, Andrei},
  note={Working paper, revised February 5, 2026},
  url={https://www.homepages.ucl.ac.uk/~uctpzel/network_ID.pdf},
  year={2026}
}

@unpublished{candelaria2024semiparametric,
  title={A Semiparametric Network Formation Model with Unobserved Linear Heterogeneity},
  author={Candelaria, Luis E},
  note={Working paper, January 2024},
  url={https://lecandelaria.github.io/papers/SemNet_2024.pdf},
  year={2024}
}

@article{charbonneau2017multiple,
  title={Multiple Fixed Effects in Binary Response Panel Data Models},
  author={Charbonneau, Karyne B},
  journal={The Econometrics Journal},
  volume={20},
  number={3},
  pages={S1--S13},
  year={2017},
  publisher={Oxford University Press}
}

@article{graham2017econometric,
  title={An Econometric Model of Network Formation with Degree Heterogeneity},
  author={Graham, Bryan S},
  journal={Econometrica},
  volume={85},
  number={4},
  pages={1033--1063},
  year={2017},
  doi={10.3982/ECTA12679},
  publisher={Wiley}
}

@misc{mauldin2020peer,
  title={Peer Relationships among Master of Social Work Students: Social Network Data},
  author={Mauldin, Rebecca L.},
  year={2021},
  doi={10.18738/T8/5ELWBV},
  note={Dataset. Texas Data Repository Dataverse, Social Networks for Social Good Dataverse}
}

@article{manski1987semiparametric,
  title={Semiparametric Analysis of Random Effects Linear Models from Binary Panel Data},
  author={Manski, Charles F},
  journal={Econometrica},
  volume={55},
  number={2},
  pages={357--362},
  year={1987},
  publisher={JSTOR}
}

@article{neyman1948consistent,
  title={Consistent Estimates Based on Partially Consistent Observations},
  author={Neyman, Jerzy and Scott, Elizabeth L},
  journal={Econometrica},
  volume={16},
  number={1},
  pages={1--32},
  year={1948},
  publisher={JSTOR}
}

@article{hughes2022estimating,
  title={A Jackknife Bias Correction for Nonlinear Network Data Models with Fixed Effects},
  author={Hughes, David W},
  journal={Journal of Econometrics},
  volume={253},
  pages={106130},
  doi={10.1016/j.jeconom.2025.106130},
  year={2026}
}

@unpublished{yan2026penalized,
  title={Penalized Likelihood for Dyadic Network Formation Models with Degree Heterogeneity},
  author={Yan, Zizhong and Li, Jingrong and Zhang, Yi},
  note={{arXiv} preprint arXiv:2605.00771},
  year={2026}
}

@article{fernandez2016individual,
  title={Individual and Time Effects in Nonlinear Panel Models with Large {N}, {T}},
  author={Fern{\'a}ndez-Val, Iv{\'a}n and Weidner, Martin},
  journal={Journal of Econometrics},
  volume={192},
  number={1},
  pages={291--312},
  year={2016},
  publisher={Elsevier}
}

@unpublished{li2024estimation,
  title={Bagging the Network},
  author={Li, Ming and Shi, Zhentao and Zheng, Yapeng},
  note={{arXiv} preprint arXiv:2410.23852},
  doi={10.48550/arXiv.2410.23852},
  year={2026}
}

@article{dzemski2019empirical,
  title={An Empirical Model of Dyadic Link Formation in a Network with Unobserved Heterogeneity},
  author={Dzemski, Andreas},
  journal={The Review of Economics and Statistics},
  volume={101},
  number={5},
  pages={763--776},
  year={2019},
  publisher={MIT Press}
}

@article{jochmans2018semiparametric,
  title={Semiparametric Analysis of Network Formation},
  author={Jochmans, Koen},
  journal={Journal of Business \& Economic Statistics},
  volume={36},
  number={4},
  pages={705--713},
  year={2018},
  publisher={Taylor \& Francis}
}

@article{de2020econometric,
  title={Econometric Models of Network Formation},
  author={De Paula, {\'A}ureo},
  journal={Annual Review of Economics},
  volume={12},
  number={1},
  pages={775--799},
  year={2020},
  publisher={Annual Reviews}
}

@incollection{graham2020network,
  title={Network Data},
  author={Graham, Bryan S},
  booktitle={Handbook of Econometrics},
  volume={7A},
  pages={111--218},
  year={2020},
  doi={10.1016/bs.hoe.2020.05.001},
  publisher={Elsevier}
}

@article{kitazawa2022transformations,
  title={Transformations and Moment Conditions for Dynamic Fixed Effects Logit Models},
  author={Kitazawa, Yoshitsugu},
  journal={Journal of Econometrics},
  volume={229},
  number={2},
  pages={350--362},
  year={2022},
  publisher={Elsevier}
}

@article{miyauchi2016structural,
  title={Structural Estimation of Pairwise Stable Networks with Nonnegative Externality},
  author={Miyauchi, Yuhei},
  journal={Journal of Econometrics},
  volume={195},
  number={2},
  pages={224--235},
  year={2016},
  publisher={Elsevier}
}

@article{de2018identifying,
  title={Identifying Preferences in Networks with Bounded Degree},
  author={De Paula, {\'A}ureo and Richards-Shubik, Seth and Tamer, Elie},
  journal={Econometrica},
  volume={86},
  number={1},
  pages={263--288},
  year={2018},
  publisher={Wiley Online Library}
}

@article{mele2017structural,
  title={A Structural Model of Dense Network Formation},
  author={Mele, Angelo},
  journal={Econometrica},
  volume={85},
  number={3},
  pages={825--850},
  year={2017},
  publisher={Wiley Online Library}
}

@article{menzel2024strategic,
  title={Strategic Network Formation with Many Agents},
  author={Menzel, Konrad},
  journal={Journal of Econometrics},
  volume={253},
  pages={106174},
  doi={10.1016/j.jeconom.2025.106174},
  year={2026}
}

@article{sheng2020structural,
  title={A Structural Econometric Analysis of Network Formation Games through Subnetworks},
  author={Sheng, Shuyang},
  journal={Econometrica},
  volume={88},
  number={5},
  pages={1829--1858},
  year={2020},
  publisher={Wiley Online Library}
}

@article{jackson1996strategic,
  title={A Strategic Model of Social and Economic Networks},
  author={Jackson, Matthew O and Wolinsky, Asher},
  journal={Journal of Economic Theory},
  volume={71},
  number={1},
  pages={44--74},
  year={1996},
  doi={10.1006/jeth.1996.0108},
  publisher={Elsevier}
}

@article{bloch2007formation,
  title={The Formation of Networks with Transfers among Players},
  author={Bloch, Francis and Jackson, Matthew O},
  journal={Journal of Economic Theory},
  volume={133},
  number={1},
  pages={83--110},
  year={2007},
  doi={10.1016/j.jet.2005.10.003},
  publisher={Elsevier}
}

@article{chamberlain1980analysis,
  title={Analysis of Covariance with Qualitative Data},
  author={Chamberlain, Gary},
  journal={The Review of Economic Studies},
  volume={47},
  number={1},
  pages={225--238},
  year={1980},
  publisher={Wiley-Blackwell}
}

@article{chamberlain1987asymptotic,
  title={Asymptotic Efficiency in Estimation with Conditional Moment Restrictions},
  author={Chamberlain, Gary},
  journal={Journal of Econometrics},
  volume={34},
  number={3},
  pages={305--334},
  year={1987},
  doi={10.1016/0304-4076(87)90015-7}
}

@article{hansen1982large,
  title={Large Sample Properties of Generalized Method of Moments Estimators},
  author={Hansen, Lars Peter},
  journal={Econometrica},
  volume={50},
  number={4},
  pages={1029--1054},
  year={1982},
  doi={10.2307/1912775}
}

@incollection{newey1994large,
  title={Large Sample Estimation and Hypothesis Testing},
  author={Newey, Whitney K. and McFadden, Daniel},
  booktitle={Handbook of Econometrics},
  editor={Engle, Robert F. and McFadden, Daniel L.},
  volume={4},
  chapter={36},
  pages={2111--2245},
  publisher={Elsevier},
  year={1994},
  doi={10.1016/S1573-4412(05)80005-4}
}

@article{honore2000panel,
  title={Panel Data Discrete Choice Models with Lagged Dependent Variables},
  author={Honor{\'e}, Bo E and Kyriazidou, Ekaterini},
  journal={Econometrica},
  volume={68},
  number={4},
  pages={839--874},
  year={2000},
  publisher={Wiley Online Library}
}

@book{myerson1991game,
  title={Game Theory: Analysis of Conflict},
  author={Myerson, Roger B},
  year={1991},
  publisher={Harvard University Press}
}

@article{gao2020nonparametric,
  title={Nonparametric Identification in Index Models of Link Formation},
  author={Gao, Wayne Yuan},
  journal={Journal of Econometrics},
  volume={215},
  number={2},
  pages={399--413},
  year={2020},
  publisher={Elsevier}
}

@article{hahn2004jackknife,
  title={Jackknife and Analytical Bias Reduction for Nonlinear Panel Models},
  author={Hahn, Jinyong and Newey, Whitney},
  journal={Econometrica},
  volume={72},
  number={4},
  pages={1295--1319},
  year={2004},
  publisher={Wiley Online Library}
}

@article{dhaene2015split,
  title={Split-Panel Jackknife Estimation of Fixed-Effect Models},
  author={Dhaene, Geert and Jochmans, Koen},
  journal={The Review of Economic Studies},
  volume={82},
  number={3},
  pages={991--1030},
  year={2015},
  doi={10.1093/restud/rdv007},
  publisher={Oxford University Press}
}

@article{davezies2023fixed,
  title={Fixed-Effects Binary Choice Models with Three or More Periods},
  author={Davezies, Laurent and D'Haultf{\oe}uille, Xavier and Mugnier, Martin},
  journal={Quantitative Economics},
  volume={14},
  number={3},
  pages={1105--1132},
  year={2023},
  doi={10.3982/QE1991},
  publisher={Wiley Online Library}
}

@article{hoeffding1948class,
  title={A Class of Statistics with Asymptotically Normal Distribution},
  author={Hoeffding, Wassily},
  journal={The Annals of Mathematical Statistics},
  volume={19},
  number={3},
  pages={293--325},
  year={1948},
  doi={10.1214/aoms/1177730196}
}

@book{serfling1980approximation,
  title={Approximation Theorems of Mathematical Statistics},
  author={Serfling, Robert J.},
  publisher={John Wiley \& Sons},
  address={New York},
  year={1980},
  doi={10.1002/9780470316481}
}

@article{rota1964foundations,
  title={On the Foundations of Combinatorial Theory {I}. Theory of {M{\"o}bius} Functions},
  author={Rota, Gian-Carlo},
  journal={Zeitschrift f{\"u}r Wahrscheinlichkeitstheorie und Verwandte Gebiete},
  volume={2},
  pages={340--368},
  year={1964},
  doi={10.1007/BF00531932}
}

@article{dejong1990central,
  title={A Central Limit Theorem for Generalized Multilinear Forms},
  author={de Jong, Peter},
  journal={Journal of Multivariate Analysis},
  volume={34},
  number={2},
  pages={275--289},
  year={1990}
}

@article{dobler2017quantitative,
  title={Quantitative {de Jong} Theorems in Any Dimension},
  author={D{\"o}bler, Christian and Peccati, Giovanni},
  journal={Electronic Journal of Probability},
  volume={22},
  number={2},
  pages={1--35},
  year={2017},
  doi={10.1214/16-EJP19}
}

@book{vandervaart1998asymptotic,
  title={Asymptotic Statistics},
  author={van der Vaart, Aad W.},
  publisher={Cambridge University Press},
  address={Cambridge},
  year={1998}
}

@book{durrett2019probability,
  title={Probability: Theory and Examples},
  author={Durrett, Rick},
  edition={5th},
  series={Cambridge Series in Statistical and Probabilistic Mathematics},
  number={49},
  publisher={Cambridge University Press},
  address={Cambridge},
  year={2019},
  doi={10.1017/9781108591034}
}

@article{graham2024sparse,
  title={Sparse Network Asymptotics for Logistic Regression under Possible Misspecification},
  author={Graham, Bryan S.},
  journal={Econometrica},
  volume={92},
  number={6},
  pages={1837--1868},
  year={2024},
  doi={10.3982/ECTA19051},
  publisher={Wiley Online Library}
}

@article{rucinski1988when,
  title={When Are Small Subgraphs of a Random Graph Normally Distributed?},
  author={Ruci{\'n}ski, Andrzej},
  journal={Probability Theory and Related Fields},
  volume={78},
  number={1},
  pages={1--10},
  year={1988},
  doi={10.1007/BF00718031},
  publisher={Springer}
}

@article{chandrasekhar2025network,
  title={A Network Formation Model Based on Subgraphs},
  author={Chandrasekhar, Arun G. and Jackson, Matthew O.},
  journal={The Review of Economic Studies},
  volume={92},
  number={6},
  pages={3741--3787},
  year={2025},
  doi={10.1093/restud/rdaf013},
  publisher={Oxford University Press}
}

@unpublished{gao2026tractable,
  title={Tractable Identification of Strategic Network Formation Models with Unobserved Heterogeneity},
  author={Gao, Wayne Yuan and Li, Ming and Xu, Zhengyan},
  note={{arXiv} preprint arXiv:2603.08634},
  year={2026}
}

@article{fuertes2000factor,
  title={Factor Structure and Short Form of the {Miville--Guzman Universality--Diversity Scale}},
  author={Fuertes, Jairo N. and Miville, Marie L. and Mohr, Jonathan J. and Sedlacek, William E. and Gretchen, Denise},
  journal={Measurement and Evaluation in Counseling and Development},
  volume={33},
  number={3},
  pages={157--169},
  year={2000},
  doi={10.1080/07481756.2000.12069007}
}

@article{kottke2011additional,
  title={Additional Evidence for the Short Form of the {Universality--Diversity Scale}},
  author={Kottke, Janet L.},
  journal={Personality and Individual Differences},
  volume={50},
  number={4},
  pages={464--469},
  year={2011},
  doi={10.1016/j.paid.2010.11.008}
}

@article{mityagin2020zero,
  title={The Zero Set of a Real Analytic Function},
  author={Mityagin, Boris S.},
  journal={Mathematical Notes},
  volume={107},
  number={3-4},
  pages={529--530},
  year={2020},
  doi={10.1134/S0001434620030189}
}

@article{karadja2017richer,
  title={Richer (and Holier) than Thou? The Effect of Relative Income Improvements on Demand for Redistribution},
  author={Karadja, Mounir and Mollerstrom, Johanna and Seim, David},
  journal={The Review of Economics and Statistics},
  volume={99},
  number={2},
  pages={201--212},
  year={2017},
  doi={10.1162/REST_a_00623}
}

@book{rudin1976principles,
  title={Principles of Mathematical Analysis},
  author={Rudin, Walter},
  edition={3rd},
  publisher={McGraw-Hill},
  address={New York},
  year={1976}
}

@article{cameron2011robust,
  title={Robust Inference with Multiway Clustering},
  author={Cameron, A. Colin and Gelbach, Jonah B. and Miller, Douglas L.},
  journal={Journal of Business \& Economic Statistics},
  volume={29},
  number={2},
  pages={238--249},
  year={2011},
  doi={10.1198/jbes.2010.07136}
}

@article{politis2011higher,
  title={Higher-Order Accurate, Positive Semidefinite Estimation of Large-Sample Covariance and Spectral Density Matrices},
  author={Politis, Dimitris N.},
  journal={Econometric Theory},
  volume={27},
  number={4},
  pages={703--744},
  year={2011},
  doi={10.1017/S0266466610000484}
}

\end{document}